\documentclass[12pt,final]{article}
\usepackage{amsmath}
\usepackage{graphicx}
\usepackage{enumerate}
\usepackage{natbib}
\setcitestyle{authoryear,open={(},close={)}} 
\usepackage{color}
\usepackage{url} 

\usepackage{setspace,graphicx,amsmath,amsfonts,amssymb,amsthm}
\graphicspath{ {./images/} }
\usepackage{bbm}
\usepackage{subfigure}

\usepackage{tikz}
\usetikzlibrary{decorations.pathreplacing,positioning, arrows.meta}
\usepackage{tikz-qtree}
\usepackage{morefloats}     
\usepackage{subfigure}     
\usepackage{ctable}
\usepackage{colortbl}
\usepackage{longtable}
\usepackage{rotating}
\usepackage[english]{babel}
\usepackage{amsmath} 
\usepackage{amssymb}           
\usepackage[all,color]{xy}
\usepackage{ragged2e}
\usepackage{ifpdf}
\usepackage{color}
\usepackage{bbm}
\usepackage{parskip}
\usepackage{marvosym}
\usepackage{afterpage}
\usepackage{booktabs}
\usepackage{multicol}
\usepackage{booktabs}
\usepackage{xparse}
\usepackage[ruled,vlined]{algorithm2e}
\SetKwInput{kwInit}{Initialize}
\usepackage{soul}

\definecolor{NewBlue}{RGB}{001,031,091}
\definecolor{NewRed}{RGB}{153,0,0}

\usepackage[plainpages=false,breaklinks=true,bookmarks=true,bookmarksopen=true,%
                                       bookmarksopenlevel=1,pdfstartview=FitH,pdfview=FitH]{hyperref}
\hypersetup{
	pdfborder = {0 0 0},
    colorlinks,
    citecolor=NewRed,
    filecolor=NewRed,
    linkcolor=NewRed,
    urlcolor=black
}

\def\Diag{\mathsf{Diag}}
\def\trace{\mathsf{tr}}  

\def\iter{\mathsf{iter}}

\usepackage[lmargin=1in, rmargin=1in, tmargin=1in, bmargin=1in,paperwidth=8.5in, paperheight=11in]{geometry}

\definecolor{darkross}{rgb}{0.008,0.412,0.471}
\definecolor{middleross}{rgb}{0.012,0.580,0.663}
\definecolor{lightross}{rgb}{0.016,0.749,0.855}
\definecolor{darkblue}{rgb}{0.067,0.008,0.471}
\definecolor{middleblue}{rgb}{0.094,0.012,0.663}
\definecolor{lightblue}{rgb}{0.122,0.016,0.855}
\definecolor{darkpurple}{rgb}{0.471,0.008,0.412}
\definecolor{middlepurple}{rgb}{0.663,0.012,0.580}
\definecolor{lightpurple}{rgb}{0.855,0.016,0.749}
\definecolor{darkbrown}{rgb}{0.471,0.067,0.008}
\definecolor{middlebrown}{rgb}{0.663,0.094,0.012}
\definecolor{lightbrown}{rgb}{0.855,0.122,0.016}
\definecolor{darkolive}{rgb}{0.412,0.471,0.008}
\definecolor{middleolive}{rgb}{0.580,0.663,0.012}
\definecolor{lightolive}{rgb}{0.749,0.855,0.016}
\definecolor{darkgreen}{rgb}{0.008,0.417,0.067}
\definecolor{middlegreen}{rgb}{0.012,0.663,0.094}
\definecolor{lightgreen}{rgb}{0.016,0.855,0.122}
\definecolor{darkocre}{rgb}{0.471,0.298,0.008}
\definecolor{middleocre}{rgb}{0.663,0.420,0.012}
\definecolor{lightocre}{rgb}{0.855,0.541,0.016}
\def\bbeta{\mbox{\boldmath $\beta$}}

\def\bmu{\mbox{\boldmath $\mu$}}

\def\bOmega{\mbox{\boldmath $\Omega$}}
\def\bxi{\mbox{\boldmath $\xi$}}

\def\bpsi{\mbox{\boldmath $\psi$}}

\def\bnu{\mbox{\boldmath $\nu$}}

\def\bTheta{\mathbf{\Theta}}
\def\bSigma{\mathbf{\Sigma}}

\def\bh{\mathbf{h}}

\def\bL{\mathbf{L}}

\def\by{\mathbf{y}} 

\def\0{\mbox{\bf{0}}}

\def\bI{\mathbf{I}}

\def\bC{\mathbf{C}}

\def\bX{\mathbf{X}} 
\def\bx{\mathbf{x}}

\def\bA{\mathbf{A}}

\def\bK{\mathbf{K}}
\def\bQ{\mathbf{Q}}

\def\bB{\mathbf{B}}

\def\bV{\mathbf{V}} 

\def\trace{\mbox{tr}}

\def\Cov{\mbox{Cov}}

\newtheorem{proposition}{Proposition}[section]

\newtheorem{result}{Result}

\newtheorem{prop}{Proposition}[subsection]

\begin{document}
\def\spacingset#1{\renewcommand{\baselinestretch}%
{#1}\small\normalsize} \spacingset{1}

\newcommand{\tit}{\vspace{2em} Variational inference for large Bayesian vector autoregressions}

\newcommand{\abs}{\small We propose a novel variational Bayes approach to estimate high-dimensional vector autoregression (VAR) models with hierarchical shrinkage priors. Our approach does not rely on a conventional structural VAR representation of the parameter space for posterior inference. Instead, we elicit hierarchical shrinkage priors directly on the matrix of regression coefficients so that (1) the prior structure directly maps into posterior inference on the reduced-form transition matrix, and (2) posterior estimates are more robust to variables permutation. An extensive simulation study provides evidence that our approach compares favourably against existing linear and non-linear Markov Chain Monte Carlo and variational Bayes methods. We investigate both the statistical and economic value of the forecasts from our variational inference approach within the context of a mean-variance investor allocating her wealth in a large set of different industry portfolios. The results show that more accurate estimates translate into substantial statistical and economic out-of-sample gains. The results hold across different hierarchical shrinkage priors and model dimensions.   
\vspace{0.1in}

\textbf{Keywords:} Bayesian methods, variational inference, hierarchical shrinkage prior, high-dimensional models, vector autoregressions, industry returns predictability.

\textbf{JEL codes:} C11, C32, C55, C53, G11 
}

\title{\vspace{-2em}{\bf \tit}\thanks{\footnotesize We are thankful to Andrea Carriero and seminar participants at the 2021 Virtual NBER-NSF SBIES, the 2021 European Summer Meeting of the Econometric Society, the 2nd Workshop on Dimensionality Reduction and Inference in High-Dimensional Time Series at Maastricht University, and the 2023 Summer Forum Workshop on Macroeconomics and Policy Evaluation at the Barcelona School of Economics for their helpful comments and suggestions. This research has been partially funded by the BERN\_BIRD2222\_01 - BIRD 2022 grant of the University of Padua. A previous version of this paper was circulating with the title ``Variational Bayes inference for large-scale multivariate predictive regressions''.}}

\author{\setcounter{footnote}{1}Mauro Bernardi\thanks{Department of Statistical Sciences, University of Padova, Italy. Email: \texttt{mauro.bernardi@unipd.it}}\and \setcounter{footnote}{2}Daniele Bianchi\thanks{School of Economics and Finance, Queen Mary University of London, United Kingdom. Email: \texttt{d.bianchi@qmul.ac.uk} Web: \texttt{whitesphd.com}} \and \setcounter{footnote}{3}Nicolas Bianco\thanks{Department of Economics and Business, Universitat Pompeu Fabra, Barcelona, Spain. Email: \texttt{nicolas.bianco@upf.edu} Web: \href{https://whitenoise8.github.io/}{\texttt{whitenoise8.github.io}} }}

 \date{\hspace{2em}}

\maketitle
\thispagestyle{empty}

\centerline{\bf Abstract}
\medskip
\abs
\normalsize

\doublespacing

\clearpage
\pagenumbering{arabic}


\section{Introduction}
\label{sec:intro}

\textcolor{black}{Hierarchical shrinkage priors have been shown to represent an effective regularization technique when estimating large vector autoregression (VAR) models.} The use of these priors often relies on a Cholesky decomposition of the residuals covariance matrix so that a large system of equations is reduced to a sequence of univariate regressions. \textcolor{black}{This allows for more efficient computations as priors can be elicited on the structural VAR representation implied by the Cholesky factorization and posterior inference is carried out equation-by-equation.} 

Such a conventional approach has two important implications for posterior inference: first, priors are not order-invariant, meaning that posterior inference is sensitive to permutations of the endogenous variables for a given prior specification. This is particularly relevant in high dimensions whereby logical orders of the endogenous variables might be unclear or a full search among all possible ordering combinations might be unfeasible (see, e.g., \citealp{chan2021large}). \textcolor{black}{Second, imposing a shrinkage prior on the structural VAR formulation does not necessarily help to pin down the significance of cross-correlations in the reduced-form VAR formulation. This is especially relevant in forecasting applications whereby the main objective is to accurately identify predictive relationships across variables, rather than to identify structural shocks.} 

{\color{black}In this paper, we take a different approach towards posterior inference with hierarchical shrinkage priors in large VAR models. Specifically, we propose a novel variational Bayes estimation approach which allows for fast and accurate estimates of the reduced-form regression coefficients without leveraging on a structural VAR representation.} This allows us to elicit hierarchical shrinkage priors directly on the matrix of regression coefficients so that (1) the prior structure directly maps into the posterior inference of the reduced-form transition matrix, and (2) posterior estimates are more robust to variables permutation. \textcolor{black}{We also account for the effect of ``exogenous'' covariates and stochastic volatility in the residuals.} 

\textcolor{black}{The key feature of our approach is that by abstracting from the linearity constraints implied by a structural VAR formulation, one can provide a more direct identification of the reduced-form regression parameters.} This could have important implications for forecasting within the context of weak predictability whereby the transition matrix and/or the coefficients on exogenous predictors are potentially sparse in nature (see, e.g., \citealp{bianchi2023dynamic}). The main advantage of our variational inference approach is that an accurate identification of the regression parameters does not translate into a higher computational cost compared to existing Bayesian estimation methods. This is particularly relevant in practice for recursive forecasting implementations with higher frequency data, such as portfolio returns.   


We investigate the accuracy of the posterior estimates based on an extensive simulation study for different model dimensions and variables permutation. As benchmarks, we consider a variety of established estimation approaches developed for large Bayesian VAR models, such as the linearized MCMC  proposed by \citet{chan2018bayesian,cross2020macroeconomic} and its variational Bayes counterpart proposed by \citet{chan_yu2020,gefang2023forecasting}. Both approaches are built upon a structural VAR formulation. \textcolor{black}{In addition, we compare our variational Bayes method against the MCMC approach developed by \citet{gruber2022forecasting}, which is not constrained by a Cholesky factorization for parameters identification, similar to our approach.} We test each estimation method for different hierarchical priors, such as the adaptive-Lasso of \citet{Leng.2014}, an adaptive version of the Normal-Gamma of \citet{griffin_brown.2010}, and the Horseshoe of \cite{carvalho_etal.2010}. 

Overall, the simulation results show that our variational inference approach represents the best trade-off between estimation accuracy and computational efficiency. Specifically, posterior inference from our variational Bayes method is as accurate as non-linear MCMC methods (see, e.g., \citealp{gruber2022forecasting}) but is considerably more efficient. At the same time, our approach is as efficient as conventional MCMC and variational Bayes methods based on a structural VAR formulation, but is considerably more accurate and less sensitive to variables permutation.

\textcolor{black}{Our approach towards posterior inference in large VARs is guided by the principle that a more accurate identification of the reduced-form transition matrix should ultimately lead to better out-of-sample forecasts and financial decision making.} To test this assumption, we investigate both the statistical and economic value of the forecasts from our variational Bayes approach within the context of a mean-variance investor who allocates her wealth between an industry portfolio and a risk-free asset based on lagged cross-industry returns and a series of macroeconomic predictors. 

Although the model is general and can be applied to any type of financial returns, as far as data are stationary, our focus on different industry portfolios is motivated by a keen interest from researchers (see, e.g., \citealp{fama1997industry,hou2006industry}) and practitioners alike. Indeed, the implications of industry returns predictability are arguably far from trivial. If all industries are unpredictable, then the market return, which is a weighted average of the industry portfolios, should also be unpredictable. As a result, the abundant evidence of aggregate market return predictability (see, e.g., \citealp{rapach2013forecasting}), implies that at least some industry portfolio return is predictable. 

The main results show that our variational inference approach fares better than competing methods in terms of out-of-sample point and density forecasts. We show that more accurate forecasts translate into larger economic gains as measured by certainty equivalent returns spreads vis-\'{a}-vis a naive investor which take investment decisions based on sample estimates of the conditional mean and variance of the returns. This holds across different hierarchical prior specifications. Overall, the empirical results support our view that by a more accurate identification of weak correlations between predictors and portfolio returns, one can significantly improve -- both statistically and economically -- the out-of-sample performance of large-scale multivariate time-series models. 

Our paper connects to a growing literature exploring the use of Bayesian methods to estimate high-dimensional VAR models with shrinkage priors. A non-exhaustive list of works on the topic contains \citet*{chan2018bayesian,carriero2019large,huber2019adaptive,chan_yu2020,cross2020macroeconomic,kastner2020sparse,chan2021large,chan2021minnesota,carriero2022corrigendum,gruber2022forecasting,gefang2023forecasting}, among others. \textcolor{black}{We contribute to this literature by providing a fast and accurate variational Bayes method which generalize posterior inference of quantities of interest by abstracting from a conventional structural VAR representation.} 

A second strand of literature we contribute to is related to the predictability of stock returns. More specifically, we contribute to the ongoing struggle to understand the dynamics of risk premiums by looking at industry-based portfolios. As highlighted by \citet{lewellen2010skeptical}, the time series variation of industry portfolios is particularly problematic to measure, since conventional risk factors do not seem to capture significant comovements and cross-signals which might improve out-of-sample predictability. Early exceptions are \citet{ferson1991variation,ferson1995arbitrage,ferson1999conditioning} and \citet{Avramov:2004}. We extend this literature by investigating the out-of-sample predictability of industry portfolios through the lens of a novel estimation method for large Bayesian VAR models.

\section{Choosing the model parametrization}
\label{sec:bayesian_model} 
Let $\mathbf{y}_t=\left(y_{1,t},\dots,y_{d,t}\right)^\intercal\in\mathbb{R}^d$ be a multivariate normal random variable and denote by $\mathbf{x}_t=\left(1,x_{1,t},\dots,x_{p,t}\right)^\intercal\in\mathbb{R}^{(p+1)}$ a vector of covariates at time $t$. {\color{black}A vector autoregressive model with exogenous covariates and stochastic volatility is defined in compact form as:}
\begin{equation}\label{eq:var1_def}
    \mathbf{y}_t= \boldsymbol{\Theta} \mathbf{z}_{t-1} +\mathbf{u}_t,\qquad \mathbf{u}_t \sim \mathsf{N}_d\left(\mathbf{0}_d,{\color{black}\boldsymbol{\Omega}_t^{-1}}\right), \qquad t=1,\ldots,T,
\end{equation}
with $\mathbf{z}_{t-1} = (\mathbf{y}_{t-1}^\intercal,\mathbf{x}_{t-1}^\intercal)^\intercal$ and $\boldsymbol{\Theta} = (\boldsymbol{\Phi},\boldsymbol{\Gamma})$ consistently partitioned, where $\boldsymbol{\Phi}\in\mathbb{R}^{d\times d}$ is the transition matrix containing the autoregression coefficients and $\boldsymbol{\Gamma}\in\mathbb{R}^{d\times (p+1)}$ is the matrix of regression parameters for the exogenous predictors. Here, $\mathbf{u}_t\in\mathbb{R}^d$ is a sequence of uncorrelated innovation terms such that $\mathbf{u}_{t-k}\perp \mathbf{u}_{t-j}$ $\forall k,j$ with $k\neq j$ and {\color{black}$\boldsymbol{\Omega}_t\in\mathbb{S}^d_{++}$ being a symmetric and positive-definite time-varying precision matrix. A modified Cholesky factorization of $\boldsymbol{\Omega}_t$ can be conveniently exploited to re-write the model in Eq.\eqref{eq:var1_def} with orthogonal innovations \citep[see, e.g.,][]{rothman_etal.2010}. 

Let $\boldsymbol{\Omega}_t = \mathbf{L}^\intercal\mathbf{V}_t\mathbf{L}$, where $\mathbf{L}\in\mathbb{R}^{d\times d}$ is unit-lower-triangular and $\mathbf{V}_t\in\mathbb{S}_{++}^d$ is diagonal with time-varying elements $\mathbf{V}_t=\text{Diag}(\nu_{1,t},\ldots,\nu_{d,t})$ \citep[see, e.g.,][]{huber2019adaptive,gefang2023forecasting}.} By multiplying both sides of Eq.\eqref{eq:var1_def} by $\mathbf{L}=\mathbf{I}_d - \mathbf{B}$ one can obtain two alternative re-parametrizations of the same model:
\begin{subequations}\label{eq:var1_orth}
\begin{align}
\mathbf{y}_t &= \mathbf{B}(\mathbf{y}_t-\boldsymbol{\Theta} \mathbf{z}_{t-1}) + \boldsymbol{\Theta} \mathbf{z}_{t-1}+\boldsymbol{\varepsilon}_t,\qquad &&\boldsymbol{\varepsilon}_t \sim \mathsf{N}_d(\mathbf{0}_d,\mathbf{V}_t^{-1}), \label{eq:var1_orth_def}\\
\mathbf{y}_t &= \mathbf{B}\mathbf{y}_t+\mathbf{A}\mathbf{z}_{t-1} +\boldsymbol{\varepsilon}_t,\qquad &&\boldsymbol{\varepsilon}_t \sim \mathsf{N}_d(\mathbf{0}_d,\mathbf{V}_t^{-1}),\label{eq:var1_orth_def2}
\end{align}
\end{subequations}
where $\mathbf{A}=\mathbf{L}\boldsymbol{\Theta}$ and $\mathbf{B}$ has a strict-lower-triangular structure with elements $\beta_{j,k}=-l_{j,k}$ for $j=2,\ldots,d$ and $k=1,\ldots,j-1$. The key difference is that Eq.\eqref{eq:var1_orth_def} is non-linear in the parameters, while Eq.\eqref{eq:var1_orth_def2} is linear. More importantly, Eq.\eqref{eq:var1_orth_def2} is known as structural VAR representation, widely used in existing MCMC and variational Bayes estimations methods for high-dimensional VAR models (see, e.g., \citealp{chan2018bayesian,chan_yu2020,gefang2023forecasting}). Instead, Eq.\eqref{eq:var1_orth_def} is the reduced-form parametrization at the core of our variational inference approach. This has also been used within the context of MCMC for smaller dimensions (see, e.g., \citealp{huber2019adaptive,gruber2022forecasting}). 

From Eq.\eqref{eq:var1_orth} one can obtain an equation-by-equation representation in which the $j$-th component of $\mathbf{y}_t$ becomes:
\begin{subequations}\label{eq:final_process}
\begin{align}
y_{j,t} &= \boldsymbol{\beta}_j \mathbf{r}_{j,t} + \boldsymbol{\vartheta}_j \mathbf{z}_{t-1} + \varepsilon_{j,t}, \quad &&\varepsilon_{j,t} \sim \mathsf{N}(0,\nu_{j,t}^{-1}), \label{eq:final_process_univ}\\
y_{j,t} &= \boldsymbol{\beta}_j \mathbf{y}_t^j + \mathbf{a}_j \mathbf{z}_{t-1} + \varepsilon_{j,t}, \quad &&\varepsilon_{j,t} \sim \mathsf{N}(0,\nu_{j,t}^{-1}),\label{eq:final_process_univ2}
\end{align}
\end{subequations}
for all $j=1,\ldots,d$ and $t=1,\ldots,T$, where $\boldsymbol{\beta}_j\in\mathbb{R}^{j-1}$ is a row vector containing the non-null elements in the $j$-th row of $\mathbf{B}$, $\boldsymbol{\vartheta}_j$ and $\mathbf{a}_j$ denote the $j$-th row of $\boldsymbol{\Theta}$ and $\mathbf{A}$, respectively. For any $j=1,\dots,d$, let $\mathbf{r}_{j,t}=\mathbf{y}_t^j - \boldsymbol{\Theta}^j \mathbf{z}_{t-1}$ denotes the the vector of residuals up to the $(j-1)$-th regression, with $\mathbf{y}_t^j = (y_{1,t},\ldots,y_{j-1,t})^\intercal\in\mathbb{R}^{j-1}$ being the sub-vector of $\mathbf{y}_t$ collecting the variables up to the $(j-1)$-th and $\boldsymbol{\Theta}^j\in\mathbb{R}^{(j-1)\times d}$ is the sub-matrix containing the first $j-1$ rows of $\boldsymbol{\Theta}$. {\color{black}We follow \cite{gefang2023forecasting,chan_yu2020} and model the time variation in $\nu_{j,t}^{-1} = \exp\left(h_{j,t}\right)$ assuming a log-volatility process $h_{j,t}=h_{j,t-1}+e_{j,t}$ with $e_{j,t}\sim\mathsf{N}(0,\psi_j)$, where the initial state $h_{0,j}\sim\mathsf{N}(0,k_0\,\psi_j)$, $k_0\gg 0$, is unknown.}

\paragraph{A discussion on variables permutation.} \textcolor{black}{Existing Bayesian approaches for large VAR models often rely on the structural representation in Eq.\eqref{eq:var1_orth_def2}, and therefore consider the elements in $\mathbf{A}$ as the parameters of interest.} This has the key merit of simplifying the implementation of MCMC (see, e.g., \citealp{chan2018bayesian}) and variational Bayes algorithms (see, e.g., \citealp{gefang2023forecasting}). Under the re-parametrization $\mathbf{A}=\mathbf{L}\boldsymbol{\Theta}$, each element $\vartheta_{i,j}$ -- which denotes the $(i,j)$-entry of $\mathbf{\Theta}$ -- is a linear combination $\vartheta_{i,j}=a_{i,j}+\sum_{k=1}^{i-1}c_{i,k}a_{k,j}$, where $a_{i,j}$ and $c_{i,j}$ are the $(i,j)$-entry of $\mathbf{A}$ and $\mathbf{L}^{-1}$, respectively. 

This raises two main issues: first, $a_{i,j}=0$ does not imply $\vartheta_{i,j}=0$, that is a shrinkage prior on $\mathbf{A}$ does not preserve the structure of $\boldsymbol{\Theta}$. Second, the estimate $\widehat{\boldsymbol{\Theta}}=\widehat{\mathbf{L}}^{-1}\widehat{\mathbf{A}}$ for a given prior is potentially highly sensitive to variables permutation due to its dependence on the Cholesky factorization (see \citealp{gruber2022forecasting} for a related discussion). Figure \ref{fig:fig1} provides a visual representation of this argument by comparing the estimates obtained based on Eq.\eqref{eq:var1_orth_def} vs Eq.\eqref{eq:var1_orth_def2}, for two different permutations of $\mathbf{y}_t$. 

\begin{figure}[h]
\centering
\includegraphics[width=.75\textwidth]{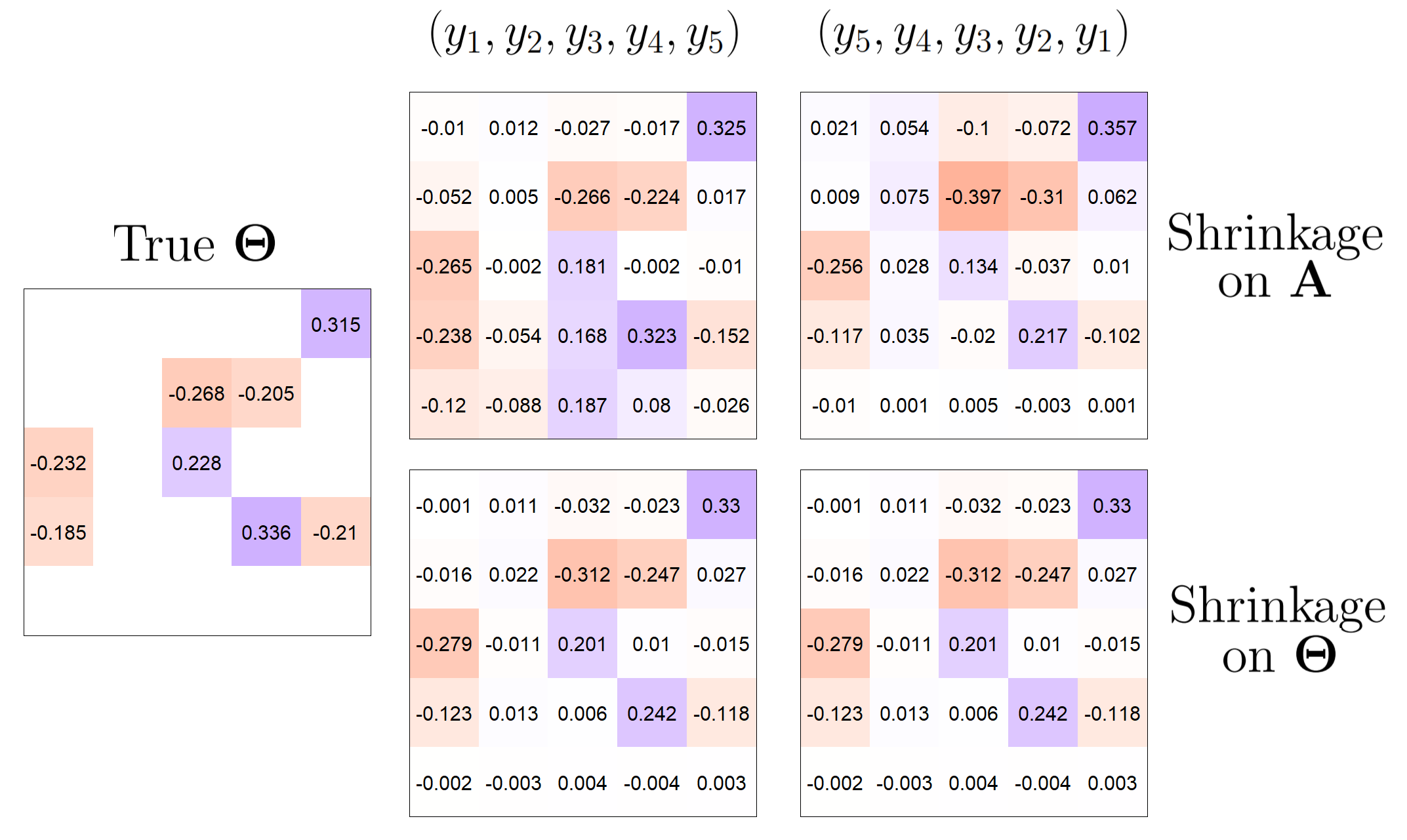}
\caption{\small Comparison between the posterior inference for the linear representation $\mathbf{A}=\mathbf{L}\boldsymbol{\Theta}$ (first row) and the original parametrization $\boldsymbol{\Theta}$ (second row), for two different permutations of $\mathbf{y}_t$.}\label{fig:fig1}
\end{figure}

The evidence confirms that the estimates based on the transformation $\widehat{\boldsymbol{\Theta}}=\widehat{\mathbf{L}}^{-1}\widehat{\mathbf{A}}$ clearly diverge from the true $\bTheta$. In addition, the posterior estimates are influenced by the variables permutation. \textcolor{black}{Instead, inference based on the representation in Eq.\eqref{eq:var1_orth_def} provides a more accurate identification of $\bTheta$ which is also less sensitive to variables permutation.} Before taking this intuition to task both in simulation and on actual forecasting, in the next Section we provide details of our variational Bayes inference approach. 

\section{Variational Bayes inference}
\label{sec:mfvb} 
A variational approach to Bayesian inference requires to minimize the Kullback-Leibler ($\mathit{KL}$) divergence between an approximating density $q(\boldsymbol{\xi})$ and the true posterior density $p(\boldsymbol{\xi}|\mathbf{y})$, where $\boldsymbol{\xi}$ denotes the set of parameters of interest. \cite{ormerod_wand.2010} show that minimizing the $\mathit{KL}$ divergence can be equivalently stated as the maximization of the ``effective lower bound'' (ELBO) denoted by $\underline{p}\left(\mathbf{y};q\right)$:
\begin{equation}
\label{eq:vb_elbo_optim}
q^*(\boldsymbol{\xi}) = \arg\max_{q(\boldsymbol{\xi}) \in \mathcal{Q}}\log\underline{p}\left(\mathbf{y};q\right),\quad
\underline{p}\left(\mathbf{y};q\right)=\int q(\boldsymbol{\xi}) \log\left\{\frac{p(\mathbf{y},\boldsymbol{\xi})}{q(\boldsymbol{\xi})}\right\}\,d\boldsymbol{\xi},
\end{equation}
where $q^*(\boldsymbol{\xi})\in\mathcal{Q}$ represents the optimal variational density and $\mathcal{Q}$ is a space of density functions. \textcolor{black}{Depending on the assumption on $\mathcal{Q}$, one falls into different variational paradigms. For instance, given a partition of the parameters vector $\boldsymbol{\xi}=\{ \boldsymbol{\xi}_1,\dots,\boldsymbol{\xi}_p\}$, a mean-field variational Bayes (MFVB) approach assumes a factorization of the form $q(\boldsymbol{\xi})=\prod_{j=1}^p q_i(\boldsymbol{\xi}_j)$.} A closed form expression for each optimal variational density $q^\ast(\boldsymbol{\xi}_j)$ can be defined as:
\begin{equation}
    q^\ast(\boldsymbol{\xi}_j) \propto \exp\left\{\mathbb{E}_{q^\star(\boldsymbol{\xi}\setminus\boldsymbol{\xi}_j)}\Big[\log p(\mathbf{y},\boldsymbol{\xi})\Big] \right\},\qquad
    q^\star(\boldsymbol{\xi}\setminus\boldsymbol{\xi}_j)=\prod_{\substack{i=1\\i\neq j}}^p q_i(\boldsymbol{\xi}_i),\label{eq:vb}
\end{equation}
where the expectation is taken with respect to the joint approximating density with the $j$-th element of the partition removed $q^\star(\boldsymbol{\xi}\setminus\boldsymbol{\xi}_j)$. This allows to implement an efficient iterative algorithm to estimate the optimal density $q^*(\boldsymbol{\xi})$, although some components $q^*(\boldsymbol{\xi}_j)$ may remain too complex to handle and further restrictions are needed. If we assume that $q^*(\boldsymbol{\xi}_j)$ belongs to a pre-specified parametric family of distributions, the MFVB outlined above is sometimes labelled as {\it semi-parametric} \citep[see][]{rohde2016semiparametric}. 


\subsection{Optimal variational densities}
\label{subsec:prior} 

We present a factorization of the variational density $q(\boldsymbol{\xi})$ for the model outlined in Eq.\eqref{eq:var1_orth_def}. As a benchmark, we consider a non-informative Normal prior for the regression coefficients. For each entry of $\boldsymbol{\Theta}$, let $\vartheta_{j,k} \sim \mathsf{N}(0,\upsilon)$, for $j=1,\dots,d$ and $k=1,\dots,d+p+1$. In addition, let $\psi_j \sim \mathsf{InvGa}(a_{\psi},b_{\psi})$ for $j=1,\dots,d$, and $\beta_{j,k} \sim \mathsf{N}(0,\tau)$, for $j=2,\dots,d$ and $k=1,\dots,j-1$. Here, $\mathsf{InvGa}(\cdot,\cdot)$ denotes the Inverse-Gamma distribution, and $a_\psi>0$, $b_\psi>0$, $\tau \gg 0$ and $\upsilon \gg 0$ are the related hyper-parameters. Let $\boldsymbol{\xi}=(\boldsymbol{\vartheta}^\intercal,\bh^\intercal,\bpsi^\intercal,\bbeta^\intercal)^\intercal$ be the set of parameters of interest, the corresponding variational density can be factorised as $q(\boldsymbol{\xi}) = q(\boldsymbol{\vartheta})q(\mathbf{h})q(\boldsymbol{\psi})q(\boldsymbol{\beta})$, where:
\begin{equation}\label{eq:q_non_info}
q(\boldsymbol{\vartheta})=\prod_{j=1}^d q(\boldsymbol{\vartheta}_j), \quad q(\mathbf{h})=\prod_{j=1}^d q(\mathbf{h}_j), \quad 
q(\boldsymbol{\psi})=\prod_{j=1}^d q(\psi_j), \quad q(\boldsymbol{\beta})=\prod_{j=2}^d q(\boldsymbol{\beta}_j).
\end{equation}

For the ease of exposition, in the main text of the paper we summarize the optimal variatonal density for the main parameters of interest $\boldsymbol{\Theta}$, with both a baseline non-informative prior and three alternative hierarchical shrinkage priors. \textcolor{black}{The parameters and the full derivations of the optimal variational densities $q^*(\mathbf{h}_j) \equiv \mathsf{N}_{T+1}(\boldsymbol{\mu}_{q(h_j)}, \bSigma_{q(h_j)})$, $q^*(\psi_j) \equiv \mathsf{InvGa}(a_{q(\psi_j)}, b_{q(\psi_j)})$, and $q^*(\boldsymbol{\beta}_j) \equiv \mathsf{N}_{j-1}(\boldsymbol{\mu}_{q(\beta_j)}, \boldsymbol{\Sigma}_{q(\beta_j)})$ for $j=1,\ldots,d$, are reported in Proposition \ref{prop:h_VB}, \ref{prop:psi_VB} and \ref{prop:beta_VB} of Appendix \ref{app:VBVAR}, respectively. Notice these optimal variational densities are invariant across different shrinkage prior specifications for $\boldsymbol{\Theta}$. We leave to Proposition \ref{prop:nu_VB} in Appendix \ref{app:VBVAR} also the derivations for the constant volatility case with $\nu_{j,t}=\nu_j$ and $\nu_j \sim \mathsf{Ga}(a_{\nu},b_{\nu})$ for $j=1,\dots,d$, where $\mathsf{Ga}(\cdot,\cdot)$ denotes the gamma distribution, and $a_\nu>0$, $b_\nu>0$. For the interested reader, Appendix \ref{app:VBVAR} also provides the analytical form of the lower bound for each set of parameters.}

Proposition \ref{eq:prop1} provides the optimal variational density for the $j$-th row of $\boldsymbol{\Theta}$ under the baseline Normal prior specification $\vartheta_{j,k} \sim \mathsf{N}(0,\upsilon)$. The proof and analytical derivations are available in Appendix \ref{app:non_sparseVBAR}.

\begin{proposition}\label{eq:prop1}
The optimal variational density for $\boldsymbol{\vartheta}_j$ is $q^*(\boldsymbol{\vartheta}_j) \equiv \mathsf{N}_{d+p+1}(\boldsymbol{\mu}_{q(\vartheta_j)}, \boldsymbol{\Sigma}_{q(\vartheta_j)})$ with hyper-parameters:
\begin{equation}\begin{aligned}
	\bSigma_{q(\mathbf{\vartheta}_j)} &= \left( \sum_{t=1}^T\bmu_{q(\omega_{j,j,t})}\mathbf{z}_{t-1}\mathbf{z}_{t-1}^\intercal+1/\upsilon\bI_{d+p+1}\right)^{-1},\\
	\bmu_{q(\mathbf{\vartheta}_j)} &= \bSigma_{q(\mathbf{\vartheta}_j)}\left(\sum_{t=1}^T\left(\bmu_{q(\mathbf{\omega}_{j,t})} \otimes\mathbf{z}_{t-1}\right)\mathbf{y}_t-\sum_{t=1}^T\left( \bmu_{q(\mathbf{\omega}_{j,-j,t})}\otimes\mathbf{z}_{t-1}\mathbf{z}_{t-1}^\intercal\right)\bmu_{q(\mathbf{\vartheta}_{-j})}\right),
\end{aligned}\end{equation}
where $\boldsymbol{\vartheta} = \left(\begin{array}{c}
	    \boldsymbol{\vartheta}_j \\ \boldsymbol{\vartheta}_{-j}
	\end{array}\right)$ and $\boldsymbol{\omega}_{j,t}$ denotes the $j$-th row of $	\bOmega_t = \left(\begin{array}{cc}
 	    \omega_{j,j,t} & \boldsymbol{\omega}_{j,-j,t} \\
 	    \boldsymbol{\omega}_{-j,j,t} & \bOmega_{-j,-j,t}
 	\end{array}\right).$
\end{proposition}

Notice that despite the multivariate model is reduced to a sequence of univariate regressions, the analytical form of the variational mean $\boldsymbol{\mu}_{q(\vartheta_j)}$ in Proposition \ref{eq:prop1} depends on all the other rows through $\bmu_{q(\mathbf{\vartheta}_{-j})}$. As a result, the variational estimates of $\boldsymbol\vartheta_j$ explicitly depend on all of the other $\boldsymbol\vartheta_{-j}$. This addresses the issue in the MCMC algorithm of \citet{carriero2019large}, which has been highlighted by \citet{bognanni2022comment} and corrected by \citet{carriero2022corrigendum}.


\paragraph{Bayesian adaptive-Lasso.} The Bayesian adaptive-Lasso of \cite{Leng.2014} extends the original work of \cite{Park_Casella.2008} by assuming a different shrinkage for each regression parameter based on a laplace distribution with an individual scaling parameter $\vartheta_{j,k}\vert\lambda_{j,k} \sim\mathsf{Lap}(\lambda_{j,k})$, for $j=1,\dots,d$ and $k=1,\dots,d+p+1$. The latter can be represented as a scale mixture of normals with an exponential mixing density, $\vartheta_{j,k}|\upsilon_{j,k} \sim \mathsf{N}(0,\upsilon_{j,k})$, $\upsilon_{j,k}|\lambda^2_{j,k} \sim \mathsf{Exp}(\lambda^2_{j,k}/2)$. The scaling parameters $\lambda^2_{j,k}$ are not fixed but inferred from the data by assuming a common hyper-prior distribution $\lambda^2_{j,k} \sim \mathsf{Ga}(h_1,h_2)$, where $h_1, h_2>0$. 

Let $\boldsymbol{\xi}_{\text{L}}=(\boldsymbol{\xi}^\intercal, \boldsymbol{\upsilon}^\intercal,(\boldsymbol{\lambda}^2)^\intercal))^\intercal$ be the vector $\boldsymbol{\xi}$ augmented with the adaptive-Lasso prior parameters. The distribution $q(\boldsymbol{\xi}_{\text{L}})$ can be factorised as,
\begin{align}\label{eq:q_lasso}
q(\boldsymbol{\xi}_{\text{L}})&=q(\boldsymbol{\xi})q(\boldsymbol{\upsilon},\boldsymbol{\lambda}^2), \qquad q(\boldsymbol{\upsilon},\boldsymbol{\lambda}^2)=\prod_{j=1}^d \prod_{k=1}^{d+p+1} q(\upsilon_{j,k})q(\lambda^2_{j,k}),
\end{align}
Proposition \ref{eq:prop2} provides the optimal variational density for the $j$-th row of $\boldsymbol{\Theta}$ under Bayesian adaptive-Lasso prior specification $\vartheta_{j,k}|\upsilon_{j,k} \sim \mathsf{N}(0,\upsilon_{j,k})$, $\upsilon_{j,k}|\lambda^2_{j,k} \sim \mathsf{Exp}(\lambda^2_{j,k}/2)$, and $\lambda^2_{j,k} \sim \mathsf{Ga}(h_1,h_2)$. The proof and analytical derivations are available in Appendix \ref{app:VBVAR-LASSO}.
\begin{proposition}
The optimal variational density for $\boldsymbol{\vartheta}_j$ is $q^*(\boldsymbol{\vartheta}_j) \equiv \mathsf{N}_{d+p+1}(\boldsymbol{\mu}_{q(\vartheta_j)}, \boldsymbol{\Sigma}_{q(\vartheta_j)})$ with $\bSigma_{q(\mathbf{\vartheta}_j)} = \left( \sum_{t=1}^T\bmu_{q(\omega_{j,j,t})}\mathbf{z}_{t-1}\mathbf{z}_{t-1}^\intercal+\mathrm{Diag}(\bmu_{q(1/\mathbf{\upsilon}_j)})\right)^{-1}$, where $\Diag(\bmu_{q(1/\mathbf{\upsilon}_j)})$ is a diagonal matrix with elements $\bmu_{q(1/\mathbf{\upsilon}_j)}=(\mu_{q(1/\upsilon_{j,1})},\mu_{q(1/\upsilon_{j,2})},\ldots,\mu_{q(1/\upsilon_{j,d+p+1})})$. The parameters $\boldsymbol{\mu}_{q(\vartheta_j)}$ and $\bmu_{q(\omega_{j,j,t})}$ are as in Proposition \ref{eq:prop1}. The optimal variational densities of the scaling parameters are $q^*(\lambda^2_{j,k}) \equiv \mathsf{Ga}(a_{q(\lambda^2_{j,k})}, b_{q(\lambda^2_{j,k})})$ with $a_{q(\lambda^2_{j,k})}, b_{q(\lambda^2_{j,k})}$ defined in Eq.\eqref{eq:up_lam_VBLASSO}, and $q^*(1/\upsilon_{j,k}) \equiv \mathsf{IG}(a_{q(\upsilon_{j,k})}, b_{q(\upsilon_{j,k})})$ with $a_{q(\upsilon_{j,k})}, b_{q(\upsilon_{j,k})}$ defined in Eq.\eqref{eq:up_ups_VBLASSO}.
\label{eq:prop2}
\end{proposition}
%
%
\paragraph{Adaptive Normal-Gamma.} We expand the original Normal-Gamma prior of \cite{griffin_brown.2010} by assuming that each regression coefficient has a different shrinkage parameter, similar to the adaptive-Lasso. The hierarchical specification requires that $\vartheta_{j,k}|\upsilon_{j,k} \sim \mathsf{N}(0,\upsilon_{j,k})$, and $\upsilon_{j,k}|\eta_j,\lambda_{j,k} \sim \mathsf{Ga}\left(\eta_j,\eta_j\lambda_{j,k}/2\right)$ for $j=1,\dots,d$ and $k=1,\dots,d+p+1$. Notice that by restricting $\eta_j=1$ one could obtain the adaptive-Lasso prior. Marginalization over the variance $\upsilon_{j,k}$ leads to $p(\vartheta_{j,k}\vert\eta_j,\lambda_{j,k})$ which corresponds to a Variance-Gamma distribution. The hyper-parameters $\eta_j$ and $\lambda_{j,k}$ are not fixed but are inferred from the data by assuming two common hyper-priors $\lambda_{j,k} \sim \mathsf{Ga}(h_1,h_2)$ and $\eta_j \sim \mathsf{Exp}(h_3)$, where $h_l>0$ for $l=1,2,3$. 

Let $\boldsymbol{\xi}_{\text{NG}}=(\boldsymbol{\xi}^\intercal, \boldsymbol{\upsilon}^\intercal,\boldsymbol{\lambda}^\intercal,\boldsymbol{\eta}^\intercal)^\intercal$ be the vector $\boldsymbol{\xi}$ augmented with the parameters of the adaptive Normal-Gamma prior. The joint distribution $q(\boldsymbol{\xi}_{\text{NG}})$ can be factorised as,
\begin{align}\label{eq:q_ng}
q(\boldsymbol{\xi}_{\text{NG}})=q(\boldsymbol{\xi})q(\boldsymbol{\upsilon},\boldsymbol{\lambda},\boldsymbol{\eta}), \qquad q(\boldsymbol{\upsilon},\boldsymbol{\lambda},\boldsymbol{\eta})=\prod_{j=1}^d q(\eta_j) \prod_{k=1}^{d+p+1} q(\upsilon_{j,k})q(\lambda_{j,k}).
\end{align}

Proposition \ref{eq:prop3} provides the optimal variational density for the $j$-th row of $\boldsymbol{\Theta}$ under an adaptive Normal-Gamma specification $\upsilon_{j,k}|\eta_j,\lambda_{j,k} \sim \mathsf{Ga}\left(\eta_j,\eta_j\lambda_{j,k}/2\right)$, $\lambda_{j,k} \sim \mathsf{Ga}(h_1,h_2)$ and $\eta_j \sim \mathsf{Exp}(h_3)$. The proof and analytical derivations are available in Appendix \ref{app:VBVAR-DG}.

\begin{proposition}
The optimal variational density for $\boldsymbol{\vartheta}_j$ is $q^*(\boldsymbol{\vartheta}_j) \equiv \mathsf{N}_{d+p+1}(\boldsymbol{\mu}_{q(\vartheta_j)}, \boldsymbol{\Sigma}_{q(\vartheta_j)})$ with $\bSigma_{q(\mathbf{\vartheta}_j)} = \left( \sum_{t=1}^T\bmu_{q(\omega_{j,j,t})}\mathbf{z}_{t-1}\mathbf{z}_{t-1}^\intercal+\mathrm{Diag}(\bmu_{q(1/\mathbf{\upsilon}_j)})\right)^{-1}$, where $\Diag(\bmu_{q(1/\mathbf{\upsilon}_j)})$ is a diagonal matrix with elements $\bmu_{q(1/\mathbf{\upsilon}_j)}=(\mu_{q(1/\upsilon_{j,1})},\mu_{q(1/\upsilon_{j,2})},\ldots,\mu_{q(1/\upsilon_{j,d+p+1})})$. The parameters $\boldsymbol{\mu}_{q(\vartheta_j)}$ and $\bmu_{q(\omega_{j,j,t})}$ are as in Proposition \ref{eq:prop1}. The optimal variational densities of the scaling parameters are $q^*(\lambda_{j,k}) \equiv \mathsf{Ga}(a_{q(\lambda_{j,k})}, b_{q(\lambda_{j,k})})$ with $a_{q(\lambda_{j,k})}, b_{q(\lambda_{j,k})}$ defined in Eq.\eqref{eq:up_lam_VBDG}, and $q^*(\upsilon_{j,k}) \equiv \mathsf{GIG}(\zeta_{q(\upsilon_{j,k})}, a_{q(\upsilon_{j,k})}, b_{q(\upsilon_{j,k})})$ is a generalized inverse normal distribution with $\zeta_{q(\upsilon_{j,k})}, a_{q(\upsilon_{j,k})}, b_{q(\upsilon_{j,k})}$ defined in Eq.\eqref{eq:up_ups_VBDG}. 
\label{eq:prop3}
\end{proposition}
Notice that the optimal density for the parameter $\eta_j$ is not a known distribution function. Proposition \ref{prop:eta_DG} in Appendix \ref{app:VBVAR-DG} provides an analytical approximation of its moments so that the optimal density can be calculated via numerical integration. 
%
\paragraph{Horseshoe prior.} As a third hierarchical shrinkage prior we consider the Horseshoe prior as proposed by \cite{Carvalho.2009,carvalho_etal.2010}. This is based on the hierarchical specification $\vartheta_{j,k}|\upsilon^2_{j,k}$, $\gamma^2 \sim \mathsf{N}(0,\gamma^2\upsilon^2_{j,k})$, $\gamma \sim \mathsf{C}^+(0,1)$, $\upsilon_{j,k}\sim \mathsf{C}^+(0,1)$, where $\mathsf{C}^+(0,1)$ denotes the standard half-Cauchy distribution with probability density function equal to $f(x) = 2/\{\pi(1 + x^2)\}\mathbbm{1}_{(0,\infty)}(x)$. The Horseshoe is a global-local prior that implies an aggressive shrinkage of weak signals without affecting the strong ones \citep[see, e.g.,][]{polson_etal.2011}. We follow \cite{Wand.2011} and leverage on a scale mixture representation of the half-Cauchy distribution as,
\begin{equation}\begin{aligned}
\vartheta_{j,k} | \upsilon^2_{j,k}, \gamma^2 \sim \mathsf{N}(0,\gamma^2\upsilon^2_{j,k}), \quad
\gamma^2|\eta &\sim \mathsf{InvGa}(1/2,1/\eta), \quad \upsilon^2_{j,k}|\lambda_{j,k}\sim \mathsf{InvGa}(1/2,1/\lambda_{j,k}), \\
\eta &\sim \mathsf{InvGa}(1/2,1), \qquad \lambda_{j,k}\sim \mathsf{InvGa}(1/2,1),\label{eq:horseshoe}
\end{aligned}\end{equation}
where the local and global shrinkage parameters are $\upsilon^2_{j,k}$ and  $\gamma^2$ respectively. 

Let $\boldsymbol{\xi}_{\text{HS}}=(\boldsymbol{\xi}^\intercal,(\boldsymbol{\upsilon}^2)^\intercal,\gamma^2,\boldsymbol{\lambda}^\intercal,\eta)^\intercal$ be the vector $\boldsymbol{\xi}$ augmented with the parameters of the Horseshoe prior. The joint distribution $\boldsymbol{\xi}_{\text{HS}}$ can be factorized as,
\begin{align}\label{eq:q_hs}
q(\boldsymbol{\xi}_{\text{HS}})=q(\boldsymbol{\xi})q(\boldsymbol{\upsilon}^2,\gamma^2,\boldsymbol{\lambda},\eta), \qquad q(\boldsymbol{\upsilon}^2,\gamma^2,\boldsymbol{\lambda},\eta) = q(\gamma^2)q(\eta)\prod_{j=1}^d\prod_{k=1}^{d+p+1}q(\upsilon_{j,k}^2)q(\lambda_{j,k}).
\end{align}
Proposition \ref{eq:prop4} provides the optimal variational density for the $j$-th row of $\boldsymbol{\Theta}$ under the Horseshoe prior outlined in Eq.\eqref{eq:horseshoe}. The proof and analytical derivations are available in Appendix \ref{app:VBVAR-HS}.
\begin{proposition}
The optimal variational density for $\boldsymbol{\vartheta}_j$ is $q^*(\boldsymbol{\vartheta}_j) \equiv \mathsf{N}_{d+p+1}(\boldsymbol{\mu}_{q(\vartheta_j)}, \boldsymbol{\Sigma}_{q(\vartheta_j)})$ with $\bSigma_{q(\mathbf{\vartheta}_j)} = \left(\sum_{t=1}^T\bmu_{q(\omega_{j,j,t})}\mathbf{z}_{t-1}\mathbf{z}_{t-1}^\intercal+\mu_{q(1/\gamma^2)}\mathrm{Diag}(\bmu_{q(1/\mathbf{\upsilon}_j^2)})\right)^{-1}$, where $\Diag(\bmu_{q(1/\mathbf{\upsilon}_j^2)})$ is a diagonal matrix with elements $\bmu_{q(1/\mathbf{\upsilon}_j^2)}=(\mu_{q(1/\upsilon_{j,1}^2)},\mu_{q(1/\upsilon_{j,2}^2)},\ldots,\mu_{q(1/\upsilon_{j,d+p+1}^2)})$. The parameters $\boldsymbol{\mu}_{q(\vartheta_j)}$ and $\bmu_{q(\omega_{j,j,t})}$ are as in Proposition \ref{eq:prop1}. The optimal variational densities for the global shrinkage is $q^*(\gamma^2) \equiv \mathsf{InvGa}\left(\frac{1}{2}\{d(d+p+1)+1\}, b_{q(\gamma^2)}\right)$ with $b_{q(\gamma^2)}$ defined in Eq.\eqref{eq:up_gamma_VBHS}, and $q^*(\eta) \equiv \mathsf{InvGa}(1, b_{q(\eta)})$ with $b_{q(\eta)}$ defined in Eq.\eqref{eq:up_eta_VBHS}. The optimal variational densities for the local shrinkage parameters are $q^*(\upsilon^2_{j,k}) \equiv \mathsf{InvGa}(1, b_{q(\upsilon^2_{j,k})})$ and $q^*(\lambda_{j,k}) \equiv \mathsf{InvGa}(1, b_{q(\lambda_{j,k})})$, with $b_{q(\upsilon^2_{j,k})}$ and $b_{q(\lambda_{j,k})}$ defined in Eq.\eqref{eq:up_ups_VBHS} and Eq.\eqref{eq:up_lam_VBHS}, respectively.
\label{eq:prop4}
\end{proposition}

%
%

\subsection{From shrinkage to sparsity}
\label{subsec:sparsity}
In addition to computational tractability, shrinking rather than selecting is a defining feature of the hierarchical priors outlined in Section \ref{subsec:prior}. That is, posterior estimates of $\bTheta$ are non-sparse, and thus can not provide exact differentiation between significant vs non-significant predictors. The latter is particularly relevant since we ultimately want to assess the accuracy of our variational inference approach -- versus existing MCMC and variational Bayes algorithms -- in identifying the exact structure of $\bTheta$.   

To address this issue, we build upon \cite{pallavi_battacharya2019savs} and implement a Signal Adaptive Variable Selector (SAVS) algorithm to induce sparsity in $\widehat{\mathbf{\Theta}}$, conditional on a given prior. The SAVS is a post-processing algorithm which divides signals and nulls on the basis of the point estimates of the regression coefficients \citep*[see, e.g.,][]{hauzenberger2021combining}. Specifically, let $\widehat{\vartheta}_j$ the posterior estimate of $\vartheta_j$ and $\mathbf{z}_j$ the associated vector of covariates. If $|\widehat{\vartheta}_j|\,||\mathbf{z}_j||^2\leq|\widehat{\vartheta}_j|^{-2}$ we set $\widehat{\vartheta}_j=0$, where $||\cdot||$ denotes the euclidean norm. 

{\color{black} The reason why we rely on the SAVS post-processing to induce sparsity in the posterior estimates is threefold. First, as highlighted by \cite{pallavi_battacharya2019savs}, the SAVS represents an automatic procedure in which the sparsity-inducing property directly depends on the effectiveness of the shrinkage performed on $\widehat{\vartheta}_j$. This refers to the precision of the posterior mean estimates; that is, the more accurate is $\widehat{\vartheta}_j$, the more precise is the identification of the non-zero elements in $\mathbf{\Theta}$. Second, the SAVS is ``agnostic'' with respect to the shrinkage prior or estimation approach adopted, so it represents a natural tool to compare different estimation methods. Third, it is decision theoretically motivated as it grounds on the idea of minimizing the posterior expected loss \citep*[see, e.g.,][]{huber2021inducing}.

In addition to SAVS, we also expand on \citet{hahn2015decoupling} (HC henceforth) and provide a multivariate extension to their least-angle regression which has originally been built for univariate regressions. Appendix \ref{subsec:postprocessing} provides the full derivation of our extended HC approach as well as a complete discussion of the drawbacks compared to SAVS. In addition, for the interested reader, Appendix \ref{app:more_sim} provides a direct comparison between the SAVS and our multivariate extension to \citet{hahn2015decoupling} based on simulated data (see also the discussion in Section \ref{sec:sim_study}).}


\subsection{Variational predictive density}
Consider the posterior distribution $p(\boldsymbol{\xi}|\mathbf{z}_{1:t})$ given the information set $\mathbf{z}_{1:t}=\left\{\mathbf{y}_{1:t},\mathbf{x}_{1:t}\right\}$ and the conditional likelihood $p(\mathbf{y}_{t+1}|\mathbf{z}_{t},\boldsymbol{\xi})$. A standard predictive density takes the form,
\begin{equation}\label{eq:pred_post}
    p(\mathbf{y}_{t+1}|\mathbf{z}_{1:t}) = \int p(\mathbf{y}_{t+1}|\mathbf{z}_{t},\boldsymbol{\xi})p(\boldsymbol{\xi}|\mathbf{z}_{1:t}) d\boldsymbol{\xi}.
\end{equation}
Given an optimal variational density $q^*(\boldsymbol{\xi})$ that approximates $p(\boldsymbol{\xi}|\mathbf{z}_{1:t})$, we follow \citet{gunawan_etal.2020} and obtain the variational predictive distribution
\begin{equation}
\label{eq:pred_post_approximated}
    q(\mathbf{y}_{t+1}|\mathbf{z}_{1:t}) = \int p(\mathbf{y}_{t+1}|\mathbf{z}_{t},\boldsymbol{\xi})q^*(\boldsymbol{\xi}) d\boldsymbol{\xi} = \int\int p(\mathbf{y}_{t+1}|\mathbf{z}_t,\boldsymbol{\vartheta},\mathbf{\Omega}_t)q^*(\boldsymbol{\vartheta})q^*(\mathbf{\Omega}_t) d\boldsymbol{\vartheta}\,d\mathbf{\Omega}_t.
\end{equation}
Although an analytical expression for Eq.\eqref{eq:pred_post_approximated} is not available, a simulation-based estimator for $q(\mathbf{y}_{t+1}|\mathbf{z}_{1:t})$ can be obtained through Monte Carlo integration by averaging $p(\mathbf{y}_{t+1}|\mathbf{z}_{t},\boldsymbol{\xi}^{(i)})$ over the draws $\boldsymbol{\xi}^{(i)}\sim q^\ast(\boldsymbol{\xi})$, such that $\widehat{q}(\mathbf{y}_{t+1}|\mathbf{z}_{1:t}) = N^{-1}\sum_{i=1}^N p(\mathbf{y}_{t+1}|\mathbf{z}_{t},\boldsymbol{\xi}^{(i)})$. {\color{black}Notice that a complete characterization of the optimal variational predictive density entails $q^*(\mathbf{\Omega}_t)$ with $\mathbf{\Omega}_t=\mathbf{L}^\intercal\mathbf{V}_t\mathbf{L}$. Proposition \ref{eq:prop5_main} shows that, conditional on $\mathbf{L}$ and $\mathbf{V}_t$, the optimal distribution of $\mathbf{\Omega}_t$ can be approximated by a $d$-dimensional Wishart distribution $\mathsf{Wishart}_d(\delta_t,\mathbf{H}_t)$, where $\delta_t$ and $\mathbf{H}_t$ are the degrees of freedom parameter and the scaling matrix, respectively.} 

\begin{proposition}
 The approximate distribution $\widetilde{q}$ of $\mathbf{\Omega}_t$ is $\mathsf{Wishart}_d(\widehat{\delta}_t,\widehat{\mathbf{H}}_t)$, where the scaling matrix is given by $\widehat{\mathbf{H}}_t=\widehat{\delta}_t^{-1}\mathbb{E}_q\left[\mathbf{\Omega}_t\right]$ and $\widehat{\delta}_t$ can be obtained numerically as the solution of a convex optimization problem.
 \label{eq:prop5_main}
\end{proposition}
The complete proof is available in Appendix \ref{app:prec_mat} and is based on the Expectation Propagation (EP) approach proposed by \cite{minka.2001}. In order to implement this approach, there is no need to know $q^*(\mathbf{\Omega}_t)$, but it is sufficient to be able to compute $\mathbb{E}_q(\mathbf{\Omega}_t)$. The latter can be reconstructed based on the optimal variational densities of the Cholesky factor $q^*(\bbeta)$ -- and therefore for $\mathbf{L}$ -- and of $q^*(\mathbf{V}_t)$. The simulation results in Appendix \ref{app:prec_mat} show that the proposed Wishart distribution provides an accurate approximation of $q^*(\mathbf{\Omega}_t)$ for both small and large dimensional models. 

{\color{black}Based on Proposition \ref{eq:prop5_main}, we can further simplify Eq.\eqref{eq:pred_post_approximated} by integrating $\mathbf{\Omega}_t$ such that:
\begin{equation}
\label{eq:pred_post_approximated2}
    q(\mathbf{y}_{t+1}|\mathbf{z}_{1:t}) = \int h(\mathbf{y}_{t+1}|\mathbf{z}_t,\boldsymbol{\vartheta})q^*(\boldsymbol{\vartheta}) d\boldsymbol{\vartheta},
\end{equation}
where $h(\mathbf{y}_{t+1}|\mathbf{z}_t,\boldsymbol{\vartheta})$ denotes the probability density function of a multivariate Student-$t$ distribution $\mathsf{t}_v(\mathbf{m},\mathbf{S})$ with mean $\mathbf{m}=\mathbf{\Theta}\mathbf{z}_t$, scaling matrix $\mathbf{S}=(v\widehat{\mathbf{H}})^{-1}$, and degrees of freedom parameter $v=\widehat{\delta}-d+1$. As a result, the predictive distribution can be approximated by averaging the density of the multivariate Student-$t$ $h(\mathbf{y}_{t+1}|\mathbf{z}_t,\boldsymbol{\vartheta}^{(i)})$ over the draws $\boldsymbol{\vartheta}^{(i)}\sim q^\ast(\boldsymbol{\vartheta})$, for $i=1,\ldots,N$, such that $\widehat{q}(\mathbf{y}_{t+1}|\mathbf{z}_{1:t}) = N^{-1}\sum_{i=1}^N h(\mathbf{y}_{t+1}|\mathbf{z}_{t},\boldsymbol{\vartheta}^{(i)})$. This allows for a more efficient sampling from the predictive density.} 

\textcolor{black}{Notice that the main advantage of the approximation obtained from Proposition 3.5 is to allow for a considerably faster computation of the variational predictive density, compared to using $q^*(\bL)$ and $q^*(\bV_t)$ as stationary distributions to sample $\bOmega_t$, similar to an MCMC. This is because the scaling matrix of the Wishart distribution is available in closed form and the computation of degrees of freedom requires only a one-dimensional optimization. In Appendix \ref{app:predictive_comput} we discuss a further simplification that minimizes the KL divergence between the multivariate Student-$t$ and a multivariate Normal distribution.} 

\section{Simulation study}
\label{sec:sim_study} 
In this section, we report the results of an extensive simulation study designed to compare the properties of our estimation approach against both MCMC and variational Bayes methods for large VAR models. To begin, we compare our {\tt VB} algorithm against the MCMC approach of \citet{chan2018bayesian,cross2020macroeconomic} and the variational inference framework proposed by \citet{chan_yu2020,gefang2023forecasting}. Both these approaches are built upon the structural VAR representation in Eq.\eqref{eq:var1_orth_def2}. {\color{black} Then, we also compare our {\tt VB} method against the MCMC approach developed by \citet{huber2019adaptive,gruber2022forecasting} which is based upon a non-linear parametrization as in Eq.\eqref{eq:var1_orth_def}, similar to our approach.} 

{\color{black}For the sake of comparability with \citet{gruber2022forecasting,gefang2023forecasting}, which do not consider the presence of exogenous predictors, we consider a standard VAR(1) as data generating process. Consistent with the empirical implementations, we set $T=360$ and $d=30,49$. The choice of $d$ is due to the two alternative industry classifications which are explored in the main empirical analysis.} We assume either a moderate -- $50\%$ of zeros -- or a high -- $90\%$ of zeros -- level of sparsity in the true matrix $\boldsymbol{\Theta}$. {\color{black} The latter is generated as follows: we fix to zero $s\cdot d^2$ entries at random, with $s=0.5,0.9$ and $d=30,49$, while the remaining non-zero coefficients are sampled from a mixture of two normal distributions with means equal to $\pm 0.08$ and standard deviation $0.1$. Appendix \ref{app:more_sim} provides additional details on the data generating process and additional simulation results for $d=15$.} 



\subsection{Estimation accuracy}

As a measure of point estimation accuracy, we first look at the Frobenius norm $\Vert\boldsymbol{\Theta}-\widehat{\boldsymbol{\Theta}}\Vert_F$, which measures the difference between the true $\boldsymbol{\Theta}$ observed at each simulation and its estimate $\widehat{\boldsymbol{\Theta}}$. {\color{black}In addition, we compare the ability of each estimation method to identify the non-zero elements in the true $\boldsymbol{\Theta}$ based on the F1 score. The latter can be expressed as a function of counts of true positives ($tp$), false positives ($fp$) and false negatives ($fn$),
\begin{align}
\text{F1} & = \frac{2tp}{2tp + fp + fn}.\notag   
\end{align}
The F1 score takes value one if identification is perfect, i.e., no false positives and no false negatives, and zero if there are no true positives. We compute both measures of estimation accuracy on $N=100$ replications to compare each estimation method and prior specification. The estimates from the MCMC specifications are based on 5,000 posterior simulations, after discarding the first 5,000 as a burn-in sample.}

\paragraph{Point estimates.} Figure \ref{fig:fig3} shows the box charts summarizing the Frobenius norm $\Vert\boldsymbol{\Theta}-\widehat{\boldsymbol{\Theta}}\Vert_F$ across $N=100$ replications. We label the linearized MCMC and variational methods with {\tt LMCMC} and {\tt LVB}, respectively, with {\tt MCMC} the non-linear method of \citet{gruber2022forecasting} and with {\tt VB} our variational inference method, respectively. To increase readability, we separate the results by prior and color-code the four different estimation methods. For instance, for a given sub-plot we report the results for the Normal, adaptive-Lasso, adaptive Normal-Gamma and Horseshoe priors from the left to the right panel. Within each panel, the simulation results for the {\tt LMCMC}, {\tt LVB}, {\tt MCMC} and {\tt VB} estimates are reported in red, yellow, light-blue and green, respectively. 

\begin{figure}[!ht]
	\centering
    \subfigure[$d=30$, moderate sparsity]{\includegraphics[width=.42\textwidth]{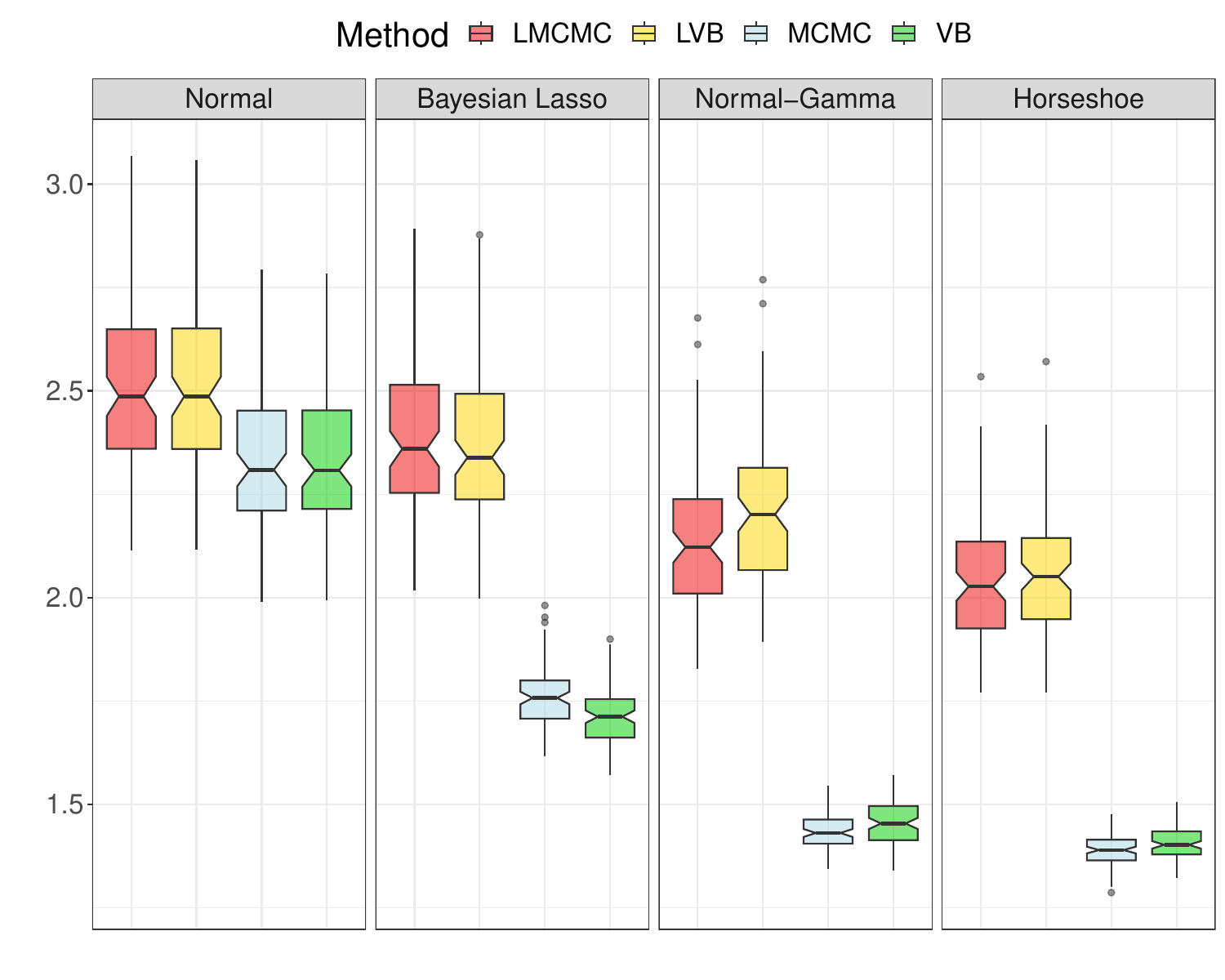}}
    \subfigure[$d=49$, moderate sparsity]{\includegraphics[width=.42\textwidth]{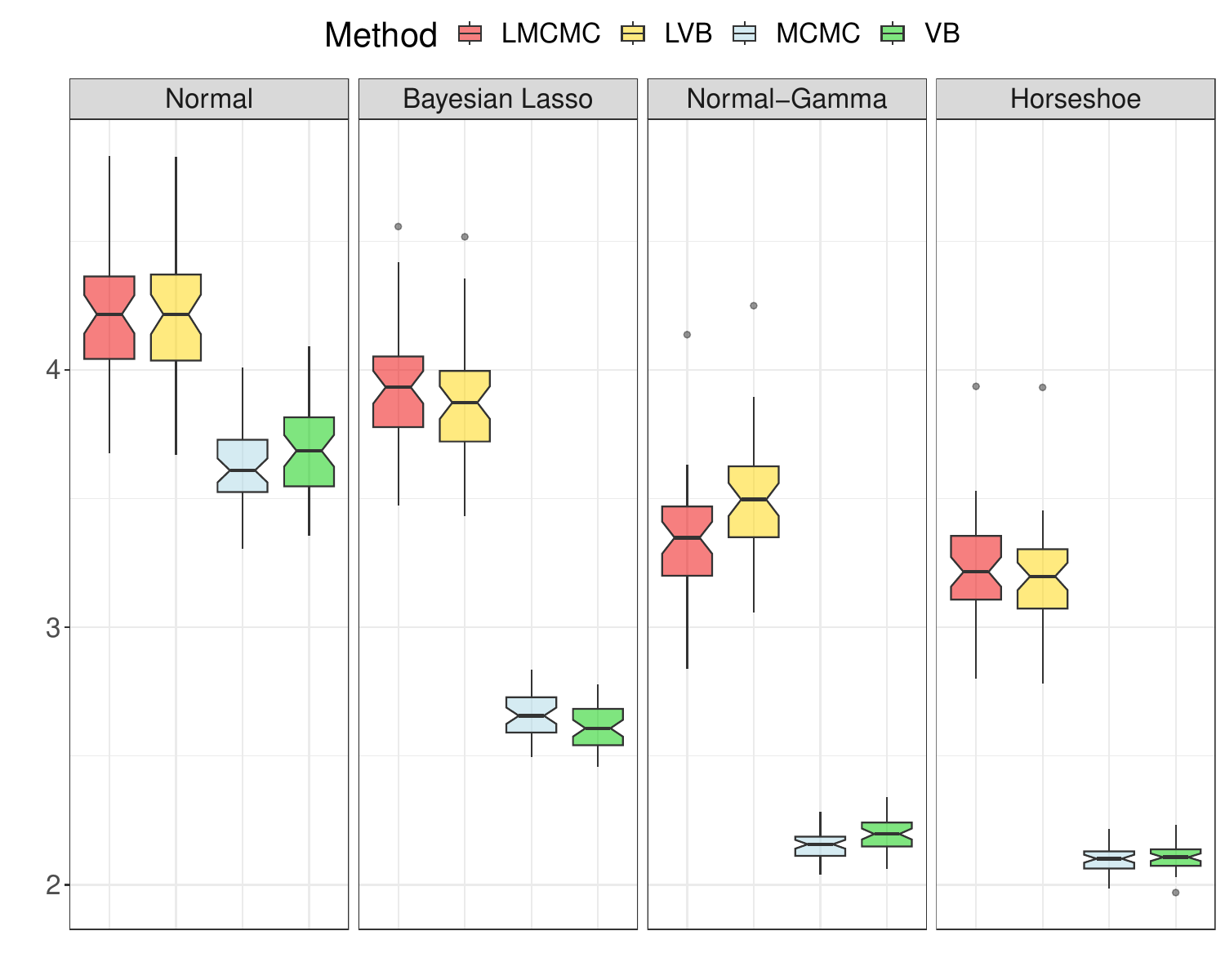}}
    
\centering\subfigure[$d=30$, high sparsity]{\includegraphics[width=.42\textwidth]{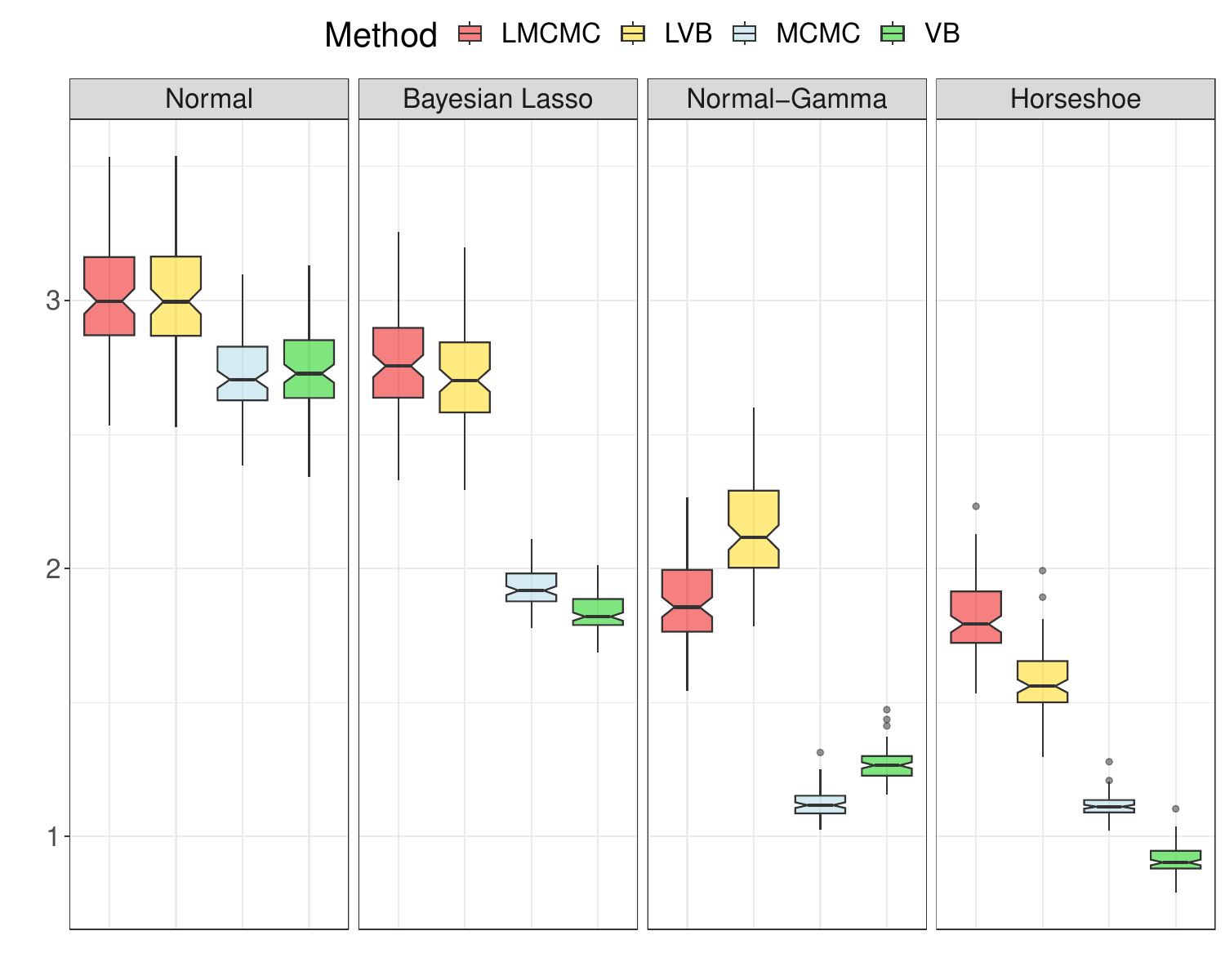}}
    \subfigure[$d=49$, high sparsity]{\includegraphics[width=.42\textwidth]{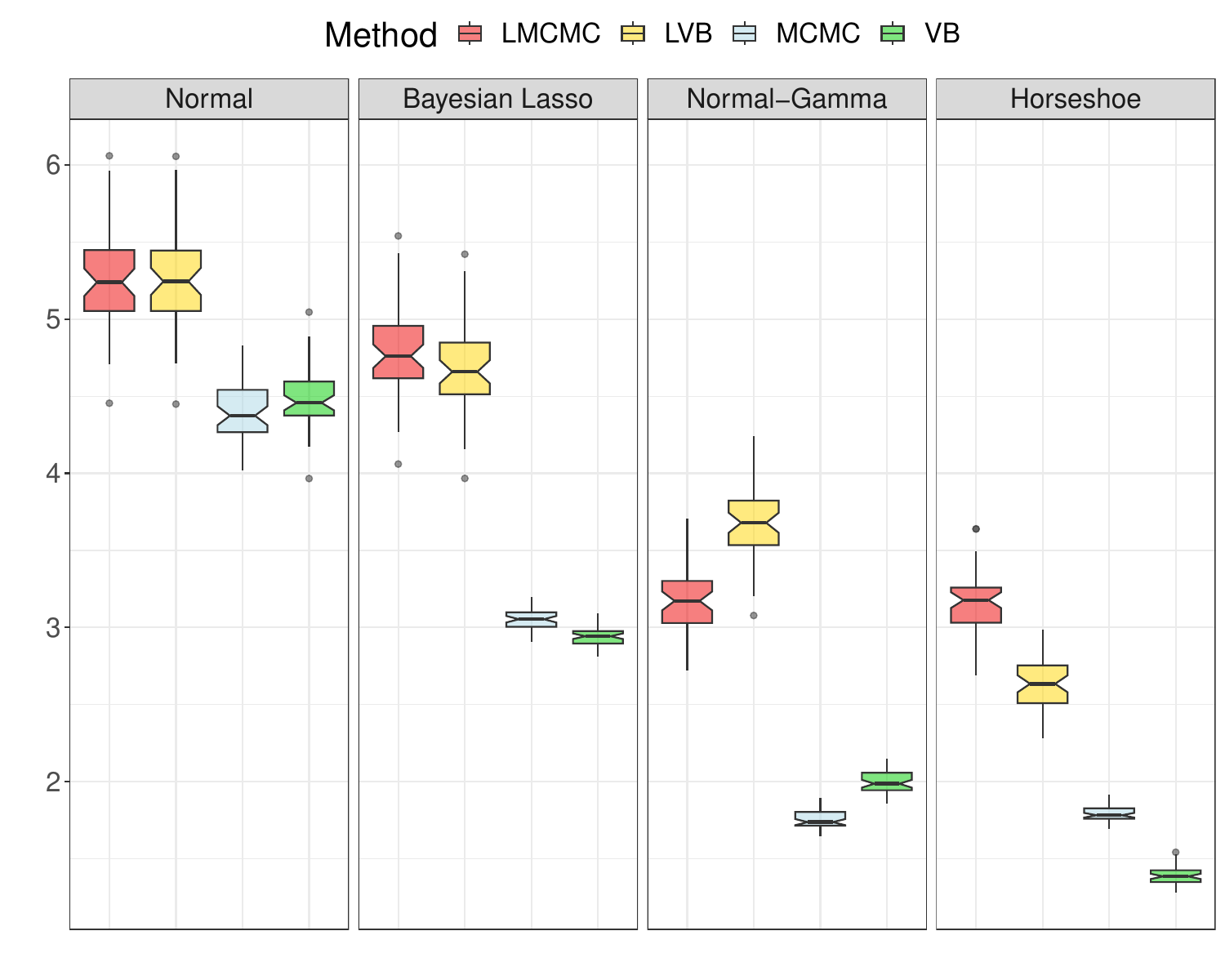}}
		\caption{\small Frobenius norm of $\boldsymbol{\Theta}-\widehat{\boldsymbol{\Theta}}$ across $N=100$ replications, for different shrinkage priors and different inference methods.}
        \label{fig:fig3}
\end{figure}

Beginning with the moderate sparsity case (top panels), the simulation results show that {\tt LMCMC} and {\tt LVB} approaches tend to perform equally across different shrinkage priors, with the only exception of the Normal-Gamma prior, in which {\tt LMCMC} slightly outperforms {\tt LVB}. However, the discrepancy between the two structural VAR representation methods tend to increase when sparsity becomes more pervasive (see bottom panels). 

Overall, the simulation results support our view that, by eliciting shrinkage priors directly on $\boldsymbol{\Theta}$ -- as per the parametrization in Eq.\eqref{eq:var1_orth_def} -- the accuracy of the posterior estimates improves. The mean squared errors obtained from {\tt MCMC} and {\tt VB} are lower compared to both {\tt LMCMC} and {\tt LVB}. This holds for all priors and the model dimension. The accuracy with $d=30$ of the {\tt MCMC} and {\tt VB} is virtually the same. Yet, with $d=49$ our {\tt VB} produces slightly more accurate estimates than {\tt MCMC} for both the adaptive-Lasso and the Horseshoe prior. 

\paragraph{Sparsity identification.} Figure \ref{fig:fig4} shows the box charts of F1 scores across $N=100$ simulations. The labeling is the same as in Figure \ref{fig:fig3}. Both {\tt LMCMC} and {\tt LVB} produce a rather dismal identification of the non-zero elements in $\bTheta$ across prios and model dimensions. This is due to the fact that $\widehat{\boldsymbol{\Theta}}=\widehat{\mathbf{L}}^{-1}\widehat{\mathbf{A}}$ in Eq.\eqref{eq:var1_orth_def2}, so that a sparse estimate of $\widehat{\mathbf{A}}$ does not map into a sparse estimate of $\widehat{\boldsymbol{\Theta}}$, and therefore produces a lower accuracy in identifying the non-zero coefficients in the true $\boldsymbol{\Theta}$. As the level of sparsity increases, the divergence between $\mathbf{A}$ and $\boldsymbol{\Theta}$ increases.

\begin{figure}[h]
	\centering
\subfigure[$d=30$, moderate sparsity]{\includegraphics[width=.42\textwidth]{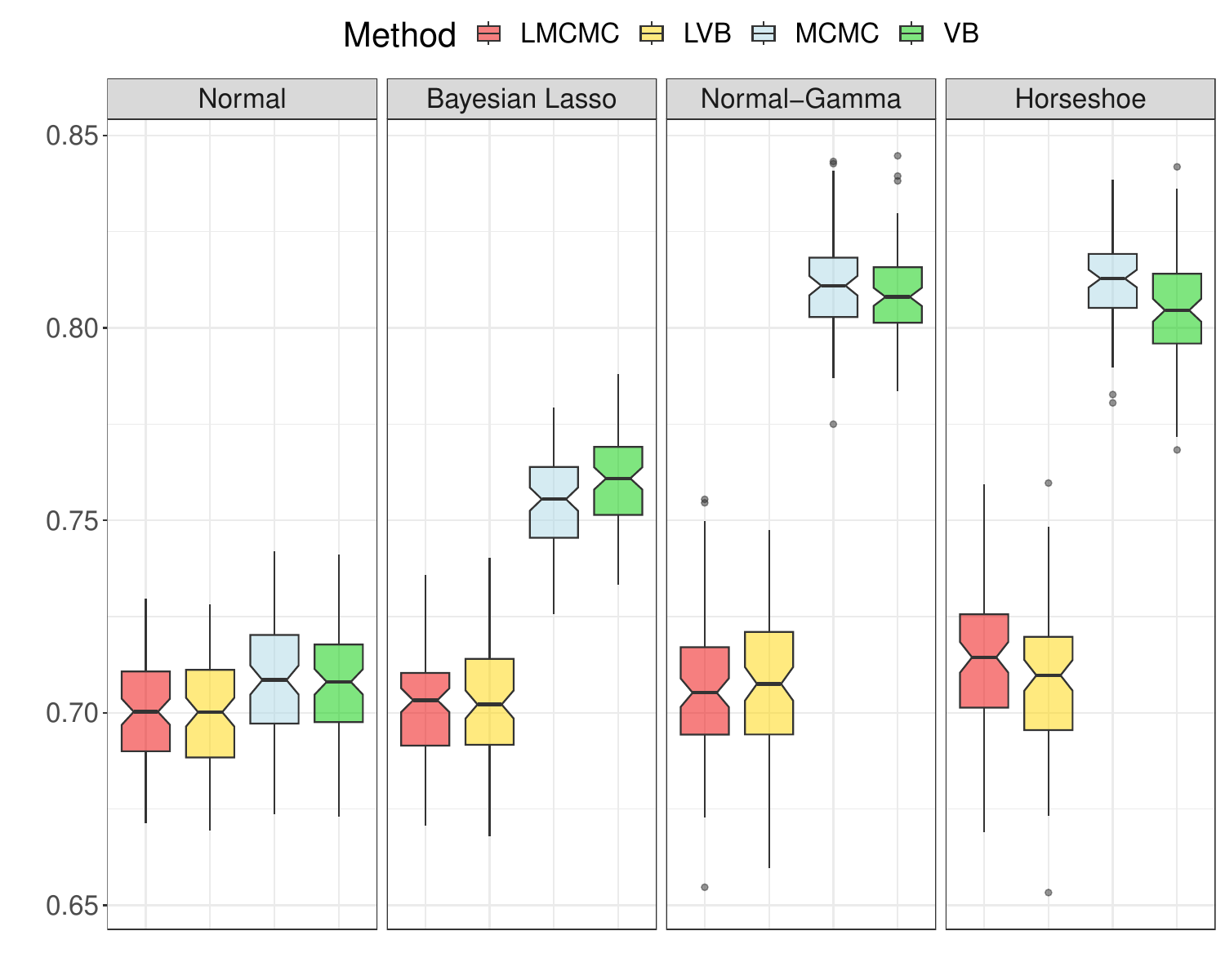}}
    \subfigure[$d=49$, moderate sparsity]{\includegraphics[width=.42\textwidth]{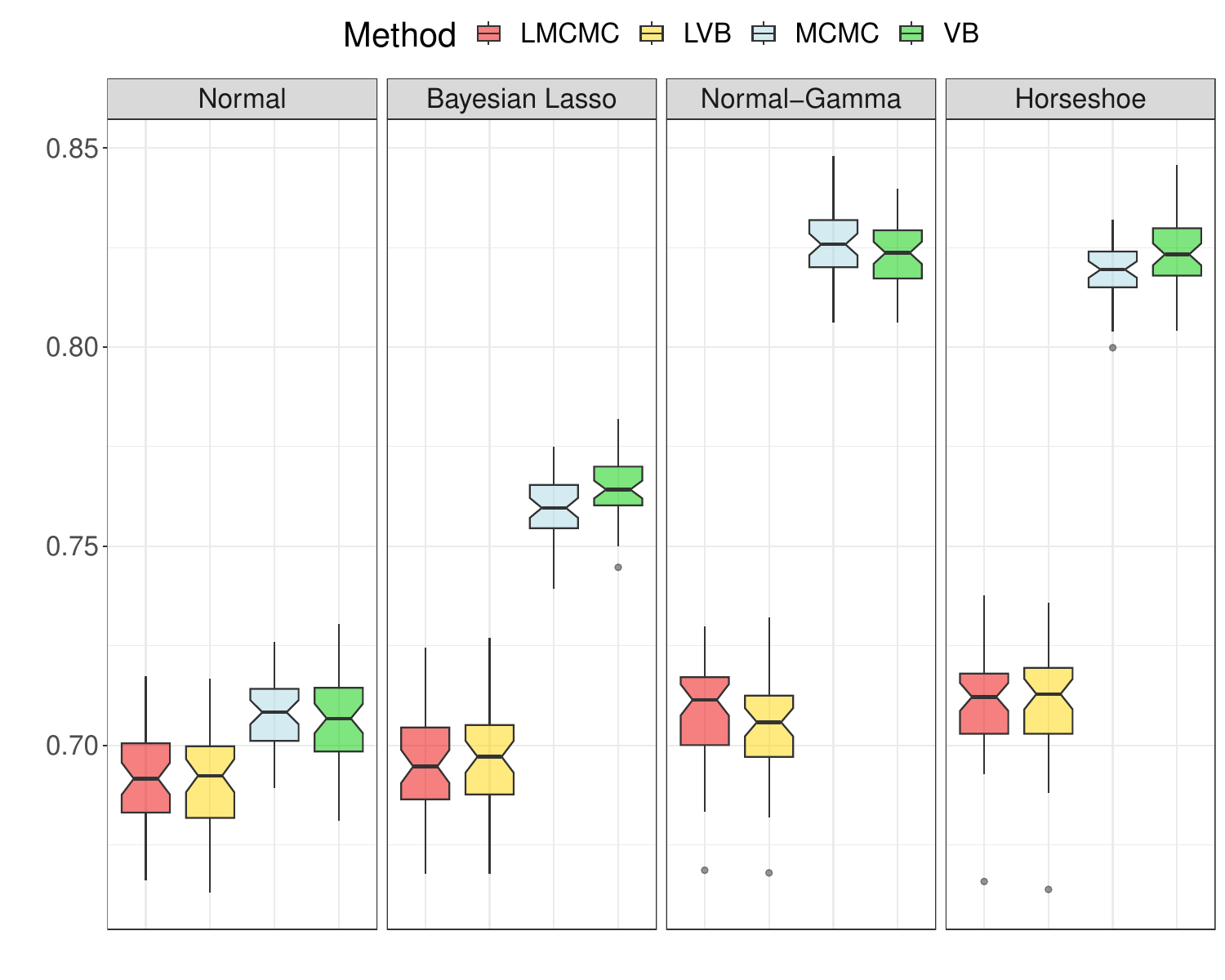}}\\
\subfigure[$d=30$, high sparsity]{\includegraphics[width=.42\textwidth]{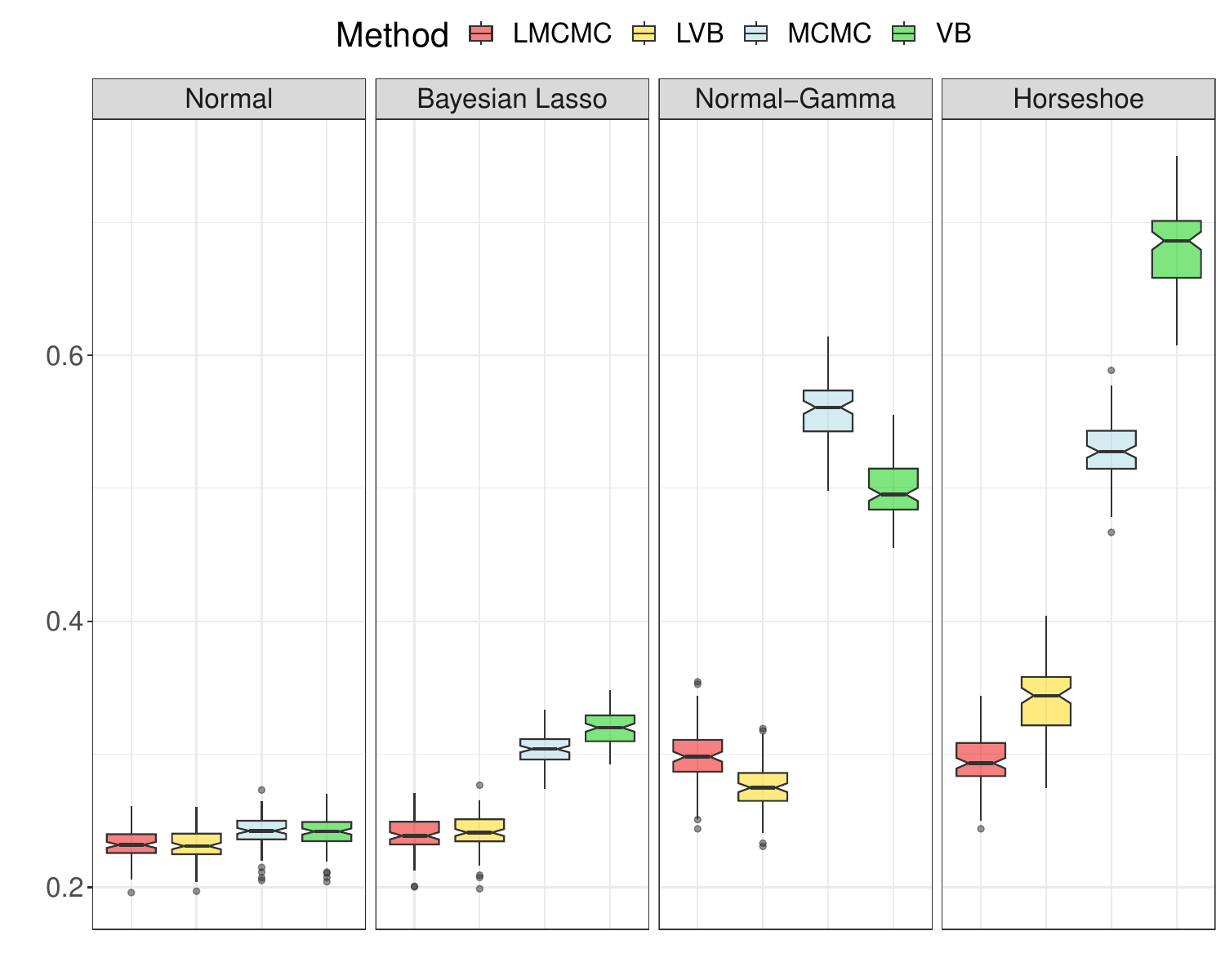}}
    \subfigure[$d=49$, high sparsity]{\includegraphics[width=.42\textwidth]{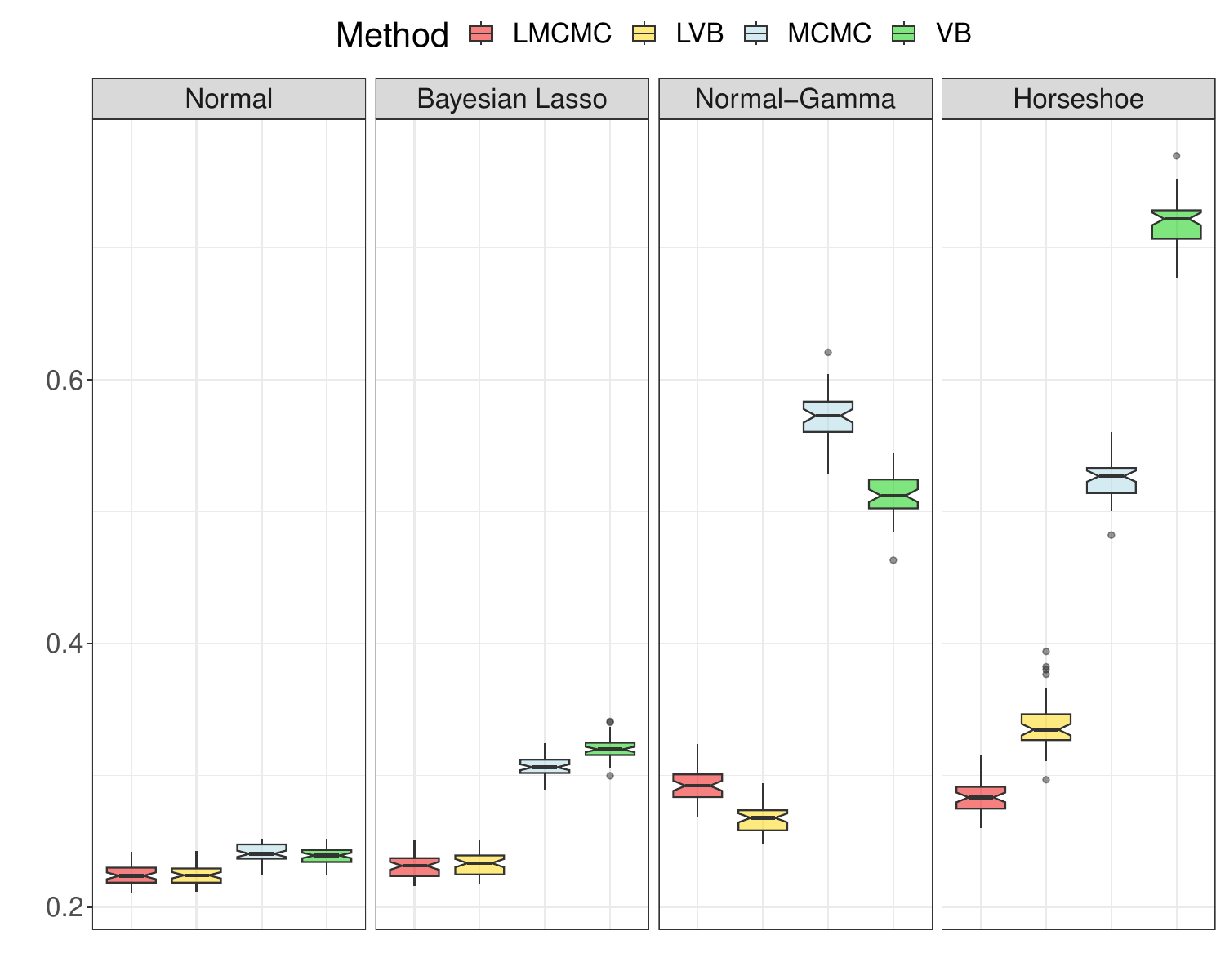}}
		\caption{\small  F1 score computed across $N=100$ replications by looking at the true non-null parameters in $\boldsymbol{\Theta}$ and the non-null parameters estimated based on $\widehat{\boldsymbol{\Theta}}$.}
        \label{fig:fig4}
\end{figure}

Consistent with our argument in favor of the parametrization in Eq.\eqref{eq:var1_orth_def}, both the {\tt MCMC} and {\tt VB} approaches produce a more accurate identification of the non-zero coefficients in $\bTheta$, as shown by the F1 score. The gap between {\tt LMCMC}, {\tt LVB} versus {\tt MCMC} and {\tt VB} becomes larger for higher levels of sparsity. This result holds across different hierarchical shrinkage priors and for different VAR dimensions. Yet, our {\tt VB} approach turns out to be more accurate than {\tt MCMC} under the adaptive-Lasso and Horseshoe priors for higher levels of sparsity. 

{\color{black}As outlined in Section \ref{sec:mfvb}, sparsity in the posterior estimates for $\widehat{\boldsymbol{\Theta}}$ for different hierarchical shrinkage priors is induced in the simulation results by using the SAVS algorithm of \citet{pallavi_battacharya2019savs}. Appendix \ref{app:more_sim} provides additional simulation results obtained by implementing a multivariate version of the post-processing method proposed by \citet{hahn2015decoupling} as an alternative to the SAVS. A full derivation is provided in Appendix \ref{subsec:postprocessing}. The F1 scores are largely the same across methods; in fact, the evidence is even more in favour of our {\tt VB}, compared to its {\tt MCMC} counterpart when using the extended \citet{hahn2015decoupling} approach: our {\tt VB} is more accurate than {\tt MCMC} with a Normal-Gamma prior.}   


\paragraph{Computational efficiency.} \citet{chan_yu2020} and \citet{gefang2023forecasting} highlight that one of the main advantages of variational Bayes methods is computational efficiency. {\color{black}Figure \ref{fig:time15} reports the computational time -- expressed in a log-minute scale -- required by each estimation approach under different shrinkage priors. To highlight the performance for a given prior, we separate the results by estimation methods and color-code the four different shrinkage priors. For instance, for a given sub-plot, we report the results for the {\tt LMCMC}, {\tt LVB}, {\tt MCMC} and {\tt VB} estimates from left to right panel. Within each panel, the Normal, adaptive-Lasso, adaptive Normal-Gamma, and Horseshoe priors are colored in shades of gray from light (left) to dark (right) grey, respectively. To guarantee a more accurate comparability, we re-coded all competing methods in {\bf Rcpp} and use the same 2.5 GHz Intel Xeon W-2175 with 32GB of RAM for all implementations.}

\begin{figure}[h!]
	\centering
\subfigure[$d=30$, moderate sparsity]{\includegraphics[width=.42\textwidth]{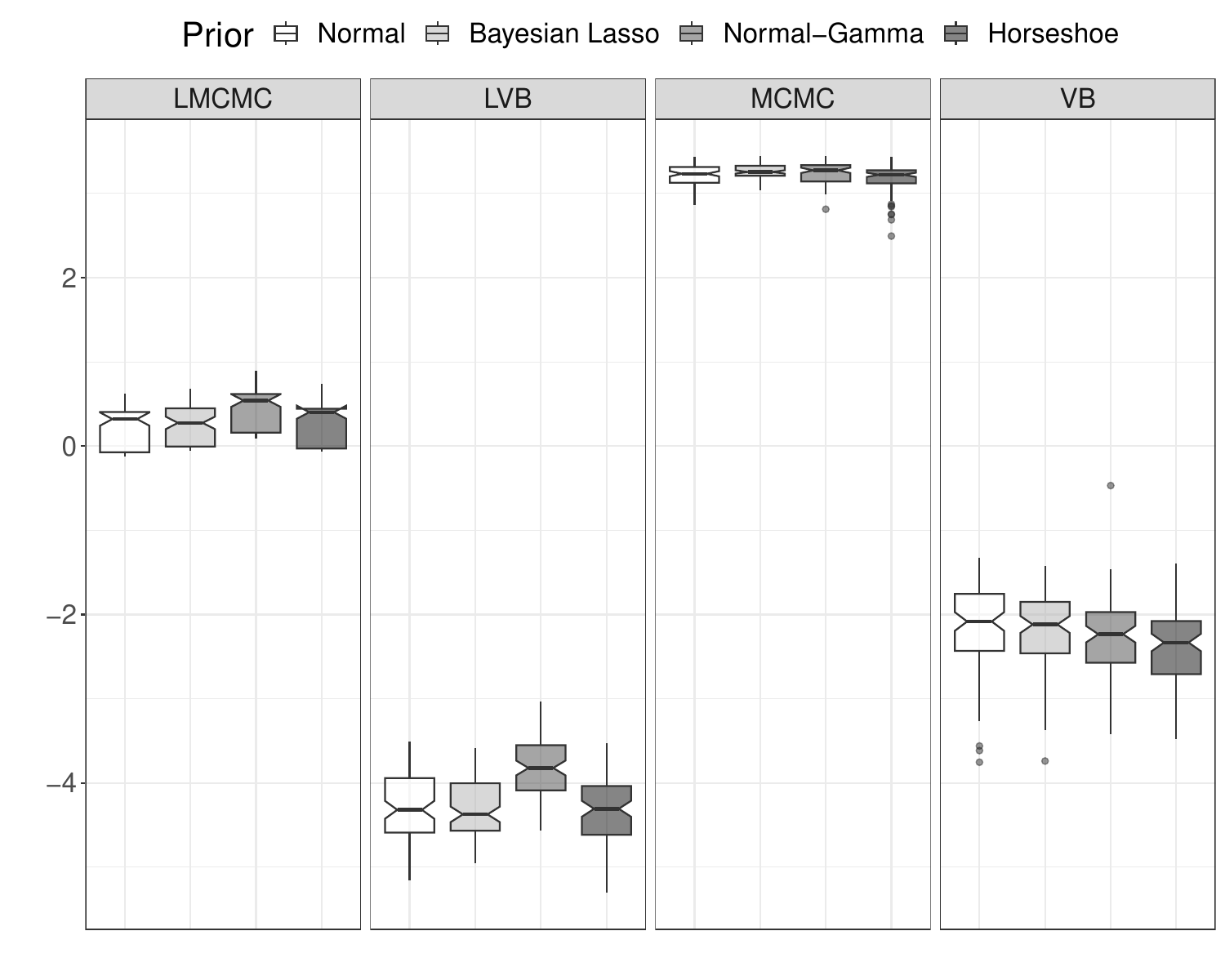}}
\subfigure[$d=49$, moderate sparsity]{\includegraphics[width=.42\textwidth]{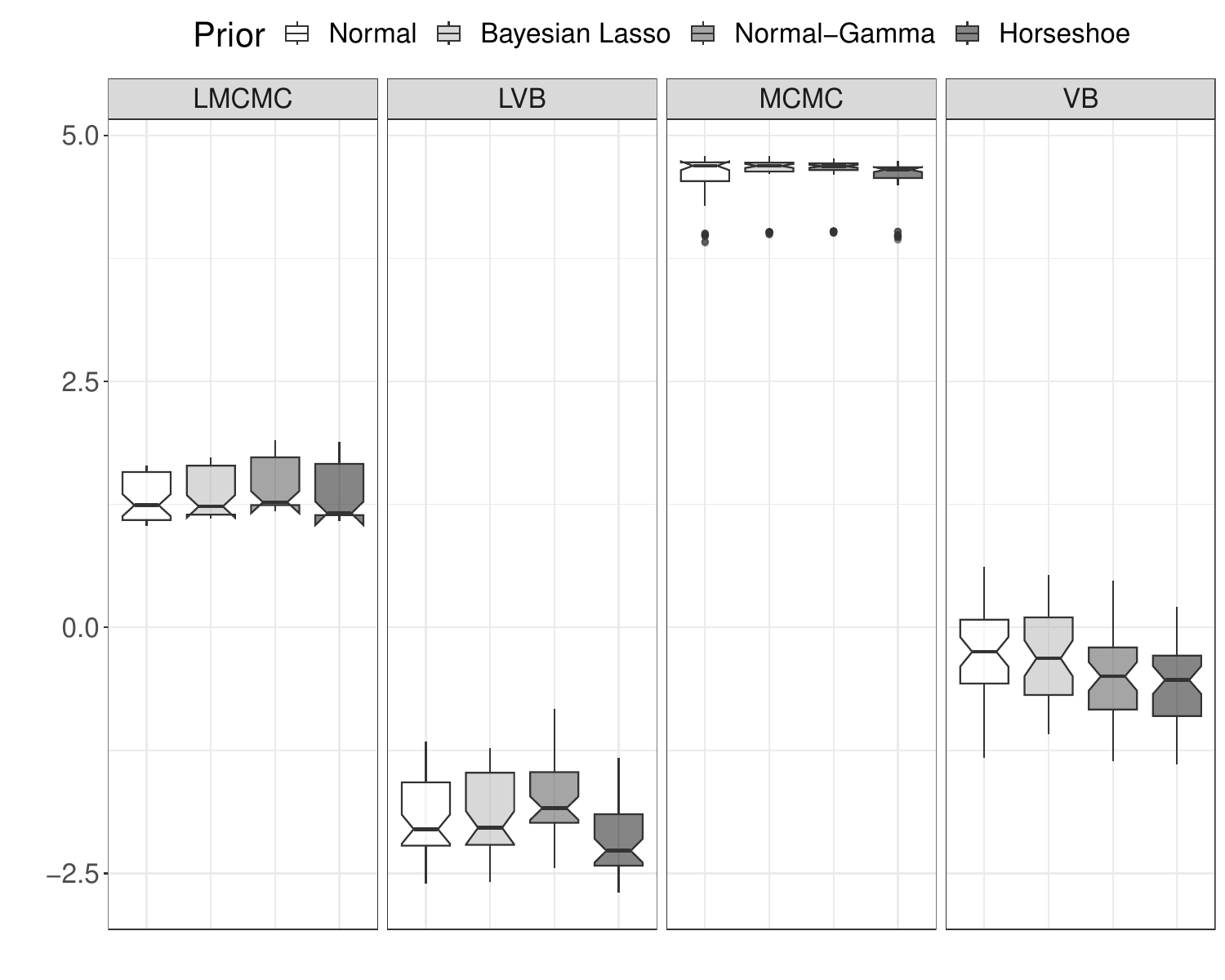}}

\subfigure[$d=30$, high sparsity]{\includegraphics[width=.42\textwidth]{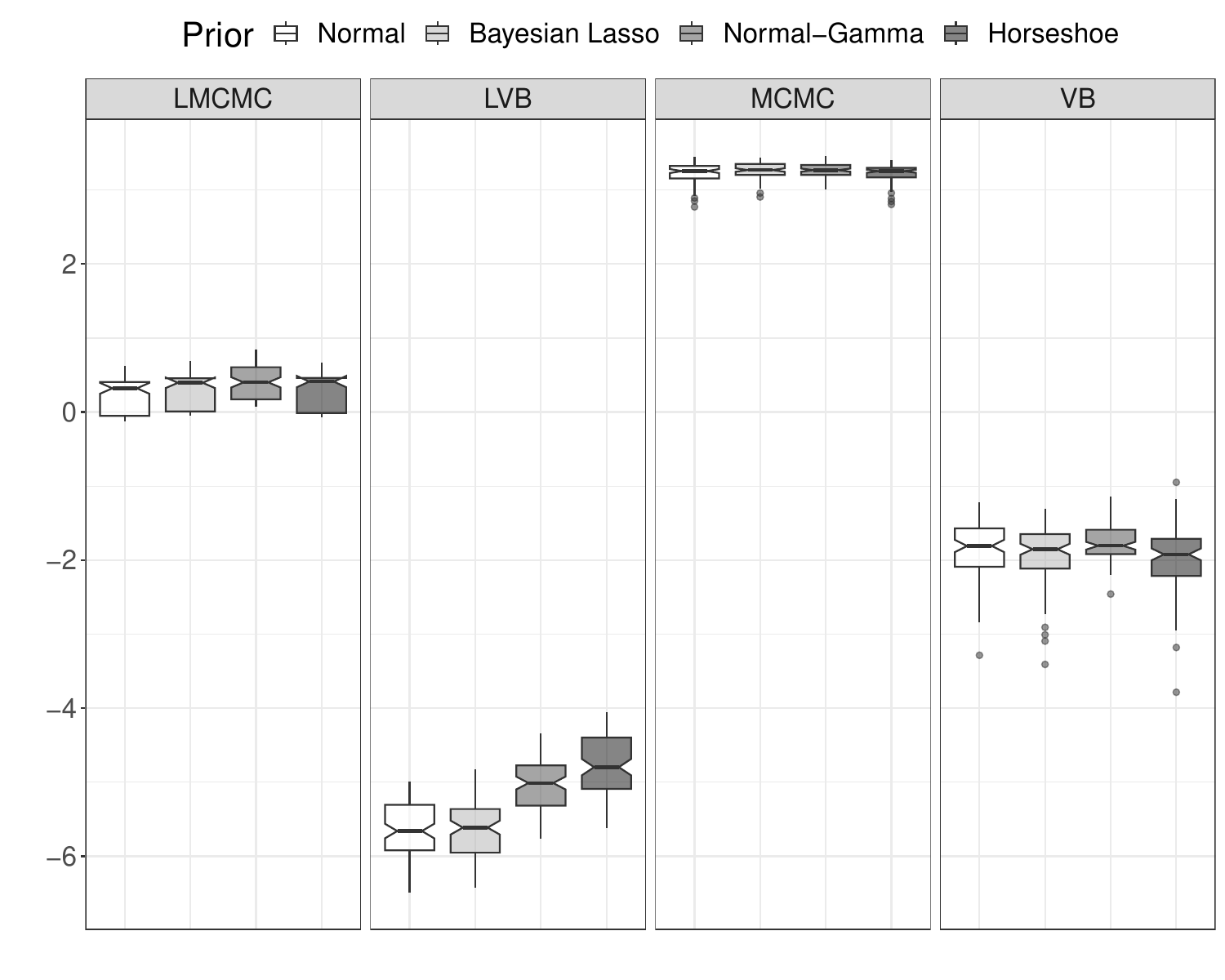}}
\subfigure[$d=49$, high sparsity]{\includegraphics[width=.42\textwidth]{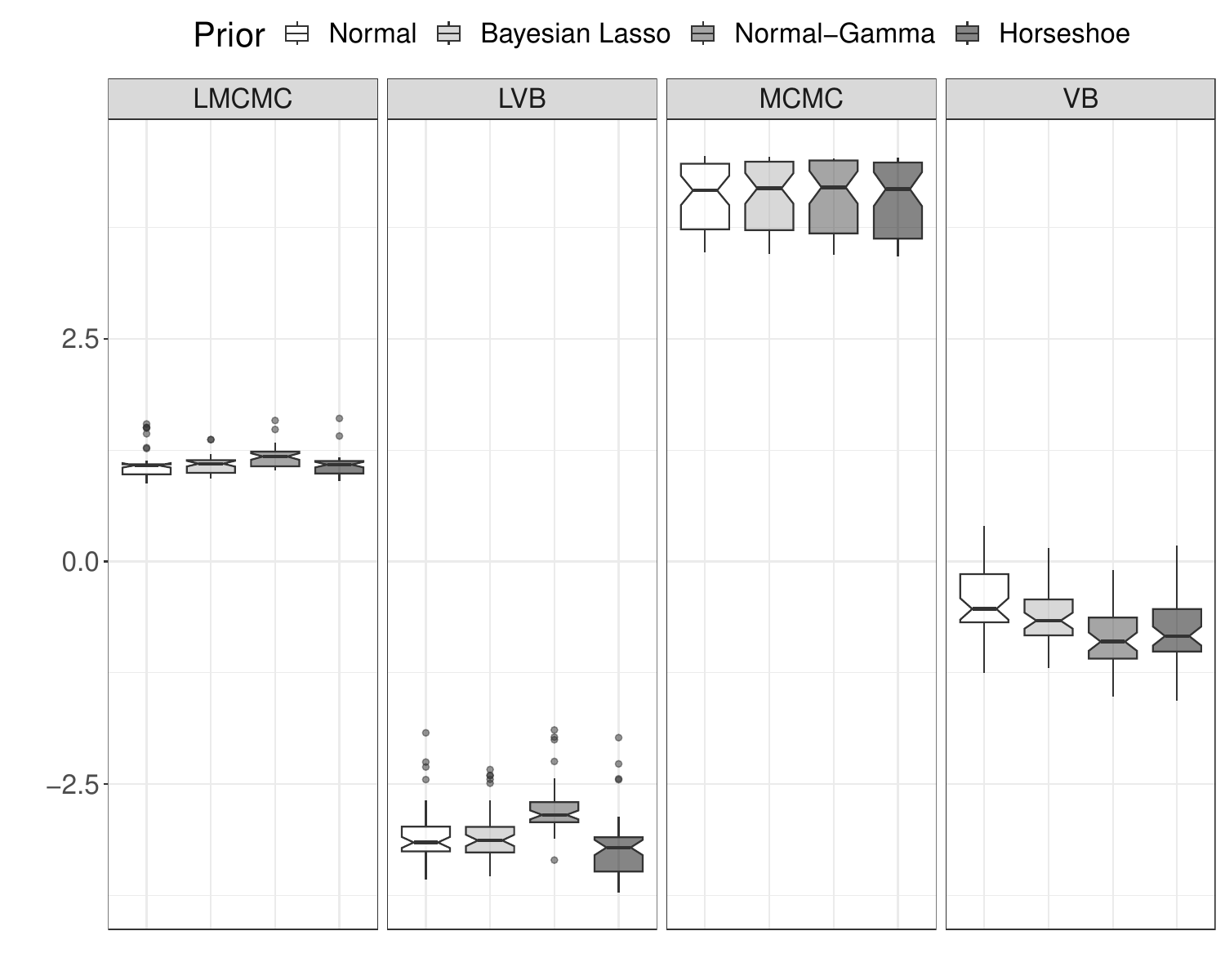}}
\caption{Computational time required by each estimation approach for different hierarchical shrinkage priors. The time is expressed on logarithmic minutes scale.}
\label{fig:time15}
\end{figure}

{\color{black}The results highlight that our {\tt VB} approach has a clear computational advantage compared to both linear and non-linear MCMC methods. For instance, for $d=30$ our {\tt VB} is more than 100 times faster than the {\tt MCMC} of \citet{gruber2022forecasting} and more than 10 times faster than the {\tt LMCMC} of \citet{cross2020macroeconomic}, respectively. The gap in favour of our {\tt VB} method compared to both {\tt LMCMC} and {\tt MCMC} increases in larger dimensions; for $d=49$ the {\tt MCMC} approach takes almost 60 minutes, on average, to generate comparably accurate posterior estimates to our {\tt VB}, which instead takes approximately between 30 to 40 seconds, on average. Such efficiency gap between {\tt VB} and {\tt MCMC} has profound implications for a practical forecasting implementation, especially within the context of recursive predictions with higher frequency data such as stock returns (see Section \ref{subsec:oos}). Perhaps not surprisingly, the {\tt LVB} approach of \citet{chan_yu2020,gefang2023forecasting} is highly competitive in terms of computational efficiency. However, being built on a structural VAR formulation, we showed in Figures \ref{fig:fig3} and \ref{fig:fig4} that such computational efficiency comes at the cost of a lower estimation accuracy.}

\textcolor{black}{Appendix \ref{subsec:computational cost} also provides a broader qualitative discussion on the computational costs of some of the existing MCMC approaches. Specifically, we review some of the results reported in the original papers and show that these largely align with our own findings. In addition, we also discuss some of the limitations of the non-linear {\tt MCMC} for the recursive forecasting implementation (see Section \ref{subsec:oos} for more details).} 

\paragraph{Robustness to variables permutation.} {\color{black}At the outset of the paper, we argue that a conventional structural VAR formulation potentially generates posterior estimates which are not permutation-invariant. That is, posterior estimates of $\bTheta$ are sensitive to the ordering imposed on the target variables $\mathbf{y}_t$, conditional on a given prior. To highlight this issue, in Appendix \ref{app:more_sim}, we report a set of additional simulation results for all estimation methods and shrinkage priors under variables permutation. 

The results show that the accuracy of the posterior estimates from both {\tt LMCMC} and {\tt LVB} changes once the variables ordering is reversed (see Figure \ref{fig:fig3a}). This is especially clear for the Normal-Gamma and Horseshoe priors, and when the amount of zero coefficients in $\boldsymbol{\Theta}$ is more pervasive. On the other hand, the estimation accuracy of both the {\tt MCMC} approach of \citet{gruber2022forecasting} and our {\tt VB} method does not substantially deteriorates by arbitrarily changing ordering of the target variables. Overall a substantially higher computational efficiency coupled with a comparable accuracy with complex MCMC, makes our {\tt VB} extremely competitive within the context of recursive forecasts with higher frequency data.}

\section{A empirical study of industry returns predictability}
\label{sec:appl}
We investigate both the statistical and economic value of our variational Bayes approach within the context of US industry returns predictability. \textcolor{black}{To expand the scope of the testing framework, we consider two alternative industry aggregations: $d=30$ industry portfolios from July 1926 to May 2020, and a larger cross section of $d=49$ industry portfolios from July 1969 to May 2020. The size of the cross sections change due to a different industry classification.} At the end of June of year $t$ each NYSE, AMEX, and NASDAQ stock is assigned to an industry portfolio based on its four-digit SIC code at that time. Thus, the returns on a given value-weighted portfolio are computed from July of $t$ to June of $t+1$. The sample periods cover major events, from the great depression to the Covid-19 outbreak. 

In addition to cross-industry portfolio returns, we consider a variety of predictors, such as the returns on the market portfolio ({\tt mkt}), and the returns on four alternative long-short investment strategies based on market capitalization ({\tt smb}), book-to-market ratios ({\tt hml}), operating profitability ({\tt rmw}) and firm investments ({\tt cma}) (see \citealp{fama2015five}). We also consider a set of additional macroeconomic predictors from \citet{Goyal2008}, such as the log price-dividend ratio ({\tt pd}), the difference between the long term yield on government bonds and the T-bill ({\tt term}), the BAA-AAA bond yields difference ({\tt credit}), the monthly log change in the CPI ({\tt infl}), the aggregate market book-to-market ratio ({\tt bm}), the net-equity issuing activity ({\tt ntis}) and the corporate bond returns ({\tt corpr}). 


\subsection{In-sample estimates of $\bTheta$}
In order to highlight some of the main properties of different estimation methods, we first report the in-sample estimates of $\bTheta$ for the $d=49$ industry case across all priors. Figure \ref{fig:theta} compares $\widehat{\bTheta}$ based on the full sample obtained from the {\tt LMCMC} and the {\tt LVB} with constant volatility, and our {\tt VB} with and without stochastic volatility. Appendix \ref{app:more_emp} reports the additional in-sample estimates for $d=30$ industry portfolios. 

\begin{figure}[h!]	
\centering
\hspace{-2.5em}\subfigure[{\tt LMCMC} w/ normal]{\includegraphics[width=.25\textwidth]{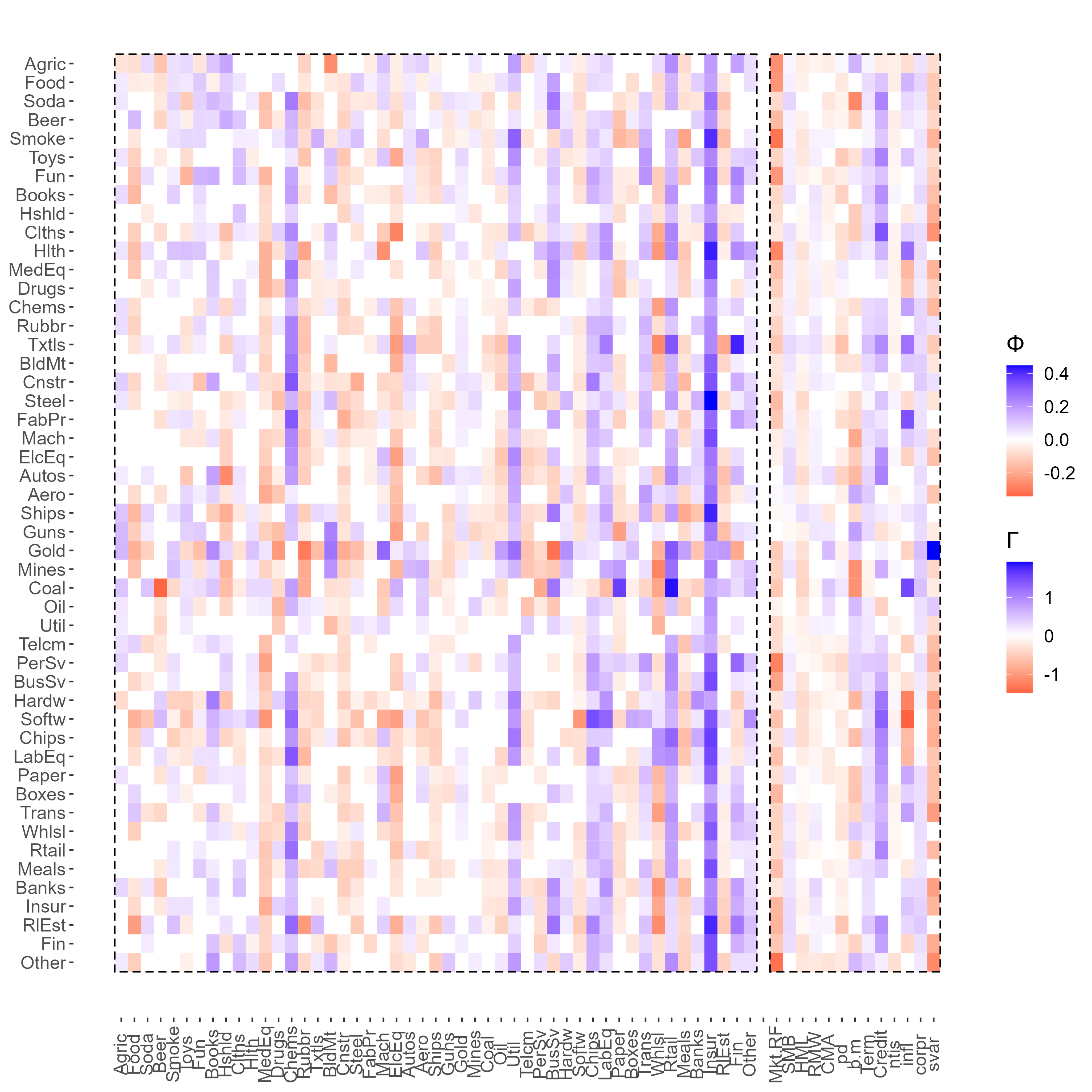}}\quad
\subfigure[{\tt LVB} w/ normal]{\includegraphics[width=.25\textwidth]{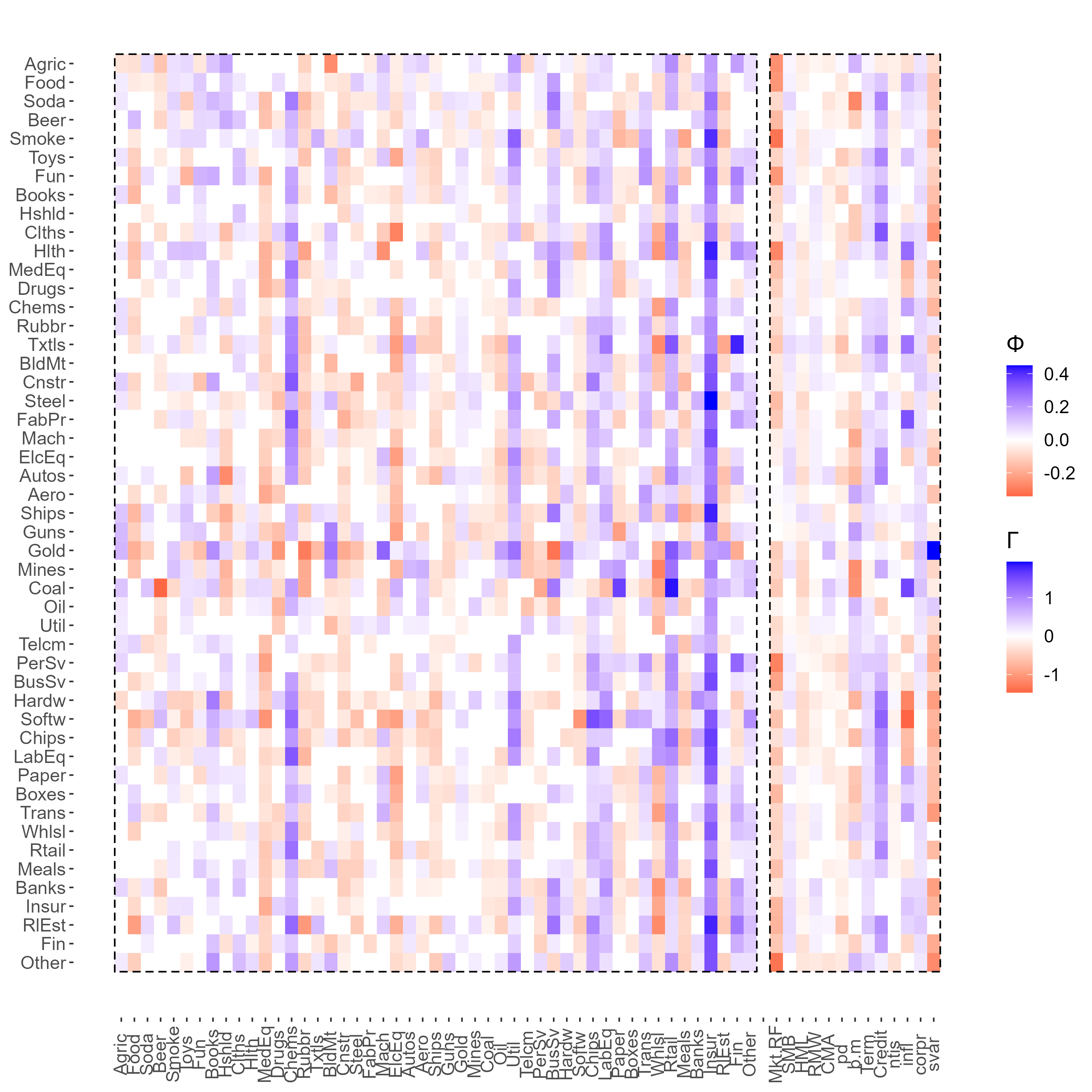}}
\subfigure[{\tt VB} w/ normal]{\includegraphics[width=.25\textwidth]{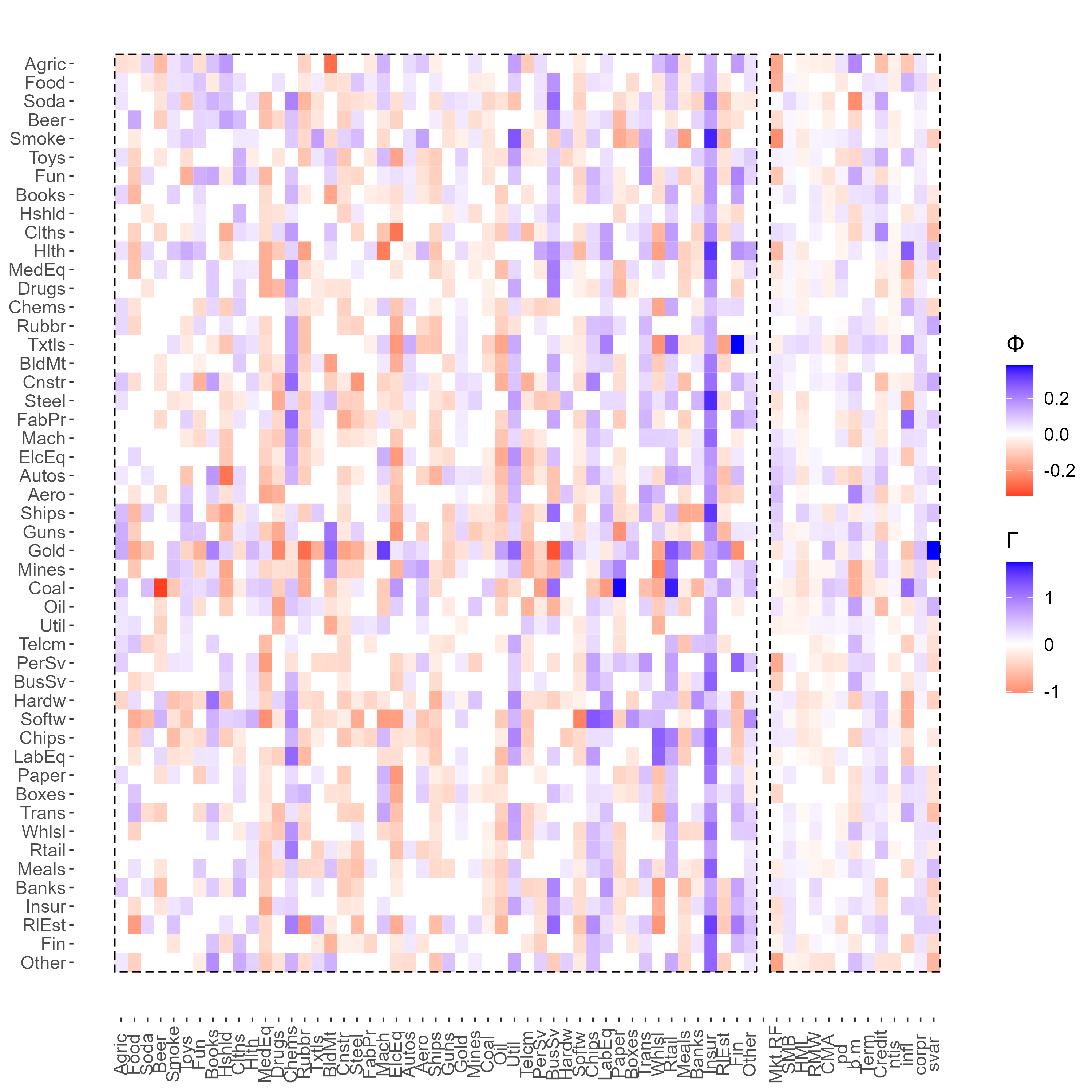}}\subfigure[{\tt VB} w/ normal + SV]{\includegraphics[width=.25\textwidth]{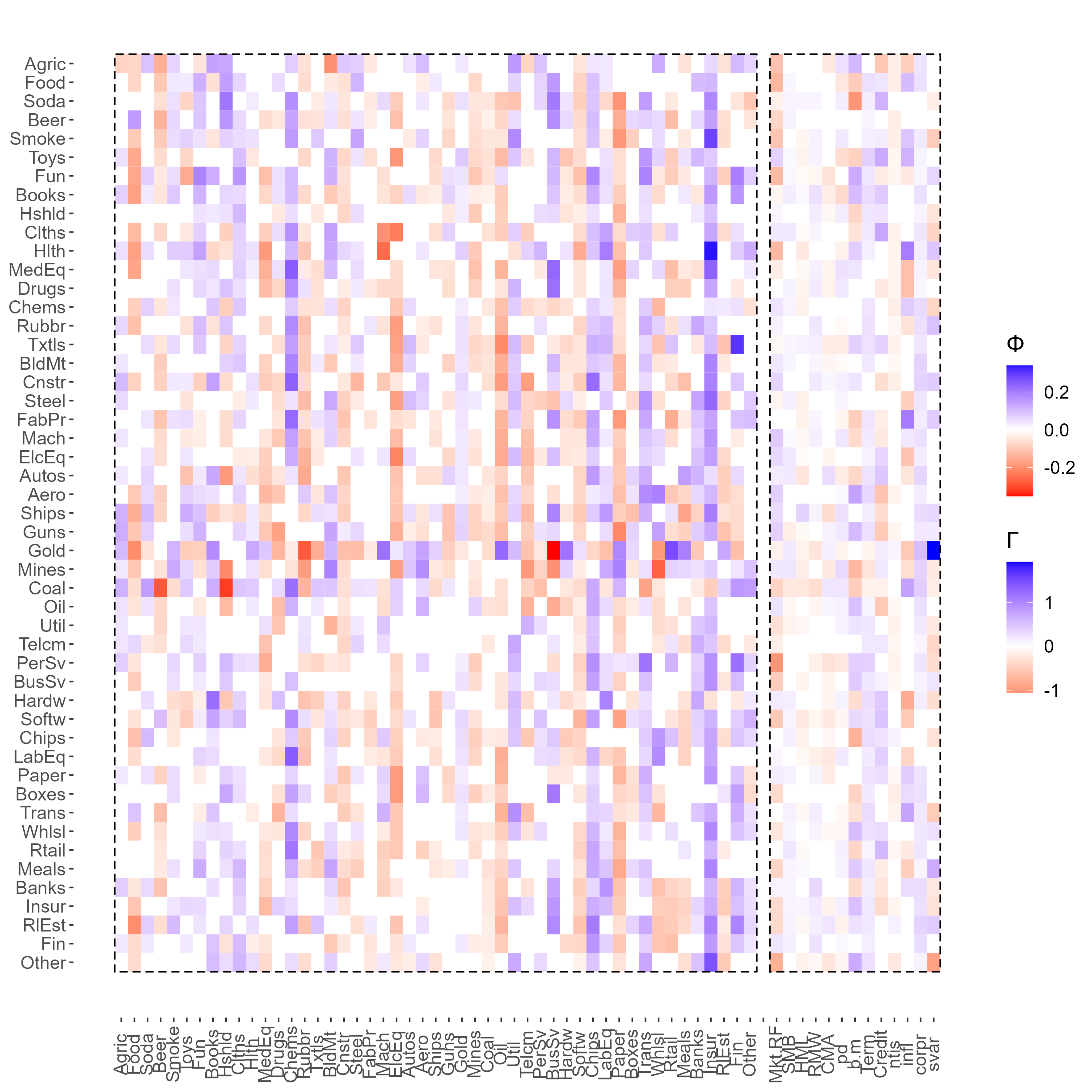}}\hspace{-2.5em}

\hspace{-2.5em}\subfigure[{\tt LMCMC} w/ Lasso]{\includegraphics[width=.25\textwidth]{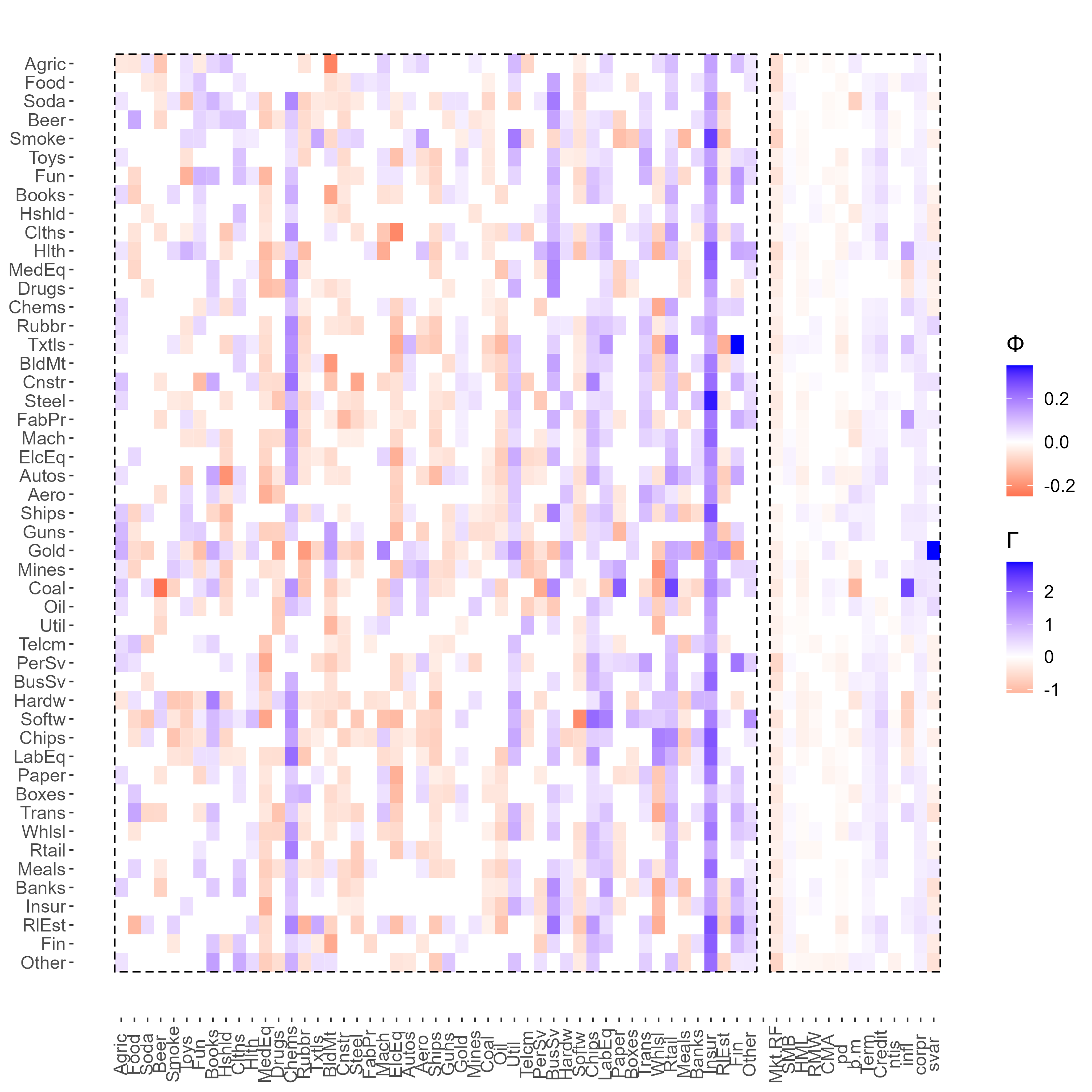}}\quad
\subfigure[{\tt LVB} w/ Lasso]{\includegraphics[width=.25\textwidth]{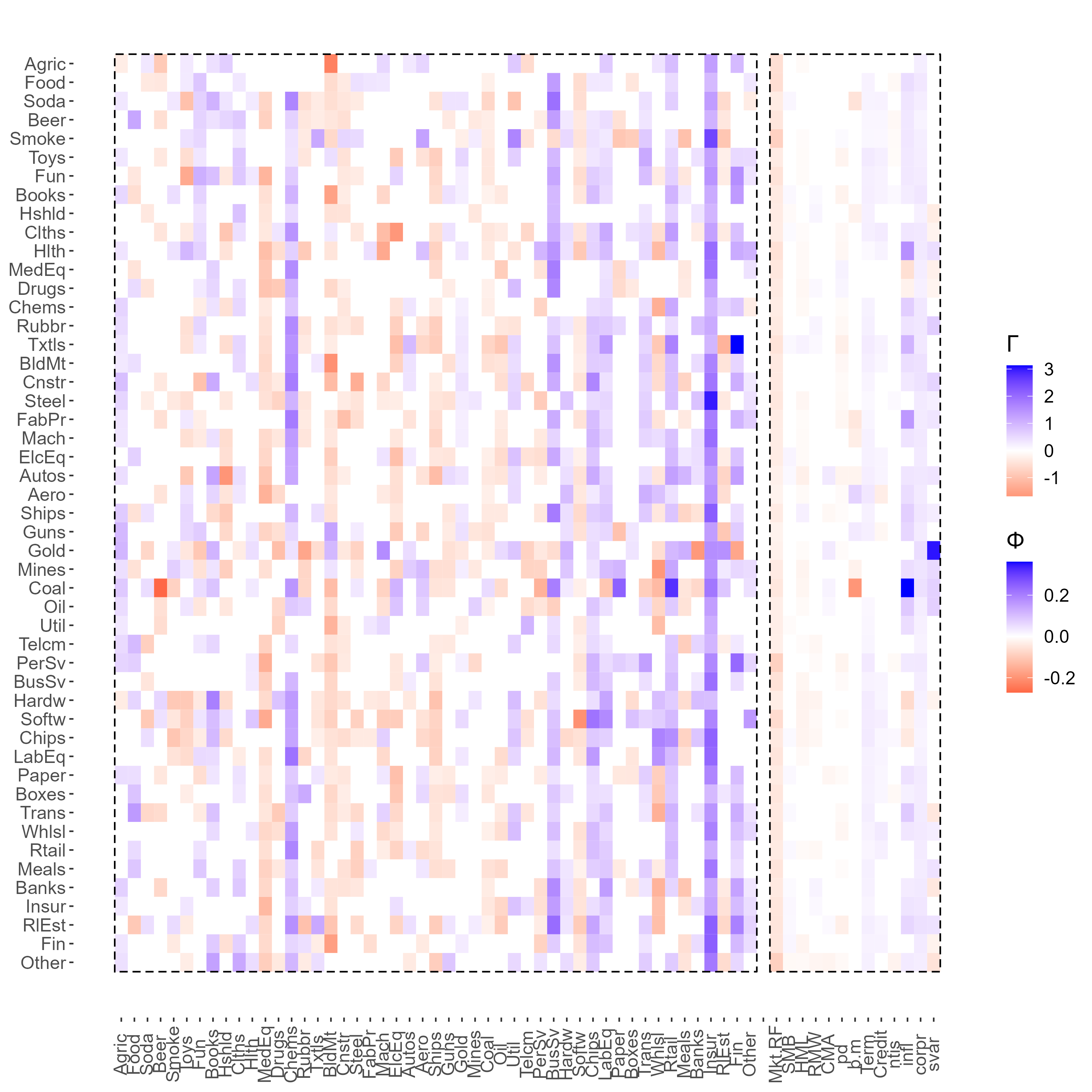}}
\subfigure[{\tt VB} w/ Lasso]{\includegraphics[width=.25\textwidth]{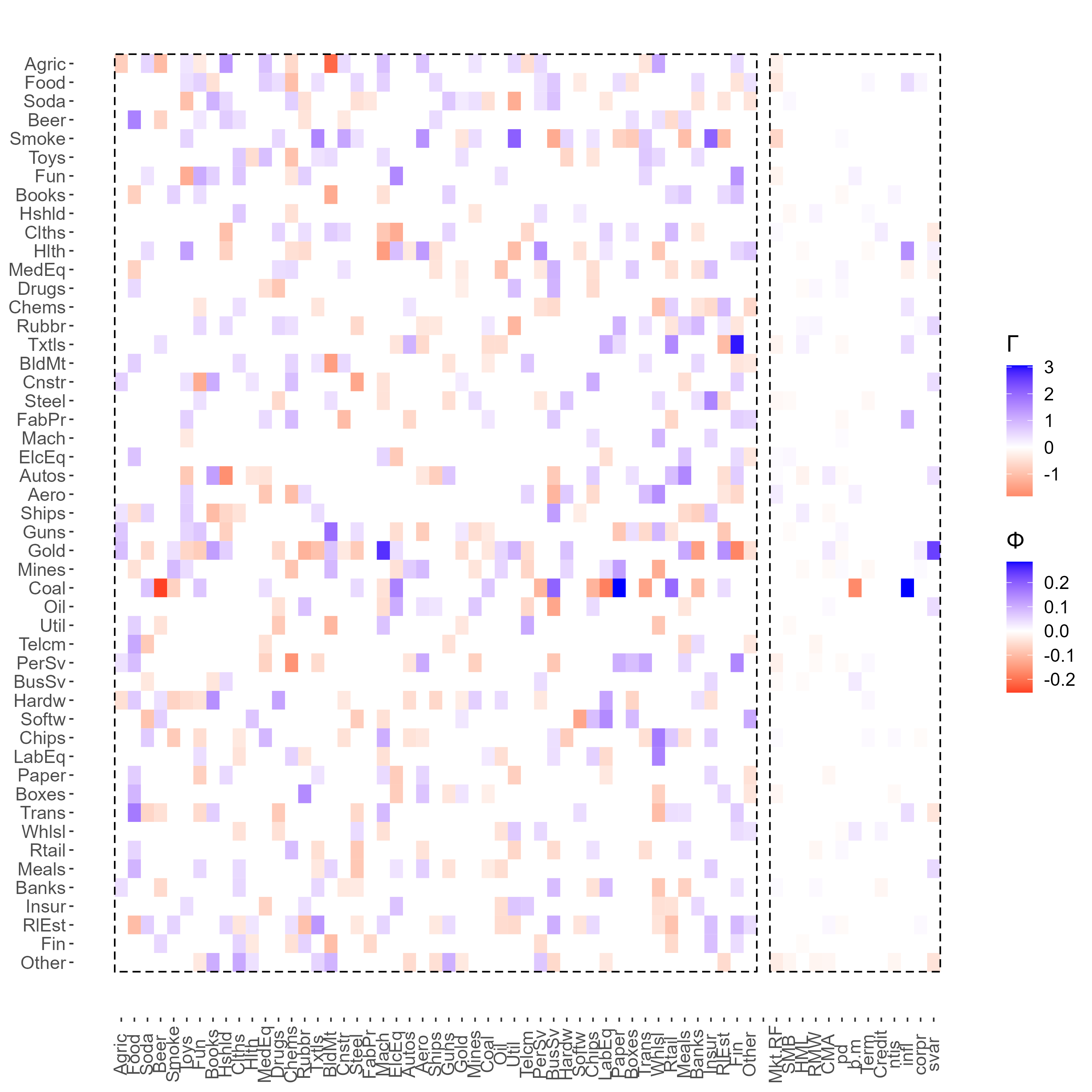}}\subfigure[{\tt VB} w/ Lasso + SV]{\includegraphics[width=.25\textwidth]{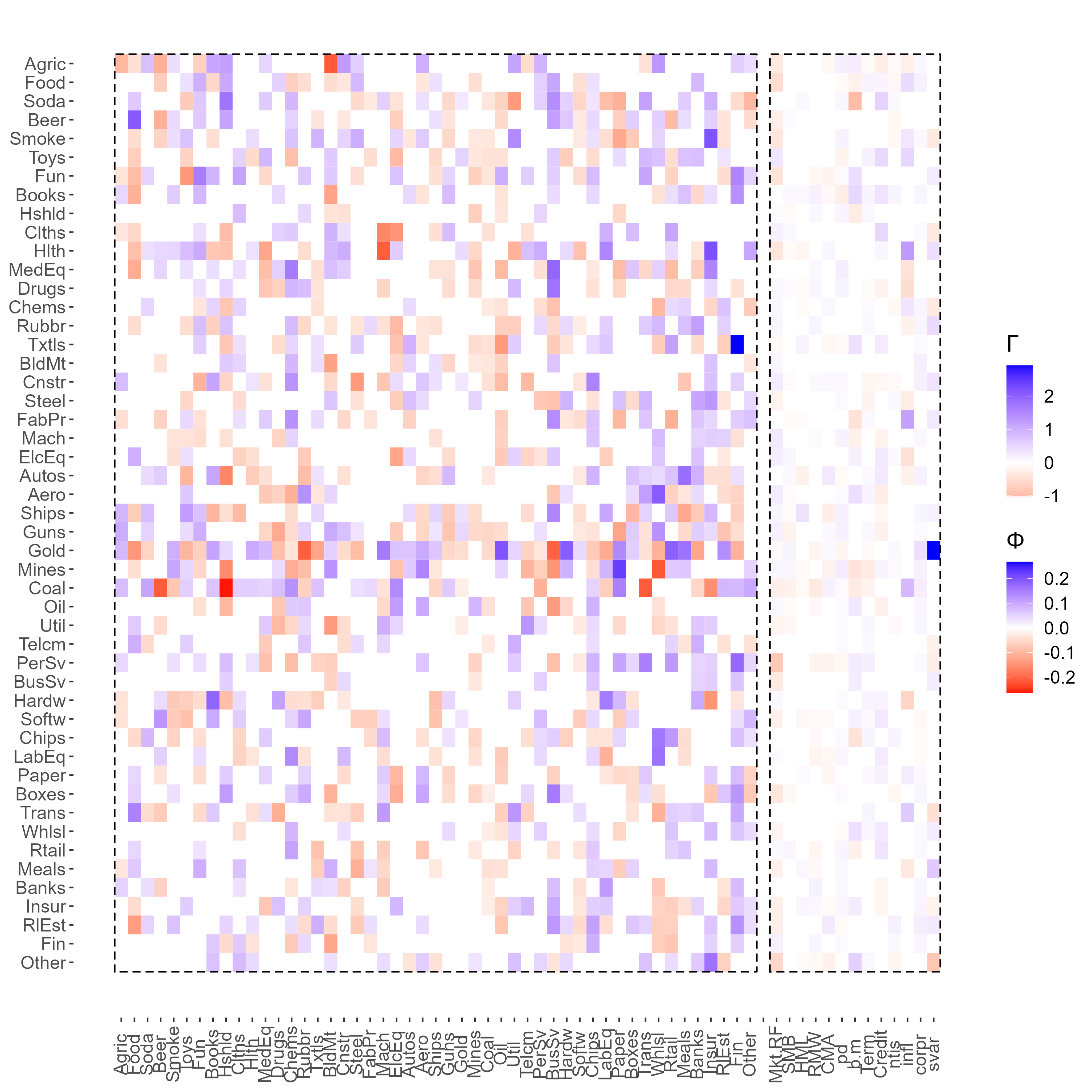}}\hspace{-2.5em}

\hspace{-2.5em}\subfigure[{\tt LMCMC} w/ HS]{\includegraphics[width=.25\textwidth]{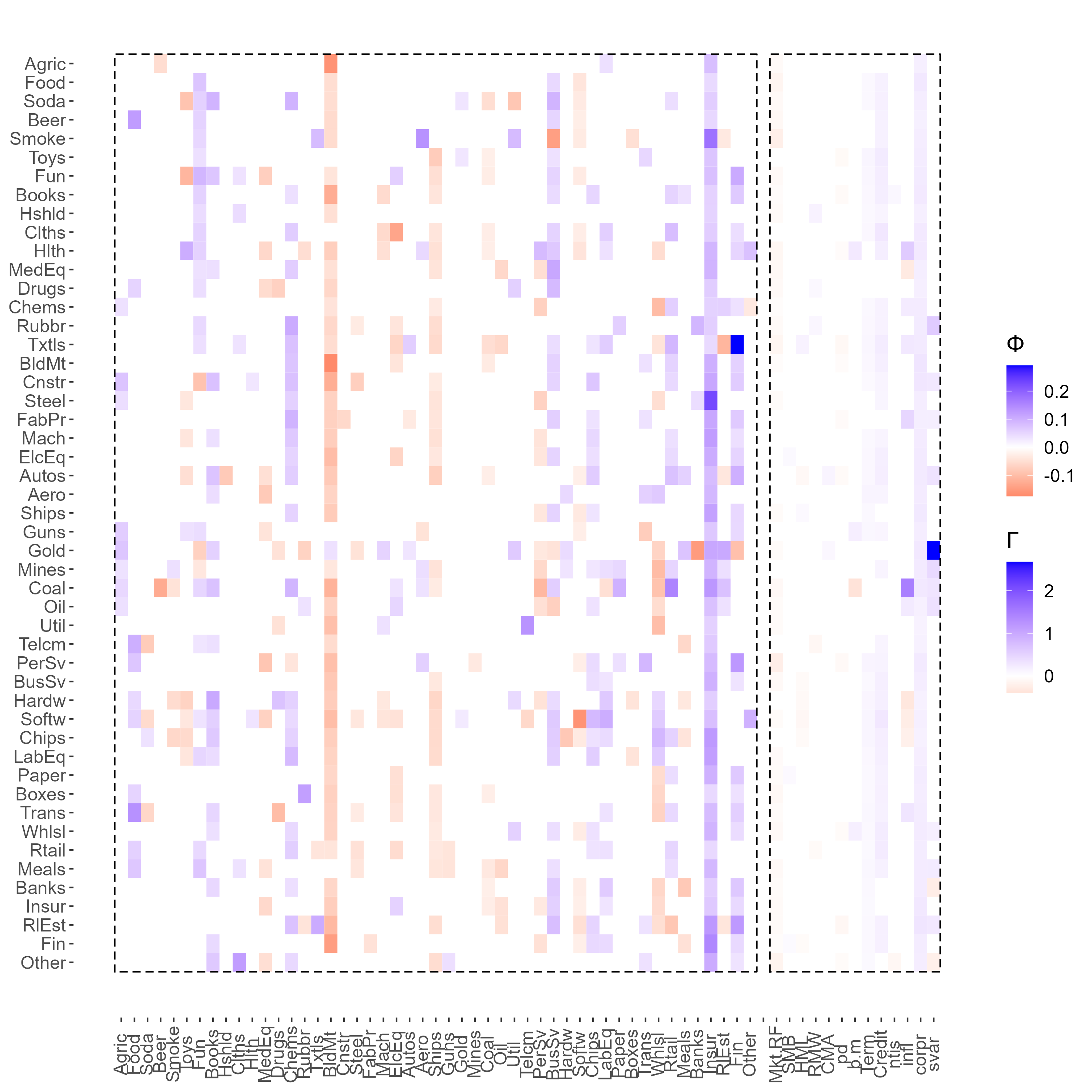}}\quad
\subfigure[{\tt LVB} w/ HS]{\includegraphics[width=.25\textwidth]{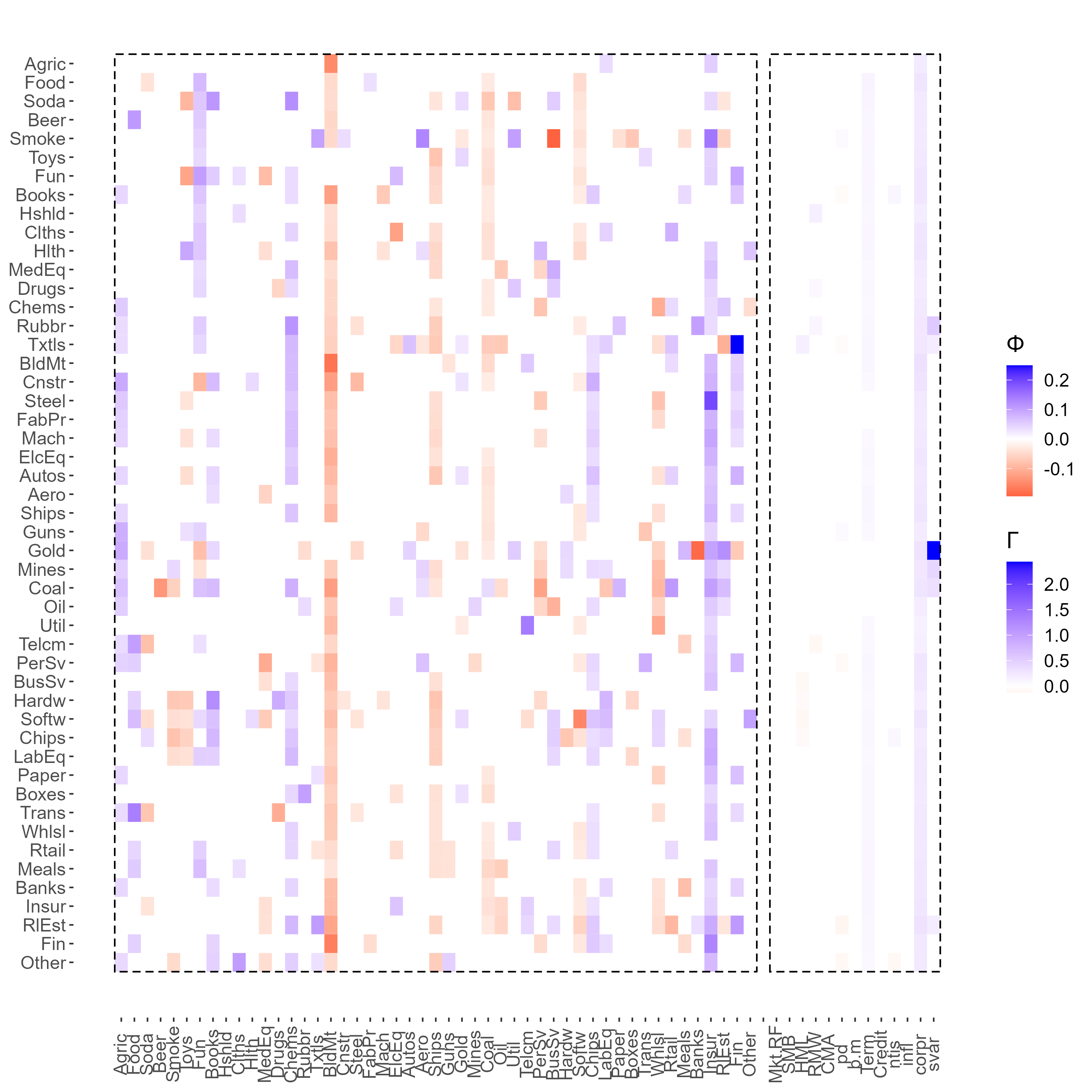}}
\subfigure[{\tt VB} w/ HS]{\includegraphics[width=.25\textwidth]{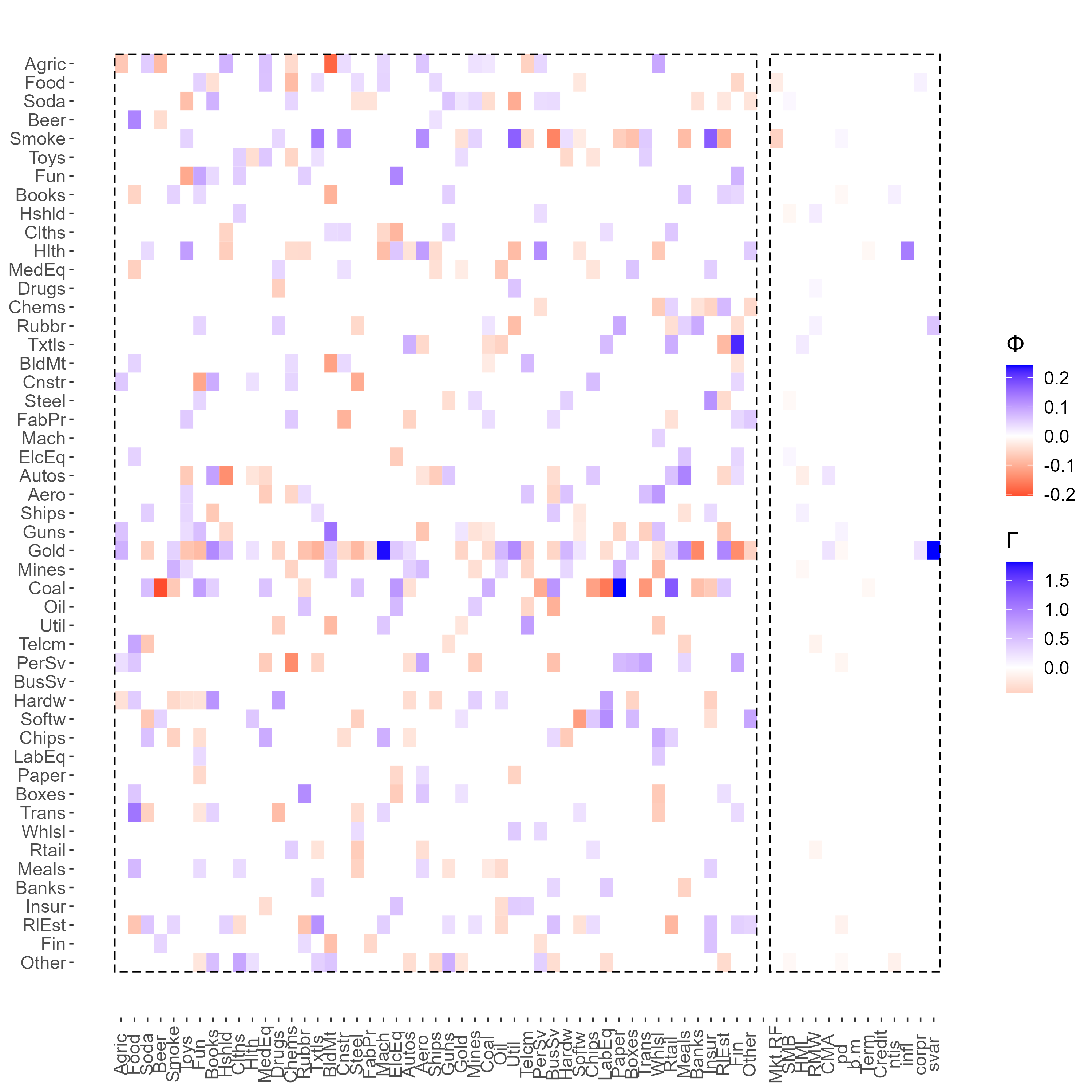}}\subfigure[{\tt VB} w/ HS + SV]{\includegraphics[width=.25\textwidth]{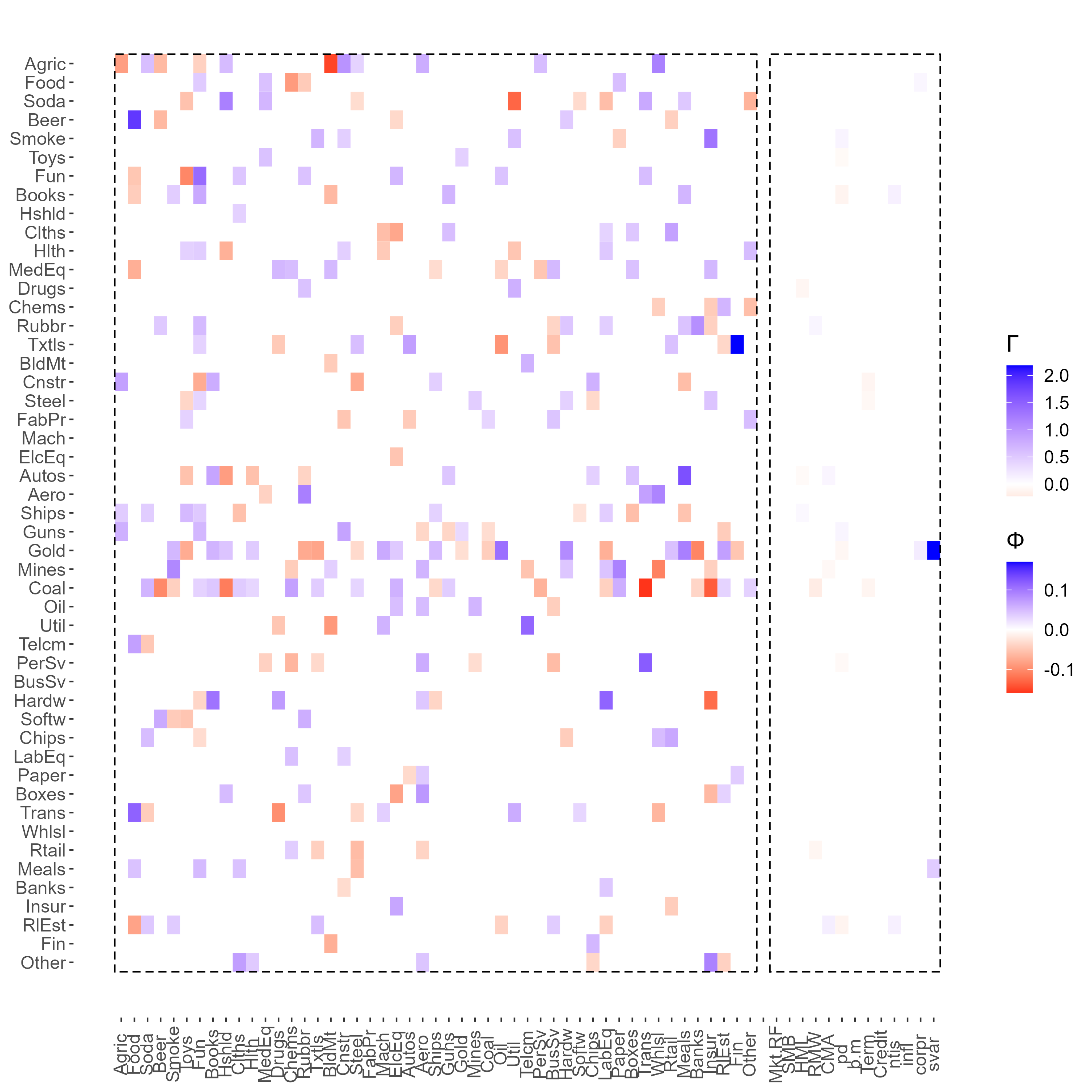}}\hspace{-2.5em}

\hspace{-2.5em}\subfigure[{\tt LMCMC} w/ NG]{\includegraphics[width=.25\textwidth]{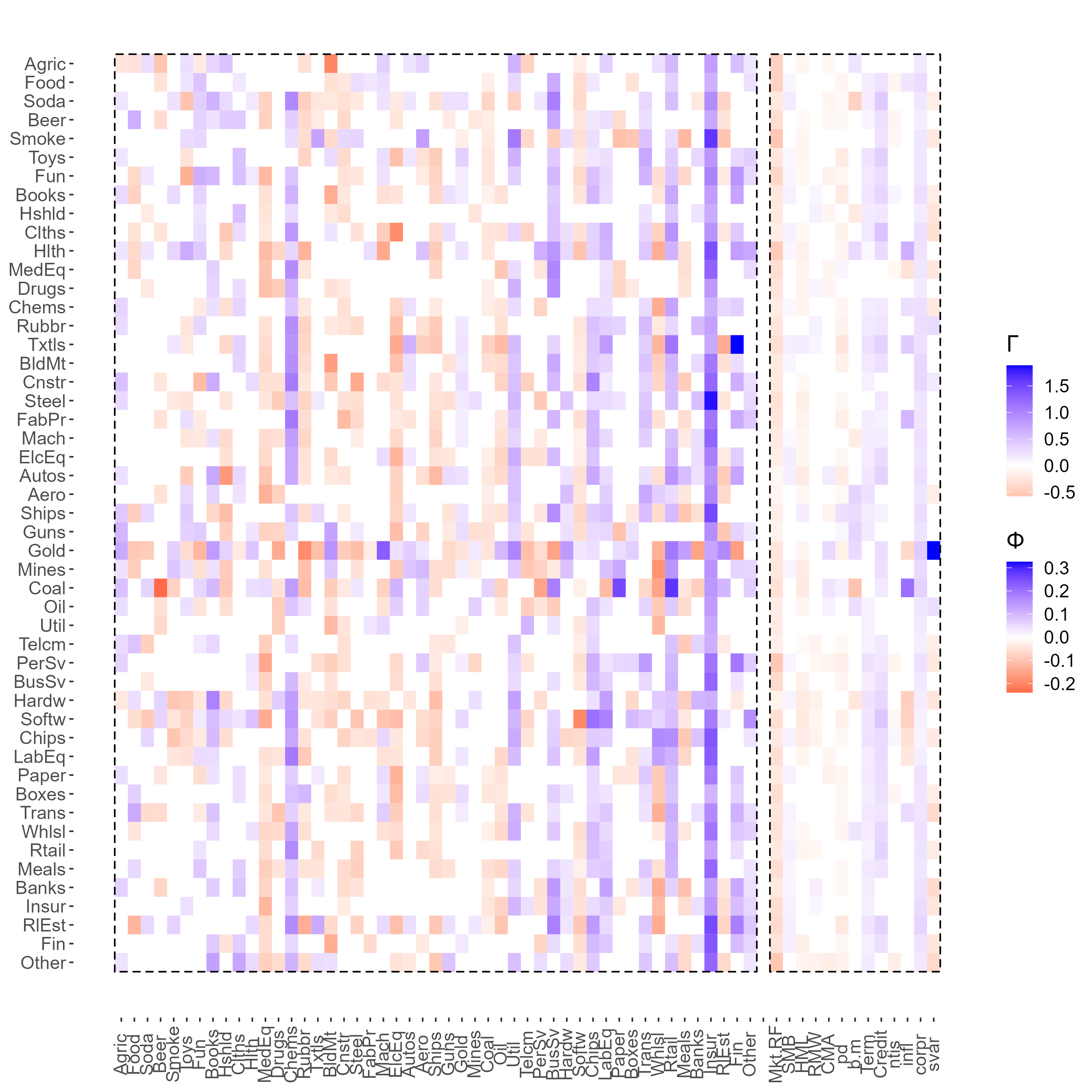}}\quad
\subfigure[{\tt LVB} w/ NG]{\includegraphics[width=.25\textwidth]{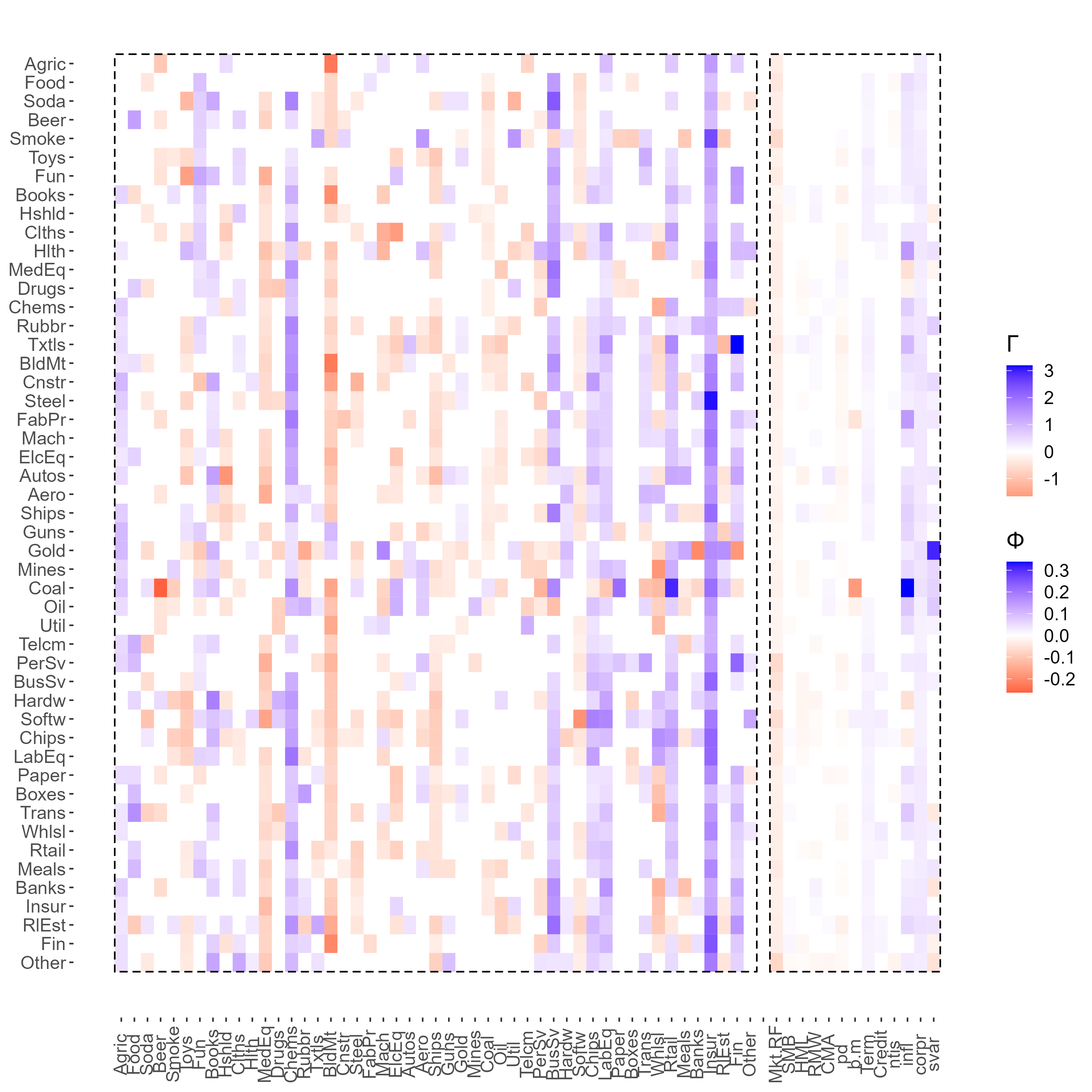}}
\subfigure[{\tt VB} w/ NG]{\includegraphics[width=.25\textwidth]{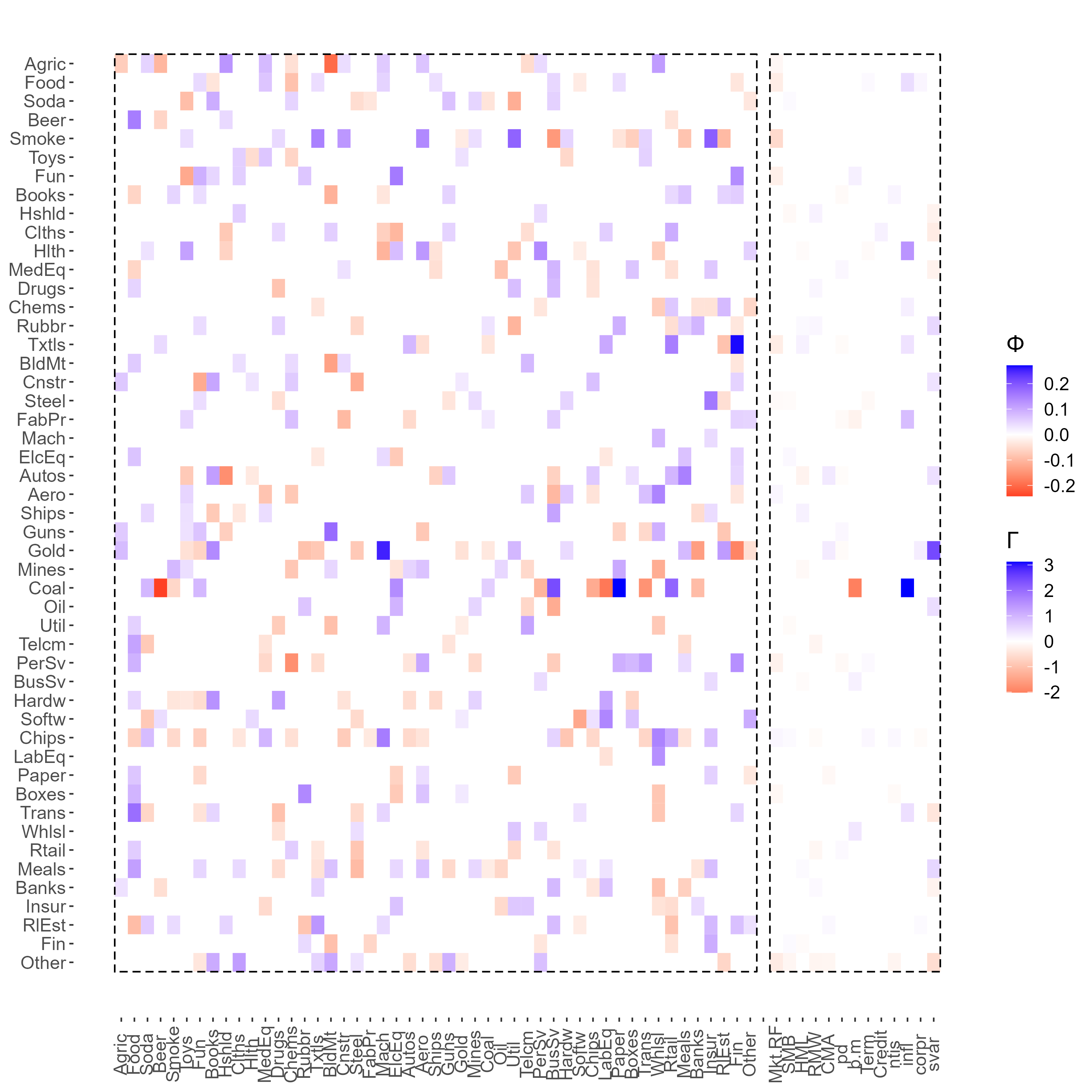}}\subfigure[{\tt VB} w/ NG + SV]{\includegraphics[width=.25\textwidth]{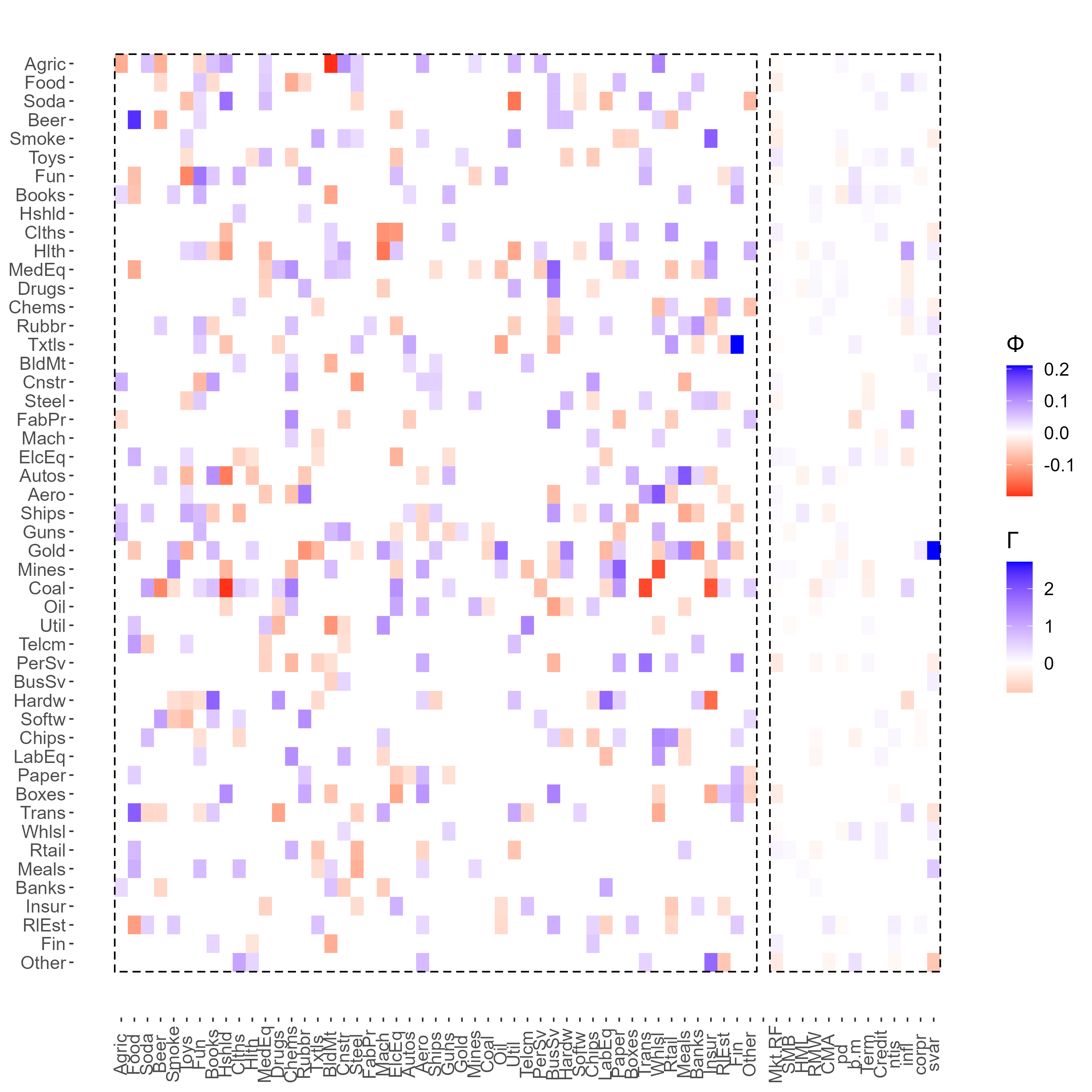}}\hspace{-2.5em}

	\caption{\small Variational Bayes estimates of the regression coefficients $\bTheta$ for different estimation methods. We report the estimates for the $d=49$ industry case obtained for all priors. We report the results for {\tt VB} with and without stochastic volatility.}
\label{fig:theta}
\end{figure}

The in-sample estimates highlight three key results. First, there are visible differences across shrinkage priors. For instance, the Horseshoe tend to shrink parameters more aggressively towards zero so that $\widehat{\bTheta}$ is more sparse compared to, for e.g., the adaptive Normal-Gamma. Second, consistent with \citet{gefang2023forecasting}, the estimates of the {\tt LMCMC} and {\tt LVB} tend to be closely related. Yet, these in-sample estimates are substantially different compared to our {\tt VB} approach. This is due to the re-parametrization $\widehat{\boldsymbol{\Theta}}=\widehat{\mathbf{L}}^{-1}\widehat{\mathbf{A}}$ in Eq.\eqref{eq:var1_orth_def2}; that is, the estimated $\widehat{\mathbf{A}}$ is not translation-invariant, unlike in our approach. Third, with the exception of the adaptive-Lasso prior, the estimates $\widehat{\boldsymbol{\Theta}}$ from {\tt VB} are remarkably stable between constant vs stochastic volatility specifications. 




\subsection{Out-of-sample forecasting accuracy}
\label{subsec:oos}
Intuitively, different estimates of $\bTheta$ should reflect in different conditional forecasts. To test this intuition we now compare the {\tt LMCMC}, {\tt LVB} and the {\tt VB} estimation approaches with and without stochastic volatility. For the sake of completeness, we also consider a series of univariate model specifications ({\tt U} henceforth), which corresponds to assuming conditional independence across industry portfolios. We consider a 360 months rolling window period for each model estimation; for instance for the 30-industry classification the out-of-sample period is from July 1957 to May 2020.  

Notice that given the recursive nature of the empirical implementation we do not consider the {\tt MCMC} approach of \citet{gruber2022forecasting}. This is because the computational cost would make such implementation prohibitive in practice, as discussed in the simulation study based on Figure \ref{fig:time15}. For instance, on a 2.5 GHz Intel Xeon W-2175 with 32GB of RAM and 14 cores it would take $20\ \text{min}\times 767\ \text{forecasts} \times 4\ \text{priors}= 61,360$ minutes, or 42 days, to implement the {\tt MCMC} approach for recursive forecasting for the 30 industry portfolios with constant volatility. The computational cost would be even more prohibitive when adding stochastic volatility and/or for the 49 industry portfolios. Appendix \ref{subsec:computational cost} provides an additional discussion on the computational costs of some of the existing MCMC approaches and the key relevance for a higher-frequency forecasting implementation such as ours. 

\paragraph{Point forecasts.} We begin by inspecting the accuracy of point forecasts for each industry based on the out-of-sample predictive R squared (see, e.g., \citealp{Goyal2008}),
\begin{equation*}
R^2_{j,oos}\left(\mathcal{M}_s\right) = 1-\frac{\sum_{t_0=2}^{T}\left(y_{jt}-\widehat{y}_{jt}\left(\mathcal{M}_s\right)\right)^2}{\sum_{t_0=2}^{T}\left(y_{jt}-\overline{y}_{jt}\right)^2},
\end{equation*}
where $t_0$ is the date of the first prediction, $\overline{y}_{jt}$ is the naive forecast from the recursive mean -- using the same rolling window of observations -- and $\widehat{y}_{jt}\left(\mathcal{M}_s\right)$ is the conditional mean returns for industry $j=1,\ldots,d$ for a given model $\mathcal{M}_s$.

\begin{figure}[h!]
	\centering
\hspace{-1em}\subfigure[$R_{j,oos}\left(\mathcal{M}_s\right)^2$ across 30 industry portfolios]{\includegraphics[width=.42\textwidth]{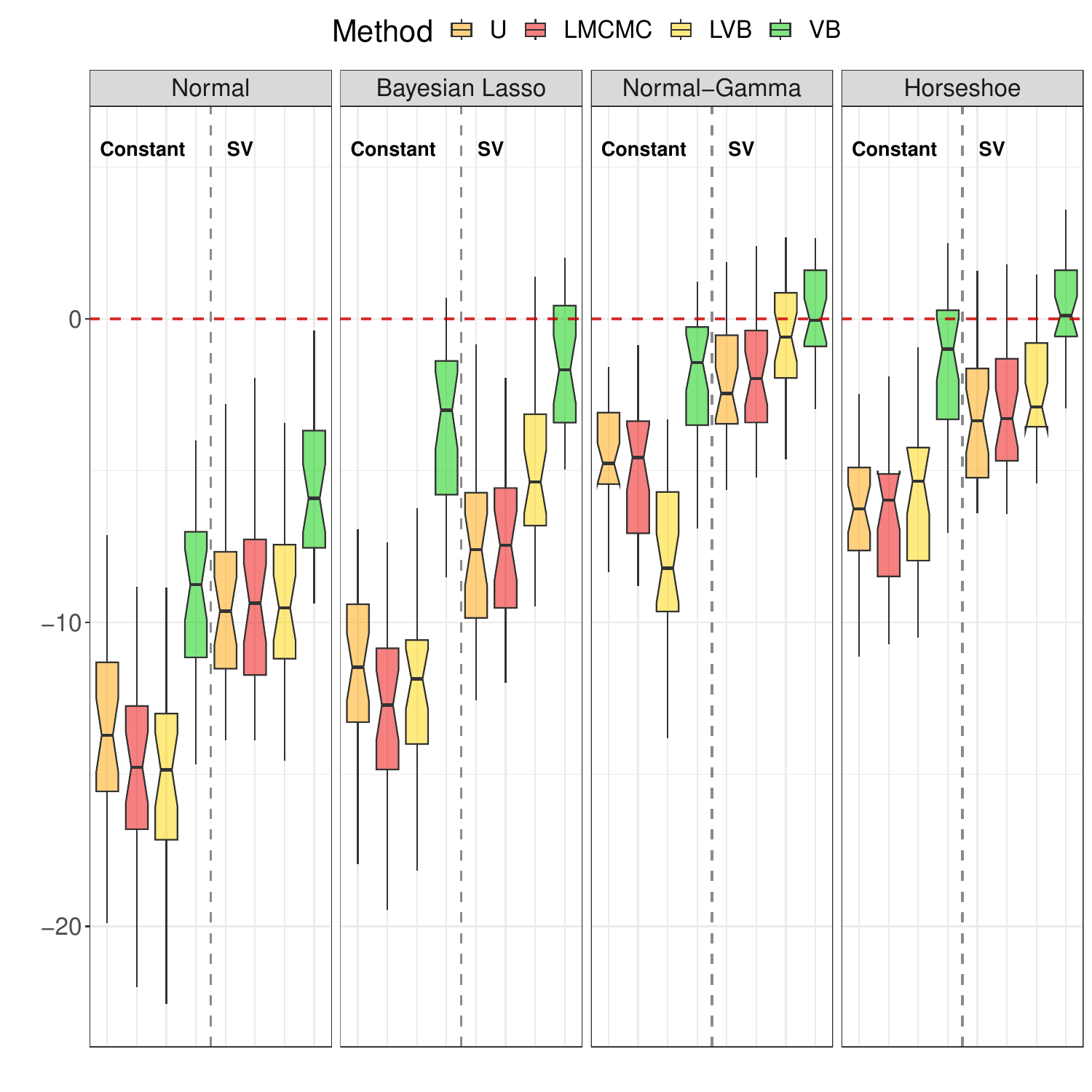}}\hspace{1em}
\subfigure[Portfolios for which $R_{j,oos}^2\left(\mathcal{M}_s\right)>0$]{\includegraphics[width=.42\textwidth]{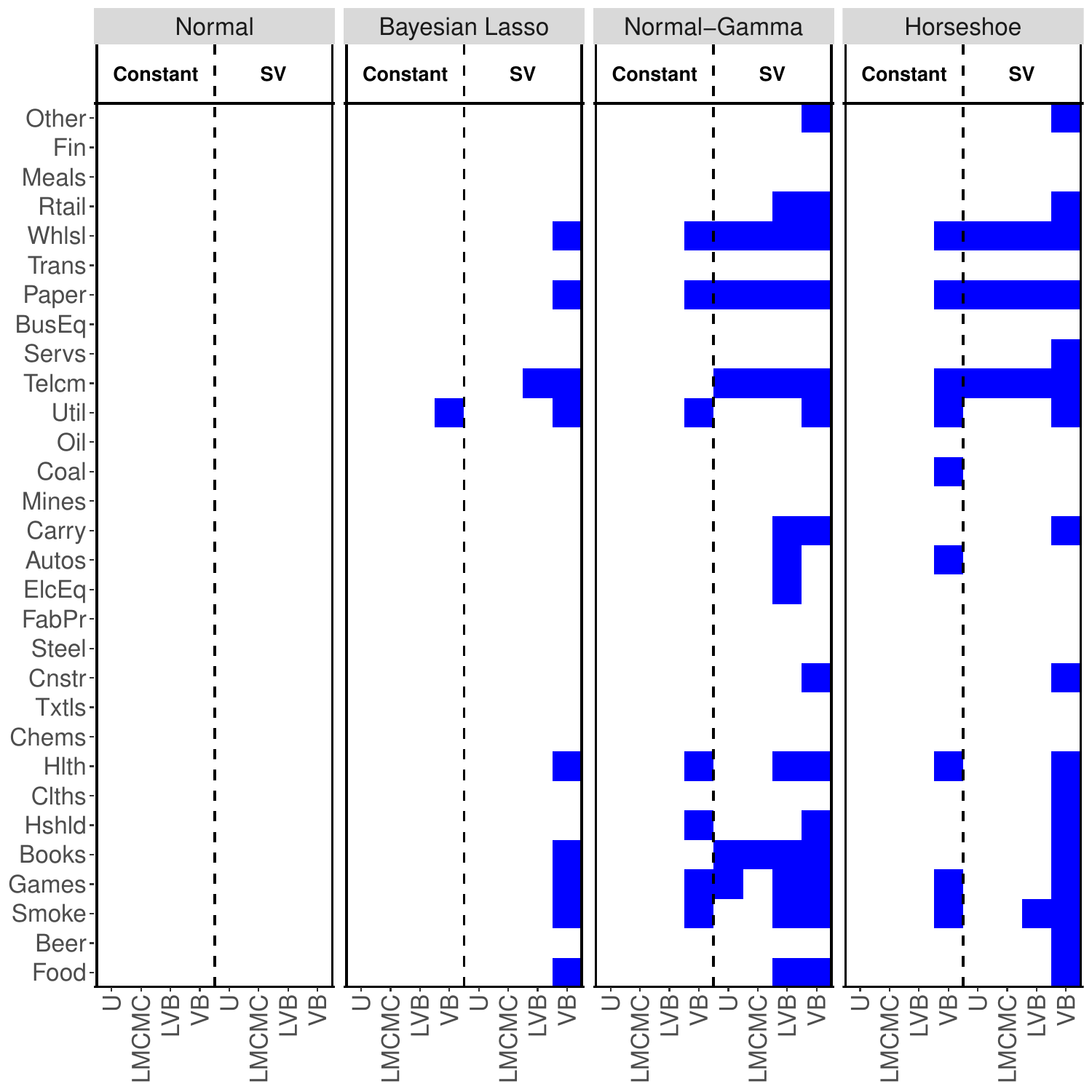}}

\hspace{-1em}\subfigure[$R_{j,oos}^2\left(\mathcal{M}_s\right)$ across 49 industry portfolios]{\includegraphics[width=.42\textwidth]{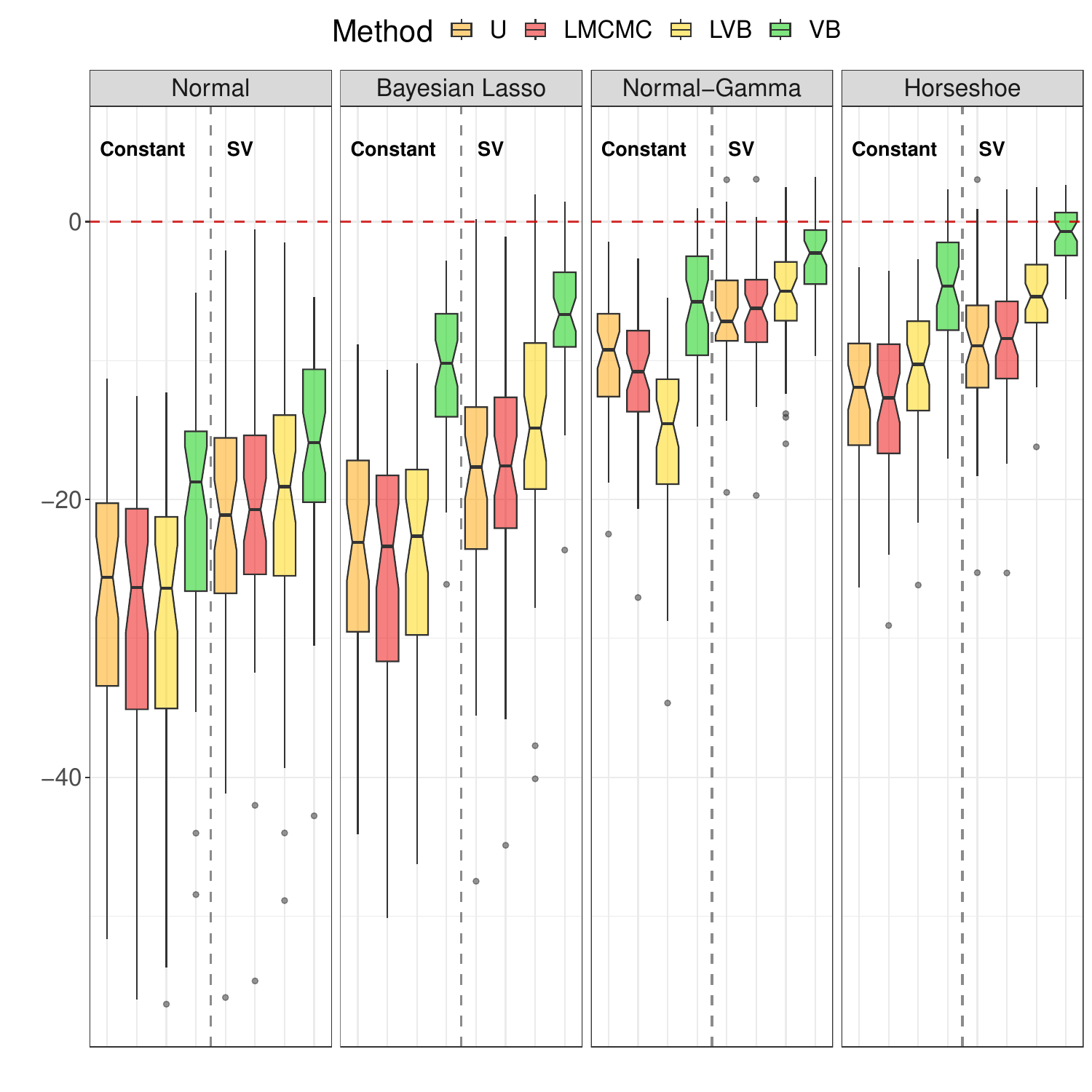}}\hspace{1em}
\subfigure[Portfolios for which $R_{j,oos}^2\left(\mathcal{M}_s\right)>0$]{\includegraphics[width=.42\textwidth]{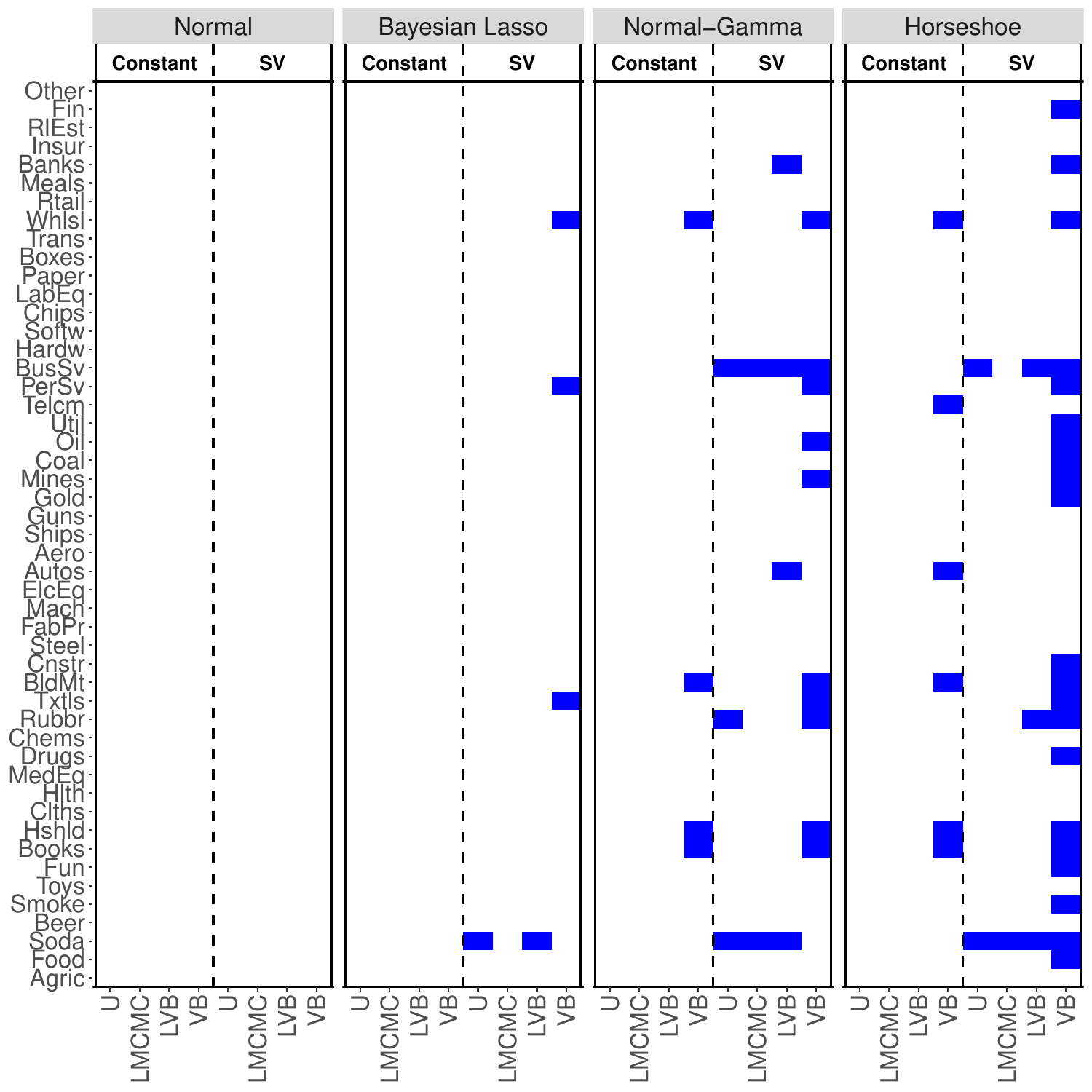}}

\caption{Left panels report the $R_{j,oos}^2\left(\mathcal{M}_s\right)$ (in \%) across industry portfolios. Right panels report the industries for which a given model can generate $R_{j,oos}^2\left(\mathcal{M}_s\right)>0$. The top (bottom) panels report the results for 30 (49) industry portfolios.}
	\label{fig:R2oos}
\end{figure}

{\color{black}The left panels of Figure \ref{fig:R2oos} show the box charts with the distribution of the $R_{j,oos}^2$ across $j=1,\ldots,d$ industries. For a given sub-plot the results for the Normal, Bayesian Lasso, Normal-Gamma and Horseshoe priors are reported from the left to the right. Within each panel of a sub-plot, the forecasting results for the {\tt U}, {\tt LMCMC}, {\tt LVB}, and {\tt VB} estimates are color coded in orange, red, yellow, and green (from left to right), respectively. The vertical dashed line within each panel separates between constant and stochastic volatility specifications. Based on the same separation across methods and priors, the right panels of Figure \ref{fig:R2oos} report a breakdown of the industries for which the corresponding $R_{j,oos}^2\left(\mathcal{M}_s\right)>0$.}


The out-of-sample $R_{j,oos}^2\left(\mathcal{M}_s\right)$ tend to be mostly negative across estimation methods and shrinkage priors. This is consistent with the existing evidence on stock returns predictability: a simple naive forecast based on a rolling sample mean represents a challenging benchmark to beat (see, e.g., \citealp{campbell2007}). However, our variational inference approach substantially improves upon univariate regressions, as well as upon the {\tt LMCMC} and {\tt LVB} methods, which are both based on a structural VAR representation. 

For instance, our {\tt VB} with stochastic volatility generates a positive $R_{j,oos}^2\left(\mathcal{M}_s\right)$ for more than half of the 30 industry portfolios based on the adaptive Normal-Gamma and the Horseshoe. This compares to 4 (adaptive Normal-Gamma) and 3 (Horseshoe) positive $R_{j,oos}^2\left(\mathcal{M}_s\right)$ obtained from {\tt LMCMC} with stochastic volatility. The gap further increases within the 49-industry classification; our {\tt VB} method is virtually the only approach that can systematically generate positive $R_{j,oos}^2\left(\mathcal{M}_s\right)$ across industries. Although concentrated on the Horseshoe prior, the out-performance of our method relative to both {\tt LMCMC} and {\tt VB} holds across different priors. 

\paragraph{Density forecasts.} {\color{black}We follow \citet{fisher2020optimal} and assess the accuracy of the density forecasts across priors and estimation methods based on the average log-score (ALS) differential with respect to a ``no-predictability'' benchmark,
\begin{align}
\text{ALS}_{j}\left(\mathcal{M}_s\right) & = \frac{1}{T-t_0}\sum_{t_0=2}^{T}\left(\ln{S_{jt}\left(\mathcal{M}_s\right)}-\ln{\overline{S}_{jt}}\right),\label{eq:logscore}   
\end{align}

where $\ln{S_{jt}}\left(\mathcal{M}_s\right)$ denotes the log-score at time $t$ for industry $j$ obtained by evaluating a Normal density with the conditional mean and variance forecast from the model $\mathcal{M}_s$. Consistent with the rationale of $R_{j,oos}^2\left(\mathcal{M}_s\right)$, the log-score for the no-predictability benchmark $\ln{\overline{S}_{j,t}}$ is constructed by evaluating a Normal density based on recursive mean and variance.} 

\begin{figure}[h!]
	\centering
\subfigure[$\text{ALS}_{j}\left(\mathcal{M}_s\right)$ for the 30 industry portfolios]{\includegraphics[width=.42\textwidth]{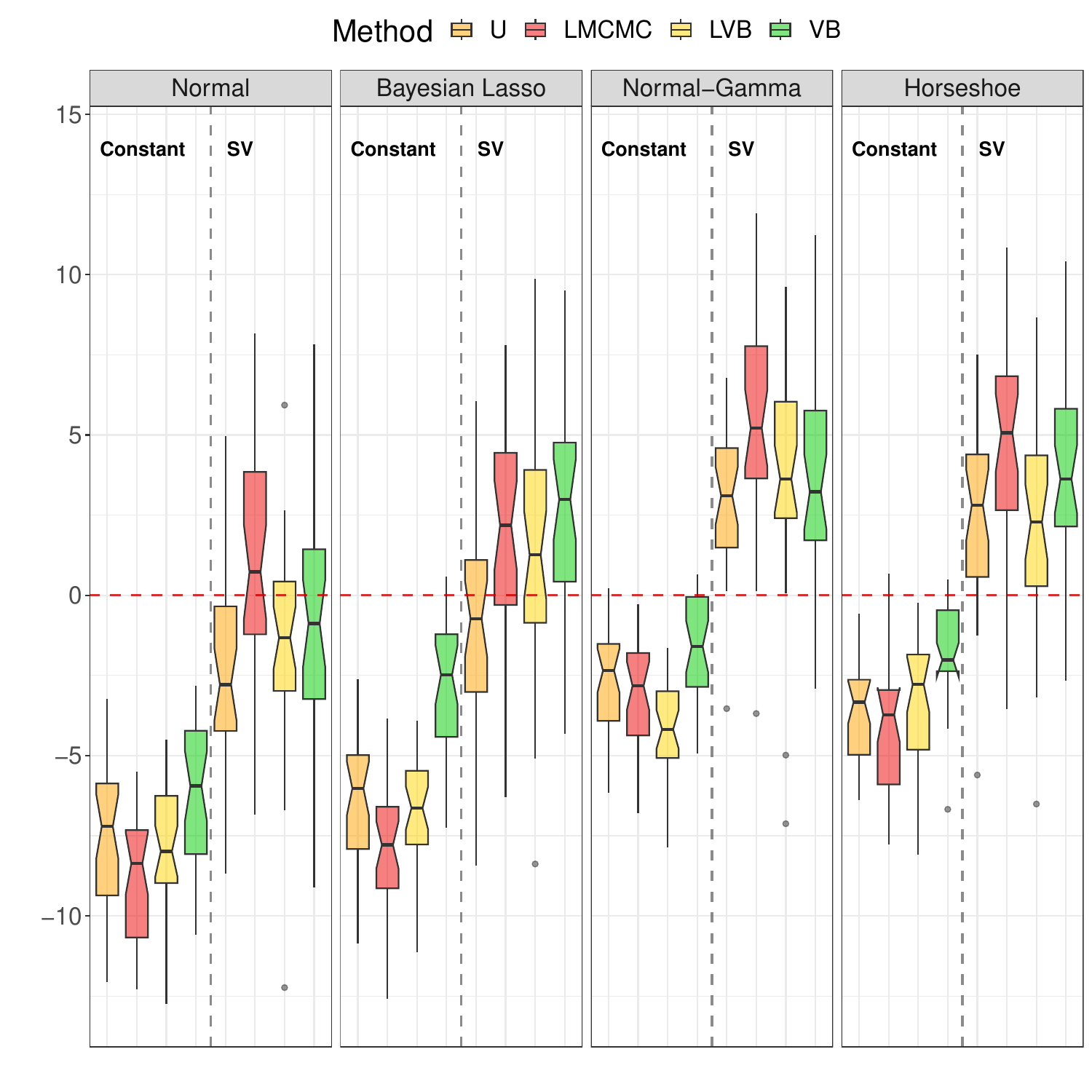}}\hspace{1em}
\subfigure[Portfolios for which $\text{ALS}_{j}\left(\mathcal{M}_s\right)>0$]{\includegraphics[width=.42\textwidth]{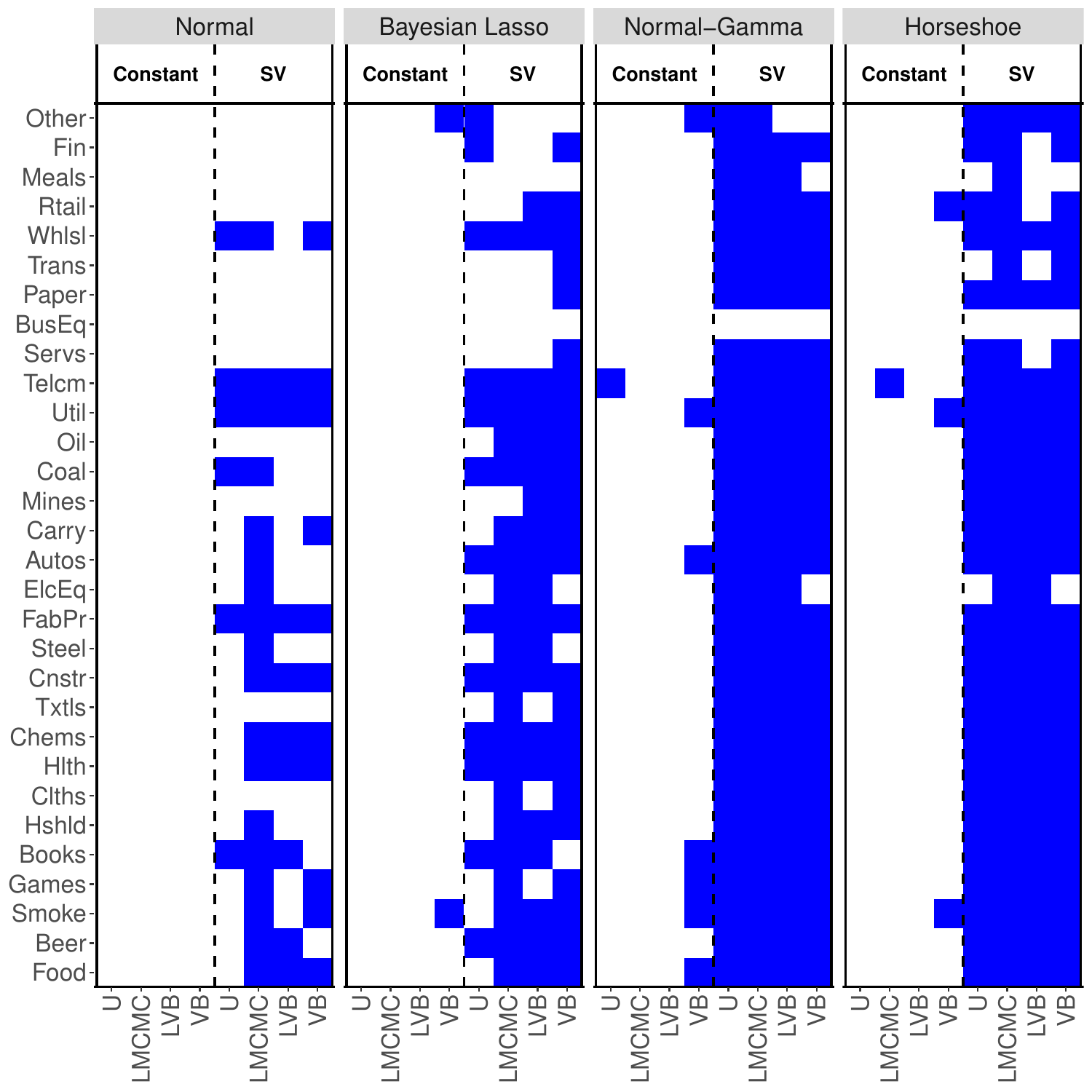}}

\subfigure[$\text{ALS}_{j}\left(\mathcal{M}_s\right)$ for the 49 industry portfolios]{\includegraphics[width=.42\textwidth]{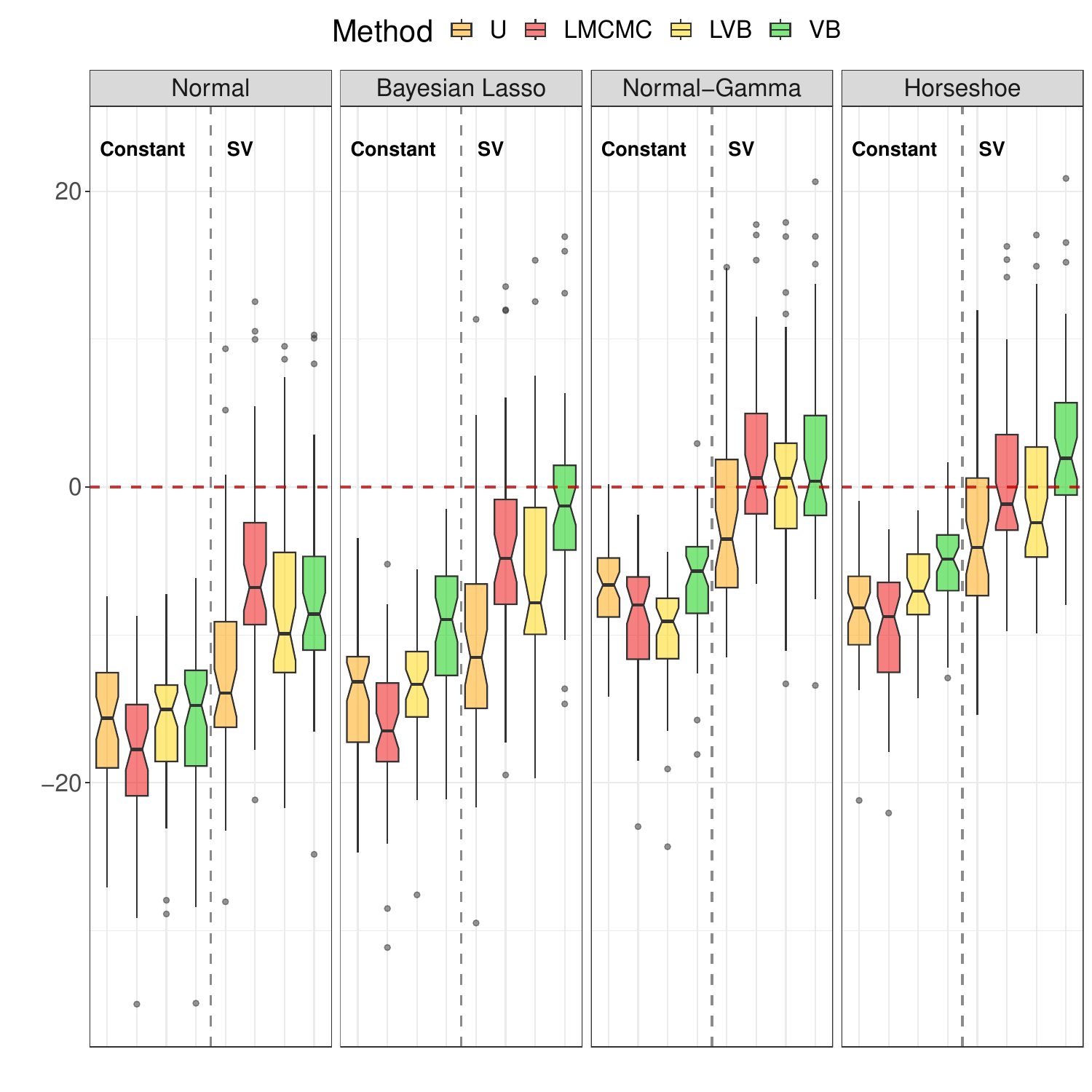}}\hspace{1em}
\subfigure[Portfolios for which $\text{ALS}_{j}\left(\mathcal{M}_s\right)>0$]{\includegraphics[width=.42\textwidth]{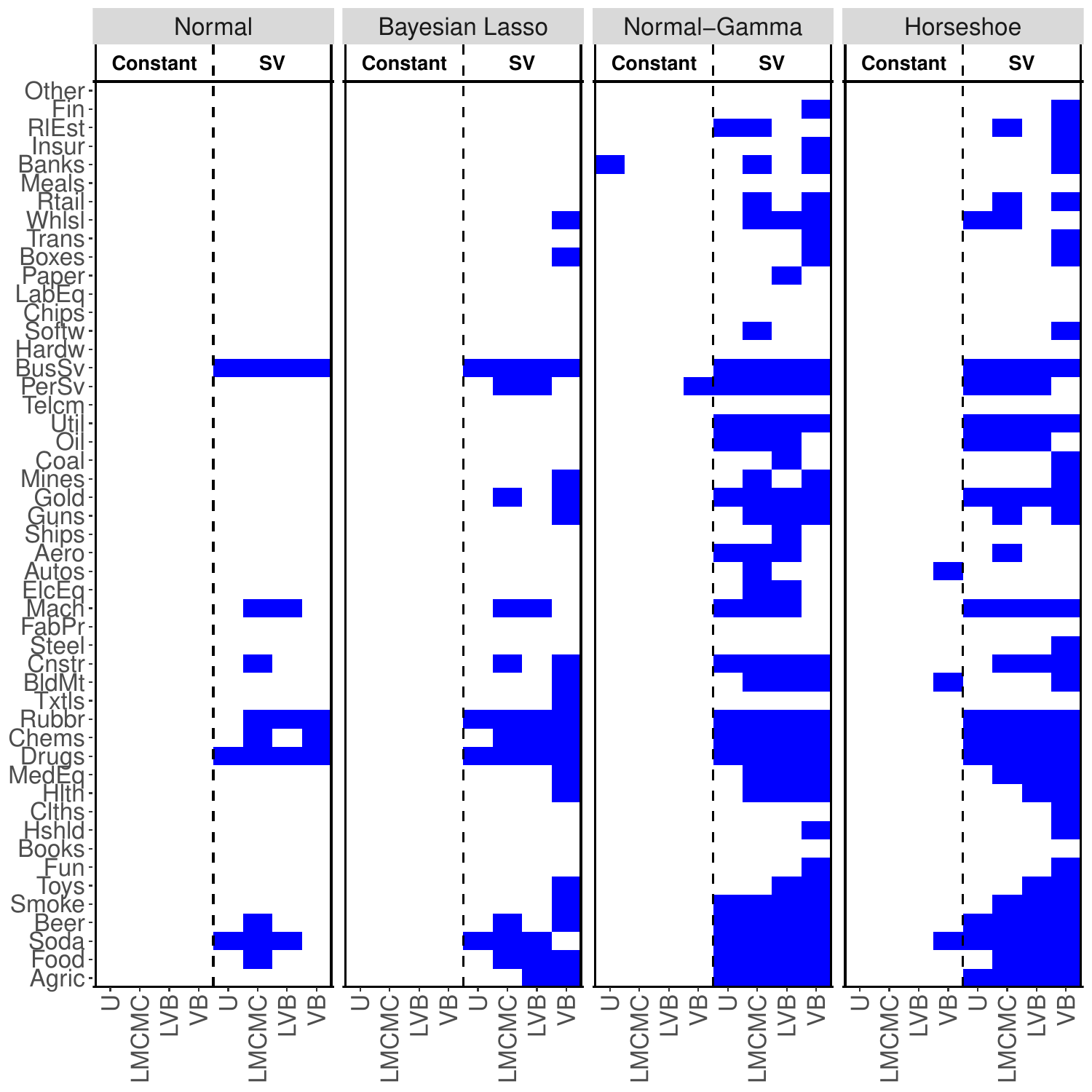}}

\caption{Left panels report the log-score differential across industry portfolios. Right panels report the industries for which a given model can generate positive log-score differential. The top (bottom) panels report the results for 30 (49) industry portfolios.}
	\label{fig:logscore}
\end{figure}

Figure \ref{fig:logscore} reports the results. The labeling is the same as in Figure \ref{fig:R2oos}. \textcolor{black}{Not surprisingly, we find that by adding stochastic volatility the accuracy of density forecasts substantially improves across priors and estimation methods.} For instance, our {\tt VB} method with stochastic volatility generate positive log-score differentials for almost all of the portfolios for the 30 industry classification and for more than half of the 49 industry portfolios. Interestingly, when it comes to density forecasts rather than modeling expected returns, the \citet{gefang2023forecasting} variational method built on a structural VAR representation performs on par with our {\tt VB} method. This is likely due to stochastic volatility alone, since our {\tt VB} still stands out within the constant volatility specifications. More generally, our {\tt VB} approach outperforms the competing estimation methods under all prior specifications.

\paragraph{Returns predictability over the business cycle.}

Existing literature suggests that expected returns are counter-cyclical and that returns predictability is more concentrated during period of economic contractions vs expansions (see, e.g., \citealp{Goyal2010}). Thus, we investigate if the forecasting performance of our modeling framework changes over the business cycle. More precisely, we split the data into recession and expansionary periods using the NBER dates of peaks and troughs. This information is considered {\it ex-post} and is not used at any time in the estimation and/or forecasting process. We compute the corresponding $R_{j,oos}^2\left(\mathcal{M}_s\right)$ for the recession periods only. 

\begin{figure}[h!]
	\centering
\subfigure[$R_{j,oos}^2\left(\mathcal{M}_s\right)>0$ for 30-industry portfolios]{\includegraphics[width=.42\textwidth]{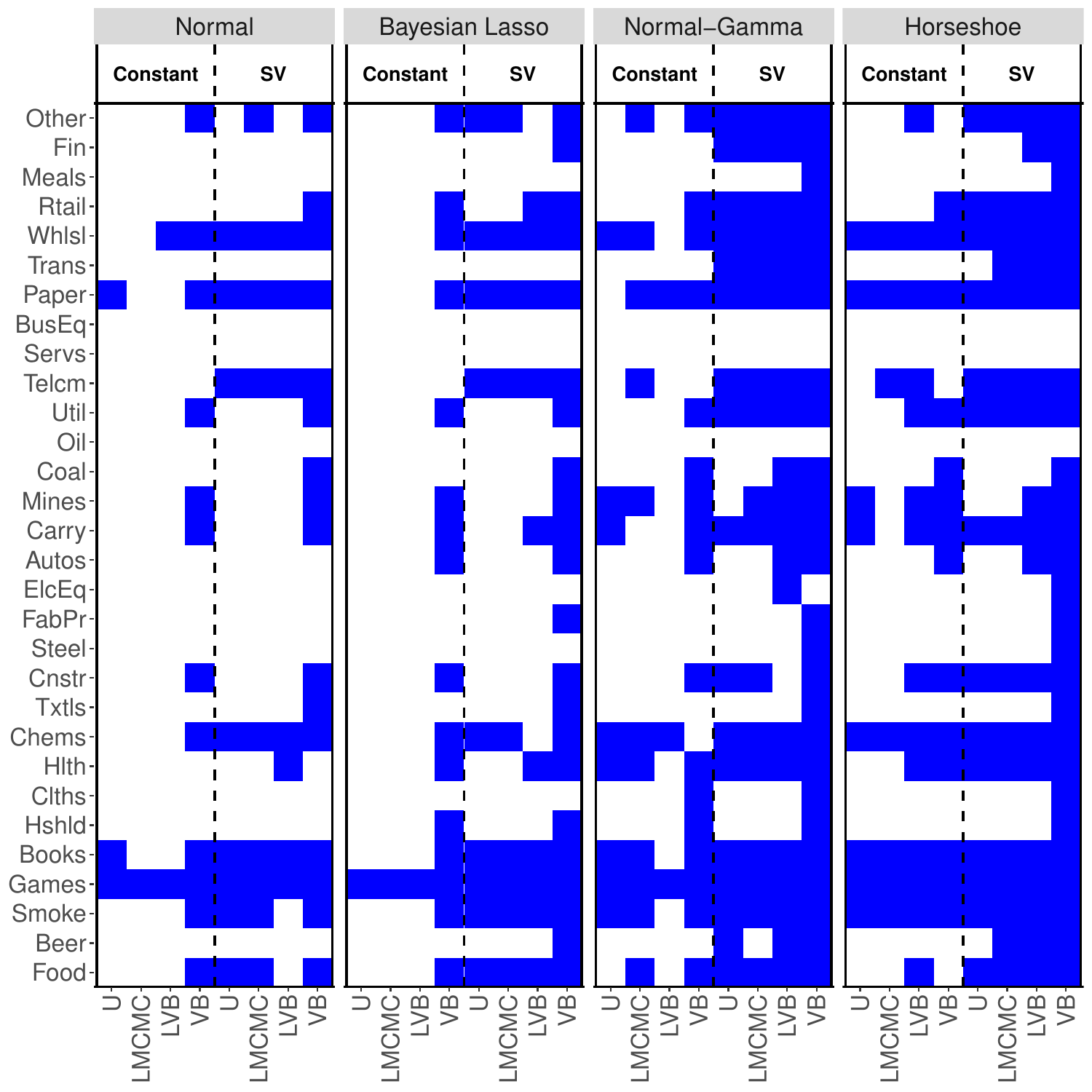}}\hspace{1em}
\subfigure[$R_{j,oos}^2\left(\mathcal{M}_s\right)>0$ for 49-industry portfolios]{\includegraphics[width=.42\textwidth]{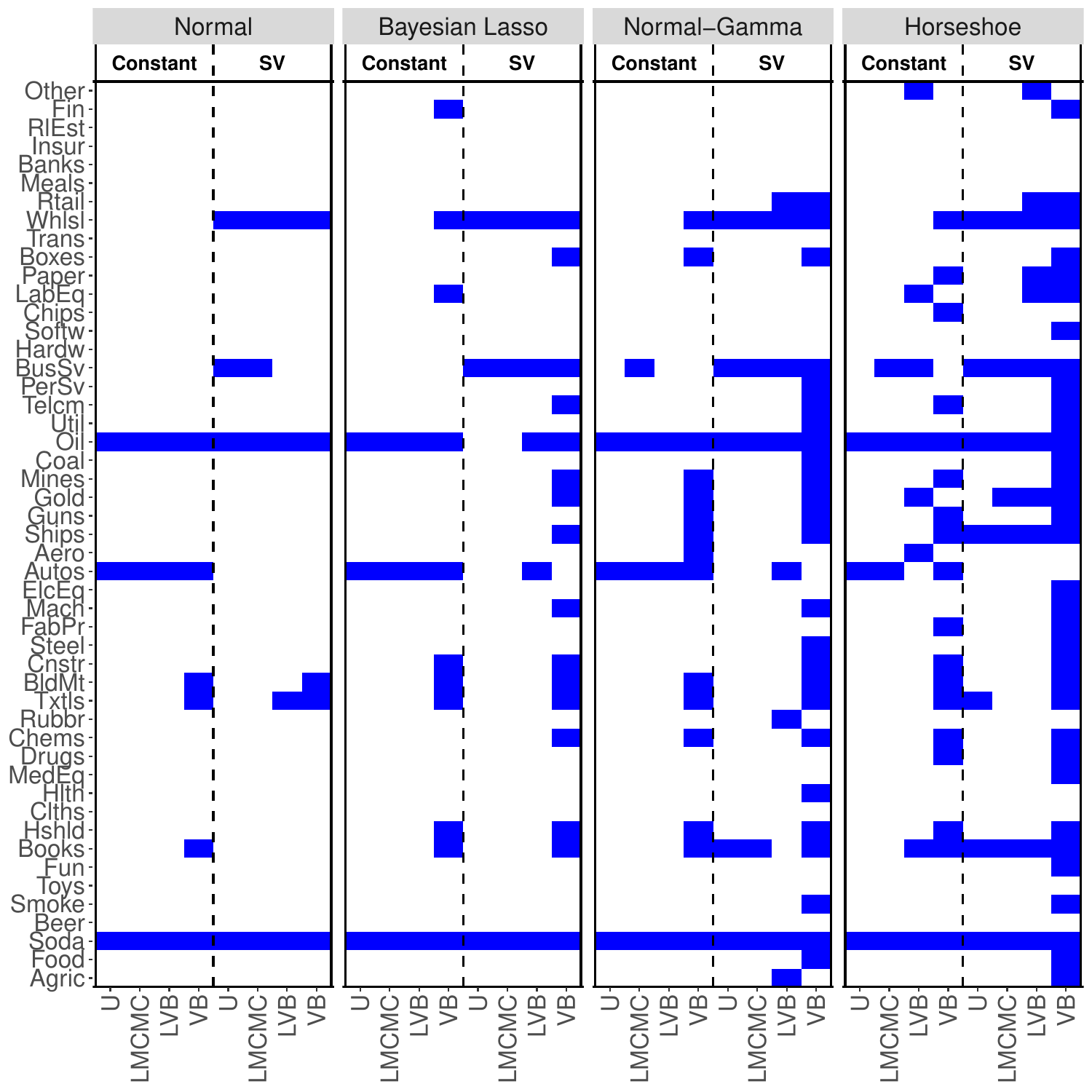}}


\caption{The figure reports the industries for which $R_{j,oos}^2\left(\mathcal{M}_s\right)>0$. The left (right) panel report the results for 30 (49) industry portfolios.}
	\label{fig:R2oos recession}
\end{figure}

Figure \ref{fig:R2oos recession} reports the industries for which $R_{j,oos}^2\left(\mathcal{M}_s\right)>0$ for both the 30 (left panel) and the 49 (right panel) industry classification. The corresponding cross-sectional distribution of the $R_{j,oos}^2\left(\mathcal{M}_s\right)$ and the relative log-scores are reported in Appendix \ref{app:more_emp}. The labeling of Figure \ref{fig:R2oos recession} is the same as in Figure \ref{fig:R2oos}. By comparing Figure \ref{fig:R2oos recession} with the results for the full sample, it suggests that the accuracy of the predictions substantially improves across methods and priors. Nevertheless, our {\tt VB} method outperforms the naive forecast from the rolling mean for a larger fraction of industry portfolios compared to other methods, in particular when stochastic volatility is considered. The difference between the recession and the full-sample performance persists when considering the 49 industry classification, especially for the adaptive Normal-Gamma and the Horseshoe prior.

\subsection{Economic evaluation}

A positive predictive performance does not necessarily translate into economic value. However, in practice an investor is obviously keenly interested in the economic value of returns predictability, perhaps even more than the statistical performance. Hence, it is of paramount importance to evaluate the extent to which apparent gains in predictive accuracy translates into better investment performances. 

Following existing literature (see, e.g., \citealp{Goyal2008,Goyal2010}), we consider a representative investor with a single-period horizon and mean-variance preferences who allocates her wealth between an industry portfolio and a risk-free asset. {\color{black}Thus, the investor optimal allocation to stocks for period $t+1$ based on information at time $t$ is given by $w_{jt} = \frac{1}{\gamma}\frac{\widehat{y}_{jt}}{\widehat{\nu}^{-1}_{jt}}$, where $\widehat{y}_{jt}$ represents the returns conditional mean forecast for industry $j=1,\ldots,d$ and $\widehat{\nu}^{-1}_{jt}$ the corresponding volatility forecast at time $t$. We also constraint the weights for each of the industry to $-0.5 \le  w_{jt} \le 1.5$ to prevent extreme short-sales and leverage positions. We assume a risk aversion coefficient of $\gamma=5$ (see, e.g., \citealp{Dangl:Halling:2012}).} 

\begin{figure}[h!]
	\centering
\subfigure[Gain$\left(\mathcal{M}_s\right)$ for 30-industry classification]{\includegraphics[width=.42\textwidth]{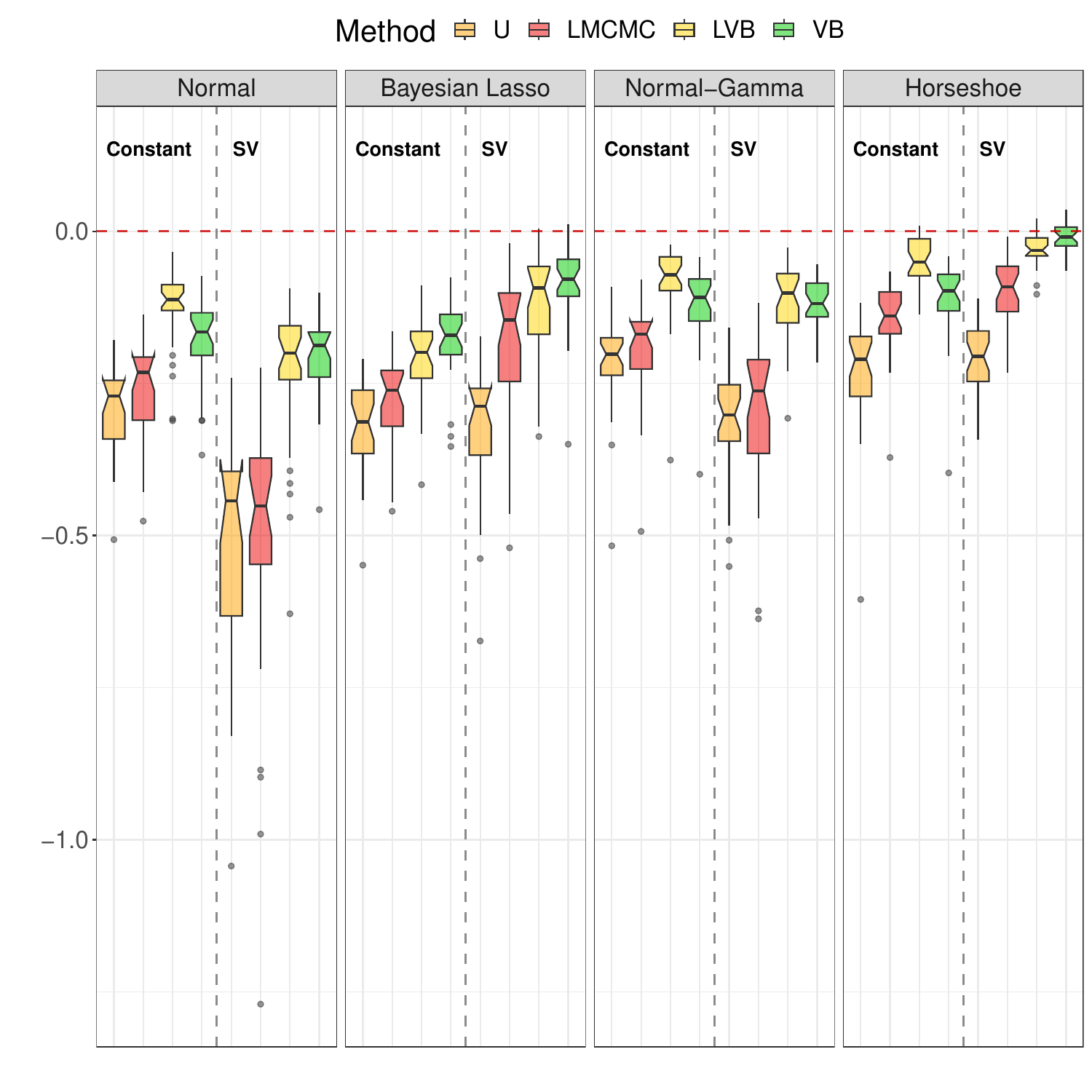}}\hspace{1em}
\subfigure[$\text{Gain}\left(\mathcal{M}_s\right)>0$ across 30 industries]{\includegraphics[width=.42\textwidth]{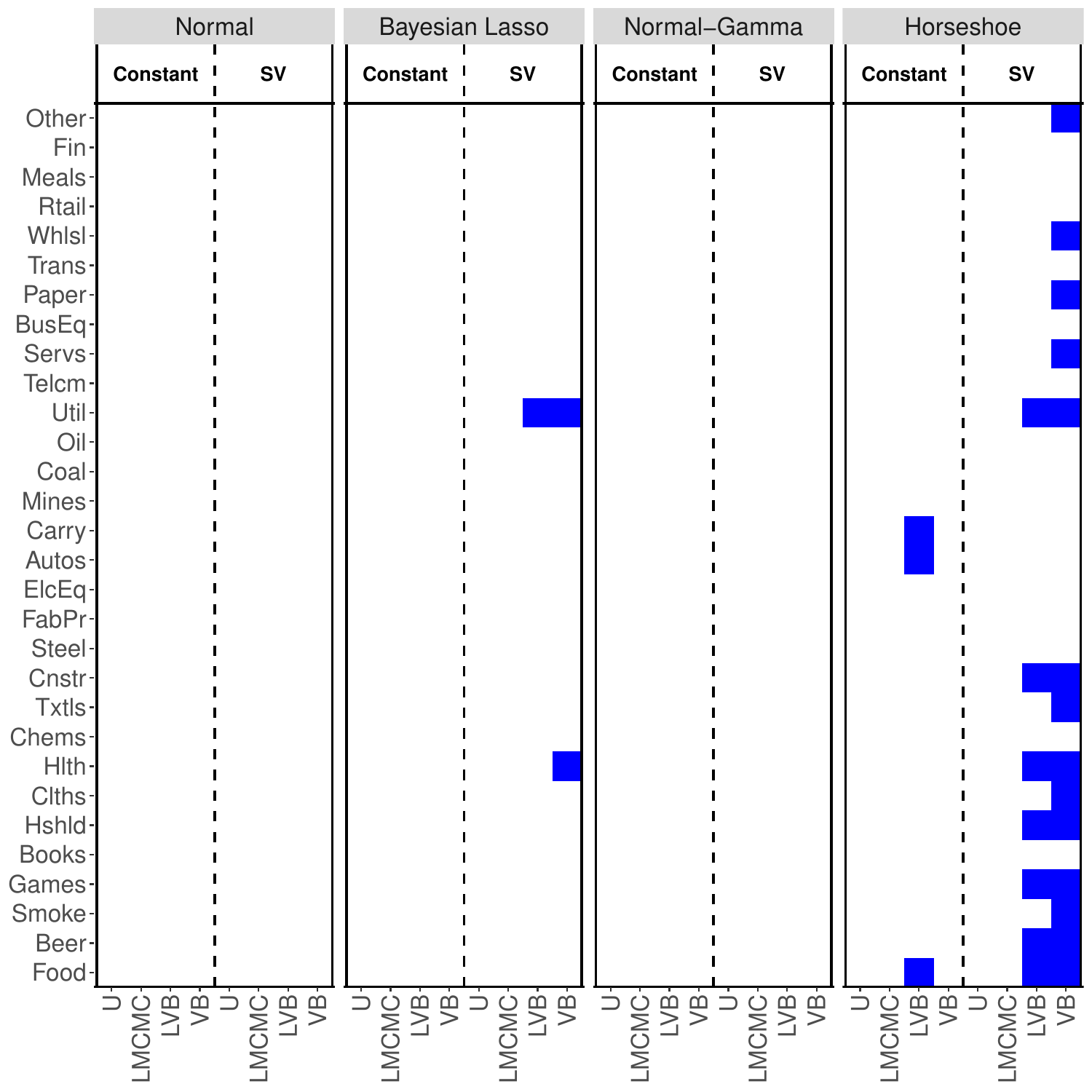}}

\subfigure[Gain$\left(\mathcal{M}_s\right)$ for 49-industry classification]{\includegraphics[width=.42\textwidth]{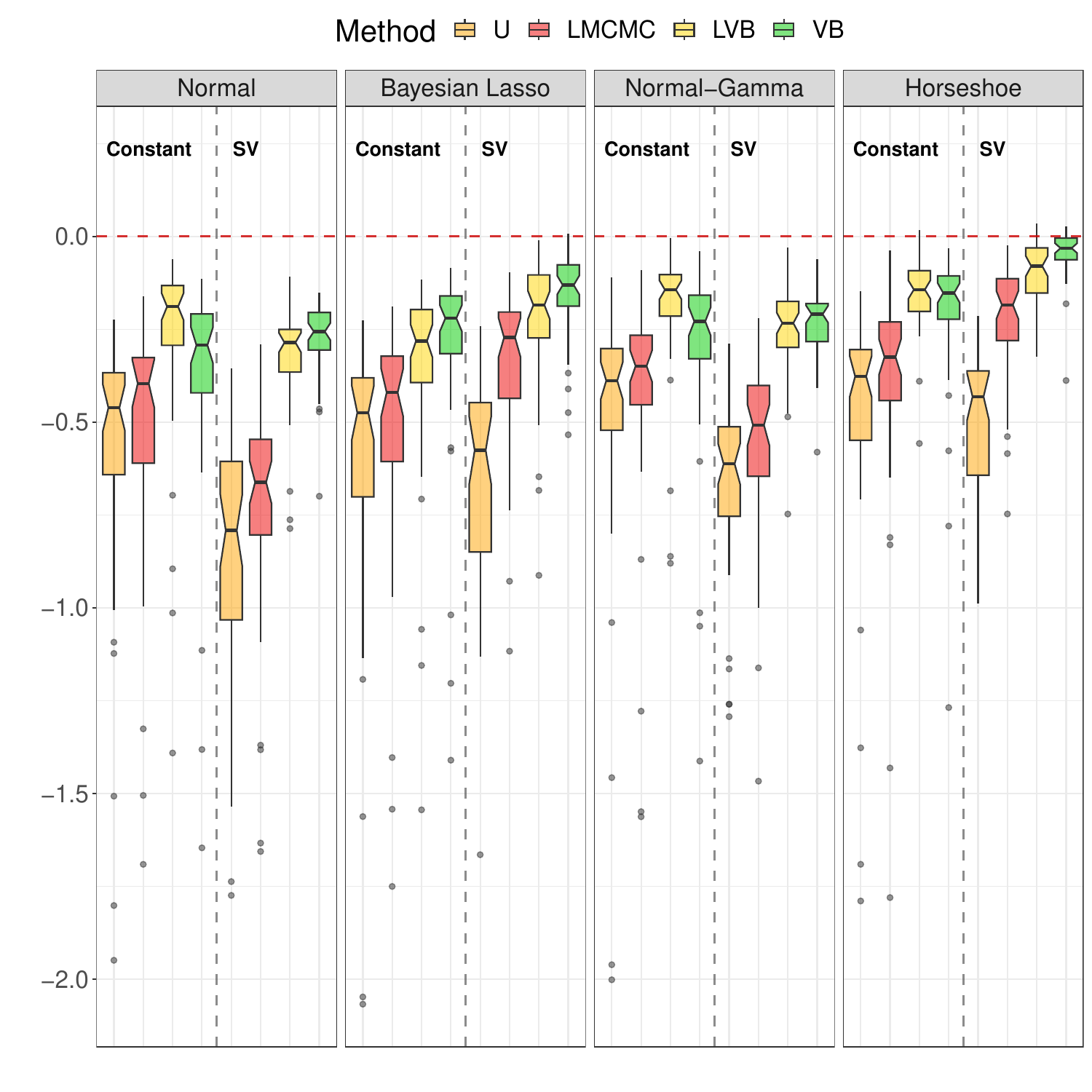}}\hspace{1em}
\subfigure[$\text{Gain}\left(\mathcal{M}_s\right)>0$ across 49 industries]{\includegraphics[width=.42\textwidth]{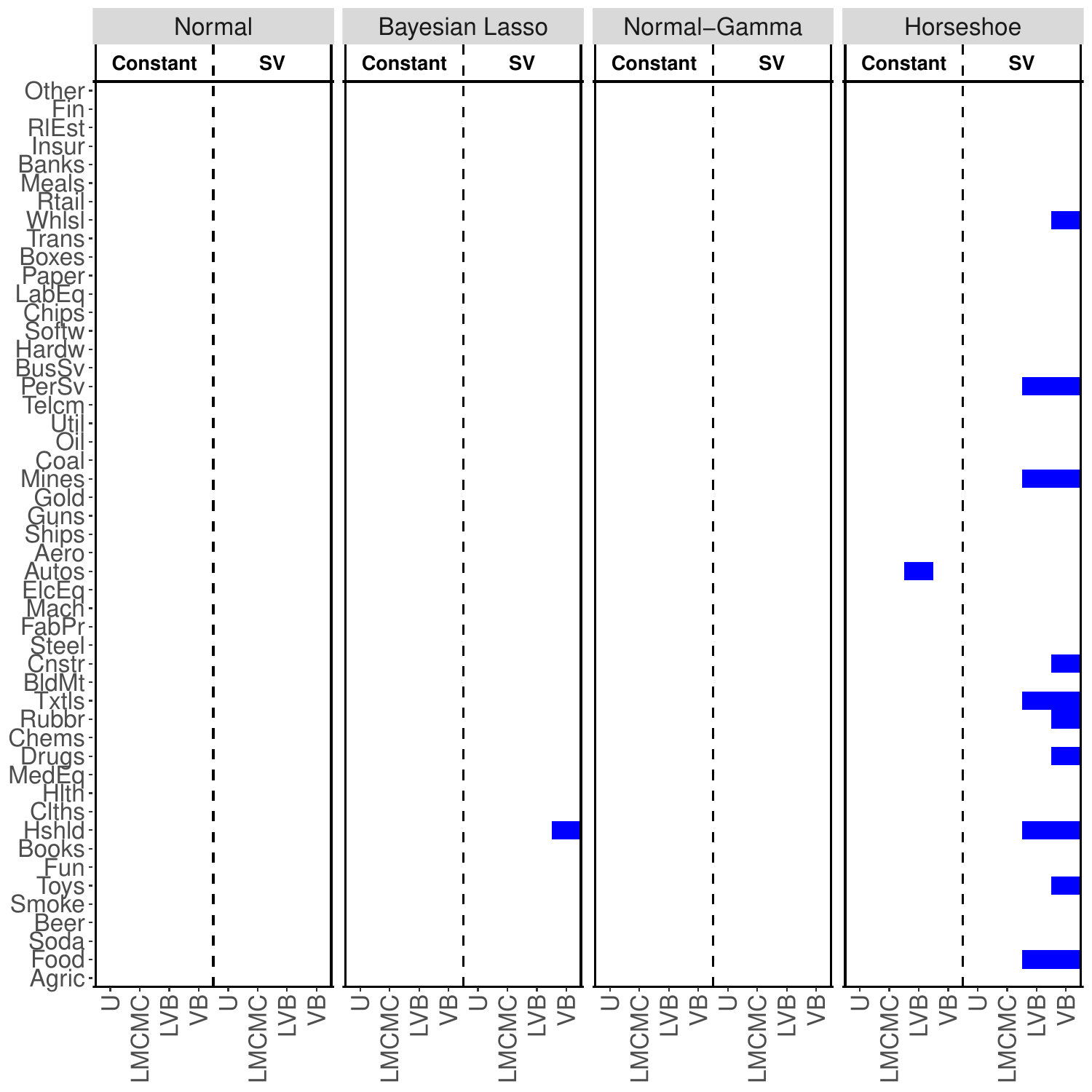}}

\caption{The left panel reports the cross-sectional distribution of the average utility gain across industry portfolios. The right panel reports the industries for which the utility gain is positive. The top (bottom) panels report the results for the 30-industry (49-industry) classification.}
	\label{fig:cer}
\end{figure}

Figure \ref{fig:cer} reports the average utility gain -- in monthly \% -- obtained by using a given forecast $\widehat{y}_{jt}$ instead of the recursive sample mean $\overline{y}_{jt}$. {\color{black}The average utility for a given model is calculated as $\widehat{u}_{j} = \overline{r}_{j} - 0.5\gamma\overline{\sigma}_{j}^2$ where $\overline{r}_{j}$ and $\overline{\sigma}_j^2$ represent the sample mean and variance, respectively, of the portfolio return $r_{jt+1} = w_{jt}y_{jt+1}$ realized over the forecasting period for the industry $j=1,\ldots,d$ under a given prior specification and estimation method. The utility gain is calculated by subtracting the average utility of a given model $\widehat{u}_{j}$ to the average utility obtained by using the naive forecast from the recursive mean and variance to calculate $w_{jt}$. A positive value for the utility gain indicates the fee that a risk-averse investor is willing to pay to access the investment strategy implied by $\mathcal{M}_s$.}  

The economic value of each forecast largely confirms the same evidence offered by the out-of-sample statistical performance. From a pure economic standpoint, the forecast from a recursive mean are quite challenging to beat: we observe that the average utility gain is mostly negative, with the only exception of those provided by {\tt VB} under an Horseshoe prior specification. Economically, the results show that a representative investor with mean-variance utility is willing to pay, on average, a monthly fee of almost 15 basis points monthly to access the strategy based on our variational inference with stochastic volatility. In addition, the right panels of Figure \ref{fig:cer} show that the positive economic value obtained from our {\tt VB} is more broadly spread across industries compared to alternative methods. This holds especially for the 30 industry classification, but also applies to the more granular 49 industry classification.



\section{Concluding remarks}
\label{sec:conclu}
\textcolor{black}{We propose a novel variational inference method for large Bayesian vector autoregressions (VAR) with exogenous predictors and stochastic volatility. Differently from most existing estimation methods for high-dimensional VAR models, our approach does not rely on a structural form representation. This allows a fast and accurate identification of the regression coefficients without leveraging on a standard Cholesky-based transformation of the parameter space. We show both in simulation and empirically that our estimation approach outperforms across different prior specifications, both statistically and economically, forecasts from existing benchmark estimation strategies, such as equivalent, non-linear MCMC algorithms (see, e.g., \citealp{gruber2022forecasting}) linearized MCMC (see, e.g., \citealp{cross2020macroeconomic}) and linearized variational inference methods (see, e.g., \citealp{gefang2023forecasting}).} 


\vskip50pt 
\singlespacing
\setlength{\parskip}{.05cm }
\begin{spacing}{0.1}
\bibliography{References}
\end{spacing}

\onehalfspacing
\normalsize
\clearpage
\appendix
\numberwithin{equation}{section}

\clearpage
\appendix
\renewcommand\thefigure{\thesection.\arabic{figure}}    
\numberwithin{equation}{section}

\centering{\Large \textbf{Supplementary Appendix of:\\}}
\vspace{4em}
\centering{\Large \textbf{Variational inference for large Bayesian}}\medskip

\centering{\Large \textbf{vector autoregressions}}
 
\setcounter{section}{0}
\justifying
\onehalfspacing
\vskip50pt 

\noindent This appendix provide the derivation of the optimal densities used in the mean-field variational Bayes algorithms. The derivation concerns the optimal densities for both the normal prior as well as the adaptive Bayesian lasso, the adaptive normal-gamma and the horseshoe. In addition, in this appendix we provide additional simulation and empirical results. 

\section{Auxiliary theoretical results}
\label{app:theoretical_res}
\noindent This section provides major results that will be repeatedly used in the proofs of the derivation of the optimal variational densities presented in Appendix \ref{app:VBVAR}.\\

\begin{result}\label{res:exp_quad_form}
	Assume that $\by$ is a $n$-dimensional vector, $\bX$ a $p\times n$ matrix and $\boldsymbol{\vartheta}$ a $p$-dimensional vector of parameters whose distribution is denoted by $q(\boldsymbol{\vartheta})$. \\
	Define $\Vert \by-\boldsymbol{\vartheta}\bX\Vert^2_2=(\by-\boldsymbol{\vartheta}\bX)(\by-\boldsymbol{\vartheta}\bX)^\intercal$, then it holds:
	\begin{equation*}\begin{aligned}
	\mathbb{E}_{\boldsymbol{\vartheta}}\left[\Vert \by-\boldsymbol{\vartheta}\bX\Vert^2_2\right] &= \by\by^\intercal + \mathbb{E}_{\boldsymbol{\vartheta}}\left[\boldsymbol{\vartheta}\bX\bX^\intercal\boldsymbol{\vartheta}^\intercal\right] -2\bmu_{q(\mathbf{\vartheta})}\bX\by^\intercal \\
	&= \by\by^\intercal + \trace\left\{\mathbb{E}_{\boldsymbol{\vartheta}}\left[\boldsymbol{\vartheta}^\intercal\boldsymbol{\vartheta}\right]\bX\bX^\intercal\right\} -2\bmu_{q(\mathbf{\vartheta})}\bX\by^\intercal\\
	&= \by\by^\intercal + \bmu_{q(\mathbf{\vartheta})}\bX\bX^\intercal\bmu_{q(\mathbf{\vartheta})}^\intercal + \trace\left\{\bSigma_{q(\mathbf{\vartheta})}\bX\bX^\intercal\right\} -2\bmu_{q(\mathbf{\vartheta})}\bX\by^\intercal \\
	&= \Vert \by-\bmu_{q(\mathbf{\vartheta})}\bX\Vert^2_2 + \trace\left\{\bSigma_{q(\mathbf{\vartheta})}\bX\bX^\intercal\right\},
	\end{aligned}\end{equation*}
	where $\mathbb{E}_{\boldsymbol{\vartheta}}(f(\boldsymbol{\vartheta}))$ denotes the expectation of the function $f(\boldsymbol{\vartheta}):\mathbb{R}^p\to \mathbb{R}^{k}$ with respect to $q(\boldsymbol{\vartheta})$, $\trace(\cdot)$ denotes the trace operator that returns the sum of the diagonal entries of a square matrix, and $\bmu_{q(\mathbf{\vartheta})}$ and $\bSigma_{q(\mathbf{\vartheta})}$ denotes the mean and variance-covariance matrix of $\boldsymbol{\vartheta}$.
\end{result}
\vspace{0.5cm}
\begin{result}\label{res:exp_XAX}
	Let $\bTheta$ be a $d\times p$ random matrix with elements $\vartheta_{i,j}$, for $i=1,\ldots,d$ and $j=1,\ldots,p$, and let $\bA$ be a $p\times p$ matrix. Our interest relies on the computation of the expectation of $\bTheta\bA\bTheta^{\intercal}$ with respect to the distribution of $\bTheta$, where the expectation is taken element-wise. The $(i,j)$-th entry of $\bTheta\bA\bTheta^{\intercal}$ is equal to $\boldsymbol{\vartheta}_i\bA\boldsymbol{\vartheta}_j^\intercal$, where $\boldsymbol{\vartheta}_i$ and $\boldsymbol{\vartheta}_j$ denote the $i$-th and $j$-th row of $\bTheta$, respectively. Therefore, the $(i,j)$-th entry of $\bTheta\bA\bTheta^{\intercal}$ is equal to:
	\begin{align*}
	\mathbb{E}\big(\boldsymbol{\vartheta}_i\bA\boldsymbol{\vartheta}_j^\intercal\big) = \mathbb{E}\big(\trace\big\{\boldsymbol{\vartheta}_j^\intercal\boldsymbol{\vartheta}_i\bA\big\}\big) = \trace\big\{\mathbb{E}(\boldsymbol{\vartheta}_j^\intercal\boldsymbol{\vartheta}_i\bA)\big\} = \trace\big\{\mathbb{E}(\boldsymbol{\vartheta}_j^\intercal\boldsymbol{\vartheta}_i)\bA\big\}.
	\end{align*}
	Let $\bmu_{\mathbf{\vartheta}_i}=\mathbb{E}(\boldsymbol{\vartheta}_i)$ and $\bSigma_{\mathbf{\vartheta}_i,\mathbf{\vartheta}_j}=\Cov(\boldsymbol{\vartheta}_i,\boldsymbol{\vartheta}_j)$, then the previous expectation reduces to:
	\begin{align*}
	\mathbb{E}(\boldsymbol{\vartheta}_i\bA\boldsymbol{\vartheta}_j^\intercal) = \trace\big\{\big(\bmu_{\mathbf{\vartheta}_j}^\intercal\bmu_{\mathbf{\vartheta}_i}+\bSigma_{\mathbf{\vartheta}_i,\mathbf{\vartheta}_j}\big)\bA\big\}=\bmu_{\mathbf{\vartheta}_i}\bA\bmu_{\mathbf{\vartheta}_j}^\intercal+\trace\left\{\bSigma_{\mathbf{\vartheta}_i,\mathbf{\vartheta}_j}\bA\right\}.
	\end{align*}
	In matrix form, $\mathbb{E}(\bTheta\bA\bTheta^{\intercal}) = \bmu_{\mathbf{\Theta}}\bA\bmu_{\mathbf{\Theta}}^\intercal + \bK_{\mathbf{\Theta}}$, where $\bmu_{\mathbf{\Theta}}$ is a $d\times p$ matrix with elements $\mu_{\vartheta_{i,j}}$, while $\bK_{\mathbf{\Theta}}$ is a $d\times d$ symmetric matrix with elements equal to $\trace\left\{\bSigma_{\mathbf{\vartheta}_i,\mathbf{\vartheta}_j}\bA\right\}$. Result \eqref{res:exp_XAX} can be further generalized to compute the expectation of $\bTheta_1\bA\bTheta_2^{\intercal}$ with respect to the joint distribution of $(\bTheta_1,\bTheta_2)$ where $\bTheta_1$ is $d_1\times p$ and $\bTheta_2$ is $d_2 \times p$.
\end{result}
\vspace{0.5cm}
\begin{result}\label{res:exp_gauss}
	Let $\boldsymbol{\vartheta}$ be a $d$-dimesnional Gaussian random vector with mean vector $\bmu_{\mathbf{\vartheta}}$ and variance-covariance matrix $\bSigma_{\mathbf{\vartheta}}$. The expectation of the quadratic form $(\boldsymbol{\vartheta}-\bmu_{\mathbf{\vartheta}})^\intercal\bSigma_{\mathbf{\vartheta}}^{-1}(\boldsymbol{\vartheta}-\bmu_{\mathbf{\vartheta}})$ with respect to $\boldsymbol{\vartheta}$ is equal to $d$. Indeed:
	\begin{align*}
	\mathbb{E}_{\mathbf{\vartheta}}\left[(\boldsymbol{\vartheta}-\bmu_{\mathbf{\vartheta}})^\intercal\bSigma_{\mathbf{\vartheta}}^{-1}(\boldsymbol{\vartheta}-\bmu_{\mathbf{\vartheta}})\right] &= \trace\left\{\mathbb{E}_{\mathbf{\vartheta}}\left[(\boldsymbol{\vartheta}-\bmu_{\mathbf{\vartheta}})(\boldsymbol{\vartheta}-\bmu_{\mathbf{\vartheta}})^\intercal\right]\bSigma_{\mathbf{\vartheta}}^{-1}\right\} = \trace\left\{\bSigma_{\mathbf{\vartheta}}\bSigma_{\mathbf{\vartheta}}^{-1}\right\} = \trace\left\{\bI_d\right\} = d.
	\end{align*}
\end{result}
%
%
%
\section{Derivation of the optimal variational densities}
\label{app:VBVAR}
\noindent This appendix explains how to obtain the relevant quantities of the mean-field variational Bayes algorithms described in Section \ref{sec:mfvb} for the prior distributions described in Section \ref{subsec:prior}. We begin by discussing the non-informative prior, then turn to the adaptive Bayesian lasso, the adaptive normal-gamma and conclude with the horseshoe prior.
%
\subsection{Normal prior specification}	
\label{app:non_sparseVBAR}			
\begin{prop}\label{prop:h_VB}
	The optimal variational density for the vector of log-volatility parameters $\mathbf{h}_j=(h_{j,0},\ldots,h_{j,T})^\intercal$ is equal to $q^*(\mathbf{h}_j) \equiv \mathsf{N}_{T+1}(\boldsymbol{\mu}_{q(h_j)}, \bSigma_{q(h_j)})$, where, for $j=1,\ldots,d$, the variational parameters $(\boldsymbol{\mu}_{q(h_j)}, \bSigma_{q(h_j)})$ are updated as:
	\begin{align}\label{eq:up_h_VB}
	\bSigma_{q(h_j)}^{new} &= \left[\nabla_{\boldsymbol{\mu}_{q(h_j)}\boldsymbol{\mu}_{q(h_j)}}^2 S(\bmu_{q(h_j)}^{old},\bSigma_{q(h_j)}^{old})\right]^{-1}, \\
    	\bmu_{q(h_j)}^{new} &= \bmu_{q(h_j)}^{new} + \bSigma_{q(h_j)}^{new}\nabla_{\boldsymbol{\mu}_{q(h_j)}} S(\bmu_{q(h_j)}^{old},\bSigma_{q(h_j)}^{old}),
\end{align}
 where $\nabla_{\bmu} S(\bmu^{old},\bSigma^{old})$ and $\nabla_{\bmu,\bmu}^2 S(\bmu^{old},\bSigma^{old})$ denote the first and second derivative of $S(\bmu,\bSigma)$ with respect to $\bmu$ and evaluated at $(\bmu^{old},\bSigma^{old})$. The function $S$ is the so called \textit{non-entropy function} which is given by $\mathbb{E}_q(\log p(\bh_j,\boldsymbol{\xi}_{-h_j},\by_j))$. In our scenario, we have that
\begin{align}\label{eq:seq_appendix}
	S(\bmu_{q(h_j)},\bSigma_{q(h_j)})&=-\frac{1}{2}[0,\boldsymbol{\iota}_n^\intercal]\bmu_{q(h_j)}-\frac{1}{2}[0,\bmu_{q(\boldsymbol{\varepsilon}^2_j)}^\intercal]\mathrm{e}^{-\bmu_{q(h_j)}+\frac{1}{2}\boldsymbol{\sigma}^2_{q(\mathbf{h}_j)}} \nonumber\\
	&\qquad
	-\frac{1}{2}\mu_{q(1/\psi_j)}\bmu_{q(h_j)}\mathbf{Q}\bmu_{q(h_j)} -\frac{1}{2}\mu_{q(1/\psi_j)}\mathsf{tr}\{\bSigma_{q(\mathbf{h}_j)}\mathbf{Q}\},
\end{align}
where $\boldsymbol{\sigma}^2_{q(h_j)}=\mathsf{diag}(\bSigma_{q(h_j)})$ is the vector of variances. In addition:
\begin{align}
\nabla_{\boldsymbol{\mu}_{q(h_j)}} S(\bmu_{q(h_j)},\bSigma_{q(h_j)}) = -\frac{1}{2}[0,\boldsymbol{\iota}_n^\intercal]^\intercal+\frac{1}{2}[0,\bmu_{q(\boldsymbol{\varepsilon}^2_j)}^\intercal]^\intercal\odot\mathrm{e}^{-\boldsymbol{\mu}_{q(h_j)}+\frac{1}{2}\boldsymbol{\sigma}^2_{q(h_j)}} -\mu_{q(1/\psi_j)}\mathbf{Q}\bmu_{q(h_j)},\\
	\nabla_{\boldsymbol{\mu}_{q(h_j)}\boldsymbol{\mu}_{q(h_j)}}^2 S(\bmu_{q(h_j)},\bSigma_{q(h_j)}) =-\frac{1}{2}\mathsf{Diag}\Bigg[[0,\bmu_{q(\boldsymbol{\varepsilon}^2_j)}^\intercal]^\intercal\odot\mathrm{e}^{-\boldsymbol{\mu}_{q(h_j)}+\frac{1}{2}\boldsymbol{\sigma}^2_{q(h_j)}}\Bigg]
	-\mu_{q(1/\psi_j)}\mathbf{Q},
\end{align}
where $\boldsymbol{\iota}_n$ is an n-dimensional vector of ones, $\mu_{q\left(1/\psi_j\right)}$ is the variational mean of $1/\psi_j$, $\mathbf{Q}$ is the precision matrix associated to the random walk process with initial state $h_0\sim\mathsf{N}(0,k_0\,\psi_j)$, and $\odot$ denotes the Hadamard product. Moreover, $\bmu_{q(\boldsymbol{\varepsilon}^2_j)}=(\mu_{q(\varepsilon^2_{j,1})},\ldots,\mu_{q(\varepsilon^2_{j,T})})^\intercal$, with elements $\mu_{q(\varepsilon^2_{j,t})}=\mathbb{E}_{q}\left[\varepsilon_{j,t}^2\right]$:	
\begin{equation*}\begin{aligned}
	\mathbb{E}_{q}\left[\varepsilon_{j,t}^2\right] &= \left(y_{j,t}-\bmu_{q(\boldsymbol{\beta}_j)}\bmu_{q(\mathbf{r}_{j,t})}-\bmu_{q(\mathbf{\vartheta}_j)}\mathbf{z}_{t-1}\right)^2 +\trace\left\{\bSigma_{q(\mathbf{\vartheta}_j)}\mathbf{z}_{t-1}\mathbf{z}_{t-1}^\intercal\right\} \\ &\qquad
    + \trace\left\{\left(\bSigma_{q(\boldsymbol{\beta}_j)}+\bmu_{q(\boldsymbol{\beta}_j)}^\intercal\bmu_{q(\boldsymbol{\beta}_j)}\right)\bK_{\mathbf{\vartheta},t}\right\}
	+\trace\left\{\bSigma_{q(\boldsymbol{\beta}_j)}\bmu_{q(\mathbf{r}_{j,t})}\bmu_{q(\mathbf{r}_{j,t})}^\intercal\right\} -2\mathbf{k}_{\boldsymbol{\vartheta},t}\bmu_{q(\boldsymbol{\beta}_j)}^\intercal,
\end{aligned}\end{equation*}
	where $\bmu_{q(\mathbf{r}_{j,t})}=\mathbf{y}_t^j -\bmu_{q(\mathbf{\Theta}^j)}\mathbf{z}_{t-1}$, and, for $i=1,\ldots,j-1$ and $k=1,\ldots,j-1$, the elements in the matrix $\bK_{\boldsymbol{\vartheta},t}$ and in the row vector $\mathbf{k}_{\boldsymbol{\vartheta},t}$ are $\left[\bK_{\mathbf{\vartheta},t}\right]_{i,k} = \trace\left\{\Cov(\boldsymbol{\vartheta}_i, \boldsymbol{\vartheta}_k) \mathbf{z}_{t-1}\mathbf{z}_{t-1}^{\intercal} \right\}$ and $\left[\mathbf{k}_{\mathbf{\vartheta},t}\right]_i = \trace\left\{\Cov(\boldsymbol{\vartheta}_i, \boldsymbol{\vartheta}_j) \mathbf{z}_{t-1}\mathbf{z}_{t-1}^{\intercal} \right\}$ respectively.  Notice that under row-factorization of $\boldsymbol{\Theta}$, we have that $\mathbf{k}_{\boldsymbol{\vartheta},t}=\mathbf{0}_j$.
\end{prop}
\begin{proof}
	Consider the model written for the $j$-th variable:
	\begin{equation*}\begin{aligned}
	y_{j,t} = \bbeta_j \mathbf{r}_{j,t} + \boldsymbol{\vartheta}_j\mathbf{z}_{t-1}+ \varepsilon_{j,t}, \quad \varepsilon_{j,t} \sim \mathsf{N}(0,\mathrm{e}^{h_{j,t}}),
	\end{aligned}\end{equation*}
	and recall that $h_{j,t}=h_{j,t-1}+e_{j,t}$ with $e_{j,t}\sim\mathsf{N}(0,\psi_j)$ and initial state $h_0\sim\mathsf{N}(0,k_0\,\psi_j)$. Define $\varepsilon_{j,t} = y_{j,t} - \bbeta_j\mathbf{r}_{j,t} - \boldsymbol{\vartheta}_j\mathbf{z}_{t-1}$ and $\bh_j=(h_{j,0},\ldots,h_{j,T})^\intercal$. Recall that the random walk can be jointly represented as a Gaussian Markov random field $\bh_j\sim\mathsf{N}_{T+1}(0,\psi\mathbf{Q}^{-1})$ with tri-diagonal precision matrix $\mathbf{Q}^{-1}$. Compute $\log p(\bh_j,\boldsymbol{\xi}_{-h_j},\by_j)\propto \ell_j(\bxi;\mathbf{y},\mathbf{x})+\log p(\bh_j)$:
	\begin{equation*}\begin{aligned}
	\log p(\bh_j,\boldsymbol{\xi}_{-h_j},\by_j) &\propto -\frac{1}{2}\sum_{t=1}^T h_{j,t}-\frac{1}{2}\sum_{t=1}^T\varepsilon^2_{j,t}\mathrm{e}^{-h_{j,t}} -\frac{1}{2\psi_j}\bh_j\mathbf{Q}\bh_j.
	\end{aligned}\end{equation*}
 Notice that the latter cannot be recognized as the kernel of a known distribution for $\bh_j$, therefore complicating the computations. To overcome this issue we exploit the parametric variational Bayes paradigm and impose a Gaussian approximation $\bh_j\sim\mathsf{N}(\bmu_{q(h_j)},\bSigma_{q(h_j)})$ similarly to \cite{bernardi2022smoothing}.
Then, we follow \cite{rohde2016semiparametric} to implement an iterative updating scheme to derive the optimal values of $(\bmu_{q(h_j)},\bSigma_{q(h_j)})$. To this aim, define the \textit{non-entropy function} $S$ as $\mathbb{E}_q(\log p(\bh_j,\boldsymbol{\xi}_{-h_j},\by_j))$:
\begin{align}
	S(\bmu_{q(h_j)},\bSigma_{q(h_j)})&=-\frac{1}{2}[0,\boldsymbol{\iota}_n^\intercal]\bmu_{q(h_j)}-\frac{1}{2}[0,\bmu_{q(\boldsymbol{\varepsilon}^2_j)}^\intercal]\mathrm{e}^{-\bmu_{q(h_j)}+\frac{1}{2}\boldsymbol{\sigma}^2_{q(\mathbf{h}_j)}} \nonumber\\
	&\qquad
	-\frac{1}{2}\mu_{q(1/\psi_j)}\bmu_{q(h_j)}\mathbf{Q}\bmu_{q(h_j)} -\frac{1}{2}\mu_{q(1/\psi_j)}\mathsf{tr}\{\bSigma_{q(\mathbf{h}_j)}\mathbf{Q}\},
\end{align}
where we exploit a vector representation of the likelihood term and $\boldsymbol{\sigma}^2_{q(h_j)}=\mathsf{diag}(\bSigma_{q(h_j)})$ is the vector of variances.  Moreover each element in the vector $\bmu_{q(\boldsymbol{\varepsilon}^2_j)}$, namely $\mu_{q(\varepsilon^2_{j,t})}=\mathbb{E}_{q}\left[\varepsilon_{j,t}^2\right]$ is given by:
	\begin{equation*}\begin{aligned}
	\mathbb{E}_{q}\left[\varepsilon_{j,t}^2\right] &= \mathbb{E}_{-\nu_j}\left[\left(y_{j,t}-\bbeta_j\mathbf{r}_{j,t}-\boldsymbol{\vartheta}_j\mathbf{z}_{t-1}\right)^2\right] \\
	&= y_{j,t}^2
	+\mathbb{E}_{\mathbf{\vartheta}}\left[\boldsymbol{\vartheta}_j\mathbf{z}_{t-1}\mathbf{z}_{t-1}^\intercal\boldsymbol{\vartheta}_j\right] +\overbrace{\mathbb{E}_{\mathbf{\vartheta},\boldsymbol{\beta}_j}\left[\bbeta_j\mathbf{r}_{j,t}\mathbf{r}_{j,t}^\intercal\bbeta_j^\intercal\right]}^{\text A}\\
	&\qquad 
	-2y_{j,t}\mathbb{E}_{\mathbf{\vartheta}}\left[\boldsymbol{\vartheta}_j\right]\mathbf{z}_{t-1}
	-2y_{j,t}\mathbb{E}_{\boldsymbol{\beta}_j}\left[\bbeta_j\right]\mathbb{E}_{\mathbf{\vartheta}}\left[\mathbf{r}_{j,t}\right]\\
	&\qquad +2\underbrace{\mathbb{E}_{\mathbf{\vartheta}}\left[\boldsymbol{\vartheta}_j\mathbf{z}_{t-1}\mathbf{r}_{j,t}^\intercal\right]\mathbb{E}_{\boldsymbol{\beta}_j}\left[\bbeta_j^\intercal\right]}_{\text B} \\
	&= y_{j,t}^2
	+\bmu_{q(\mathbf{\vartheta}_j)}\mathbf{z}_{t-1}\mathbf{z}_{t-1}^\intercal\bmu_{q(\mathbf{\vartheta}_j)} +\bmu_{q(\boldsymbol{\beta}_j)}\bmu_{q(\mathbf{r}_{j,t})}\bmu_{q(\mathbf{r}_{j,t})}^\intercal\bmu_{q(\boldsymbol{\beta}_j)}^\intercal\\
	&\qquad
	-2y_{j,t}\bmu_{q(\mathbf{\vartheta}_j)}\mathbf{z}_{t-1}
	-2y_{j,t}\bmu_{q(\boldsymbol{\beta}_j)}\bmu_{q(\mathbf{r}_{j,t})}\\
	&\qquad +2\bmu_{q(\mathbf{\vartheta}_j)}\mathbf{z}_{t-1}\bmu_{q(\mathbf{r}_{j,t})}^\intercal\bmu_{q(\boldsymbol{\beta}_j)}^\intercal \\
	&\qquad +\trace\left\{\bSigma_{q(\mathbf{\vartheta}_j)}\mathbf{z}_{t-1}\mathbf{z}_{t-1}^\intercal\right\} + \trace\left\{\left(\bSigma_{q(\boldsymbol{\beta}_j)}+\bmu_{q(\boldsymbol{\beta}_j)}^\intercal\bmu_{q(\boldsymbol{\beta}_j)}\right)\bK_{\mathbf{\vartheta},t}\right\}\\
	&\qquad
	+\trace\left\{\bSigma_{q(\boldsymbol{\beta}_j)}\bmu_{q(\mathbf{r}_{j,t})}\bmu_{q(\mathbf{r}_{j,t})}^\intercal\right\}
	-2\mathbf{k}_{\mathbf{\vartheta},t}\bmu_{q(\boldsymbol{\beta}_j)}^\intercal\\
	&= \left(y_{j,t}-\bmu_{q(\boldsymbol{\beta}_j)}\bmu_{q(\mathbf{r}_{j,t})}-\bmu_{q(\mathbf{\vartheta}_j)}\mathbf{z}_{t-1}\right)^2 \\
	&\qquad +\trace\left\{\bSigma_{q(\mathbf{\vartheta}_j)}\mathbf{z}_{t-1}\mathbf{z}_{t-1}^\intercal\right\} + \trace\left\{\left(\bSigma_{q(\boldsymbol{\beta}_j)}+\bmu_{q(\boldsymbol{\beta}_j)}^\intercal\bmu_{q(\boldsymbol{\beta}_j)}\right)\bK_{\mathbf{\vartheta},t}\right\}\\
	&\qquad
	+\trace\left\{\bSigma_{q(\boldsymbol{\beta}_j)}\bmu_{q(\mathbf{r}_{j,t})}\bmu_{q(\mathbf{r}_{j,t})}^\intercal\right\}
	-2\mathbf{k}_{\mathbf{\vartheta},t}\bmu_{q(\boldsymbol{\beta}_j)}^\intercal ,
	\end{aligned}\end{equation*}
	where $\bmu_{q(\mathbf{r}_{j,t})}=\mathbf{y}_t^j -\bmu_{q(\mathbf{\Theta}^j)}\mathbf{z}_{t-1}$.
	The computations involving terms A and B are presented in the following equations.
	Firs of all, define $\bbeta_j\mathbf{r}_{j,t}\mathbf{r}_{j,t}^\intercal\bbeta_j^\intercal=\Vert\bbeta_j\mathbf{r}_{j,t}\Vert^2_2$, then the term A above is equal to:
	\begin{equation*}\begin{aligned}
	\mathbb{E}_{\mathbf{\vartheta},\boldsymbol{\beta}_j}\left[\Vert\bbeta_j\mathbf{r}_{j,t}\Vert^2_2\right] &= \mathbb{E}_{\boldsymbol{\beta}_j}\Big[\bbeta_j\overbrace{\mathbb{E}_{\mathbf{\vartheta}}\left[\mathbf{r}_{j,t}\mathbf{r}_{j,t}^\intercal\right]}^{\text{See Results \ref{res:exp_quad_form} and \ref{res:exp_XAX}}}\bbeta_j^\intercal\Big]\\
	&= \mathbb{E}_{\boldsymbol{\beta}_j}\left[\bbeta_j\left\{\bmu_{q(\mathbf{r}_{j,t})}\bmu_{q(\mathbf{r}_{j,t})}^\intercal+\bK_{\mathbf{\vartheta},t}\right\}\bbeta_j^\intercal\right] \\
	&=\bmu_{q(\boldsymbol{\beta}_j)}\left\{\bmu_{q(\mathbf{r}_{j,t})}\bmu_{q(\mathbf{r}_{j,t})}^\intercal+\bK_{\mathbf{\vartheta},t}\right\}\bmu_{q(\boldsymbol{\beta}_j)}^\intercal +\trace\left\{\bSigma_{q(\boldsymbol{\beta}_j)}\left[\bmu_{q(\mathbf{r}_{j,t})}\bmu_{q(\mathbf{r}_{j,t})}^\intercal+\bK_{\mathbf{\vartheta},t}\right]\right\} \\
	&=\Vert\bmu_{q(\boldsymbol{\beta}_j)}\bmu_{q(\mathbf{r}_{j,t})}\Vert^2_2 
	+\trace\left\{\left(\bSigma_{q(\boldsymbol{\beta}_j)}+\bmu_{q(\boldsymbol{\beta}_j)}^\intercal\bmu_{q(\boldsymbol{\beta}_j)}\right)\bK_{\mathbf{\vartheta},t}\right\}\\
	&\qquad
	+\trace\left\{\bSigma_{q(\boldsymbol{\beta}_j)}\bmu_{q(\mathbf{r}_{j,t})}\bmu_{q(\mathbf{r}_{j,t})}^\intercal\right\},
	\end{aligned}
	\end{equation*}
	while the term B is:
	\begin{equation*}\begin{aligned}
	\mathbb{E}_{\mathbf{\vartheta}}\left[\boldsymbol{\vartheta}_j\mathbf{z}_{t-1}\mathbf{r}_{j,t}^\intercal\right]\mathbb{E}_{\boldsymbol{\beta}_j}\left[\bbeta_j^\intercal\right] &= \mathbb{E}_{\mathbf{\vartheta}}\bigg[\boldsymbol{\vartheta}_j\mathbf{z}_{t-1}\mathbf{y}_t^{j\intercal}-\overbrace{\boldsymbol{\vartheta}_j\mathbf{z}_{t-1}\mathbf{z}_{t-1}^\intercal\bTheta^{j\intercal}}^{\text{See Result \ref{res:exp_XAX}}}\bigg]\bmu_{q(\boldsymbol{\beta}_j)}^\intercal \\
	&= \left(\bmu_{q(\mathbf{\vartheta}_j)}\mathbf{z}_{t-1}\mathbf{y}_t^{j\intercal}-\bmu_{q(\mathbf{\vartheta}_j)}\mathbf{z}_{t-1}\mathbf{z}_{t-1}^\intercal\bmu_{q(\mathbf{\Theta}^j)}^\intercal-\mathbf{k}_{\mathbf{\vartheta},t}\right)\bmu_{q(\boldsymbol{\beta}_j)}^\intercal\\
	&= \bmu_{q(\mathbf{\vartheta}_j)}\mathbf{z}_{t-1}\bmu_{q(\mathbf{r}_{j,t})}^\intercal\bmu_{q(\boldsymbol{\beta}_j)}^\intercal-\mathbf{k}_{\mathbf{\vartheta},t}\bmu_{q(\boldsymbol{\beta}_j)}^\intercal.
	\end{aligned}\end{equation*}
	Notice that for the latter derivation we use Results \ref{res:exp_quad_form} and \ref{res:exp_XAX}.
\end{proof}
\begin{prop}\label{prop:nu_t_VB}
	The optimal variational density for the vector of time-varying precision parameters $\boldsymbol{\nu}_j=(\nu_{j,1},\ldots,\nu_{j,T})^\intercal$ is equal to $q^*(\boldsymbol{\nu}_j) \equiv \mathsf{logN}_T(-\boldsymbol{\mu}_{q(h_j)}, \bSigma_{q(h_j)})$, where, for each $j=1,\ldots,d$:
	\begin{equation}\begin{aligned}\label{eq:up_nu_t_VB}
    \mathbb{E}_q[\nu_t] &= \exp\{-\mu_{q(h_{j,t})}+1/2\sigma^2_{q(h_{j,t})}\},\\
    \mathsf{Var}_q[\nu_t] &= \exp\{-2\mu_{q(h_{j,t})}+\sigma^2_{q(h_{j,t})}\}(\exp\{\sigma^2_{q(h_{j,t})}\}-1), \\
    \mathsf{Cov}_q[\nu_t,\nu_{t+1}] &= \exp\{-\mu_{q(h_{j,t})}-\mu_{q(h_{j,t+1})}+1/2(\sigma^2_{q(h_{j,t})}+\sigma^2_{q(h_{j,t+1})})\}(\exp\{\mathsf{Cov}_q[h_t,h_{t+1}]\}-1).
\end{aligned}\end{equation}
\end{prop}
\begin{proof}
	The proof immediately follows from the fact that $\nu_{j,t}=\mathrm{e}^{-h_{j,t}}$ for $t=1,\ldots,T$ and the distribution of $\bh_j$ is Gaussian, as defined in Proposition \ref{prop:h_VB}.
\end{proof}
\begin{prop}\label{prop:nu_VB}
	The optimal variational density for the constant precision parameter (homoskedastic modeling) $\nu_j$ is equal to $q^*(\nu_j) \equiv \mathsf{Ga}(a_{q(\nu_j)}, b_{q(\nu_j)})$, where, for $j=1,\ldots,d$:
	\begin{equation}\begin{aligned}\label{eq:up_nu_VB}
	a_{q(\nu_j)} =  a_{\nu} + T/2, \quad
	b_{q(\nu_j)} =  b_{\nu}+\frac{1}{2}\sum_{t=1}^T\mathbb{E}_{-\nu_j}\left[\varepsilon_{j,t}^2\right],
	\end{aligned}\end{equation}
	where $\mathbb{E}_{-\nu_j}\left[\varepsilon_{j,t}^2\right]$ is defined in Proposition \ref{prop:h_VB}. 
\end{prop}
\begin{proof}
	Consider the model written for the $j$-th variable:
	\begin{equation*}\begin{aligned}
	y_{j,t} = \bbeta_j \mathbf{r}_{j,t} + \boldsymbol{\vartheta}_j\mathbf{z}_{t-1}+ \varepsilon_{j,t}, \quad \varepsilon_{j,t} \sim \mathsf{N}(0,1/\nu_j),
	\end{aligned}\end{equation*}
	and notice that $\varepsilon_{j,t} = y_{j,t} - \bbeta_j\mathbf{r}_{j,t} - \boldsymbol{\vartheta}_j\mathbf{z}_{t-1}$. Recall that a priori $\nu_j\sim\mathsf{Ga}(a_{\nu},b_{\nu})$ and compute $\log q^*(\nu_j)\propto \mathbb{E}_{-\nu_j}\left[\ell_j(\bxi;\mathbf{y},\mathbf{x})+\log p(\nu_j)\right]$:
	\begin{equation*}\begin{aligned}
	\log q^*(\nu_j) &\propto \mathbb{E}_{-\nu_j}\left[\frac{T}{2}\log\nu_j -\frac{\nu_j}{2}\sum_{t=1}^T\varepsilon_{j,t}^2 +(a_{\nu}-1)\log\nu_j -b_{\nu}\nu_j \right]\\
	&\propto \left(\frac{T}{2}+a_{\nu}-1\right)\log\nu_j  -\nu_j\left(b_{\nu}+\frac{1}{2}\sum_{t=1}^T\mathbb{E}_{-\nu_j}\left[\varepsilon_{j,t}^2\right]\right),
	\end{aligned}\end{equation*}
	where the computations for $\mathbb{E}_{-\nu_j}\left[\varepsilon_{j,t}^2\right]$ have been previously considered in the Proof of Proposition \ref{prop:h_VB}. Take the exponential of the latter equation, and notice that it is the kernel of a gamma random variable $\mathsf{Ga}(a_{q(\nu_j)},b_{q(\nu_j)})$ as defined in Proposition \ref{prop:nu_VB}.
\end{proof}
\begin{prop}\label{prop:beta_VB}
	The optimal variational density for the parameter $\bbeta_j$ for $j=2,\ldots,d$ is equal to $q^*(\bbeta_j) \equiv \mathsf{N}_{j-1}(\bmu_{q(\boldsymbol{\beta}_j)}, \bSigma_{q(\boldsymbol{\beta}_j)})$, where:
	\begin{equation}\begin{aligned}\label{eq:up_beta_VB}
	\bSigma_{q(\boldsymbol{\beta}_j)} &= \left(\sum_{t=1}^T\mu_{q(\nu_{j,t})}\left(\bmu_{q(\mathbf{r}_{j,t})}\bmu_{q(\mathbf{r}_{j,t})}^\intercal + \bK_{\mathbf{\vartheta},t}\right) + 1/\tau\bI_{j-1} \right)^{-1},\\
	\bmu_{q(\boldsymbol{\beta}_j)} &= \bSigma_{q(\boldsymbol{\beta}_j)} \sum_{t=1}^T\mu_{q(\nu_{j,t})}\left(\bmu_{q(\mathbf{r}_{j,t})}(y_{j,t}-\bmu_{q(\mathbf{\vartheta}_j)}\mathbf{z}_{t-1})^\intercal + \mathbf{k}_{\mathbf{\vartheta},t}\right).
	\end{aligned}\end{equation}
 The optimal variational density for the parameter $\bbeta_j$ under homoskedastic assumption is obtained by substituting $\mu_{q(\nu_{j,t})}$ by $\mu_{q(\nu_j)}$ in the latter equations.
\end{prop}
\begin{proof}
	Consider the model written for the $j$-th variable:
	\begin{equation*}\begin{aligned}
	y_{j,t} = \bbeta_j\mathbf{r}_{j,t} + \boldsymbol{\vartheta}_j\mathbf{z}_{t-1}+ \varepsilon_{j,t}, \quad \varepsilon_{j,t} \sim \mathsf{N}(0,1/\nu_{j,t}).
	\end{aligned}\end{equation*}
	Recall that a priori $\bbeta_j\sim\mathsf{N}_{j-1}(\boldsymbol{0},\tau\bI_{j-1})$ and compute the optimal variational density as $\log q^*(\bbeta_j)\propto \mathbb{E}_{-\boldsymbol{\beta}_j}\left[\ell_j(\bxi;\mathbf{y},\mathbf{x})+\log p(\bbeta_j)\right]$:
	\begin{equation*}\begin{aligned}\label{eq:deriv_beta}
	\log q^*(\bbeta_j) &\propto \mathbb{E}_{-\boldsymbol{\beta}_j}\left[ -\frac{1}{2}\sum_{t=1}^T\nu_{j,t}\left(y_{j,t}-\boldsymbol{\vartheta}_j\mathbf{z}_{t-1}-\bbeta_j\mathbf{r}_{j,t} \right)^2 -\frac{1}{2\tau}\bbeta_j\bbeta_j^\intercal \right] \\
	&\propto \mathbb{E}_{-\boldsymbol{\beta}_j}\left[ -\frac{1}{2}\left\{\bbeta_j\left(\sum_{t=1}^T\nu_{j,t}\mathbf{r}_{j,t}\mathbf{r}_{j,t}^\intercal+1/\tau\bI_{j-1}\right)\bbeta_j^\intercal-2\bbeta_j\nu_j\sum_{t=1}^T\nu_{j,t}\mathbf{r}_{j,t}(y_{j,t}-\boldsymbol{\vartheta}_j\mathbf{z}_{t-1})^\intercal\right\}\right],
	\end{aligned}\end{equation*}
    and, applying some results defined is Appendix \ref{app:theoretical_res}, we get:
	\begin{equation*}\begin{aligned}
	\log q^*(\bbeta_j) &\propto -\frac{1}{2}\bigg\{\bbeta_j\bigg(\sum_{t=1}^T\mu_{q(\nu_{j,t})}\mathbb{E}_{\mathbf{\vartheta}}\overbrace{\left[\mathbf{r}_{j,t}\mathbf{r}_{j,t}^\intercal\right]}^{\text{Result \ref{res:exp_XAX}}}+\frac{1}{\tau}\bI_{j-1}\bigg)\bbeta_j^\intercal -2\bbeta_j\sum_{t=1}^T\mu_{q(\nu_{j,t})}\mathbb{E}_{\mathbf{\vartheta}}\overbrace{\left[\mathbf{r}_{j,t}(y_{j,t}-\boldsymbol{\vartheta}_j\mathbf{z}_{t-1})^\intercal\right]}^{\text{Result \ref{res:exp_XAX}}}\bigg\} \\
	&\propto -\frac{1}{2}\bigg\{\bbeta_j\bigg(\sum_{t=1}^T\mu_{q(\nu_{j,t})}\big(\bmu_{q(\mathbf{r}_{j,t})}\bmu_{q(\mathbf{r}_{j,t})}^\intercal + \bK_{\mathbf{\vartheta},t}\big)+\frac{1}{\tau}\bI_{j-1}\bigg)\bbeta_j^\intercal \\
	&\qquad\qquad-2\bbeta_j\sum_{t=1}^T\mu_{q(\nu_{j,t})}\big(\bmu_{q(\mathbf{r}_{j,t})}(y_{j,t}-\bmu_{q(\mathbf{\vartheta}_j)}\mathbf{z}_{t-1})^\intercal + \mathbf{k}_{\mathbf{\vartheta},t}\big)\bigg\}.
	\end{aligned}\end{equation*}
	Take the exponential and notice that the latter is the kernel of a Gaussian random variable $\mathsf{N}_{j-1}(\bmu_{q(\boldsymbol{\beta}_j)}, \bSigma_{q(\boldsymbol{\beta}_j)})$, as defined in Proposition \ref{prop:beta_VB}.
\end{proof}
\begin{prop}\label{prop:theta_VB}
	The optimal variational density for the parameter $\boldsymbol{\vartheta}$ is equal to a multivariate Gaussian $q^*(\boldsymbol{\vartheta}) \equiv \mathsf{N}_{d(d+p+1)}(\bmu_{q(\mathbf{\vartheta})}, \bSigma_{q(\mathbf{\vartheta})})$, where:
	\begin{equation}\begin{aligned}\label{eq:up_theta_VB}
	\bSigma_{q(\mathbf{\vartheta})} &= \left( \sum_{t=1}^T(\bmu_{q(\mathbf{\Omega}_t)} \otimes\mathbf{z}_{t-1}\mathbf{z}_{t-1}^\intercal)+1/\upsilon\bI_{d(d+p+1)}\right)^{-1},\quad\bmu_{q(\mathbf{\vartheta})} &= \bSigma_{q(\mathbf{\vartheta})} \sum_{t=1}^T\left(\bmu_{q(\mathbf{\Omega}_t)} \otimes\mathbf{z}_{t-1}\right)\mathbf{y}_t,
	\end{aligned}\end{equation}
	where $\bmu_{q(\mathbf{\Omega}_t)}=\mathbb{E}_q\left[\mathbf{\Omega}_t\right]=\mathbb{E}_q\left[\mathbf{L}^\intercal\mathbf{V}_t\mathbf{L}\right]=(\bI_d-\boldsymbol{\mu}_{q(\mathbf{B})})^{\intercal}\boldsymbol{\mu}_{q(\mathbf{V}_t)}(\bI_d-\boldsymbol{\mu}_{q(\mathbf{B})}) +\bC_{\mathbf{\vartheta},t}$ and $\bC_{\mathbf{\vartheta},t}$ is a $d\times d$ symmetric matrix whose generic element is given by:
	\begin{align*}
	\left[\bC_{\mathbf{\vartheta},t}\right]_{i,j} = \sum_{k=j+1}^d \Cov(\beta_{k,i},\beta_{k,j})\mu_{q(\nu_{k,t})}.
	\end{align*}
 The optimal variational density for the parameter $\boldsymbol{\vartheta}$ under homoskedastic assumption is obtained by substituting $\bmu_{q(\mathbf{\Omega}_t)}$ by $\bmu_{q(\mathbf{\Omega})}=(\bI_d-\boldsymbol{\mu}_{q(\mathbf{B})})^{\intercal}\boldsymbol{\mu}_{q(\mathbf{V})}(\bI_d-\boldsymbol{\mu}_{q(\mathbf{B})}) +\bC_{\mathbf{\vartheta}}$ and $\bC_{\mathbf{\vartheta}}$ is a constant $d\times d$ symmetric matrix whose generic element is given by:
	\begin{align*}
	\left[\bC_{\mathbf{\vartheta}}\right]_{i,j} = \sum_{k=j+1}^d \Cov(\beta_{k,i},\beta_{k,j})\mu_{q(\nu_{k})}.
	\end{align*}
\end{prop}
\begin{proof}
	Consider the model written as $\bL\mathbf{y}_t = \bL\bTheta\mathbf{z}_{t-1}+ \boldsymbol{\varepsilon}_t$ with $\boldsymbol{\varepsilon}_t \sim \mathsf{N}_d(0,\bV_t^{-1})$ and then apply the vectorisation operation on the transposed and get:
	\begin{equation*}\begin{aligned}
	\bL\mathbf{y}_t = (\bL\otimes\mathbf{z}_{t-1}^\intercal)\boldsymbol{\vartheta}+ \boldsymbol{\varepsilon}_t, \quad \boldsymbol{\varepsilon}_t \sim \mathsf{N}_d(0,\bV_t^{-1}).
	\end{aligned}\end{equation*}
	Recall that a priori $\boldsymbol{\vartheta}\sim\mathsf{N}_{d(d+p+1)}(\boldsymbol{0},\upsilon\bI_{d(d+p+1)})$. Compute the optimal variational density for the parameter $\boldsymbol{\vartheta}$ as $\log q^*(\boldsymbol{\vartheta})\propto \mathbb{E}_{-\mathbf{\vartheta}}\left[\ell(\bxi;\mathbf{y},\mathbf{x})+\log p(\boldsymbol{\vartheta})\right]$:
	\begin{equation*}\begin{aligned}
	\log q^*(\boldsymbol{\vartheta}) &\propto -\frac{1}{2}\mathbb{E}_{-\mathbf{\vartheta}}\Bigg[\sum_{t=1}^T\left(\bL\mathbf{y}_t-(\bL\otimes\mathbf{z}_{t-1}^\intercal)\boldsymbol{\vartheta}\right)^\intercal\bV_t\left(\bL\mathbf{y}_t-(\bL\otimes\mathbf{z}_{t-1}^\intercal)\boldsymbol{\vartheta}\right)\Bigg] -\frac{1}{2\upsilon}\mathbb{E}_{-\mathbf{\vartheta}}\Bigg[\boldsymbol{\vartheta}^\intercal\boldsymbol{\vartheta}\Bigg] \\
	&\propto -\frac{1}{2}\mathbb{E}_{-\mathbf{\vartheta}}\left[\sum_{t=1}^T\left(\boldsymbol{\vartheta}^\intercal(\mathbf{\Omega}_t \otimes\mathbf{z}_{t-1}\mathbf{z}_{t-1}^\intercal)\boldsymbol{\vartheta}\right)-2\sum_{t=1}^T\boldsymbol{\vartheta}^\intercal\bigg((\mathbf{\Omega}_t \otimes\mathbf{z}_{t-1})\mathbf{y}_t\bigg)\right] -\frac{1}{2\upsilon}\boldsymbol{\vartheta}^\intercal\boldsymbol{\vartheta}\\
	&\propto -\frac{1}{2}\left\{\boldsymbol{\vartheta}^\intercal\left(\sum_{t=1}^T(\bmu_{q(\mathbf{\Omega}_t)} \otimes\mathbf{z}_{t-1}\mathbf{z}_{t-1}^\intercal)+\frac{1}{\upsilon}\bI_{d(d+p+1)} \right)\boldsymbol{\vartheta} -2\boldsymbol{\vartheta}^\intercal\sum_{t=1}^T\left(\bmu_{q(\mathbf{\Omega})} \otimes\mathbf{z}_{t-1}\right)\mathbf{y}_t\right\}.
	\end{aligned}\end{equation*}
	To compute the expectation $\bmu_{q(\mathbf{\Omega}_t)}=\mathbb{E}_{-\mathbf{\vartheta}}\left[(\bI_d-\bB)^\intercal\bV_t(\bI_d-\bB)\right]$ we use the following:
	\begin{align*}
	\mathbb{E}_{\mathbf{B},\mathbf{V}_t}\left[(\bI_d-\bB)^{\intercal}\bV_t(\bI_d-\bB)\right] &= \mathbb{E}_{\mathbf{B},\mathbf{V}_t}\left[\bV_t-2\bB^{\intercal}\bV_t-\bB^{\intercal}\bV_t\bB\right] \\
	&=  \boldsymbol{\mu}_{q(\bV_t)}-2\boldsymbol{\mu}_{q(\bB)}^{\intercal}\boldsymbol{\mu}_{q(\bV_t)}-\mathbb{E}_{\mathbf{B},\mathbf{V}_t}\left[\bB^{\intercal}\bV_t\bB\right] \\
	&=  \boldsymbol{\mu}_{q(\bV_t)}-2\boldsymbol{\mu}_{q(\bB)}^{\intercal}\boldsymbol{\mu}_{q(\bV_t)}+\boldsymbol{\mu}_{q(\bB)}^{\intercal}\boldsymbol{\mu}_{q(\bV_t)}\boldsymbol{\mu}_{q(\bB)} + \bC_{\mathbf{\vartheta},t} \\
	&= (\bI_d-\boldsymbol{\mu}_{q(\bB)})^{\intercal}\boldsymbol{\mu}_{q(\bV_t)}(\bI_d-\boldsymbol{\mu}_{q(\bB)}) + \bC_{\mathbf{\vartheta},t},
	\end{align*}
	where we exploit the fact that the $(i,j)$-th element of $\bB^{\intercal}\bV_t\bB$ is given by:
	\begin{align*}
	\left[\bB^{\intercal}\bV_t\bB\right]_{i,j} & = \sum_{k=j+1}^d \beta_{k,i}\beta_{k,j}\nu_{k,t}, \quad i \leq j \quad \text{and} \quad \left[\bB^{\intercal}\bV_t\bB\right]_{i,j} = \left[\bB^{\intercal}\bV_t\bB\right]_{j,i}
	\end{align*}
	hence
	\begin{align*}
 \mathbb{E}_{\mathbf{B},\mathbf{V}_t}\left[\bB^{\intercal}\bV_t\bB\right]_{i,j} & = \mathbb{E}_{\mathbf{B},\mathbf{V}_t}\left[\sum_{k=j+1}^d \beta_{k,i}\beta_{k,j}\nu_{k,t}\right] \\
	&= \sum_{k=j+1}^d \left(\mu_{q(\beta_{k,i})}\mu_{q(\beta_{k,j})}+ \Cov(\beta_{k,i},\beta_{k,j})\right)\mu_{q(\nu_{k,t})} \\
	&= \sum_{k=j+1}^d \mu_{q(\beta_{k,i})}\mu_{q(\beta_{k,j})}\mu_{q(\nu_{k,t})} + \sum_{k=j+1}^d \Cov\left(\beta_{k,i},\beta_{k,j}\right)\mu_{q(\nu_{k,t})} \\
	&= \left[\boldsymbol{\mu}_{q(\bB^{\intercal})}\boldsymbol{\mu}_{q(\bV_t)}\boldsymbol{\mu}_{q(\bB)}\right]_{i,j} + \sum_{k=j+1}^d \Cov\left(\beta_{k,i},\beta_{k,j})\right)\mu_{q(\nu_{k,t})}.
	\end{align*}
	Thus, each element of $\bC_{\mathbf{\vartheta},t}$ is given by
	\begin{align*}
	\left[\bC_{\mathbf{\vartheta},t}\right]_{i,j} = \sum_{k=j+1}^d \Cov(\beta_{k,i},\beta_{k,j})\mu_{q(\nu_{k,t})} = \left[\bC_{\mathbf{\vartheta},t}\right]_{j,i}.
	\end{align*}
	Take the exponential of the $\log q^*(\boldsymbol{\vartheta})$ derived above and notice that it coincides with the kernel of a Gaussian random variable $\mathsf{N}_{d(d+p+1)}(\bmu_{q(\mathbf{\vartheta})}, \bSigma_{q(\mathbf{\vartheta})})$, as defined in Proposition \ref{prop:theta_VB}.
\end{proof}
\begin{prop}\label{prop:theta_VB_row}
	The optimal variational density for the parameter $\boldsymbol{\vartheta}_j$ is equal to a multivariate Gaussian $q^*(\boldsymbol{\vartheta}_j) \equiv \mathsf{N}_{d+p+1}(\bmu_{q(\mathbf{\vartheta}_j)}, \bSigma_{q(\mathbf{\vartheta}_j)})$, where, for each row $j=1,\ldots,d$ of $\bTheta$:
	\begin{equation}\begin{aligned}\label{eq:up_theta_VB_row}
	\bSigma_{q(\mathbf{\vartheta}_j)} &= \left(\sum_{t=1}^T\bmu_{q(\omega_{j,j,t})} \mathbf{z}_{t-1}\mathbf{z}_{t-1}^\intercal+1/\upsilon\bI_{d+p+1}\right)^{-1},\\
	\bmu_{q(\mathbf{\vartheta}_j)} &= \bSigma_{q(\mathbf{\vartheta}_j)}\left(\sum_{t=1}^T\left(\bmu_{q(\mathbf{\omega}_{j,t})} \otimes\mathbf{z}_{t-1}\right)\mathbf{y}_t-\sum_{t=1}^T\left( \bmu_{q(\mathbf{\omega}_{j,-j,t})}\otimes\mathbf{z}_{t-1}\mathbf{z}_{t-1}^\intercal\right)\bmu_{q(\mathbf{\vartheta}_{-j})}\right) .
	\end{aligned}\end{equation}
	Under this setting the vector $\mathbf{k}_{\mathbf{\vartheta},t}$ computed for $q^*(\nu_j)$ and $q^*(\bbeta_j)$ is a null vector since the independence among rows of $\bTheta$ is assumed. Again, the homoskedastic scenario is recovered with constant elements $\bmu_{q(\omega_{j,j})}$, $\bmu_{q(\omega_{j})}$, and $\bmu_{q(\omega_{j,-j})}$.
\end{prop}
\begin{proof}
	Consider the setting as in Proposition \ref{prop:theta_VB}, define $\bmu_{q(\mathbf{\Omega}_t)}=\mathbb{E}_{-\mathbf{\vartheta}}\left[(\bI_d-\bB)^\intercal\bV_t(\bI_d-\bB)\right]$ the expectation of the precision matrix and compute the optimal variational density for the parameter $\boldsymbol{\vartheta}_j$ as $\log q^*(\boldsymbol{\vartheta}_j)\propto \mathbb{E}_{-\mathbf{\vartheta}_j}\left[\ell(\bxi;\mathbf{y},\mathbf{x})+\log p(\boldsymbol{\vartheta}_j)\right]$:
	\begin{equation*}\begin{aligned}
	\log q^*(\boldsymbol{\vartheta}_j) &\propto -\frac{1}{2}\mathbb{E}_{-\boldsymbol{\vartheta}_j}\left[\boldsymbol{\vartheta}\right]^\intercal\left(\sum_{t=1}^T(\bmu_{q(\mathbf{\Omega}_t)}\otimes\mathbf{z}_{t-1}\mathbf{z}_{t-1}^\intercal)\right)\mathbb{E}_{-\boldsymbol{\vartheta}_j}\left[\boldsymbol{\vartheta}\right] -\frac{1}{2\upsilon}\boldsymbol{\vartheta}_j^\intercal\boldsymbol{\vartheta}_j \\
	&\qquad +\mathbb{E}_{-\boldsymbol{\vartheta}_j}\left[\boldsymbol{\vartheta}\right]^\intercal\sum_{t=1}^T\left(\bmu_{q(\mathbf{\Omega}_t)} \otimes\mathbf{z}_{t-1}\right)\mathbf{y}_t \\
	&\propto -\frac{1}{2}\boldsymbol{\vartheta}_j^\intercal\left( \sum_{t=1}^T\bmu_{q(\omega_{j,j,t})}\mathbf{z}_{t-1}\mathbf{z}_{t-1}^\intercal\right)\boldsymbol{\vartheta}_j -\frac{1}{2\upsilon}\boldsymbol{\vartheta}_j^\intercal\boldsymbol{\vartheta}_j\\
	&\qquad +\boldsymbol{\vartheta}_j^\intercal\sum_{t=1}^T\left(\bmu_{q(\mathbf{\omega}_{j,t})} \otimes\mathbf{z}_{t-1}\right)\mathbf{y}_t - \boldsymbol{\vartheta}_j^\intercal \sum_{t=1}^T\left(\bmu_{q(\mathbf{\omega}_{j,-j,t})}\otimes\mathbf{z}_{t-1}\mathbf{z}_{t-1}^\intercal\right)\bmu_{q(\mathbf{\vartheta}_{-j})}.
	\end{aligned}\end{equation*}
	Where we used the following partitions:
	\begin{align*}
	\boldsymbol{\vartheta} = \left(\begin{array}{c}
	    \boldsymbol{\vartheta}_j \\ \boldsymbol{\vartheta}_{-j}
	\end{array}\right), \qquad
	\bOmega_t = \left(\begin{array}{cc}
	    \omega_{j,j,t} & \boldsymbol{\omega}_{j,-j,t} \\
	    \boldsymbol{\omega}_{-j,j,t} & \bOmega_{-j,-j,t}
	\end{array}\right),
	\end{align*}
	and we denote with $\boldsymbol{\omega}_{j,t}$ the $j$-th row of $\bOmega_t$.
	Re-arrange the terms, take the exponential of the $\log q^*(\boldsymbol{\vartheta}_j)$ derived above and notice that it coincides with the kernel of a Gaussian random variable $\mathsf{N}_{d+p+1}(\bmu_{q(\mathbf{\vartheta}_j)}, \bSigma_{q(\mathbf{\vartheta}_j)})$, as defined in Proposition \ref{prop:theta_VB_row}.
\end{proof}
\begin{prop}\label{prop:psi_VB}
	The optimal variational density for the conditional variance parameter $\psi_j$ is an inverse-gamma distribution $q(\psi_j)\equiv\mathsf{InvGa}(A_{q(\psi_j)},B_{q(\psi_j)})$, where:
	\begin{equation}\begin{aligned}\label{eq:up_psi}
		A_{q(\psi_j)} &=  A_\psi + \frac{n+1}{2}\\
		B_{q(\psi_j)} &=  B_\psi + \frac{1}{2}\bmu_{q(\mathbf{h}_j)}^\intercal\mathbf{Q}\bmu_{q(\mathbf{h}_j)} +\frac{1}{2}\mathsf{tr}\left\{\bSigma_{q(\mathbf{h}_j)}\mathbf{Q}\right\},
	\end{aligned}\end{equation}
	and recall that $\mu_{q(1/\psi_j)}=A_{q(\psi_j)}/B_{q(\psi_j)}$.
\end{prop}
\begin{proof}
	Recall that a priori $\psi_j\sim\mathsf{InvGa}(A_\psi,B_\psi)$ and compute the optimal variational density as $\log q^*(\psi_j)\propto \mathbb{E}_{-\psi_j}\left[\log p(\bh_j|\psi_j)+\log p(\psi_j)\right]$:
	\begin{align*}
	    \log q(\eta^2) &\propto \mathbb{E}_{-\psi_j}\left[-\frac{n+1}{2}\log\psi_j-\frac{1}{2\psi_j}\bh_j^\intercal\mathbf{Q}\bh_j -(A_\psi+1)\log\psi_j-B_\psi/\psi_j\right]\\
     &\propto-\left(A_\psi+\frac{n+1}{2}+1\right)\log\psi_j-\frac{1}{\psi_j}\left(B_\psi+\frac{1}{2}\mathbb{E}_{h_j}\left[\bh_j^\intercal\bQ\bh_j\right]\right),
	\end{align*}
	where
	\begin{align*}
	    \mathbb{E}_{h_j}\left[\bh_j^\intercal\bQ\bh_j\right] &=  \bmu_{q(h_j)}^\intercal\mathbf{q}\bmu_{q(h_j)}+\mathsf{tr}\left\{\bSigma_{q(h_j)}\mathbf{Q}\right\}.
	\end{align*}
	Take the exponential and end up with the kernel of an inverse gamma distribution with parameters as in \eqref{eq:up_psi}.
\end{proof}
In what follows we derive analytically the variational lower bound. Notice that we consider the case of joint approximation $q(\boldsymbol{\vartheta})$, since it represents the more general case, while the lower bound under the further restriction $q(\boldsymbol{\vartheta})=\prod_{j=1}^dq(\boldsymbol{\vartheta}_j)$ can be recovered assuming a block-diagonal structure of $\bSigma_{q(\vartheta)}$ in \eqref{eq:elbo_VB} and \eqref{eq:elbo_VBSV}.
\begin{prop}\label{prop:elbo_VB}
	The variational lower bound for the non-sparse homoskedastic multivariate regression model can be derived analytically and it is equal to:
	\begin{equation}\begin{aligned}\label{eq:elbo_VB}
	\log\underline{p}(\by;q) &= d\left(-\frac{T}{2}\log 2\pi +a_{\nu}\log b_{\nu} -\log\Gamma(a_{\nu})\right) -\sum_{j=1}^d \left(a_{q(\nu_j)}\log b_{q(\nu_j)} -\log\Gamma(a_{q(\nu_j)})\right) \\
	&\qquad -\frac{1}{2}\sum_{j=2}^d\sum_{k=1}^{j-1}\left(\log\tau + 1/\tau\mu_{q(\beta^2_{j,k})}\right) +\frac{1}{2}\sum_{j=2}^d\left(\log\vert\bSigma_{q(\boldsymbol{\beta}_j)}\vert+(j-1)\right)\\
	&\qquad -\frac{1}{2}\sum_{j=1}^d\sum_{k=1}^{d+p+1}\left(\log\upsilon +1/\upsilon\mu_{q(\vartheta^2_{j,k})}\right) + \frac{1}{2}\left(\log|\bSigma_{q(\boldsymbol{\vartheta})}| +d(d+p+1)\right).
	\end{aligned}\end{equation}
\end{prop}
\begin{proof}
	First of all, notice that the lower bound can be written in terms of expected values with respect to the density $q$ as:
	\begin{align*}
	\log\underline{p}(\by;q) = \int q(\boldsymbol{\xi})\log\frac{p(\boldsymbol{\xi},\by)}{q(\boldsymbol{\xi})} \,d\boldsymbol{\xi} = \mathbb{E}_q\left[\log p(\boldsymbol{\xi},\by)\right]-\mathbb{E}_q\left[\log q(\boldsymbol{\xi})\right],
	\end{align*}
	where $\log p(\boldsymbol{\xi},\by)  = \ell(\boldsymbol{\xi}; \by)+ \log p(\boldsymbol{\xi})$. Following our model specification, we have that
	\begin{align*}
	\log p(\boldsymbol{\xi},\by) = \sum_{j=1}^d\left(\ell_j(\boldsymbol{\xi}; \by,\bx) + \log p(\nu_j) \right) +\sum_{j=2}^d\log p(\bbeta_j) +\log p(\boldsymbol{\vartheta}),
	\end{align*}
	where $\ell_j(\boldsymbol{\vartheta}; \by,\bx)$ denotes the log-likelihood for the $j$-th variable:
	\begin{align*}
	\ell_j(\bxi; \by,\bx) = -\frac{T}{2}\log 2\pi+\frac{T}{2}\log\nu_j-\frac{\nu_j}{2}\sum_{t=1}^T\left(y_{j,t}-\bbeta_j\mathbf{r}_{j,t}-\boldsymbol{\vartheta}_j\mathbf{z}_{t-1} \right)^2.
	\end{align*}
	Similarly for the variational density we have:
	\begin{align*}
	\log q(\bxi) = \sum_{j=1}^d\log q(\nu_j) +\sum_{j=2}^d\log q(\bbeta_j) +\log q(\boldsymbol{\vartheta}),
	\end{align*}
	and the lower bound can be divided into terms referring to each parameter:
	\begin{equation}\begin{aligned}\label{eq:elbo_terms}
	\log\underline{p}(\by;q) &= \sum_{j=1}^d\mathbb{E}_q\left[\ell_j(\bxi; \by,\bx)+\log p(\nu_j)-\log q(\nu_j) \right] \\
	&\qquad +\sum_{j=2}^d\mathbb{E}_q\left[\log p(\bbeta_j)-\log q(\bbeta_j)\right] +\mathbb{E}_q\left[\log p(\boldsymbol{\vartheta})-\log q(\boldsymbol{\vartheta})\right] \\
	&= \sum_{j=1}^d\big(\underbrace{\mathbb{E}_q\left[\ell_j(\bxi; \by,\bx) +\log\underline{p}(\by;\nu_j)\right]}_{A} +\sum_{j=2}^d\underbrace{\mathbb{E}_q\left[\log\underline{p}(\by;\bbeta_j)\right]}_{B} +\underbrace{\mathbb{E}_q\left[\log\underline{p}(\by;\boldsymbol{\vartheta})\right]}_{C},
	\end{aligned}\end{equation}
	thus our strategy will be to evaluate each piece in the latter separately and then put the results together.
	The first part of the lower bound we compute is $A=\ell_j(\bxi; \by,\bx) +\log\underline{p}(\by;\nu_j)$:
	\begin{align*}
	A &= \mathbb{E}_q\left[-\frac{T}{2}\log 2\pi+\frac{T}{2}\log\nu_j-\frac{\nu_j}{2}\sum_{t=1}^T\left(y_{j,t}-\bbeta_j\mathbf{r}_{j,t}-\boldsymbol{\vartheta}_j\mathbf{z}_{t-1} \right)^2\right] \\
	&\qquad + \mathbb{E}_q\left[a_{\nu}\log b_{\nu} -\log\Gamma(a_{\nu}) +(a_{\nu}-1)\log\nu_j-\nu_j b_{\nu}\right] \\
	&\qquad - \mathbb{E}_q\left[a_{q(\nu_j)}\log b_{q(\nu_j)} -\log\Gamma(a_{q(\nu_j)}) +(a_{q(\nu_j)}-1)\log\nu_j-\nu_j b_{q(\nu_j)}\right] \\
	&= -\frac{T}{2}\log 2\pi+\frac{T}{2}\mu_{q(\log\nu_j)}-\frac{\mu_{q(\nu_j)}}{2}\sum_{t=1}^T\mathbb{E}_q\left[\varepsilon_{j,t}^2\right] \\
	&\qquad + a_{\nu}\log b_{\nu} -\log\Gamma(a_{\nu}) +(a_{\nu}-1)\mu_{q(\log\nu_j)}-\mu_{q(\nu_j)} b_{\nu} \\
	&\qquad -a_{q(\nu_j)}\log b_{q(\nu_j)} +\log\Gamma(a_{q(\nu_j)}) -(a_{q(\nu_j)}-1)\mu_{q(\log\nu_j)}+\mu_{q(\nu_j)} b_{q(\nu_j)} \\
	&= -\frac{T}{2}\log 2\pi + a_{\nu}\log b_{\nu} -\log\Gamma(a_{\nu}) -a_{q(\nu_j)}\log b_{q(\nu_j)} +\log\Gamma(a_{q(\nu_j)}),
	\end{align*}
	where we exploit the definitions of $\mathbb{E}_q\left[\varepsilon_{j,t}^2\right], a_{q(\nu_j)}, b_{q(\nu_j)}$ given in Proposition \ref{prop:nu_VB}. The second term to compute is equal to:
	\begin{align*}
	B &= \mathbb{E}_q\left[-\frac{j-1}{2}\log 2\pi-\frac{1}{2}\sum_{k=1}^{j-1}\log\tau-\frac{1}{2\tau}\sum_{k=1}^{j-1}\beta^2_{j,k}\right] \\
	&\qquad -\mathbb{E}_q\bigg[-\frac{j-1}{2}\log 2\pi-\frac{1}{2}\log|\bSigma_{q(\boldsymbol{\beta}_j)}|-\frac{1}{2}\overbrace{(\bbeta_j-\bmu_{q(\boldsymbol{\beta}_j)})\bSigma_{q(\boldsymbol{\beta}_j)}^{-1}(\bbeta_j-\bmu_{q(\boldsymbol{\beta}_j)})^\intercal}^{\text{See Result \ref{res:exp_gauss}}}\bigg] \\
	&=-\frac{1}{2}\sum_{k=1}^{j-1}\log\tau-\frac{1}{2\tau}\sum_{k=1}^{j-1}\mu_{q(\beta^2_{j,k})} +\frac{1}{2}\log|\bSigma_{q(\boldsymbol{\beta}_j)}|+\frac{j-1}{2},
	\end{align*}
	where $\mu_{q(\beta^2_{j,k})}=\mu^2_{q(\beta_{j,k})}+\sigma^2_{q(\beta_{j,k})}$ and $\sigma^2_{q(\beta_{j,k})}$ denotes the $k$-th element on the diagonal of $\bSigma_{q(\boldsymbol{\beta}_j)}$. To conclude, we compute the last term:
	\begin{align*}
	C &= \mathbb{E}_q\left[-\frac{d(d+p+1)}{2}\log 2\pi-\frac{1}{2}\sum_{j=1}^{d}\sum_{k=1}^{d+p+1}\log\upsilon-\frac{1}{2\upsilon}\sum_{j=1}^{d}\sum_{k=1}^{d+p+1}\vartheta^2_{j,k}\right] \\
	&\qquad -\mathbb{E}_q\bigg[-\frac{d(d+p+1)}{2}\log 2\pi-\frac{1}{2}\log|\bSigma_{q(\mathbf{\vartheta})}|-\frac{1}{2}\overbrace{(\boldsymbol{\vartheta}-\bmu_{q(\mathbf{\vartheta})})^\intercal\bSigma_{q(\mathbf{\vartheta})}^{-1}(\boldsymbol{\vartheta}-\bmu_{q(\mathbf{\vartheta})})}^{\text{See Result \ref{res:exp_gauss}}}\bigg] \\
	&=-\frac{1}{2}\sum_{j=1}^{d}\sum_{k=1}^{d+p+1}\log\upsilon-\frac{1}{2\upsilon}\sum_{j=1}^{d}\sum_{k=1}^{d+p+1}\mu_{q(\vartheta^2_{j,k})} +\frac{1}{2}\log|\bSigma_{q(\mathbf{\vartheta})}|+\frac{d(d+p+1)}{2}.
	\end{align*}
	Put together the terms $A,B,C$ as in \eqref{eq:elbo_terms} and notice that the variational lower bound here computed coincides with the one presented in Proposition \ref{prop:elbo_VB}.
\end{proof}
\begin{prop}\label{prop:elbo_VBSV}
	The variational lower bound for the non-sparse multivariate regression model with stochastic volatility can be derived analytically and it is equal to:
	\begin{equation}\begin{aligned}\label{eq:elbo_VBSV}
	\log\underline{p}(\by;q) &= d\left(-\frac{T}{2}\log 2\pi +\frac{T+1}{2}-\frac{1}{2}\log k_0+a_{\psi}\log b_{\psi} -\log\Gamma(a_{\psi})\right) \\
 &\qquad +\frac{1}{2}\sum_{j=1}^d\sum_{t=1}^T\mu_{q(h_{j,t})}-\frac{1}{2}\sum_{j=1}^d\sum_{t=1}^T\exp(-\mu_{q(h_{j,t})}+1/2\sigma^2_{q(h_{j,t})})\mathbb{E}_q\left[\varepsilon_{j,t}^2\right] \\
 &\qquad+\frac{1}{2}\sum_{j=1}^d\log|\bSigma_{q(h_j)}|-\sum_{j=1}^d \left(a_{q(\psi_j)}\log b_{q(\psi_j)} -\log\Gamma(a_{q(\psi_j)})\right) \\
	&\qquad -\frac{1}{2}\sum_{j=2}^d\sum_{k=1}^{j-1}\left(\log\tau + 1/\tau\mu_{q(\beta^2_{j,k})}\right) +\frac{1}{2}\sum_{j=2}^d\left(\log\vert\bSigma_{q(\boldsymbol{\beta}_j)}\vert+(j-1)\right)\\
	&\qquad -\frac{1}{2}\sum_{j=1}^d\sum_{k=1}^{d+p+1}\left(\log\upsilon +1/\upsilon\mu_{q(\vartheta^2_{j,k})}\right) + \frac{1}{2}\left(\log|\bSigma_{q(\boldsymbol{\vartheta})}| +d(d+p+1)\right).
	\end{aligned}\end{equation}
\end{prop}
\begin{proof}
	Under the heteroskedastic model specification, we have that
	\begin{align*}
	\log p(\boldsymbol{\xi},\by) = \sum_{j=1}^d\left(\ell_j(\boldsymbol{\xi}; \by,\bx) + \log p(\bh_j) +\log p(\psi_j) \right) +\sum_{j=2}^d\log p(\bbeta_j) +\log p(\boldsymbol{\vartheta}),
	\end{align*}
	where $\ell_j(\boldsymbol{\vartheta}; \by,\bx)$ denotes the log-likelihood for the $j$-th variable:
	\begin{align*}
	\ell_j(\bxi; \by,\bx) = -\frac{T}{2}\log 2\pi-\frac{1}{2}\sum_{t=1}^T h_{j,t}-\frac{1}{2}\sum_{t=1}^T\exp(-h_{j,t})\left(y_{j,t}-\bbeta_j\mathbf{r}_{j,t}-\boldsymbol{\vartheta}_j\mathbf{z}_{t-1} \right)^2.
	\end{align*}
	Similarly for the variational density we have:
	\begin{align*}
	\log q(\bxi) = \sum_{j=1}^d(\log q(\bh_j) +\log q(\psi_j)) +\sum_{j=2}^d\log q(\bbeta_j) +\log q(\boldsymbol{\vartheta}),
	\end{align*}
	and the lower bound can be divided into terms referring to each parameter:
	\begin{equation}\begin{aligned}\label{eq:elbo_terms_SV}
	\log\underline{p}(\by;q) &= \sum_{j=1}^d\mathbb{E}_q\left[\ell_j(\bxi; \by,\bx)+\log p(\bh_j)-\log q(\bh_j) +\log p(\psi_j)-\log q(\psi_j) \right] \\
	&\qquad +\sum_{j=2}^d\mathbb{E}_q\left[\log p(\bbeta_j)-\log q(\bbeta_j)\right] +\mathbb{E}_q\left[\log p(\boldsymbol{\vartheta})-\log q(\boldsymbol{\vartheta})\right] \\
	&= \sum_{j=1}^d\big(\underbrace{\mathbb{E}_q\left[\ell_j(\bxi; \by,\bx) +\log\underline{p}(\by;\bh_j)+\log\underline{p}(\by;\psi_j)\right]}_{A} \\
 &\qquad+\sum_{j=2}^d\underbrace{\mathbb{E}_q\left[\log\underline{p}(\by;\bbeta_j)\right]}_{B} +\underbrace{\mathbb{E}_q\left[\log\underline{p}(\by;\boldsymbol{\vartheta})\right]}_{C},
	\end{aligned}\end{equation}
	thus our strategy will be to evaluate each piece in the latter separately and then put the results together.
 The terms B and C are the same computed for the homoskedastic model. The term $A=\ell_j(\bxi; \by,\bx) +\log\underline{p}(\by;\nu_j)$ is equal to:
	\begin{align*}
	A &= \mathbb{E}_q\left[-\frac{T}{2}\log 2\pi-\frac{1}{2}\sum_{t=1}^T h_{j,t}-\frac{1}{2}\sum_{t=1}^T\exp(-h_{j,t})\left(y_{j,t}-\bbeta_j\mathbf{r}_{j,t}-\boldsymbol{\vartheta}_j\mathbf{z}_{t-1} \right)^2\right] \\
 &\qquad +\mathbb{E}_q\left[-\frac{T+1}{2}\log 2\pi -\frac{T+1}{2}\log\psi_j +\frac{1}{2}\underbrace{\log|\mathbf{Q}|}_{=-\log k_0}-\frac{1}{2\psi_j}\bh_j^\intercal\mathbf{Q}\bh_j\right] \\
 &\qquad -\mathbb{E}_q\left[-\frac{T+1}{2}\log 2\pi -\frac{1}{2}\log|\mathbf{\Sigma}_{q(h_j)}|-\frac{1}{2}\overbrace{(\bh_j-\bmu_{q(h_j)})^\intercal\mathbf{\Sigma}_{q(h_j)}^{-1}(\bh_j-\bmu_{q(h_j)})}^{\text{See Result \ref{res:exp_gauss}}}\right] \\
	&\qquad + \mathbb{E}_q\left[a_{\psi}\log b_{\psi} -\log\Gamma(a_{\psi}) -(a_{\psi}+1)\log\psi_j-b_{\psi}/\psi_j\right] \\
	&\qquad - \mathbb{E}_q\left[a_{q(\psi_j)}\log b_{q(\psi_j)} -\log\Gamma(a_{q(\psi_j)}) -(a_{q(\psi_j)}+1)\log\psi_j-b_{q(\psi_j)}/\psi_j \right] \\
	&= -\frac{T}{2}\log 2\pi+\frac{1}{2}\sum_{t=1}^T\mu_{q(h_{j,t})}-\frac{1}{2}\sum_{t=1}^T\exp(-\mu_{q(h_{j,t})}+1/2\sigma^2_{q(h_{j,t})})\mathbb{E}_q\left[\varepsilon_{j,t}^2\right] \\
 &\qquad -\frac{T+1}{2}\mu_{q(\log\psi_j)}-\frac{1}{2}\log k_0 -\frac{1}{2}\mu_{q(1/\psi_j)}\mathbb{E}_{h_j}\left[\bh_j\mathbf{Q}\bh_j\right] +\frac{1}{2}\log|\bSigma_{q(h_j)}|+\frac{T+1}{2}\\
	&\qquad + a_{\psi}\log b_{\psi} -\log\Gamma(a_{\psi}) -(a_{\psi}+1)\mu_{q(\log\psi_j)}-\mu_{q(1/\psi_j)} b_{\psi} \\
	&\qquad -a_{q(\psi_j)}\log b_{q(\psi_j)} +\log\Gamma(a_{q(\psi_j)}) +(a_{q(\psi_j)}+1)\mu_{q(\log\psi_j)}+\mu_{q(1/\psi_j)} b_{q(\psi_j)} \\
	&= -\frac{T}{2}\log 2\pi+\frac{1}{2}\sum_{t=1}^T\mu_{q(h_{j,t})}-\frac{1}{2}\sum_{t=1}^T\exp(-\mu_{q(h_{j,t})}+1/2\sigma^2_{q(h_{j,t})})\mathbb{E}_q\left[\varepsilon_{j,t}^2\right]+\frac{1}{2}\log|\bSigma_{q(h_j)}| \\
 &\qquad+\frac{T+1}{2}-\frac{1}{2}\log k_0 + a_{\psi}\log b_{\psi} -\log\Gamma(a_{\psi}) -a_{q(\psi_j)}\log b_{q(\psi_j)} +\log\Gamma(a_{q(\psi_j)}),
	\end{align*}
	where $\mathbb{E}_q\left[\varepsilon_{j,t}^2\right]$ is defined in Proposition \ref{prop:h_VB}, and to make some simplifications we exploit the definitions of $a_{q(\psi_j)}, b_{q(\psi_j)}$ given in Proposition \ref{prop:psi_VB}. 
	Put together the terms $A,B,C$ as in \eqref{eq:elbo_terms_SV} and notice that the variational lower bound here computed coincides with the one presented in Proposition \ref{prop:elbo_VBSV}.
\end{proof}
The moments of the optimal variational densities are updated at each iteration of the Algorithm \ref{code:VB} and the convergence is assessed by checking the variation both in the lower bound and the parameters.
\begin{table}[!ht]
	\begin{algorithm}[H]
		\SetAlgoLined
		\kwInit{$q^*(\bxi)$, $\Delta_{\mathbf{\xi}}$, $\Delta_{\text{ELBO}}$}
		\While{$\big(\widehat{\Delta}_{\text{ELBO}}>\Delta_{\text{ELBO}}\big) \lor \big(\widehat{\Delta}_{\mathbf{\xi}}>\Delta_{\mathbf{\xi}}\big)$}{
			Update $q^*(\nu_1)$ as in \eqref{eq:up_nu_VB} (homoskedastic); \\
                Update $q^*(\bh_{1})$ and therefore $q^*(\boldsymbol{\nu}_1)$ as in \eqref{eq:up_h_VB} and \eqref{eq:up_nu_t_VB}  (heteroskedastic); \\
                Update $q^*(\psi_1)$ as in \eqref{eq:up_psi};\\
			\For{$j=2,\ldots,d$}{
                Update $q^*(\nu_j)$ as in \eqref{eq:up_nu_VB} (homoskedastic); \\
                Update $q^*(\bh_{j})$ and therefore $q^*(\boldsymbol{\nu}_j)$ as in \eqref{eq:up_h_VB} and \eqref{eq:up_nu_t_VB}  (heteroskedastic); \\
                Update $q^*(\psi_j)$ as in \eqref{eq:up_psi};\\
				Update $q^*(\bbeta_j)$ as in \eqref{eq:up_beta_VB}; \\
			}
			Update $q^*(\boldsymbol{\vartheta})$ as in \eqref{eq:up_theta_VB} or \eqref{eq:up_theta_VB_row}; \\
			Compute $\log\underline{p}\left(\by;q\right)$ as in \eqref{eq:elbo_VB} (homoskedastic) or \eqref{eq:elbo_VBSV} (heteroskedastic); \\
			Compute $\widehat{\Delta}_{\text{ELBO}} = \log\underline{p}\left(\by;q\right)^{(\iter)}-\log\underline{p}\left(\by;q\right)^{(\iter-1)}$;\\
			Compute $\widehat{\Delta}_{\mathbf{\xi}} = q^*(\boldsymbol{\xi})^{(\iter)}-q^*(\boldsymbol{\xi})^{(\iter-1)}$ ;
		}
		\caption{MFVB with non-informative prior.}
		\label{code:VB}
	\end{algorithm}
\end{table}

\subsection{Bayesian adaptive lasso}
\label{app:VBVAR-LASSO}
In order to induce shrinkage towards zero in the estimates of the coefficients $\boldsymbol{\vartheta}$, we assume an adaptive lasso prior. Notice that the optimal densities for $\bh_j$, $\nu_j$, and for the cholesky factor rows $\bbeta_j$ remain exactly the same computed in Section \ref{app:non_sparseVBAR}. The changes in the optimal densities $q^*(\boldsymbol{\vartheta})$ consist in the fact that now the prior variances are no more fixed, but random variables themselves.
\begin{prop}\label{prop:LASSO_onTheta}
	The joint optimal variational density for the parameter $\boldsymbol{\vartheta}$ is equal to $q^*(\boldsymbol{\vartheta}) \equiv \mathsf{N}_{d(d+p+1)}(\bmu_{q(\mathbf{\vartheta})}, \bSigma_{q(\mathbf{\vartheta})})$, where:
	\begin{equation}\begin{aligned}\label{eq:up_theta_VBLASSO}
	\bSigma_{q(\mathbf{\vartheta})} &= \left( \sum_{t=1}^T\bmu_{q(\mathbf{\Omega}_t)} \otimes\mathbf{z}_{t-1}\mathbf{z}_{t-1}^\intercal+\mathsf{Diag}(\bmu_{q(1/\mathbf{\upsilon})})\right)^{-1},\qquad\bmu_{q(\mathbf{\vartheta})} &= \bSigma_{q(\mathbf{\vartheta})} \sum_{t=1}^T\left(\bmu_{q(\mathbf{\Omega}_t)}\otimes\mathbf{z}_{t-1}\right)\mathbf{y}_t,
	\end{aligned}\end{equation}
	where $\Diag(\bmu_{q(1/\mathbf{\upsilon})})$ is a diagonal matrix where $\bmu_{q(1/\mathbf{\upsilon})}=(\mu_{q(1/\upsilon_{1,1})},\mu_{q(1/\upsilon_{1,2})},\ldots,\mu_{q(1/\upsilon_{d,d+p+1})})$.\par
	Under the row-independence assumption, the optimal variational density for the parameter $\boldsymbol{\vartheta}_j$ is equal to $q^*(\boldsymbol{\vartheta}_j) \equiv \mathsf{N}_{d+p+1}(\bmu_{q(\mathbf{\vartheta}_j)}, \bSigma_{q(\mathbf{\vartheta}_j)})$, where:
 \begin{equation}\begin{aligned}\label{eq:up_theta_VBLASSO_row}
	\bSigma_{q(\mathbf{\vartheta}_j)} &= \left( \sum_{t=1}^T\bmu_{q(\omega_{j,j,t})}\mathbf{z}_{t-1}\mathbf{z}_{t-1}^\intercal+\Diag(\bmu_{q(1/\mathbf{\upsilon}_j)})\right)^{-1},\\
	\bmu_{q(\mathbf{\vartheta}_j)} &= \bSigma_{q(\mathbf{\vartheta}_j)}\left(\sum_{t=1}^T\left(\bmu_{q(\mathbf{\omega}_{j,t})} \otimes\mathbf{z}_{t-1}\right)\mathbf{y}_t-\sum_{t=1}^T\left(\bmu_{q(\mathbf{\omega}_{j,-j,t})}\otimes\mathbf{z}_{t-1}\mathbf{z}_{t-1}^\intercal\right)\bmu_{q(\mathbf{\vartheta}_{-j})}\right),
	\end{aligned}\end{equation}
	where $\Diag(\bmu_{q(1/\mathbf{\upsilon}_j)})$ is a diagonal matrix where $\bmu_{q(1/\mathbf{\upsilon}_j)}=(\mu_{q(1/\upsilon_{j,1})},\mu_{q(1/\upsilon_{j,2})},\ldots,\mu_{q(1/\upsilon_{j,d+p+1})})$.
\end{prop}
\noindent Hereafter we describe the optimal densities for the parameters used in hierarchical specification of the prior here assumed.
\begin{prop}\label{prop:ups_LASSO}
	The optimal density for the prior variance $1/\upsilon_{j,k}$ is equal to an inverse Gaussian distribution $q^*(1/\upsilon_{j,k}) \equiv \mathsf{IG}(a_{q(1/\upsilon_{j,k})},b_{q(1/\upsilon_{j,k})})$, where, for each $j=1,\ldots,d$ and $k=1,\ldots,d+p+1$:
	\begin{align}\label{eq:up_ups_VBLASSO}
	a_{q(1/\upsilon_{j,k})} = \mu_{q(\vartheta^2_{j,k})}, \quad b_{q(1/\upsilon_{j,k})} = \mu_{q(\lambda^2_{j,k})}.
	\end{align}
	Moreover, it is useful to know that
	\begin{align*}
	&\mu_{q(1/\upsilon_{j,k})} = \sqrt{b_{q(1/\upsilon_{j,k})}/a_{q(1/\upsilon_{j,k})}},\quad \mu_{q(\upsilon_{j,k})} = \sqrt{a_{q(1/\upsilon_{j,k})}/b_{q(1/\upsilon_{j,k})}}+1/b_{q(1/\upsilon_{j,k})}.
	\end{align*}
\end{prop}
\begin{proof}
	Consider the prior specification which involves the parameter $\upsilon_{j,k}$:
	\begin{equation*}\begin{aligned}
	\vartheta_{j,k}|\upsilon_{j,k} \sim \mathsf{N}(0,\upsilon_{j,k}), \qquad \upsilon_{j,k}|\lambda^2_{j,k} \sim \mathsf{Exp}\left(\lambda^2_{j,k}/2\right).
	\end{aligned}\end{equation*}
	Compute the optimal variational density $\log q^*(\upsilon_{j,k})\propto \mathbb{E}_{-\mathbf{\upsilon_{j,k}}}\left[\log p(\vartheta_{j,k})+\log p(\upsilon_{j,k})\right]$:
	\begin{equation*}\begin{aligned}
	\log q^*(\upsilon_{j,k}) &\propto \mathbb{E}_{-\upsilon_{j,k}} \left[-\frac{1}{2}\log\upsilon_{j,k} -\frac{1}{2\upsilon_{j,k}}\vartheta_{j,k}^2-\upsilon_{j,k}\frac{\lambda^2_{j,k}}{2}\right]\\
	&\propto -1/2\log\upsilon_{j,k} -\frac{1}{2\upsilon_{j,k}}\mu_{q(\vartheta_{j,k}^2)}-\upsilon_{j,k}\frac{\mu_{q(\lambda^2_{j,k})}}{2} ,
	\end{aligned}\end{equation*}
	and, as a consequence, we obtain:
	\begin{equation*}\begin{aligned}\label{eq:deriv_ups_LASSO}
	\log q^*(1/\upsilon_{j,k}) &\propto -3/2\log(1/\upsilon_{j,k}) -\frac{1}{2}(1/\upsilon_{j,k})\mu_{q(\vartheta_{j,k}^2)}-\frac{\mu_{q(\lambda^2_{j,k})}}{2(1/\upsilon_{j,k})}.
	\end{aligned}\end{equation*}
	Take the exponential and notice that the latter is the kernel of an inverse Gaussian random variable $\mathsf{IG}(a_{q(1/\upsilon_{j,k})},b_{q(1/\upsilon_{j,k})})$, as defined in Proposition \ref{prop:ups_LASSO}. 
\end{proof}
\begin{prop}\label{prop:lam_LASSO}
	The optimal density for the latent parameter $\lambda^2_{j,k}$ for $j=1,\ldots,d$ and $k=1,\ldots,d+p+1$ is equal to a $q^*(\lambda^2_{j,k}) \equiv \mathsf{Ga}(a_{q(\lambda^2_{j,k})},b_{q(\lambda^2_{j,k})})$, where:
	\begin{align}\label{eq:up_lam_VBLASSO}
	a_{q(\lambda^2_{j,k})} = h_1+1 ,\quad b_{q(\lambda^2_{j,k})} = \mu_{q(\upsilon_{j,k})}/2 + h_2 .
	\end{align}
\end{prop}
\begin{proof}
	Consider the prior specification which involves the parameter $\lambda^2_{j,k}$:
	\begin{equation*}\begin{aligned}
	\upsilon_{j,k}|\lambda^2_{j,k} \sim \mathsf{Exp}\left(\lambda^2_{j,k}/2\right), \quad \lambda^2_{j,k}\sim\mathsf{Ga}(h_1,h_2).
	\end{aligned}\end{equation*}
	Compute the optimal variational density as $\log q^*(\lambda^2_{j,k})\propto \mathbb{E}_{-\mathbf{\lambda^2_{j,k}}}\left[\log p(\upsilon_{j,k})+\log p(\lambda^2_{j,k})\right]$:
	\begin{equation*}\begin{aligned}\label{eq:deriv_lam_LASSO}
	\log q^*(\lambda^2_{j,k}) &\propto \mathbb{E}_{-\lambda^2_{j,k}} \left[h_1\log\lambda^2_{j,k}-\lambda^2_{j,k}\left(\upsilon_{j,k}/2+h_2\right)\right] \\
	&\propto h_1\log\lambda^2_{j,k}-\lambda^2_{j,k}\left(\mu_{q(\upsilon_{j,k})}/2+ h_2\right),
	\end{aligned}\end{equation*}
	then take the exponential and notice that the latter is the kernel of a gamma random variable $\mathsf{Ga}(a_{q(\lambda^2_{j,k})},b_{q(\lambda^2_{j,k})})$, as defined in Proposition \ref{prop:lam_LASSO}. 
\end{proof}
\begin{prop}\label{prop:elbo_VBLASSO}
	The variational lower bound for the multivariate regression model with adaptive Bayesian lasso prior can be derived analytically and it is equal to:
	\begin{equation}\begin{aligned}\label{eq:elbo_VBLASSO}
	\log\underline{p}(\by;q) &= \log\underline{p}^{\text{SV}}(\by;\boldsymbol{\beta},\bh,\bpsi) \,\,\,\, \left(\text{or } \log\underline{p}^{\text{C}}(\by;\boldsymbol{\beta},\bnu) \,\, \text{if homoskedastic}\right) \\
	&\qquad + \frac{1}{2}\left(\log|\bSigma_{q(\boldsymbol{\vartheta})}| +d(d+p+1)\right)+\sum_{j=1}^d\sum_{k=1}^{d+p+1}\frac{1}{2}\mu_{q(\lambda^2_{j,k})}\mu_{q(\upsilon_{j,k})} \\
	&\qquad -\sum_{j=1}^d\sum_{k=1}^{d+p+1} (1/4\log(b_{q(1/\upsilon_{j,k})}/a_{q(1/\upsilon_{j,k})})-\log K_{1/2}(\sqrt{b_{q(1/\upsilon_{j,k})}a_{q(1/\upsilon_{j,k})}}))\\
	&\qquad +d(d+p+1)\left(h_1\log h_2 -\log\Gamma(h_1)\right)-\sum_{j=1}^d\sum_{k=1}^{d+p+1}\left(a_{q(\lambda^2_{j,k})}\log b_{q(\lambda^2_{j,k})} -\log\Gamma(a_{q(\lambda^2_{j,k})})\right),
	\end{aligned}\end{equation}
 where
 \begin{align*}
     \log\underline{p}^{\text{C}}(\by;\boldsymbol{\beta},\bnu) &= d\left(-\frac{T}{2}\log 2\pi +a_{\nu}\log b_{\nu} -\log\Gamma(a_{\nu})\right) -\sum_{j=1}^d \left(a_{q(\nu_j)}\log b_{q(\nu_j)} -\log\Gamma(a_{q(\nu_j)})\right) \\
     &\qquad -\frac{1}{2}\sum_{j=2}^d\sum_{k=1}^{j-1}\left(\log\tau + 1/\tau\mu_{q(\beta^2_{j,k})}\right) +\frac{1}{2}\sum_{j=2}^d\left(\log\vert\bSigma_{q(\boldsymbol{\beta}_j)}\vert+(j-1)\right)\\
     \log\underline{p}^{\text{SV}}(\by;\bbeta,\bh,\bpsi) &= d\left(-\frac{T}{2}\log 2\pi +\frac{T+1}{2}-\frac{1}{2}\log k_0+a_{\psi}\log b_{\psi} -\log\Gamma(a_{\psi})\right) \\
 &\qquad +\frac{1}{2}\sum_{j=1}^d\sum_{t=1}^T\mu_{q(h_{j,t})}-\frac{1}{2}\sum_{j=1}^d\sum_{t=1}^T\exp(-\mu_{q(h_{j,t})}+1/2\sigma^2_{q(h_{j,t})})\mathbb{E}_q\left[\varepsilon_{j,t}^2\right] \\
 &\qquad+\frac{1}{2}\sum_{j=1}^d\log|\bSigma_{q(h_j)}|-\sum_{j=1}^d \left(a_{q(\psi_j)}\log b_{q(\psi_j)} -\log\Gamma(a_{q(\psi_j)})\right) \\
	&\qquad -\frac{1}{2}\sum_{j=2}^d\sum_{k=1}^{j-1}\left(\log\tau + 1/\tau\mu_{q(\beta^2_{j,k})}\right) +\frac{1}{2}\sum_{j=2}^d\left(\log\vert\bSigma_{q(\boldsymbol{\beta}_j)}\vert+(j-1)\right)\\
 \end{align*}
\end{prop}
\begin{proof}
	As we did in \eqref{eq:elbo_terms} for Proposition \ref{prop:elbo_VB}, the lower bound can be divided into terms referring to each parameter:
	\begin{align}\label{eq:elboLASSO_terms}
	\log\underline{p}(\by;q) = A +\sum_{j=1}^d\sum_{k=1}^{d+p+1}\bigg(\underbrace{\mathbb{E}_q\left[\log\underline{p}(\by;\upsilon_{j,k})\right]}_{B} &+\underbrace{\mathbb{E}_q\left[\log\underline{p}(\by;\lambda^2_{j,k})\right]}_{C}\bigg), \nonumber
	\end{align}
	where A is equal to \eqref{eq:elbo_terms} in the previous non-informative model specification.
	Our strategy will be to evaluate each piece in the latter separately and then put the results together. Notice that the computations for the piece $A$ are already available from Proposition \ref{prop:elbo_VB} and they are equal to the lower bound for the model with the non-informative prior where we still have to take the expectations with respect to the latent parameters $\upsilon_{j,k}$. Thus, we have that:
	\begin{equation}\begin{aligned}\label{eq:A_lasso}
	A &= \log\underline{p}^{\text{SV}}(\by;\boldsymbol{\beta},\bh,\bpsi) \,\,\,\, \left(\text{or } \log\underline{p}^{\text{C}}(\by;\boldsymbol{\beta},\bnu) \,\, \text{if homoskedastic}\right) \\
	&\qquad -\frac{1}{2}\sum_{j=1}^d\sum_{k=1}^{d+p+1}\left(\mu_{q(\log\upsilon_{j,k})} +\mu_{q(1/\upsilon_{j,k})}\mu_{q(\vartheta^2_{j,k})}\right) + \frac{1}{2}\left(\log|\bSigma_{q(\boldsymbol{\vartheta})}| +d(d+p+1)\right).
	\end{aligned}\end{equation}
	Consider now the piece $B$ and recall that, since $q^*(1/\upsilon_{j,k})\equiv\mathsf{IG}(a_{q(\upsilon_{j,k})},b_{q(\upsilon_{j,k})})$, then its inverse follows $q^*(\upsilon_{j,k})\equiv\mathsf{GIG}(1/2,b_{q(1/\upsilon_{j,k})},a_{q(1/\upsilon_{j,k})})$. We have that
	\begin{align*}
	B &= \mathbb{E}_q\left[\log\lambda^2_{j,k}-\log 2-\upsilon_{j,k}\frac{\lambda^2_{j,k}}{2}\right] \\
	&\qquad-\mathbb{E}_q\left[h(1/2,b_{q(1/\upsilon_{j,k})},a_{q(1/\upsilon_{j,k})})-1/2\log\upsilon_{j,k}-\frac{1}{2}\left(b_{q(1/\upsilon_{j,k})}\upsilon_{j,k}+\frac{a_{q(1/\upsilon_{j,k})}}{\upsilon_{j,k}}\right)\right] \\
	&= \mu_{q(\log\lambda^2_{j,k})}-\log 2-h(1/2,b_{q(1/\upsilon_{j,k})},b_{q(1/\upsilon_{j,k})}) +1/2\mu_{q(\log\upsilon_{j,k})} \\
	&\qquad-\frac{1}{2}\left(\mu_{q(\upsilon_{j,k})}\mu_{q(\lambda^2_{j,k})}-b_{q(1/\upsilon_{j,k})}\mu_{q(\upsilon_{j,k})}-a_{q(1/\upsilon_{j,k})}\mu_{q(1/\upsilon_{j,k})}\right),
	\end{align*}
	where $h(\zeta,a,b)$ denotes the logarithm of the normalizing constant of a $\mathsf{GIG}$ distribution, i.e.
	\begin{equation*}
	h(\zeta,a,b) = \zeta/2\log(a/b)-\log 2-\log K_{\zeta}(\sqrt{ab}).
	\end{equation*}
	The term involving $\lambda^2_{j,k}$, for $j=1,\ldots,d$ and $k=1,\ldots,d+p+1$, is equal to:
	\begin{align*}
	C &= \mathbb{E}_q\left[h_1\log h_2 -\log\Gamma(h_1)+(h_1-1)\log\lambda^2_{j,k}-\lambda^2_{j,k}h_2\right]\\
	&\qquad-\mathbb{E}_q\left[a_{q(\lambda^2_{j,k})}\log b_{q(\lambda^2_{j,k})} -\log\Gamma(a_{q(\lambda^2_{j,k})})+(a_{q(\lambda^2_{j,k})}-1)\log\lambda^2_{j,k}-\lambda^2_{j,k}b_{q(\lambda^2_{j,k})}\right] \\
	&= h_1\log h_2 -\log\Gamma(h_1)+(h_1-1)\mu_{q(\log\lambda^2_{j,k})}-\mu_{q(\lambda^2_{j,k})}h_2\\
	&\qquad-a_{q(\lambda^2_{j,k})}\log b_{q(\lambda^2_{j,k})} +\log\Gamma(a_{q(\lambda^2_{j,k})})-(a_{q(\lambda^2_{j,k})}-1)\mu_{q(\log\lambda^2_{j,k})}+\mu_{q(\lambda^2_{j,k})}b_{q(\lambda^2_{j,k})}.
	\end{align*}
	Group together the terms and exploit the analytical form of the optimal parameters to perform some simplifications. The remaining terms form the lower bound for a multivariate regression model with adaptive lasso prior.
\end{proof}
\noindent The moments of the optimal variational densities are updated at each iteration of the Algorithm \ref{code:VBLASSO} and the convergence is assessed by checking the variation both in the lower bound and the parameters.
\begin{table}[!ht]
	\begin{algorithm}[H]
		\SetAlgoLined
		\kwInit{$q^*(\boldsymbol{\xi})$, $\Delta_{\mathbf{\xi}}$, $\Delta_{\text{ELBO}}$}
		\While{$\big(\widehat{\Delta}_{\text{ELBO}}>\Delta_{\text{ELBO}}\big) \lor \big(\widehat{\Delta}_{\mathbf{\xi}}>\Delta_{\mathbf{\xi}}\big)$}{
			Update $q^*(\nu_1)$ as in \eqref{eq:up_nu_VB} (homoskedastic); \\
                Update $q^*(\bh_{1})$ and therefore $q^*(\boldsymbol{\nu}_1)$ as in \eqref{eq:up_h_VB} and \eqref{eq:up_nu_t_VB}  (heteroskedastic); \\
                Update $q^*(\psi_1)$ as in \eqref{eq:up_psi};\\
			\For{$j=2,\ldots,d$}{
				Update $q^*(\nu_j)$ as in \eqref{eq:up_nu_VB} (homoskedastic); \\
                Update $q^*(\bh_{j})$ and therefore $q^*(\boldsymbol{\nu}_j)$ as in \eqref{eq:up_h_VB} and \eqref{eq:up_nu_t_VB}  (heteroskedastic); \\
                Update $q^*(\psi_j)$ as in \eqref{eq:up_psi};\\
				Update $q^*(\bbeta_j)$ as in \eqref{eq:up_beta_VB}; \\
			}
			Update $q^*(\boldsymbol{\vartheta})$ as in \eqref{eq:up_theta_VBLASSO} or \eqref{eq:up_theta_VBLASSO_row}; \\
			\For{$j=1,\ldots,d$}{
				\For{$k=1,\ldots,d+p+1$}{
					Update $q^*(\upsilon_{j,k})$, $q^*(\lambda^2_{j,k})$ as in \eqref{eq:up_ups_VBLASSO}-\eqref{eq:up_lam_VBLASSO}; \\
				}	
			}
			Compute $\log\underline{p}\left(\by;q\right)$ as in \eqref{eq:elbo_VBLASSO}; \\
			Compute $\widehat{\Delta}_{\text{ELBO}} = \log\underline{p}\left(\by;q\right)^{(\iter)}-\log\underline{p}\left(\by;q\right)^{(\iter-1)}$;\\
			Compute $\widehat{\Delta}_{\mathbf{\xi}} = q^*(\boldsymbol{\xi})^{(\iter)}-q^*(\boldsymbol{\xi})^{(\iter-1)}$ ;
		}
		\caption{MFVB with Bayesian adaptive lasso prior.}
		\label{code:VBLASSO}
	\end{algorithm}
\end{table}

\subsection{Adaptive normal-gamma}
\label{app:VBVAR-DG}
In order to induce shrinkage towards zero in the estimates of the coefficients, we assume an adaptive normal-gamma prior on $\boldsymbol{\vartheta}$. Notice that the optimal densities for $\bh_j$, $\nu_j$, and for the cholesky factor rows $\bbeta_j$ remain exactly the same computed in Section \ref{app:non_sparseVBAR}. The optimal density $q^*(\boldsymbol{\vartheta})$ has the same structure as the one computed in Proposition \eqref{prop:LASSO_onTheta} for the lasso prior.\\
\noindent Hereafter we describe the optimal densities for the parameters used in hierarchical specification of the normal-gamma prior.
\begin{prop}\label{prop:ups_DG}
	The optimal density for the prior variance $\upsilon_{j,k}$ is equal to a generalized inverse Gaussian distribution $q^*(\upsilon_{j,k}) \equiv \mathsf{GIG}(\zeta_{q(\upsilon_{j,k})},a_{q(\upsilon_{j,k})},b_{q(\upsilon_{j,k})})$, where, for $j=1,\ldots,d$ and $k=1,\ldots,d+p+1$:
	\begin{align}\label{eq:up_ups_VBDG}
	\zeta_{q(\upsilon_{j,k})} = \mu_{q(\eta_{j})}-1/2, \quad a_{q(\upsilon_{j,k})} = \mu_{q(\eta_{j})}\mu_{q(\lambda_{j,k})}, \quad b_{q(\upsilon_{j,k})} = \mu_{q(\vartheta_{j,k}^2)}.
	\end{align}
	Moreover, it is useful to know that
	\begin{align*}
	&\mu_{q(\upsilon_{j,k})} = \frac{\sqrt{b_{q(\upsilon_{j,k})}}K_{\zeta_{q(\upsilon_{j,k})}+1}\left(\sqrt{a_{q(\upsilon_{j,k})}b_{q(\upsilon_{j,k})}}\right)}{\sqrt{a_{q(\upsilon_{j,k})}}K_{\zeta_{q(\upsilon_{j,k})}}\left(\sqrt{a_{q(\upsilon_{j,k})}b_{q(\upsilon_{j,k})}}\right)} ,\\
	&\mu_{q(1/\upsilon_{j,k})} = \frac{\sqrt{a_{q(\upsilon_{j,k})}}K_{\zeta_{q(\upsilon_{j,k})}+1}\left(\sqrt{a_{q(\upsilon_{j,k})}b_{q(\upsilon_{j,k})}}\right)}{\sqrt{b_{q(\upsilon_{j,k})}}K_{\zeta_{q(\upsilon_{j,k})}}\left(\sqrt{a_{q(\upsilon_{j,k})}b_{q(\upsilon_{j,k})}}\right)}-\frac{2\zeta_{q(\upsilon_{j,k})}}{b_{q(\upsilon_{j,k})}},\\
	&\mu_{q(\log\upsilon_{j,k})} = \log\frac{\sqrt{b_{q(\upsilon_{j,k})}}}{\sqrt{a_{q(\upsilon_{j,k})}}}+\frac{\partial}{\partial \zeta_{q(\upsilon_{j,k})}}\log K_{\zeta_{q(\upsilon_{j,k})}}\left(\sqrt{a_{q(\upsilon_{j,k})}b_{q(\upsilon_{j,k})}}\right),
	\end{align*}
	where $K_{\zeta}(\cdot)$ denotes the modified Bessel function of second kind.
\end{prop}
\begin{proof}
	Consider the prior specification which involves the parameter $\upsilon_{j,k}$:
	\begin{equation*}\begin{aligned}
	\vartheta_{j,k}|\upsilon_{j,k} \sim \mathsf{N}(0,\upsilon_{j,k}), \qquad \upsilon_{j,k}|\eta_{j},\lambda_{j,k} \sim \mathsf{Ga}\left(\eta_{j},\frac{\eta_{j}\lambda_{j,k}}{2}\right).
	\end{aligned}\end{equation*}
	Compute the optimal variational density as $\log q^*(\upsilon_{j,k})\propto \mathbb{E}_{-\mathbf{\upsilon_{j,k}}}\left[\log p(\vartheta_{j,k})+\log p(\upsilon_{j,k})\right]$:
	\begin{equation*}\begin{aligned}\label{eq:deriv_ups_DG}
	\log q^*(\upsilon_{j,k}) &\propto \mathbb{E}_{-\upsilon_{j,k}} \left[-\frac{1}{2}\log\upsilon_{j,k} -\frac{1}{2\upsilon_{j,k}}\beta_{j,k}^2+(\eta_{j}-1)\log\upsilon_{j,k}-\upsilon_{j,k}\frac{\eta_{j}\lambda_{j,k}}{2}\right]\\
	&\propto \left(\mu_{q(\eta_{j})}-\frac{1}{2}-1\right)\log\upsilon_{j,k} -\frac{1}{2\upsilon_{j,k}}\mu_{q(\vartheta_{j,k}^2)}-\upsilon_{j,k}\frac{\mu_{q(\eta_{j})}\mu_{q(\lambda_{j,k})}}{2} ,
	\end{aligned}\end{equation*}
	where $\mu_{q(\vartheta_{j,k}^2)} = \sigma^2_{q(\vartheta_{j,k})}+\mu_{q(\vartheta_{j,k})}^2$.
	Take the exponential and notice that the latter is the kernel of a generalized inverse Gaussian random variable $\mathsf{GIG}(\zeta_{q(\upsilon_{j,k})},a_{q(\upsilon_{j,k})},b_{q(\upsilon_{j,k})})$, as defined in Proposition \ref{prop:ups_DG}. 
\end{proof}
\begin{prop}\label{prop:lam_DG}
	The optimal density for the latent parameter $\lambda_{j,k}$ for $j=1,\ldots,d$ and $k=1,\ldots,d+p+1$ is equal to a $q^*(\lambda_{j,k}) \equiv \mathsf{Ga}(a_{q(\lambda_{j,k})},b_{q(\lambda_{j,k})})$, where:
	\begin{align}\label{eq:up_lam_VBDG}
	a_{q(\lambda_{j,k})} = \mu_{q(\eta_{j})}+h_1 ,\quad b_{q(\lambda_{j,k})} = \frac{\mu_{q(\eta_{j})}\mu_{q(\upsilon_{j,k})}}{2}+ h_2 .
	\end{align}
	Moreover, it is useful to know that
	\begin{align*}
	\mu_{q(\lambda_{j,k})} = \frac{a_{q(\lambda_{j,k})}}{b_{q(\lambda_{j,k})}},\quad \mu_{q(\log\lambda_{j,k})} = -\log b_{q(\lambda_{j,k})}+\frac{\Gamma^{\prime}(a_{q(\lambda_{j,k})})}{\Gamma(a_{q(\lambda_{j,k})})}.
	\end{align*}
\end{prop}
\begin{proof}
	Consider the prior specification which involves the parameter $\lambda_{j,k}$:
	\begin{equation*}\begin{aligned}
	\upsilon_{j,k}|\eta_{j},\lambda_{j,k} \sim \mathsf{Ga}\left(\eta_{j},\frac{\eta_{j}\lambda_{j,k}}{2}\right), \quad \lambda_{j,k}\sim\mathsf{Ga}(h_1,h_2).
	\end{aligned}\end{equation*}
	Compute the optimal variational density as $\log q^*(\lambda_{j,k})\propto \mathbb{E}_{-\mathbf{\lambda_{j,k}}}\left[\log p(\upsilon_{j,k})+\log p(\lambda_{j,k})\right]$:
	\begin{equation}\begin{aligned}\label{eq:deriv_lam_DG}
	\log q^*(\lambda_{j,k}) &\propto \mathbb{E}_{-\lambda_{j,k}} \left[\left(\eta_{j}+h_1-1\right)\log\lambda_{j,k}-\lambda_{j,k}\left(\frac{\eta_{j}\upsilon_{j,k}}{2}+h_2\right)\right] \\
	&\propto\left(\mu_{q(\eta_{j})}+h_1-1\right)\log\lambda_{j,k}-\lambda_{j,k}\left(\frac{\mu_{q(\eta_{j})}\mu_{q(\upsilon_{j,k})}}{2}+ h_2\right),
	\end{aligned}\end{equation}
	then take the exponential and notice that the latter is the kernel of a gamma random variable $\mathsf{Ga}(a_{q(\lambda_{j,k})},b_{q(\lambda_{j,k})})$, as defined in Proposition \ref{prop:lam_DG}. 
\end{proof}
\begin{prop}\label{prop:eta_DG}
	The optimal density for the latent parameter $\eta_{j}$ for $j=1,\ldots,d$ is equal to:
	\begin{align}\label{eq:up_eta_VBDG}
	q^*(\eta_{j}) = \frac{h(\eta_{j})}{c_{\eta_{j}}}\exp\left\{-\eta_{j}\sum_{k=1}^{d+p+1}\left(\frac{\mu_{q(\lambda_{j,k})}\mu_{q(\upsilon_{j,k})}}{2}-\mu_{q(\log\lambda_{j,k})}-\mu_{q(\log\upsilon_{j,k})}+\log2+h_3\right)\right\},
	\end{align}
	where $\log h(\eta_{j}) = (d+p+1)(\eta_{j}\log\eta_{j}-\log\Gamma(\eta_{j}))$ and
	\begin{align*}
	c_{\eta_{j}} = {\displaystyle\int_{\mathbb{R^+}}} h(\eta_{j}) \exp\left\{-\eta_{j}\sum_{k=1}^{d+p+1}\left(\frac{\mu_{q(\lambda_{j,k})}\mu_{q(\upsilon_{j,k})}}{2}-\mu_{q(\log\lambda_{j,k})}-\mu_{q(\log\upsilon_{j,k})}+(d+p+1)\log2+h_3\right)\right\} \, d\eta_{j}.
	\end{align*}
	Then, we have that $\mu_{q(\eta_{j})} = \int_{\mathbb{R^+}} \eta_{j} q^*(\eta_{j})\, d\eta_{j}$.
\end{prop}
\begin{proof}
	Consider the prior specification which involves the parameter $\eta_{j}$:
	\begin{equation*}\begin{aligned}
	\upsilon_{j,k}|\eta_{j},\lambda_{j,k} \sim \mathsf{Ga}\left(\eta_{j},\frac{\eta_{j}\lambda_{j,k}}{2}\right), \quad \eta_{j}\sim\mathsf{Exp}(h_3).
	\end{aligned}\end{equation*}
	Compute the optimal variational density as $\log q^*(\eta_{j})\propto \mathbb{E}_{-\mathbf{\eta_{j}}}\left[\sum_{k=1}^{d+p+1}\log p(\upsilon_{j,k})+\log p(\eta_{j})\right]$:
	\begin{equation}\begin{aligned}\label{eq:deriv_eta_DG}
	\log q^*(\eta_{j}) &\propto \mathbb{E}_{-\eta_{j}}\bigg[(d+p+1)\left(\eta_{j}\log\eta_{j}-\log\Gamma(\eta_{j})\right) -\eta_{j}\sum_{k=1}^{d+p+1}\left(\left(\frac{\lambda_{j,k}\upsilon_{j,k}}{2}-\log\frac{\lambda_{j,k}\upsilon_{j,k}}{2}\right)+h_3 \right)\bigg] \\
	&= (d+p+1)\left(\eta_{j}\log\eta_{j}-\log\Gamma(\eta_{j})\right) \\
	&\qquad -\eta_{j}\sum_{k=1}^{d+p+1}\left( \frac{\mu_{q(\lambda_{j,k})}\mu_{q(\upsilon_{j,k})}}{2}-\mathbb{E}_{\upsilon_{j,k}\lambda_{j,k}} \left[\log\frac{\lambda_{j,k}\upsilon_{j,k}}{2} \right]+h_3 \right),
	\end{aligned}\end{equation}
	which is not the kernel of a know distribution, but since $\mathbb{E}\left[\log x\right] \le \log\mathbb{E}\left[x\right] < \mathbb{E}\left[x\right]$,
	it holds that 
	\begin{equation*}
	\frac{\mu_{q(\lambda_{j,k})}\mu_{q(\upsilon_{j,k})}}{2} > \mathbb{E}_{\upsilon_{j,k}\lambda_{j,k}} \left[\log\frac{\lambda_{j,k}\upsilon_{j,k}}{2} \right] = \mu_{q(\log\lambda_{j,k})}+\mu_{q(\log\upsilon_{j,k})}-\log2,
	\end{equation*}
	hence the exponential of term in \eqref{eq:deriv_eta_DG} is integrable and thus we can compute the normalizing constant and its expectation.
\end{proof}
\begin{prop}\label{prop:elbo_VBDG}
	The variational lower bound for the multivariate regression model with adaptive normal-gamma prior can be derived analytically and it is equal to:
	\begin{equation}\begin{aligned}\label{eq:elbo_VBDG}
	\log\underline{p}(\by;q) &= \log\underline{p}^{\text{SV}}(\by;\boldsymbol{\beta},\bh,\bpsi) \,\,\,\, \left(\text{or } \log\underline{p}^{\text{C}}(\by;\boldsymbol{\beta},\bnu) \,\, \text{if homoskedastic}\right)\\
	&\qquad + \frac{1}{2}\left(\log|\bSigma_{q(\boldsymbol{\vartheta})}| +d(d+p+1)\right)  -\sum_{j=1}^d\sum_{k=1}^{d+p+1} h(\zeta_{q(\upsilon_{j,k})},a_{q(\upsilon_{j,k})},b_{q(\upsilon_{j,k})})\\
	&\qquad +d(d+p+1)\left(h_1\log h_2 -\log\Gamma(h_1)\right)-\sum_{j=1}^d\sum_{k=1}^{d+p+1}\left(a_{q(\lambda_{j,k})}\log b_{q(\lambda_{j,k})} -\log\Gamma(a_{q(\lambda_{j,k})})\right) \\
	&\qquad +d\log h_3 + \sum_{j=1}^d\log c_{\eta_{j}} +\sum_{j=1}^d\mu_{q(\eta_{j})}\sum_{k=1}^{d+p+1}\left(\mu_{q(\lambda_{j,k})}\mu_{q(\upsilon_{j,k})} -\mu_{q(\log\lambda_{j,k})}-\mu_{q(\log\upsilon_{j,k})}\right),
	\end{aligned}\end{equation}
 where $\log\underline{p}^{\text{SV}}(\by;\boldsymbol{\beta},\bh,\bpsi)$ and $\log\underline{p}^{\text{C}}(\by;\boldsymbol{\beta},\bnu)$ are defined in \ref{eq:elbo_VBLASSO}.
\end{prop}
\begin{proof}
	As we did in \eqref{eq:elbo_terms} for Proposition \ref{prop:elbo_VB}, the lower bound can be divided into terms referring to each parameter:
	\begin{align}\label{eq:elboDG_terms}
	\log\underline{p}(\by;q) = A &+\sum_{j=1}^d\sum_{k=1}^{d+p+1}\bigg(\underbrace{\mathbb{E}_q\left[\log\underline{p}(\by;\upsilon_{j,k})\right]}_{B} +\underbrace{\mathbb{E}_q\left[\log\underline{p}(\by;\lambda_{j,k})\right]}_{C}+\underbrace{\mathbb{E}_q\left[\log\underline{p}(\by;\eta_{j})\right]}_{D}\bigg),
	\end{align}
	where A is equal to \eqref{eq:A_lasso}.
	Our strategy will be to evaluate each piece in the latter separately and then put the results together.
	Consider the piece $B$:
	\begin{align*}
	B &= \mathbb{E}_q\left[\eta_{j}\log\eta_{j}+\eta_{j}\left(\log\lambda_{j,k}-\log 2\right)-\log\Gamma(\eta_{j})+(\eta_{j}-1)\log\upsilon_{j,k}-\upsilon_{j,k}\frac{\eta_{j}\lambda_{j,k}}{2}\right] \\
	&\qquad-\mathbb{E}_q\left[h(\zeta_{q(\upsilon_{j,k})},a_{q(\upsilon_{j,k})},b_{q(\upsilon_{j,k})})+(\zeta_{q(\upsilon_{j,k})}-1)\log\upsilon_{j,k}-\frac{a_{q(\upsilon_{j,k})}\upsilon_{j,k}}{2}-\frac{b_{q(\upsilon_{j,k})}}{2\upsilon_{j,k}}\right] \\
	&= \mu_{q(\eta_{j}\log\eta_{j})}+\mu_{q(\eta_{j})}\left(\mu_{q(\log\lambda_{j,k})}-\log 2\right)-\mu_{q(\log\Gamma(\eta_{j}))}-h(\zeta_{q(\upsilon_{j,k})},a_{q(\upsilon_{j,k})},b_{q(\upsilon_{j,k})}) \\
	&\qquad +(\mu_{q(\eta_{j})}-1)\mu_{q(\log\upsilon_{j,k})}-(\zeta_{q(\upsilon_{j,k})}-1)\mu_{q(\log\upsilon_{j,k})} \\
	&\qquad-\frac{1}{2}\left(\mu_{q(\upsilon_{j,k})}\mu_{q(\eta_{j})}\mu_{q(\lambda_{j,k})}-a_{q(\upsilon_{j,k})}\mu_{q(\upsilon_{j,k})}-b_{q(\upsilon_{j,k})}\mu_{q(1/\upsilon_{j,k})}\right),
	\end{align*}
	where $h(\zeta,a,b)$ denotes the logarithm of the normalizing constant of a $\mathsf{GIG}$ distribution, i.e.
	\begin{equation*}
	h(\zeta,a,b) = \zeta/2\log(a/b)-\log 2-\log K_{\zeta}(\sqrt{ab}).
	\end{equation*}
	The term involving $\lambda_{j,k}$, for $j=1,\ldots,d$ and $k=1,\ldots,d+p+1$, is equal to:
	\begin{align*}
	C &= \mathbb{E}_q\left[h_1\log h_2 -\log\Gamma(h_1)+(h_1-1)\log\lambda_{j,k}-\lambda_{j,k}h_2\right]\\
	&\qquad-\mathbb{E}_q\left[a_{q(\lambda_{j,k})}\log b_{q(\lambda_{j,k})} -\log\Gamma(a_{q(\lambda_{j,k})})+(a_{q(\lambda_{j,k})}-1)\log\lambda_{j,k}-\lambda_{j,k}b_{q(\lambda_{j,k})}\right] \\
	&= h_1\log h_2 -\log\Gamma(h_1)+(h_1-1)\mu_{q(\log\lambda_{j,k})}-\mu_{q(\lambda_{j,k})}h_2\\
	&\qquad-a_{q(\lambda_{j,k})}\log b_{q(\lambda_{j,k})} +\log\Gamma(a_{q(\lambda_{j,k})})-(a_{q(\lambda_{j,k})}-1)\mu_{q(\log\lambda_{j,k})}+\mu_{q(\lambda_{j,k})}b_{q(\lambda_{j,k})},
	\end{align*}
	and, to conclude, compute the term $D$:
	\begin{align*}
	D &= \mathbb{E}_q\left[\log h_3-\eta_{j}h_3\right]\\
	&\qquad-\mathbb{E}_q\left[\log h(\eta_{j})-\log c_{\eta_{j}}-\eta_{j}\sum_{k=1}^{d+p+1}\left(\frac{\mu_{q(\lambda_{j,k})}\mu_{q(\upsilon_{j,k})}}{2}-\mu_{q(\log\lambda_{j,k})}-\mu_{q(\log\upsilon_{j,k})}+\log2+h_3\right)\right] \\
	&= \log h_3-\mu_{q(\eta_{j})}h_3\\
	&\qquad-\mu_{q(\log h(\eta_{j}))}+\log c_{\eta_{j}}+\mu_{q(\eta_{j})}\sum_{k=1}^{d+p+1}\left(\frac{\mu_{q(\lambda_{j,k})}\mu_{q(\upsilon_{j,k})}}{2}-\mu_{q(\log\lambda_{j,k})}-\mu_{q(\log\upsilon_{j,k})}+\log2+h_3\right).
	\end{align*}
	Group together the terms and exploit the analytical form of the optimal parameters to perform some simplifications. The remaining terms form the lower bound for a multivariate regression model with adaptive normal-gamma prior.
\end{proof}
\noindent The moments of the optimal variational densities are updated at each iteration of the Algorithm \ref{code:VBDG} and the convergence is assessed by checking the variation both in the lower bound and the parameters.
\begin{table}[!ht]
	\begin{algorithm}[H]
		\SetAlgoLined
		\kwInit{$q^*(\boldsymbol{\xi})$, $\Delta_{\mathbf{\xi}}$, $\Delta_{\text{ELBO}}$}
		\While{$\big(\widehat{\Delta}_{\text{ELBO}}>\Delta_{\text{ELBO}}\big) \lor \big(\widehat{\Delta}_{\mathbf{\xi}}>\Delta_{\mathbf{\xi}}\big)$}{
			Update $q^*(\nu_1)$ as in \eqref{eq:up_nu_VB} (homoskedastic); \\
                Update $q^*(\bh_{1})$ and therefore $q^*(\boldsymbol{\nu}_1)$ as in \eqref{eq:up_h_VB} and \eqref{eq:up_nu_t_VB}  (heteroskedastic); \\
                Update $q^*(\psi_1)$ as in \eqref{eq:up_psi};\\
			\For{$j=2,\ldots,d$}{
				Update $q^*(\nu_j)$ as in \eqref{eq:up_nu_VB} (homoskedastic); \\
                Update $q^*(\bh_{j})$ and therefore $q^*(\boldsymbol{\nu}_j)$ as in \eqref{eq:up_h_VB} and \eqref{eq:up_nu_t_VB}  (heteroskedastic); \\
                Update $q^*(\psi_j)$ as in \eqref{eq:up_psi};\\
				Update $q^*(\bbeta_j)$ as in \eqref{eq:up_beta_VB}; \\
			}
			Update $q^*(\boldsymbol{\vartheta})$ as in \eqref{eq:up_theta_VBLASSO} or \eqref{eq:up_theta_VBLASSO_row}; \\
			\For{$j=1,\ldots,d$}{
				\For{$k=1,\ldots,d+p+1$}{
					Update $q^*(\upsilon_{j,k})$, $q^*(\lambda_{j,k})$ as in \eqref{eq:up_ups_VBDG}-\eqref{eq:up_lam_VBDG}; \\
				}	
			Update $q^*(\eta_j)$ as in \eqref{eq:up_eta_VBDG}; \\
			}
			Compute $\log\underline{p}\left(\by;q\right)$ as in \eqref{eq:elbo_VBDG}; \\
			Compute $\widehat{\Delta}_{\text{ELBO}} = \log\underline{p}\left(\by;q\right)^{(\iter)}-\log\underline{p}\left(\by;q\right)^{(\iter-1)}$;\\
			Compute $\widehat{\Delta}_{\mathbf{\xi}} = q^*(\boldsymbol{\xi})^{(\iter)}-q^*(\boldsymbol{\xi})^{(\iter-1)}$ ;
		}
		\caption{MFVB with adaptive normal-gamma prior.}
		\label{code:VBDG}
	\end{algorithm}
\end{table}

\subsection{Horseshoe prior}
\label{app:VBVAR-HS}
First of all, notice that the optimal densities for $\bh_j$, $\nu_j$, and for the coefficients $\bbeta_j$ remain the same computed in Section \ref{app:non_sparseVBAR}. The changes in the optimal densities $q^*(\boldsymbol{\vartheta})$ are stated in the next proposition.
\begin{prop}\label{prop:HS_onTheta}
	The joint optimal variational density for the parameter $\boldsymbol{\vartheta}$ is equal to $q^*(\boldsymbol{\vartheta}) \equiv \mathsf{N}_{d(d+p+1)}(\bmu_{q(\mathbf{\vartheta})}, \bSigma_{q(\mathbf{\vartheta})})$, where:
	\begin{equation}\begin{aligned}\label{eq:up_theta_VBHS}
    \bSigma_{q(\mathbf{\vartheta})} &= \left( \sum_{t=1}^T\bmu_{q(\mathbf{\Omega}_t)} \otimes\mathbf{z}_{t-1}\mathbf{z}_{t-1}^\intercal+\bmu_{q(1/\mathbf{\gamma^2})}\mathsf{Diag}(\bmu_{q(1/\mathbf{\upsilon}^2)})\right)^{-1},\\
    \bmu_{q(\mathbf{\vartheta})} &= \bSigma_{q(\mathbf{\vartheta})} \sum_{t=1}^T\left(\bmu_{q(\mathbf{\Omega}_t)}\otimes\mathbf{z}_{t-1}\right)\mathbf{y}_t,
	\end{aligned}\end{equation}
	where $\Diag(\bmu_{q(1/\mathbf{\upsilon}^2)})$ is a diagonal matrix and $\bmu_{q(1/\mathbf{\upsilon}^2)}=(\mu_{q(1/\upsilon_{1,1}^2)},\mu_{q(1/\upsilon_{1,2}^2)},\ldots,\mu_{q(1/\upsilon_{d,d+p+1}^2)})$.\par
	Under the row-independence assumption, the optimal variational density for the parameter $\boldsymbol{\vartheta}_j$ is equal to $q^*(\boldsymbol{\vartheta}_j) \equiv \mathsf{N}_{d+p+1}(\bmu_{q(\mathbf{\vartheta}_j)}, \bSigma_{q(\mathbf{\vartheta}_j)})$, where:
	\begin{equation}\begin{aligned}\label{eq:up_theta_VBHS_row}
	\bSigma_{q(\mathbf{\vartheta}_j)} &= \left( \sum_{t=1}^T\bmu_{q(\omega_{j,j,t})}\mathbf{z}_{t-1}\mathbf{z}_{t-1}^\intercal+\bmu_{q(1/\mathbf{\gamma^2})}\mathsf{Diag}(\bmu_{q(1/\mathbf{\upsilon}_j^2)})\right)^{-1},\\
	\bmu_{q(\mathbf{\vartheta}_j)} &= \bSigma_{q(\mathbf{\vartheta}_j)}\left(\sum_{t=1}^T\left(\bmu_{q(\mathbf{\omega}_{j,t})} \otimes\mathbf{z}_{t-1}\right)\mathbf{y}_t-\sum_{t=1}^T\left( \bmu_{q(\mathbf{\omega}_{j,-j,t})}\otimes\mathbf{z}_{t-1}\mathbf{z}_{t-1}^\intercal\right)\bmu_{q(\mathbf{\vartheta}_{-j})}\right),
	\end{aligned}\end{equation}
	where $\Diag(\bmu_{q(1/\mathbf{\upsilon}^2_j)})$ is a diagonal matrix and $\bmu_{q(1/\mathbf{\upsilon}^2_j)}=(\mu_{q(1/\upsilon^2_{j,1})},\mu_{q(1/\upsilon^2_{j,2})},\ldots,\mu_{q(1/\upsilon^2_{j,d+p+1})})$.
\end{prop}
\medskip
\noindent Hereafter we describe the optimal densities for the parameters used in hierarchical specification of the prior.
\begin{prop}\label{prop:ups_HS}
	The optimal density for the prior local variance $\upsilon^2_{j,k}$ is equal to an inverse gamma distribution $q^*(\upsilon^2_{j,k}) \equiv \mathsf{InvGa}(1,b_{q(\upsilon^2_{j,k})})$, where, for $j=1,\ldots,d$ and $k=1,\ldots,d+p+1$:
	\begin{align}\label{eq:up_ups_VBHS}
	b_{q(\upsilon^2_{j,k})} = \mu_{q(1/\lambda_{j,k})}+\frac{1}{2}\mu_{q(\vartheta^2_{j,k})}\mu_{q(1/\gamma^2)}.
	\end{align}
\end{prop}
\begin{proof}
	Consider the prior specification which involves the parameter $\upsilon^2_{j,k}$:
	\begin{equation*}\begin{aligned}
	\vartheta_{j,k}|\gamma^2,\upsilon^2_{j,k} \sim \mathsf{N}(0,\gamma^2\upsilon^2_{j,k}), \qquad \upsilon^2_{j,k}|\lambda_{j,k} \sim \mathsf{InvGa}\left(1/2,1/\lambda_{j,k}\right).
	\end{aligned}\end{equation*}
	Compute the optimal variational density $\log q^*(\upsilon^2_{j,k})\propto \mathbb{E}_{-\mathbf{\upsilon^2_{j,k}}}\left[\log p(\vartheta_{j,k})+\log p(\upsilon^2_{j,k})\right]$:
	\begin{equation*}\begin{aligned}
	\log q^*(\upsilon^2_{j,k}) &\propto \mathbb{E}_{-\upsilon^2_{j,k}} \left[-\frac{1}{2}\log\upsilon^2_{j,k} -\frac{1}{2\gamma^2\upsilon^2_{j,k}} \vartheta_{j,k}^2-(1/2+1)\log\upsilon^2_{j,k}-\frac{1}{\upsilon^2_{j,k}\lambda_{j,k}}\right]\\
	&\propto -2\log\upsilon^2_{j,k} -\frac{1}{\upsilon^2_{j,k}}\left(\mu_{q(1/\gamma^2)}\mu_{q(\vartheta^2_{j,k})}/2 +\mu_{q(1/\lambda_{j,k})}\right).
	\end{aligned}\end{equation*}
	Take the exponential and notice that the latter is the kernel of an inverse gamma random variable $\mathsf{InvGa}(1,b_{q(\upsilon^2_{j,k})})$, as defined in Proposition \ref{prop:ups_HS}. 
\end{proof}
\begin{prop}\label{prop:gamma_HS}
	The optimal density for the prior global variance $\gamma^2$ is equal to an inverse gamma distribution $q^*(\gamma^2) \equiv \mathsf{InvGa}(a_{q(\gamma^2)},b_{q(\gamma^2)})$, where:
	\begin{align}\label{eq:up_gamma_VBHS}
	a_{q(\gamma^2)}=\frac{d(d+p+1)+1}{2}, \quad b_{q(\gamma^2)} = \mu_{q(1/\eta)}+\frac{1}{2}\sum_{j=1}^d\sum_{k=1}^{d+p+1}\mu_{q(1/\upsilon^2_{j,k})}\mu_{q(\vartheta^2_{j,k})}.
	\end{align}
\end{prop}
\begin{proof}
	Consider the prior specification which involves the parameter $\gamma^2$:
	\begin{equation*}\begin{aligned}
	\vartheta_{j,k}|\gamma^2,\upsilon^2_{j,k} \sim \mathsf{N}(0,\gamma^2\upsilon^2_{j,k}), \qquad \gamma^2|\eta \sim \mathsf{InvGa}\left(1/2,1/\eta\right).
	\end{aligned}\end{equation*}
	Compute the optimal variational density $\log q^*(\gamma^2)\propto \mathbb{E}_{-\mathbf{\gamma^2}}\left[\sum_{j=1}^d\sum_{k=1}^{d+p+1}\log p(\vartheta_{j,k})+\log p(\gamma^2)\right]$:
	\begin{equation*}\begin{aligned}
	\log q^*(\gamma^2) &\propto \mathbb{E}_{-\gamma^2} \left[-\frac{d(d+p+1)}{2}\log\gamma^2 -\frac{1}{2\gamma^2\upsilon^2_{j,k}} \vartheta_{j,k}^2-(1/2+1)\log\gamma^2-\frac{1}{\gamma^2\eta}\right]\\
	&\propto -\left(\frac{d(d+p+1)+1}{2}+1\right)\log\gamma^2 -\frac{1}{\gamma^2}\left(\sum_{j=1}^d\sum_{k=1}^{d+p+1}\mu_{q(1/\upsilon^2_{j,k})}\mu_{q(\vartheta^2_{j,k})}/2 +\mu_{q(1/\eta)}\right).
	\end{aligned}\end{equation*}
	Take the exponential and notice that the latter is the kernel of an inverse gamma random variable $\mathsf{InvGa}(a_{q(\gamma^2)},b_{q(\gamma^2)})$, as defined in Proposition \ref{prop:gamma_HS}. 
\end{proof}
\begin{prop}\label{prop:lam_HS}
	The optimal density for the latent parameter $\lambda_{j,k}$ is equal to an inverse gamma distribution $q^*(\lambda_{j,k}) \equiv \mathsf{InvGa}(1,b_{q(\lambda_{j,k})})$, where, for $j=1,\ldots,d$ and $k=1,\ldots,d+p+1$:
	\begin{align}\label{eq:up_lam_VBHS}
	b_{q(\lambda_{j,k})} = 1+\mu_{q(1/\upsilon^2_{j,k})}.
	\end{align}
\end{prop}
\begin{proof}
	Consider the prior specification which involves the parameter $\lambda_{j,k}$:
	\begin{equation*}\begin{aligned}
	\upsilon^2_{j,k}|\lambda_{j,k} \sim \mathsf{InvGa}\left(1/2,1/\lambda_{j,k}\right), \qquad \lambda_{j,k} \sim \mathsf{InvGa}\left(1/2,1\right).
	\end{aligned}\end{equation*}
	Compute the optimal variational density $\log q^*(\lambda_{j,k})\propto \mathbb{E}_{-\mathbf{\lambda_{j,k}}}\left[\log p(\upsilon^2_{j,k})+\log p(\lambda_{j,k})\right]$:
	\begin{equation*}\begin{aligned}
	\log q^*(\lambda_{j,k}) &\propto \mathbb{E}_{-\lambda_{j,k}} \left[-\frac{1}{2}\log\lambda_{j,k}-\frac{1}{\upsilon^2_{j,k}\lambda_{j,k}}-(1/2+1)\log\lambda_{j,k}-\frac{1}{\lambda_{j,k}}\right]\\
	&\propto -2\log\lambda_{j,k} -\frac{1}{\lambda_{j,k}}\left(1+\mu_{q(1/\upsilon^2_{j,k})}\right).
	\end{aligned}\end{equation*}
	Take the exponential and notice that the latter is the kernel of an inverse gamma random variable $\mathsf{InvGa}(1,b_{q(\lambda_{j,k})})$, as defined in Proposition \ref{prop:lam_HS}. 
\end{proof}
\begin{prop}\label{prop:eta_HS}
	The optimal density for the latent parameter $\eta$ is equal to an inverse gamma distribution $q^*(\eta) \equiv \mathsf{InvGa}(1,b_{q(\eta)})$, where:
	\begin{align}\label{eq:up_eta_VBHS}
	b_{q(\eta)} = 1+\mu_{q(1/\gamma^2)}.
	\end{align}
\end{prop}
\begin{proof}
	Consider the prior specification which involves the parameter $\eta$:
	\begin{equation*}\begin{aligned}
	\gamma^2|\eta \sim \mathsf{InvGa}\left(1/2,1/\eta\right), \qquad \eta \sim \mathsf{InvGa}\left(1/2,1\right).
	\end{aligned}\end{equation*}
	Compute the optimal variational density $\log q^*(\eta)\propto \mathbb{E}_{-\mathbf{\eta}}\left[\log p(\gamma^2)+\log p(\eta)\right]$:
	\begin{equation*}\begin{aligned}
	\log q^*(\eta) &\propto \mathbb{E}_{-\eta} \left[-\frac{1}{2}\log\eta-\frac{1}{\gamma^2\eta}-(1/2+1)\log\eta-\frac{1}{\eta}\right]\\
	&\propto -2\log\eta -\frac{1}{\eta}\left(1+\mu_{q(1/\gamma^2)}\right).
	\end{aligned}\end{equation*}
	Take the exponential and notice that the latter is the kernel of an inverse gamma random variable $\mathsf{InvGa}(1,b_{q(\eta)})$, as defined in Proposition \ref{prop:eta_HS}. 
\end{proof}
\begin{prop}\label{prop:elbo_VBHS}
	The variational lower bound for the multivariate regression model with Horseshoe prior can be derived analytically and it is equal to:
	\begin{equation}\begin{aligned}\label{eq:elbo_VBHS}
	\log\underline{p}(\by;q) &= \log\underline{p}^{\text{SV}}(\by;\boldsymbol{\beta},\bh,\bpsi) \,\,\,\, \left(\text{or } \log\underline{p}^{\text{C}}(\by;\boldsymbol{\beta},\bnu) \,\, \text{if homoskedastic}\right)\\
	&\qquad + \frac{1}{2}\left(\log|\bSigma_{q(\boldsymbol{\vartheta})}| +d(d+p+1)\right) + \mu_{q(1/\gamma^2)}\left(\mu_{q(1/\eta)}+\sum_{j=1}^d\sum_{k=1}^{d+p+1}\mu_{q(\vartheta^2_{j,k})}\mu_{q(1/\upsilon^2_{j,k})}\right) \\
	&\qquad +\sum_{j=1}^d\sum_{k=1}^{d+p+1}\left(\mu_{q(1/\upsilon^2_{j,k})}\mu_{q(1/\lambda_{j,k})} -\log b_{q(\upsilon^2_{j,k})}  -\log b_{q(\lambda_{j,k})} -\log\pi \right) \\
	&\qquad -a_{q(\gamma^2)}\log b_{q(\gamma^2)}  -\log b_{q(\eta)} -\log\pi,
	\end{aligned}\end{equation}
 where $\log\underline{p}^{\text{SV}}(\by;\boldsymbol{\beta},\bh,\bpsi)$ and $\log\underline{p}^{\text{C}}(\by;\boldsymbol{\beta},\bnu)$ are defined in \ref{eq:elbo_VBLASSO}.
\end{prop}
\begin{proof}
	As we did in \eqref{eq:elbo_terms} for Proposition \ref{prop:elbo_VB}, the lower bound can be divided into terms referring to each parameter:
	\begin{equation}\begin{aligned}\label{eq:elboHS_terms}
	\log\underline{p}(\by;q) = A &+\underbrace{\mathbb{E}_q\left[\log\underline{p}(\by;\gamma^2)\right]}_{B} +\underbrace{\mathbb{E}_q\left[\log\underline{p}(\by;\eta)\right]}_{C}\\
	&\qquad +\sum_{j=1}^d\sum_{k=1}^{d+p+1}\bigg(\underbrace{\mathbb{E}_q\left[\log\underline{p}(\by;\upsilon^2_{j,k})\right]}_{D} +\underbrace{\mathbb{E}_q\left[\log\underline{p}(\by;\lambda_{j,k})\right]}_{E}\bigg) ,
	\end{aligned}\end{equation}
	where A is similar to \eqref{eq:elbo_terms} in the previous non-informative model specification.
	Our strategy will be to evaluate each piece in the latter separately and then put the results together. Notice that the computations for the piece $A$ are similar to Proposition \ref{prop:elbo_VB}. Hence, we have that:
	\begin{equation}
	\begin{aligned}
	A &= \log\underline{p}^{\text{SV}}(\by;\boldsymbol{\beta},\bh,\bpsi) \,\,\,\, \left(\text{or } \log\underline{p}^{\text{C}}(\by;\boldsymbol{\beta},\bnu) \,\, \text{if homoskedastic}\right) \\
	&\qquad -\frac{1}{2}\sum_{j=1}^d\sum_{k=1}^{d+p+1}\left(\mu_{q(\log\delta^2)} + \mu_{q(\log\upsilon^2_{j,k})}+ \mu_{q(1/\delta^2)}\mu_{q(1/\upsilon^2_{j,k})}\mu_{q(\vartheta^2_{j,k})}\right) + \frac{1}{2}\left(\log|\bSigma_{q(\boldsymbol{\vartheta})}| +d(d+p+1)\right).
	\end{aligned}
	\end{equation}
	Consider now the piece $B$. We have that:
	\begin{align*}
	B &= \mathbb{E}_q\left[-\frac{1}{2}\log\eta-\frac{1}{2}\log\pi-(1/2+1)\log\gamma^2-1/(\gamma^2\eta)\right] \\
	&\qquad-\mathbb{E}_q\left[a_{q(\gamma^2)}\log b_{q(\gamma^2)}-\log\Gamma(a_{q(\gamma^2)})-(a_{q(\gamma^2)}+1)\log\gamma^2-b_{q(\gamma^2)}/\gamma^2\right] \\
	&= -\frac{1}{2}\mu_{q(\log\eta)}-\frac{1}{2}\log\pi-(1/2+1)\mu_{q(\log\gamma^2)}-\mu_{q(1/\gamma^2)}\mu_{q(1/\eta)} \\
	&\qquad -a_{q(\gamma^2)}\log b_{q(\gamma^2)}+\log\Gamma(a_{q(\gamma^2)})+(a_{q(\gamma^2)}+1)\mu_{q(\log\gamma^2)}+\mu_{q(1/\gamma^2)}b_{q(\gamma^2)},
	\end{align*}
	while, $C$ reduces to:
	\begin{align*}
	C &= \mathbb{E}_q\left[-\frac{1}{2}\log\pi-(1/2+1)\log\eta-1/\eta\right] -\mathbb{E}_q\left[\log b_{q(\eta)}-2\log\eta-b_{q(\eta)}/\eta\right] \\
	&= -\frac{1}{2}\log\pi-(1/2+1)\mu_{q(\log\eta)}-\mu_{q(1/\eta)} -\log b_{q(\eta)}+2\mu_{q(\log\eta)}+\mu_{q(1/\eta)}b_{q(\eta)}.
	\end{align*}
	The remaining terms behave likely $B$ and $C$. In particular, for $j=1,\ldots,d$ and $k=1,\ldots,d+p+1$:
	\begin{align*}
	D &= \mathbb{E}_q\left[-\frac{1}{2}\log\lambda_{j,k}-\frac{1}{2}\log\pi-(1/2+1)\log\upsilon^2_{j,k}-1/(\upsilon^2_{j,k}\lambda_{j,k})\right] \\
	&\qquad-\mathbb{E}_q\left[\log b_{q(\upsilon^2_{j,k})}-2\log\upsilon^2_{j,k}-b_{q(\upsilon^2_{j,k})}/\upsilon^2_{j,k}\right] \\
	&= -\frac{1}{2}\mu_{q(\log\lambda_{j,k})}-\frac{1}{2}\log\pi-(1/2+1)\mu_{q(\log\upsilon^2_{j,k})}-\mu_{q(1/\upsilon^2_{j,k})}\mu_{q(1/\lambda_{j,k})} \\
	&\qquad -\log b_{q(\upsilon^2_{j,k})}+2\mu_{q(\log\upsilon^2_{j,k})}+\mu_{q(1/\upsilon^2_{j,k})}b_{q(\upsilon^2_{j,k})},
	\end{align*}
	and
	\begin{align*}
	E &= \mathbb{E}_q\left[-\frac{1}{2}\log\pi-(1/2+1)\log\lambda_{j,k}-1/\lambda_{j,k}\right] -\mathbb{E}_q\left[\log b_{q(\lambda_{j,k})}-2\log\lambda_{j,k}-b_{q(\lambda_{j,k})}/\lambda_{j,k}\right] \\
	&= -\frac{1}{2}\log\pi-(1/2+1)\mu_{q(\log\lambda_{j,k})}-\mu_{q(1/\lambda_{j,k})} -\log b_{q(\lambda_{j,k})}+2\mu_{q(\log\lambda_{j,k})}+\mu_{q(1/\lambda_{j,k})}b_{q(\lambda_{j,k})}.
	\end{align*}
	Group together the terms and exploit the analytical form of the optimal parameters to perform some simplifications. The remaining terms form the lower bound for a multivariate regression model with Horseshoe prior.
\end{proof}
\noindent The moments of the optimal variational densities are updated at each iteration of the Algorithm \ref{code:VBHS} and the convergence is assessed by checking the variation both in the lower bound and the parameters.
\begin{table}[!ht]
	\begin{algorithm}[H]
		\SetAlgoLined
		\kwInit{$q^*(\bxi)$, $\Delta_{\mathbf{\xi}}$, $\Delta_{\text{ELBO}}$}
		\While{$\big(\widehat{\Delta}_{\text{ELBO}}>\Delta_{\text{ELBO}}\big) \lor \big(\widehat{\Delta}_{\mathbf{\xi}}>\Delta_{\mathbf{\xi}}\big)$}{
			Update $q^*(\nu_1)$ as in \eqref{eq:up_nu_VB} (homoskedastic); \\
                Update $q^*(\bh_{1})$ and therefore $q^*(\boldsymbol{\nu}_1)$ as in \eqref{eq:up_h_VB} and \eqref{eq:up_nu_t_VB}  (heteroskedastic); \\
                Update $q^*(\psi_1)$ as in \eqref{eq:up_psi};\\
			\For{$j=2,\ldots,d$}{
				Update $q^*(\nu_j)$ as in \eqref{eq:up_nu_VB} (homoskedastic); \\
                Update $q^*(\bh_{j})$ and therefore $q^*(\boldsymbol{\nu}_j)$ as in \eqref{eq:up_h_VB} and \eqref{eq:up_nu_t_VB}  (heteroskedastic); \\
                Update $q^*(\psi_j)$ as in \eqref{eq:up_psi};\\
				Update $q^*(\bbeta_j)$ as in \eqref{eq:up_beta_VB}; \\
			}
			Update $q^*(\boldsymbol{\vartheta})$ as in \eqref{eq:up_theta_VBHS} or \eqref{eq:up_theta_VBHS_row} ; \\
			\For{$j=1,\ldots,d$}{
				\For{$k=1,\ldots,d+p+1$}{
					Update $q^*(\upsilon^2_{j,k})$, $q^*(\lambda_{j,k})$ as in \eqref{eq:up_ups_VBHS}-\eqref{eq:up_lam_VBHS}; \\
				}
			}
			Update $q^*(\gamma^2)$, $q^*(\eta)$ as in \eqref{eq:up_gamma_VBHS}-\eqref{eq:up_eta_VBHS}; \\
			Compute $\log\underline{p}\left(\by;q\right)$ as in \eqref{eq:elbo_VBHS}; \\
			Compute $\widehat{\Delta}_{\text{ELBO}} = \log\underline{p}\left(\by;q\right)^{(\iter)}-\log\underline{p}\left(\by;q\right)^{(\iter-1)}$;\\
			Compute $\widehat{\Delta}_{\mathbf{\xi}} = q^*(\boldsymbol{\xi})^{(\iter)}-q^*(\boldsymbol{\xi})^{(\iter-1)}$ ;
		}
		\caption{MFVB with Horseshoe prior.}
		\label{code:VBHS}
	\end{algorithm}
\end{table}

\setcounter{table}{0}  
\section{Variational predictive density}

In this section we first discuss the approximation of $q^*(\mathbf{\Omega}_t)$. This is instrumental to the derivation of the optimal variational predictive density.

\subsection{Inference on the time-varying precision matrix}	
\label{app:prec_mat}
Proposition \ref{eq:prop5_main} shows that, conditional on $\mathbf{L}$ and $\mathbf{V}_t$, the optimal distribution of $\mathbf{\Omega}_t$ can be approximated by a $d$-dimensional Wishart distribution $\mathsf{Wishart}_d(\delta_t,\mathbf{H}_t)$, where $\delta_t$ and $\mathbf{H}_t$ are the degrees of freedom and the scaling matrix, respectively. The complete proof is based on the Expectation Propagation (EP) approach proposed by \cite{minka.2001}. This has the goal of minimizing the {\it KL} divergence between the true and unknown optimal variational distribution $q^*(\mathbf{\Omega}_t)$ and a sub-optimal approximating density $\tilde{q}(\mathbf{\Omega}_t)$. In order to implement this approach, there is no need to know $q^*(\mathbf{\Omega}_t)$, but it is sufficient to be able to compute $\mathbb{E}_q(\mathbf{\Omega}_t)$. The latter can be reconstructed based on the optimal variational densities of the Cholesky factor $q^*(\bbeta)$ -- and therefore for $\mathbf{L}$ --, and of $\mathbf{V}_t$.\\

\begin{proposition}
 The approximate distribution $q$ of $\mathbf{\Omega}_t$ is $\mathsf{Wishart}_d(\widehat{\delta}_t,\widehat{\mathbf{H}}_t)$, where the scaling matrix is given by $\widehat{\mathbf{H}}_t=\widehat{\delta}_t^{-1}\mathbb{E}_q\left[\mathbf{\Omega}_t\right]$ and $\widehat{\delta}$ can be obtained numerically as the solution of a convex optimization problem.
 \label{eq:prop5}
\end{proposition}
\begin{proof}

The Kullback-Leibler divergence between $q(\mathbf{\Omega}_t)$ and the new approximating distribution $\tilde{q}(\mathbf{\Omega}_t)$ is $\mathcal{D}_{\mathit{KL}}(q(\mathbf{\Omega}_t)\Vert \tilde{q}(\mathbf{\Omega}_t)) \propto-\mathbb{E}_q(\log \tilde{q}(\mathbf{\Omega}_t))$, where the expectation is taken with respect to the variational distribution $q(\mathbf{\Omega})$. Therefore the optimal parameters are $(\widehat{\delta}_t,\widehat{\mathbf{H}}_t)=\arg\min_{\delta_t,\mathbf{H}_t}\psi(\delta_t,\mathbf{H}_t)$, where $\psi(\delta_t,\mathbf{H}_t)=-\mathbb{E}_q(\log \tilde{q}(\mathbf{\Omega}_t))$:
\begin{equation}
\label{eq:approximating_omega_1}
\psi(\delta_t,\mathbf{H}_t) \propto
\frac{d\delta_t}{2}\log2 +\frac{\delta_t}{2}\log|\mathbf{H}_t|+\log\Gamma_d(\delta_t/2) -\frac{\delta_t}{2}\mathbb{E}_q\left[\log|\mathbf {\Omega}_t|\right] +\frac{1}{2}\trace\left\{\mathbf{H}_t^{-1}\mathbb{E}_q\left[\mathbf{\Omega}_t\right]\right\}.
\end{equation}
Note that $\mathbb{E}_q\left[\log|\mathbf{\Omega}_t|\right] = \mathbb{E}_{q(V_t)}\left[\log|\mathbf{V}_t|\right] = \sum_{j=1}^d\mu_{q(\log\nu_{j,t})}$ and $\mathbb{E}_q\left[\mathbf{\Omega}_t\right] = \mathbb{E}_{q(L),q(V_t)}\left[\mathbf{L}^\intercal\mathbf{V}_t\mathbf{L}\right]$ are available as byproduct of the mean-field Variational Bayes algorithm. Differentiating \eqref{eq:approximating_omega_1} with respect to the scaling matrix $\mathbf{H}_t$, and solving $\partial \psi(\delta_t,\mathbf{H}_t)/\partial \mathbf{H}_t = 0$ provides $\widehat{\mathbf{H}}_{t}(\delta_t) = \delta_t^{-1}\mathbb{E}_q\left[\mathbf{\Omega}_t\right]$ that depends on the degrees of freedom $\delta_t$. Plugging-in the latter in the objective function $\psi(\delta_t,\widehat{\mathbf{H}}_{t}(\delta_t))$ and proceeding with the minimization of the resulting functional with respect to $\delta_t$ provides $\widehat{\delta}_t$, which completes the proof.
\end{proof}

Table \ref{tab:accuracy_omega} compares the sampled distributions with the marginals of the Wishart with $(\widehat{\delta}_t,\widehat{\mathbf{H}}_t)$ in terms of approximation accuracy $\mathcal{ACC}=100\left\{1-0.5\int|\tilde{q}(\omega_t)-q(\omega_t)|\,d\omega_t\right\}\%$,
where $\omega_t$ is a generic element of $\mathbf{\Omega}_t$. 

\begin{table}[!ht]
    \centering
    \resizebox{.9\textwidth}{!}{
    \begin{tabular}{lrrrrrrrrrrrrr}
    \toprule
          & \multicolumn{2}{c}{$d=15$} & & \multicolumn{2}{c}{$d=30$} & & \multicolumn{2}{c}{$d=50$} & & \multicolumn{2}{c}{$d=100$} \\
          \cmidrule{2-3}\cmidrule{5-6}\cmidrule{8-9}\cmidrule{11-12}
          & \multicolumn{1}{c}{$\omega_{j,j,t}$} & \multicolumn{1}{c}{$\omega_{j,k,t}$} & & \multicolumn{1}{c}{$\omega_{j,j,t}$} & \multicolumn{1}{c}{$\omega_{j,k,t}$} & &
          \multicolumn{1}{c}{$\omega_{j,j,t}$} & \multicolumn{1}{c}{$\omega_{j,k,t}$} & &
          \multicolumn{1}{c}{$\omega_{j,j,t}$} & \multicolumn{1}{c}{$\omega_{j,k,t}$} \\
          \midrule 
            Median & 98.41 & 98.46 & & 98.56 & 98.35 & & 98.43 & 98.28 & & 97.42 & 98.14\\
            Min & 97.66 & 97.13 & & 97.60 & 96.69 & & 96.76 & 94.80 & & 94.47 & 90.66\\
            Max & 99.02 & 99.03 & & 99.34 & 99.18 & & 99.21 & 99.24 & & 99.35 & 99.24\\
    \bottomrule
    \end{tabular}}
    \caption{Accuracy (\%) of the Wishart approximation $\tilde{q}(\mathbf{\Omega}_t)$ for dimensions $d=15,30,50,100$ separately for the diagonal ($\omega_{j,j,t}$) and out-of-diagonal ($\omega_{j,k,t}$) elements of $\mathbf{\Omega}_t$.}
    \label{tab:accuracy_omega}
\end{table}

The simulation results suggest that our variational inference approach provides an accurate approximation of the optimal distribution of $\boldsymbol{\Omega}_t$ for different dimensions. 

\subsection{Derivation of the variational predictive density}
\label{app:predictive_comput}

Recall that the variational predictive posterior can be computed as:
\begin{equation}\label{eq:first_step}
    q(\mathbf{y}_{t+1}|\mathbf{z}_{1:t}) = \int p(\mathbf{y}_{t+1}|\mathbf{z}_{t},\bxi)q^*(\bxi) d\bxi = \int\int p(\mathbf{y}_{t+1}|\mathbf{z}_t,\boldsymbol{\vartheta},\mathbf{\Omega})q^*(\boldsymbol{\vartheta})q^*(\mathbf{\Omega}_t) d\boldsymbol{\vartheta}\,d\mathbf{\Omega}_t,
\end{equation}
which requires only a simulation step according to the first methodology presented in the main paper. If we wish to make the estimation simpler, we can integrate out the precision parameter $\mathbf{\Omega}_t$ (as discussed in Section \ref{app:prec_mat}) in the following way:
\begin{equation}
    q(\mathbf{y}_{t+1}|\mathbf{z}_{1:t}) = \int q(\boldsymbol{\vartheta})\underbrace{\left[\int \mathsf{N}_d(\mathbf{y}_{t+1};\boldsymbol{\Theta}\mathbf{z}_t,\mathbf{\Omega}_t^{-1})\mathsf{Wishart}_d(\mathbf{\Omega}_t;\delta_t,\mathbf{H}_t) d\mathbf{\Omega}_t\right]}_{A}d\boldsymbol{\vartheta},
\end{equation}
where 
\begin{equation}\begin{aligned}
    A &= \frac{2^{-d(\delta_t+1)/2}|\mathbf{H}_t|^{\delta_t/2}}{\pi^{d/2}\Gamma_d(\delta_t/2)}\int\underbrace{|\mathbf{\Omega}_t|^{(\delta_t-d)/2}\exp\left\{-\frac{1}{2}\trace\left\{\mathbf{\Omega}_t\left(\mathbf{H}_t^{-1}+(\mathbf{y}_{t+1}-\mathbf{\Theta}\mathbf{z}_{t})(\mathbf{y}_{t+1}-\mathbf{\Theta}\mathbf{z}_{t})^\intercal\right)\right\}\right\}}_{\text{Kernel of a }\mathsf{Wishart}_d(\delta_t+1,\left(\mathbf{H}_t^{-1}+(\mathbf{y}_{t+1}-\mathbf{\Theta}\mathbf{z}_{t})(\mathbf{y}_{t+1}-\mathbf{\Theta}\mathbf{z}_{t})^\intercal\right)^{-1})}\,d\mathbf{\Omega}_t \\
    & \\
    &=\frac{|1+\frac{1}{v_t}(\mathbf{y}_{t+1}-\mathbf{\Theta}\mathbf{z}_{t})^\intercal v_t\mathbf{H}_t(\mathbf{y}_{t+1}-\mathbf{\Theta}\mathbf{z}_{t})|^{-\frac{v_t+d}{2}}\Gamma(\frac{v_t+d}{2})}{\pi^{d/2}v_t^{d/2}|\mathbf{H}_t^{-1}|^{1/2}\Gamma(v_t/2)} = h(\mathbf{y}_{t+1}|\mathbf{z}_t,\boldsymbol{\vartheta}),
\end{aligned}\end{equation}
is the density function of a multivariate Student-t distribution with dimension $d$, $v_t=\delta_t-d+1$ degrees of freedom, mean vector $\mathbf{\Theta}\mathbf{z}_t$ and scaling matrix $\mathbf{S}_t=(v_t\mathbf{H}_t)^{-1}$, i.e. $\mathsf{t}_{v_t}(\mathbf{\Theta}\mathbf{z}_t,\mathbf{S}_t)$. Then, the integral in Eq.\eqref{eq:first_step} becomes
\begin{equation}\label{eq:second_step}
    q(\mathbf{y}_{t+1}|\mathbf{z}_{1:t}) = \int h(\mathbf{y}_{t+1}|\mathbf{z}_t,\boldsymbol{\vartheta})q(\boldsymbol{\vartheta}) d\boldsymbol{\vartheta},
\end{equation}
which requires to simulate only from the optimal multivariate Gaussian distribution of $\boldsymbol{\vartheta}$ according to the second methodology presented in the main paper.\par
A second-order approximation can be implemented in order to further increase the computational efficiency. To this aim, we propose to approximate the multivariate Student-t in \eqref{eq:second_step} with the closest multivariate normal distribution in terms of {\it KL} divergence:
\begin{equation}\begin{aligned}
\mathcal{D}_{\mathit{KL}}(h\Vert \phi) &\propto
  -\int \log \phi(\mathbf{y}_{t+1}|\mathbf{m}_t,\mathbf{R}_t^{-1})h(\mathbf{y}_{t+1}|\mathbf{z}_t,\boldsymbol{\vartheta}) \,d\mathbf{y}_{t+1} \\
  &=-\mathbb{E}_h(\log \phi(\mathbf{y}_{t+1}|\mathbf{m}_t,\mathbf{R}_t^{-1}))=\psi(\mathbf{m}_t,\mathbf{R}_t),
\end{aligned}\end{equation}
where, in particular,
\begin{equation}\begin{aligned}
    \psi(\mathbf{m}_t,\mathbf{R}_t) &\propto \mathbb{E}_h\left(-\frac{1}{2}\log\mathbf{R}_t+\frac{1}{2}(\mathbf{y}_{t+1}-\mathbf{m}_t)^\intercal\mathbf{R}_t(\mathbf{y}_{t+1}-\mathbf{m}_t)\right) \\
    &=-\frac{1}{2}\log\mathbf{R}_t+\frac{1}{2}(\mathbf{\Theta}\mathbf{z}_t-\mathbf{m}_t)^\intercal\mathbf{R}_t(\mathbf{\Theta}\mathbf{z}_t-\mathbf{m}_t)+\frac{v_t}{2(v_t-2)}\trace\left\{\mathbf{R}_t\mathbf{S}_t\right\},
\end{aligned}\end{equation}
which turns out to be minimized when $\mathbf{m}_t=\mathbf{\Theta}\mathbf{z}_t$ and $\mathbf{R}_t=\frac{v_t-2}{v_t}\mathbf{S}_t^{-1}$. If we substitute the function $h(\cdot)$ with its Gaussian approximation we get
\begin{equation}\label{eq:third_step}
    q(\mathbf{y}_{t+1}|\mathbf{z}_{1:t}) = \int \phi(\mathbf{y}_{t+1}|\mathbf{m}_t,\mathbf{R}_t^{-1})q(\boldsymbol{\vartheta}) d\boldsymbol{\vartheta},
\end{equation}
where now $\phi(\mathbf{y}_{t+1}|\boldsymbol{\Theta}\mathbf{z}_t,\mathbf{R}_t^{-1})$ denotes the density of the multivariate normal distribution that is closest in a {\it KL} sense to the multivariate Student-t $h(\mathbf{y}_{t+1}|\mathbf{z}_t,\boldsymbol{\vartheta})$.
The advantage of this procedure is that the integral in \eqref{eq:third_step} can be solved analytically leading to a closed form variational predictive density $q(\mathbf{y}_{t+1}|\mathbf{z}_{1:t})$ which is a multivariate Gaussian distribution with variance matrix $\mathbf{\Sigma}_{pred,t}$ and mean vector $\boldsymbol{\mu}_{pred,t}$. 
Define $\mathbf{Z}_t = (\bI_d \otimes \mathbf{z}_t^\intercal)$ and compute the integral above:
\begin{equation}\begin{aligned}
    q(\mathbf{y}_{t+1}|\mathbf{z}_{1:t}) &\propto \int\exp\left\{-\frac{1}{2}\left[(\mathbf{y}_{t+1}-\mathbf{Z}_t\boldsymbol{\vartheta})^\intercal\mathbf{R}_t(\mathbf{y}_{t+1}-\mathbf{Z}_t\boldsymbol{\vartheta})+(\boldsymbol{\vartheta}-\boldsymbol{\mu}_{q(\vartheta)})^\intercal\mathbf{\Sigma}_{q(\vartheta)}^{-1}(\boldsymbol{\vartheta}-\boldsymbol{\mu}_{q(\vartheta)})\right]\right\} d\boldsymbol{\vartheta} \\
    &\propto \exp\left\{-\frac{1}{2}\mathbf{y}_{t+1}^\intercal\mathbf{R}_t\mathbf{y}_{t+1}\right\} \\
    &\qquad\qquad\times\int\exp\left\{-\frac{1}{2}\left[\boldsymbol{\vartheta}^\intercal(\mathbf{\Sigma}_{q(\vartheta)}^{-1}+\mathbf{Z}_t^\intercal\mathbf{R}_t\mathbf{Z}_t)\boldsymbol{\vartheta}-2\boldsymbol{\vartheta}^\intercal(\mathbf{\Sigma}_{q(\vartheta)}^{-1}\boldsymbol{\mu}_{q(\vartheta)}+\mathbf{Z}_t\mathbf{R}_t\mathbf{y}_{t+1})\right]\right\} d\boldsymbol{\vartheta},
\end{aligned}\end{equation}
where the term in the integral is the kernel of a multivariate Gaussian random variable with variance matrix $\tilde{\bSigma}_t=(\mathbf{\Sigma}_{q(\vartheta)}^{-1}+\mathbf{Z}_t^\intercal\mathbf{R}_t\mathbf{Z}_t)^{-1}$ and mean $\tilde{\bmu}_t=\tilde{\bSigma}_t(\mathbf{\Sigma}_{q(\vartheta)}^{-1}\boldsymbol{\mu}_{q(\vartheta)}+\mathbf{Z}_t\mathbf{R}_t\mathbf{y}_{t+1})$. Solve the integral and get:
\begin{equation}\begin{aligned}
    q(\mathbf{y}_{t+1}|\mathbf{z}_{1:t}) &\propto \exp\left\{-\frac{1}{2}(\mathbf{y}_{t+1}^\intercal\mathbf{R}_t\mathbf{y}_{t+1}-\tilde{\bmu}_t^\intercal\tilde{\bSigma}_t\tilde{\bmu}_t)\right\} \\
    &\propto \exp\left\{-\frac{1}{2}(\mathbf{y}_{t+1}^\intercal\mathbf{R}_t\mathbf{y}_{t+1}-\mathbf{y}_{t+1}^\intercal\mathbf{R}_t\mathbf{Z}_t\tilde{\bSigma}_t\mathbf{Z}_t^\intercal\mathbf{R}_t\mathbf{y}_{t+1}-2\mathbf{y}_{t+1}\mathbf{R}_t\mathbf{Z}_t\tilde{\bSigma}_t\mathbf{\Sigma}_{q(\vartheta)}^{-1}\boldsymbol{\mu}_{q(\vartheta)})\right\} \\
    &= \exp\left\{-\frac{1}{2}(\mathbf{y}_{t+1}^\intercal(\mathbf{R}_t-\mathbf{R}_t\mathbf{Z}_t\tilde{\bSigma}_t\mathbf{Z}_t^\intercal\mathbf{R}_t)\mathbf{y}_{t+1}-2\mathbf{y}_{t+1}\mathbf{R}_t\mathbf{Z}_t\tilde{\bSigma}_t\mathbf{\Sigma}_{q(\vartheta)}^{-1}\boldsymbol{\mu}_{q(\vartheta)})\right\},
\end{aligned}\end{equation}
which is the kernel of a multivariate Gaussian with variance matrix $\bSigma_{pred,t} = (\mathbf{R}_t-\mathbf{R}_t\mathbf{Z}_t\tilde{\bSigma}_t\mathbf{Z}_t^\intercal\mathbf{R}_t)^{-1}$ and mean $\bmu_{pred,t} = \bSigma_{pred,t}\mathbf{R}_t\mathbf{Z}_t\tilde{\bSigma}_t\mathbf{\Sigma}_{q(\vartheta)}^{-1}\boldsymbol{\mu}_{q(\vartheta)}$. To conclude, the second-order Gaussian approximation to the variational predictive posterior is such that $q(\mathbf{y}_{t+1}|\mathbf{z}_{1:t})\equiv\mathsf{N}_d(\bmu_{pred,t},\bSigma_{pred,t})$.\par

\begin{figure}[h]
	\centering
	\includegraphics[width=.7\textwidth]{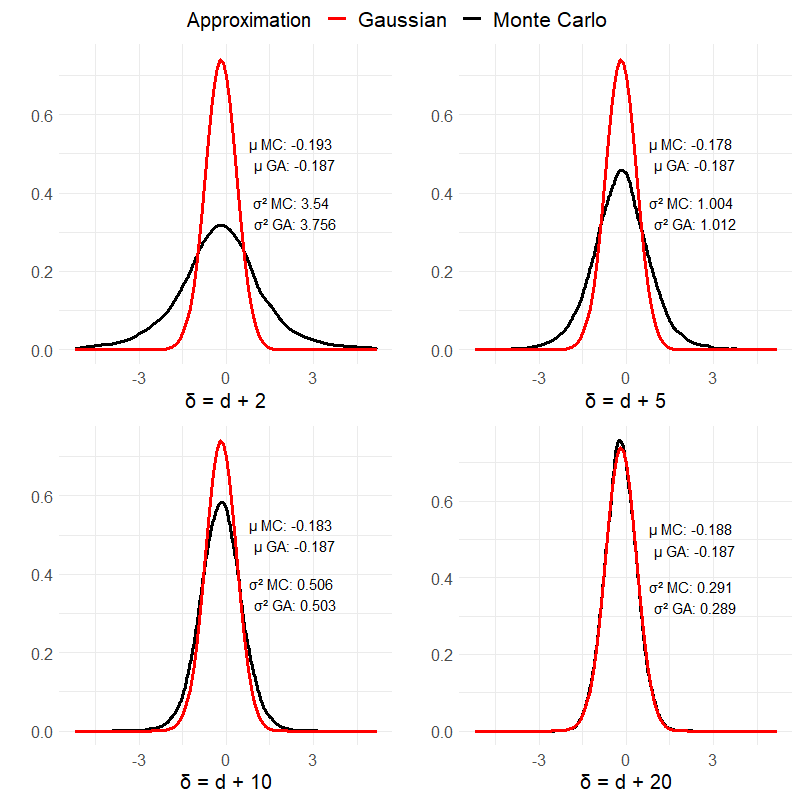}
	\caption{\small Second-order approximation of the predictive density.}
	\label{fig:pred_approx}
\end{figure}

Figure \ref{fig:pred_approx} shows the approximation of variational predictive posterior with Monte Carlo methods (MC) and via Gaussian approximation (GA) varying the degrees of freedom $\widehat{\delta}_t$ for the distribution of $\mathbf{\Omega}_t$. We can see that if $\widehat{\delta}_t\gg d$ the approximation is rather accurate, while the accuracy decreases as $\widehat{\delta}_t$ approaches $d$. However, even for the case $\widehat{\delta}_t\approx d$, we can still obtain precise estimates of the first and second moments of the predictive density.

\setcounter{figure}{0}  
\setcounter{table}{0}  
\section{Simulation details and additional results}
\label{app:more_sim}
In this section we report additional details and results on the simulation study we highlighted in Section \ref{sec:sim_study}. The true data generating process is an homoskedastic VAR(1):
\begin{equation*}
    \by_t = \bTheta\,\by_{t-1} + \mathbf{u}_t, \qquad \mathbf{u}_t\sim\mathsf{N}_d(\mathbf{0}_d,\bOmega^{-1}), \qquad t=1,\ldots,T.
\end{equation*}
The reason why we focus on a VAR(1) data generating process is for direct comparability with the competing estimation methods, such as \citet{gruber2022forecasting} and \citet{gefang2023forecasting}, which do not consider the presence of exogenous predictors. 

We set the length of the time series equal to $T=360$, corresponding to $30$ years of monthly data, the dimension of the multivariate regression model equal to $d=15,30,49$ and we further assume both moderate level of sparsity ($50\%$ of zeros) and high level of sparsity ($90\%$ of zeros). The true matrix $\bTheta$ is generated as follows: we fix to zero $s\cdot d^2$ entries at random, where $s=0.5,0.9$, while the remaining non zero coefficients are sampled from a mixutre of two Gaussian with means $-0.08$ and $0.08$, and standard deviation $0.1$. Figure \ref{fig:distr_phi_true} reports the distribution of the non-zero parameters. Note the draws from the Normal distributions are truncated at $-0.05$ and $0.05$ respectively, to avoid very small values for the non zero parameters. 

\begin{figure}[ht]
\centering
\includegraphics[width=0.6\textwidth]{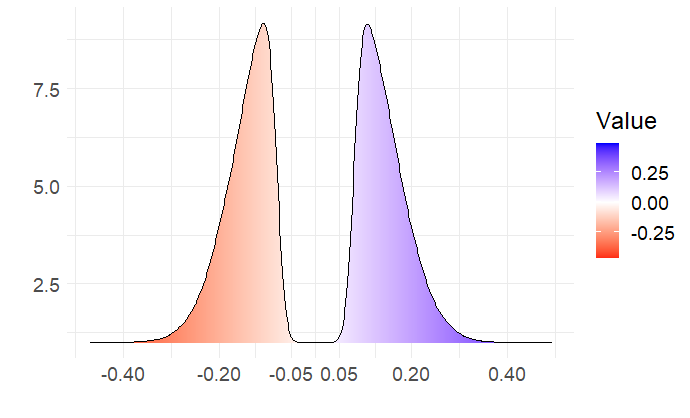}
	\caption{ Distribution of non-zero parameters in the true regression matrix. This figure plots the distribution from which we sample the non-zero entries of the regression matrices used to generate the data for the simulation study.}
	\label{fig:distr_phi_true}
\end{figure}

The variance-covariance matrix $\bOmega^{-1}$ coincides with the sample variance covariance matrix computed on the real-data used in the empirical application. The initial state $\by_0$ is sampled from the marginal distribution of the VAR(1) defined above, and we consider a burn-in period of $t_{\text{burn}}=1,\ldots,1000$ before sampling $(\by_1,\ldots,\by_T)$ from the VAR(1). Figure \ref{fig:phi_sim} shows examples of the true regression matrixes for different dimensions $d=15,30,49$ and for two alternative levels of sparsity $s=0.5, 0.9$, that is 50\% and 90\% of the entries in the matrix $\bTheta$ are set to zero. 
\begin{figure}[tp]
\centering
\subfigure[$d=15$ moderate sparsity]{\includegraphics[width=.3\textwidth]{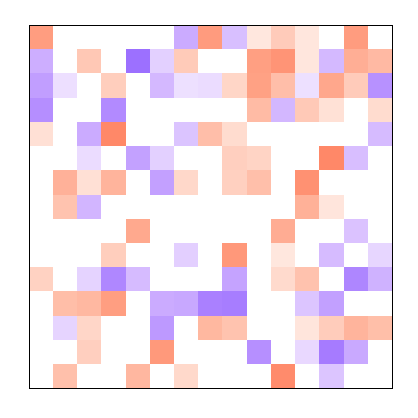}}
\subfigure[$d=30$ moderate sparsity]{\includegraphics[width=.3\textwidth]{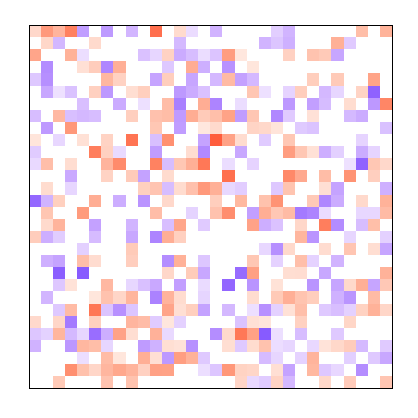}}
\subfigure[$d=49$ moderate sparsity]{\includegraphics[width=.34\textwidth]{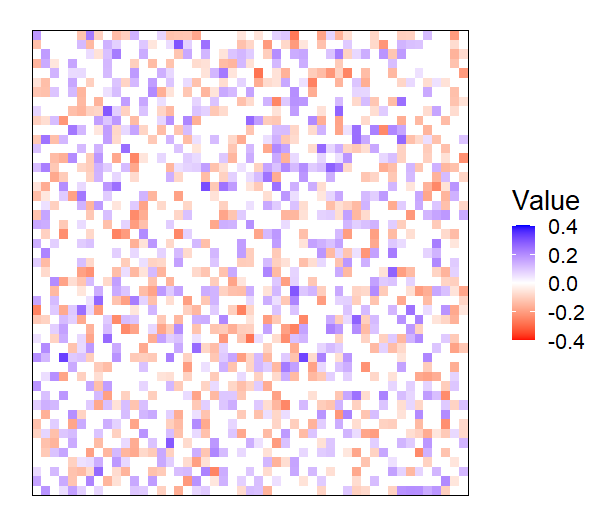}}\\
\subfigure[$d=15$ high sparsity]{\includegraphics[width=.3\textwidth]{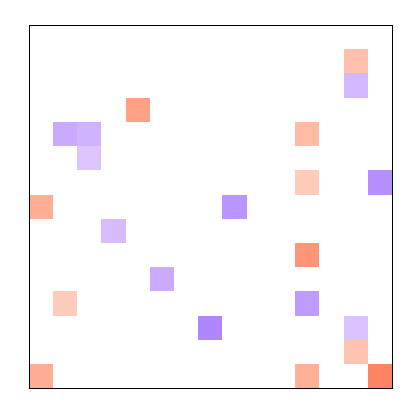}}
\subfigure[$d=30$ high sparsity]{\includegraphics[width=.3\textwidth]{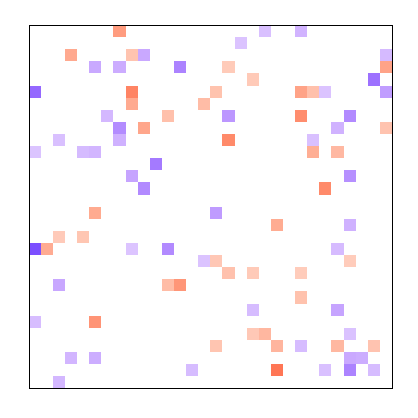}}
\subfigure[$d=49$ high sparsity]{\includegraphics[width=.34\textwidth]{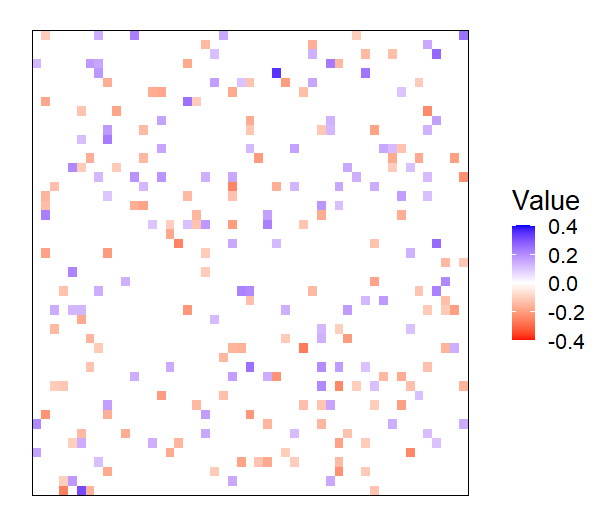}}
	\caption{True regression matrices for the simulation study. This figure plots the regression matrices used in the simulation study. We assume both moderate level of sparsity (top panels, $50\%$ of true zeros) and high level of sparsity (bottom panels, $90\%$ of true zeros).}
	\label{fig:phi_sim}
\end{figure}

\subsection{Additional simulation results}
We complement the results in the main text and show some of the additional results on a smaller model dimension of $d=15$. Figure \ref{fig:frobenius15} reports the Frobenius norm (top panels) and the F1 score (bottom panels) as in the main text. The labeling and structure of figure is the same as in Figure \ref{fig:fig3}. Similar to the larger VAR cases, our {\tt VB} estimation procedure outperform both MCMC and variational methods based on a structural VAR formulation. On the other hand, the non-linear MCMC proposed by \citet{gruber2022forecasting} turns out to be quite competitive. Nevertheless, our {\tt VB} approach is more accurate for both the adaptive lasso and horseshoe priors, especially when sparsity is more pervasive. 

\begin{figure}[t!]

\subfigure[Frobenius norm $d=15$, moderate sparsity]{\includegraphics[width=.48\textwidth]{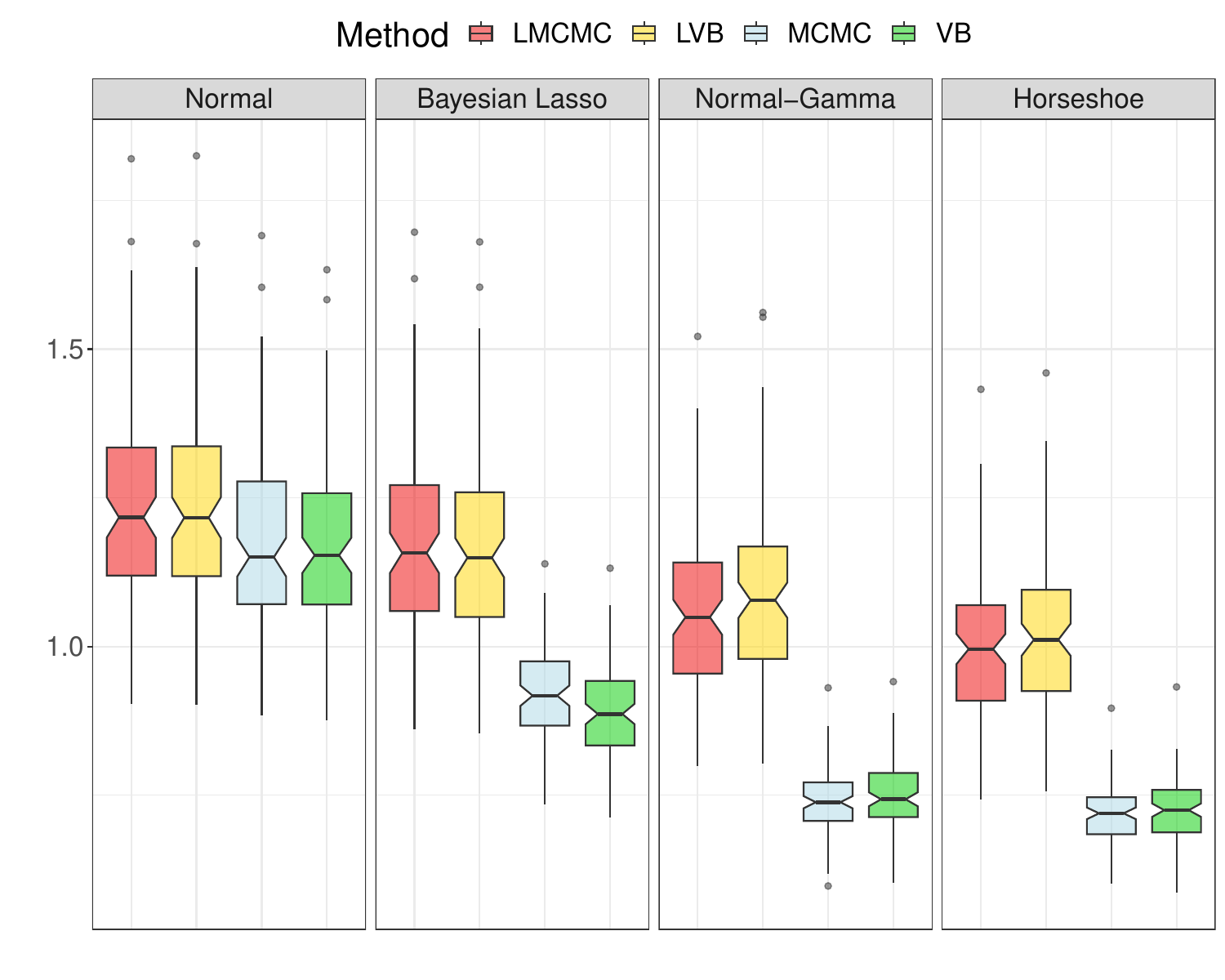}}
\subfigure[Frobenius norm $d=15$, high sparsity]{\includegraphics[width=.48\textwidth]{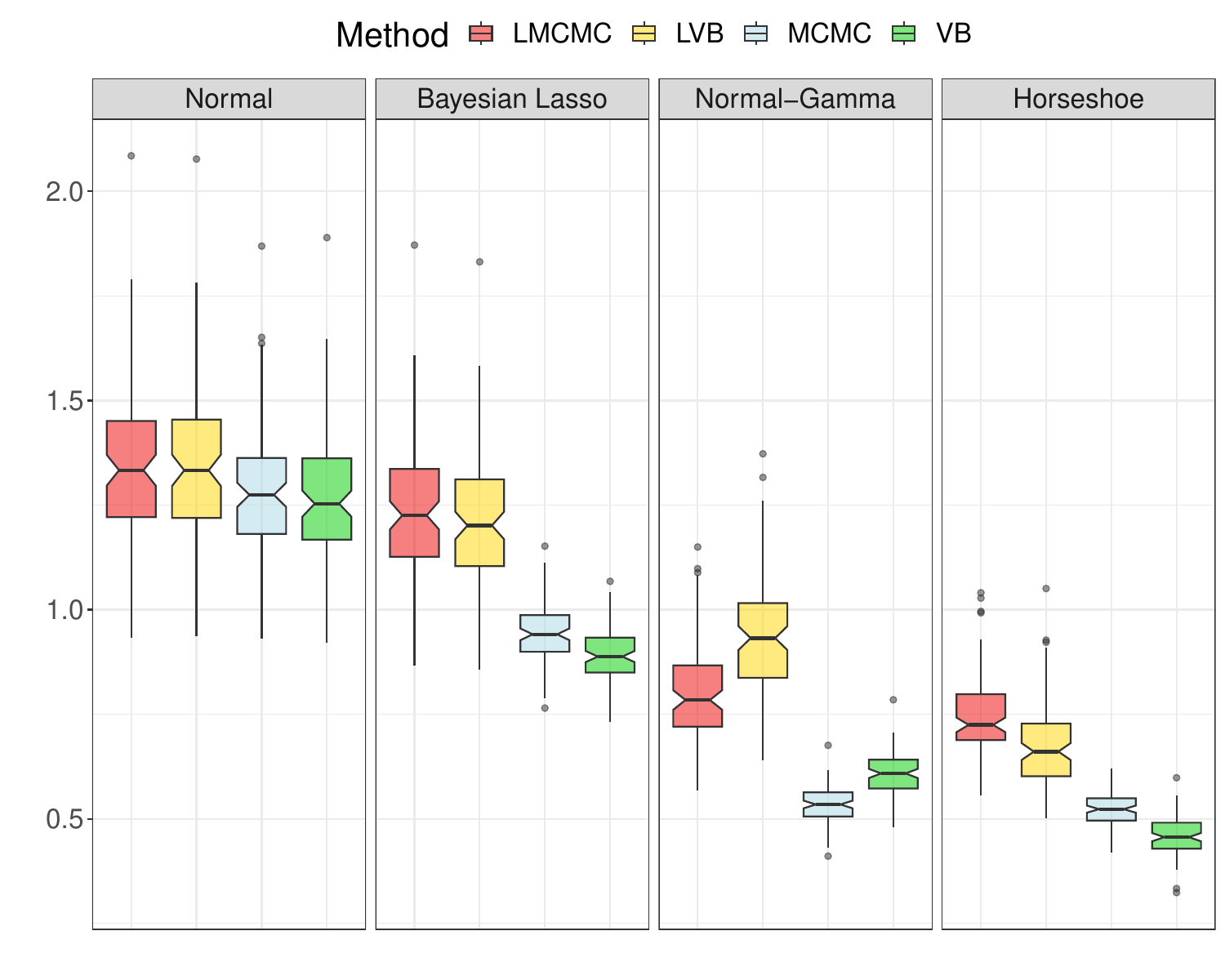}}

\subfigure[F1 score norm $d=15$, moderate sparsity]{\includegraphics[width=.48\textwidth]{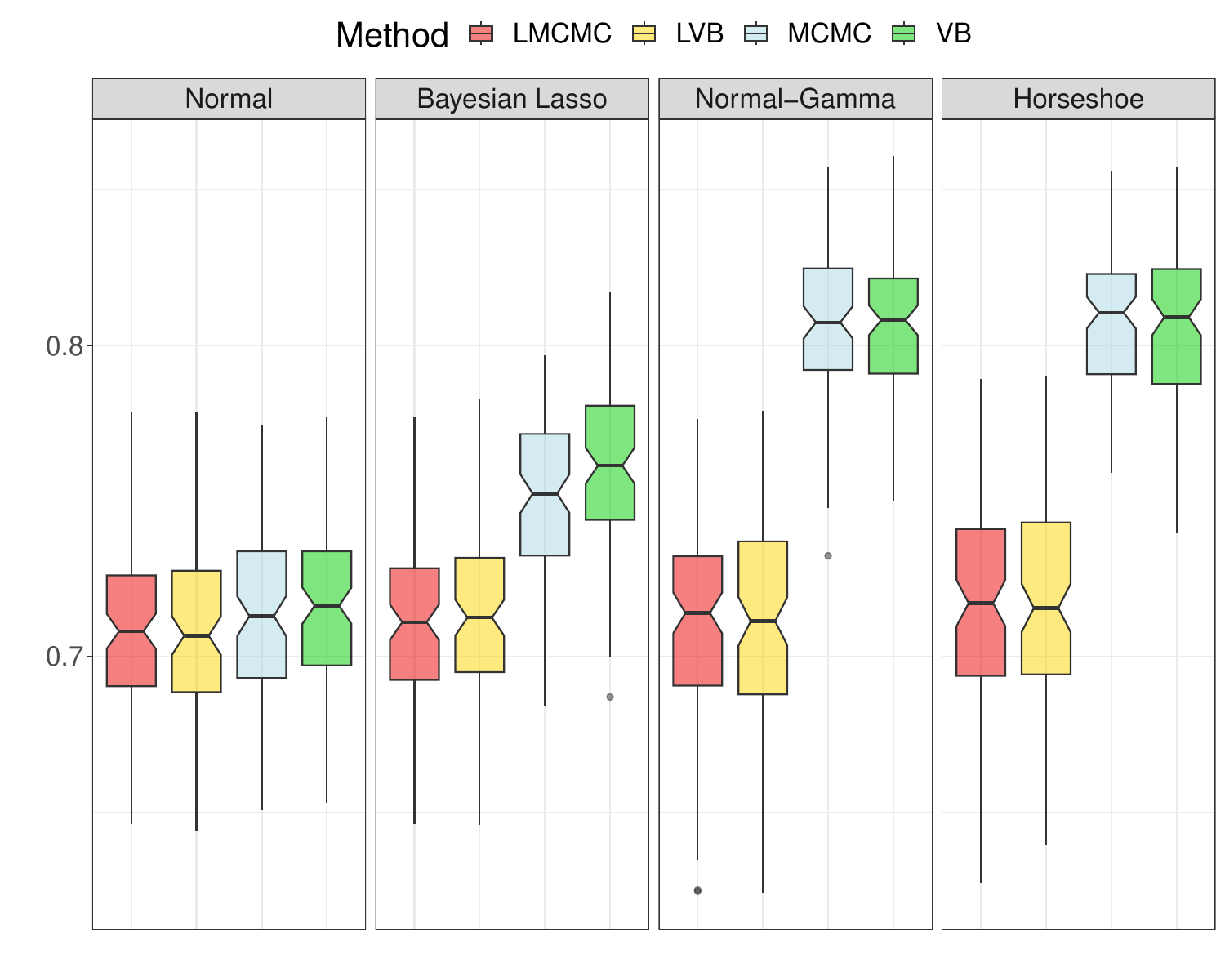}}
\subfigure[F1 score norm $d=15$, high sparsity]{\includegraphics[width=.48\textwidth]{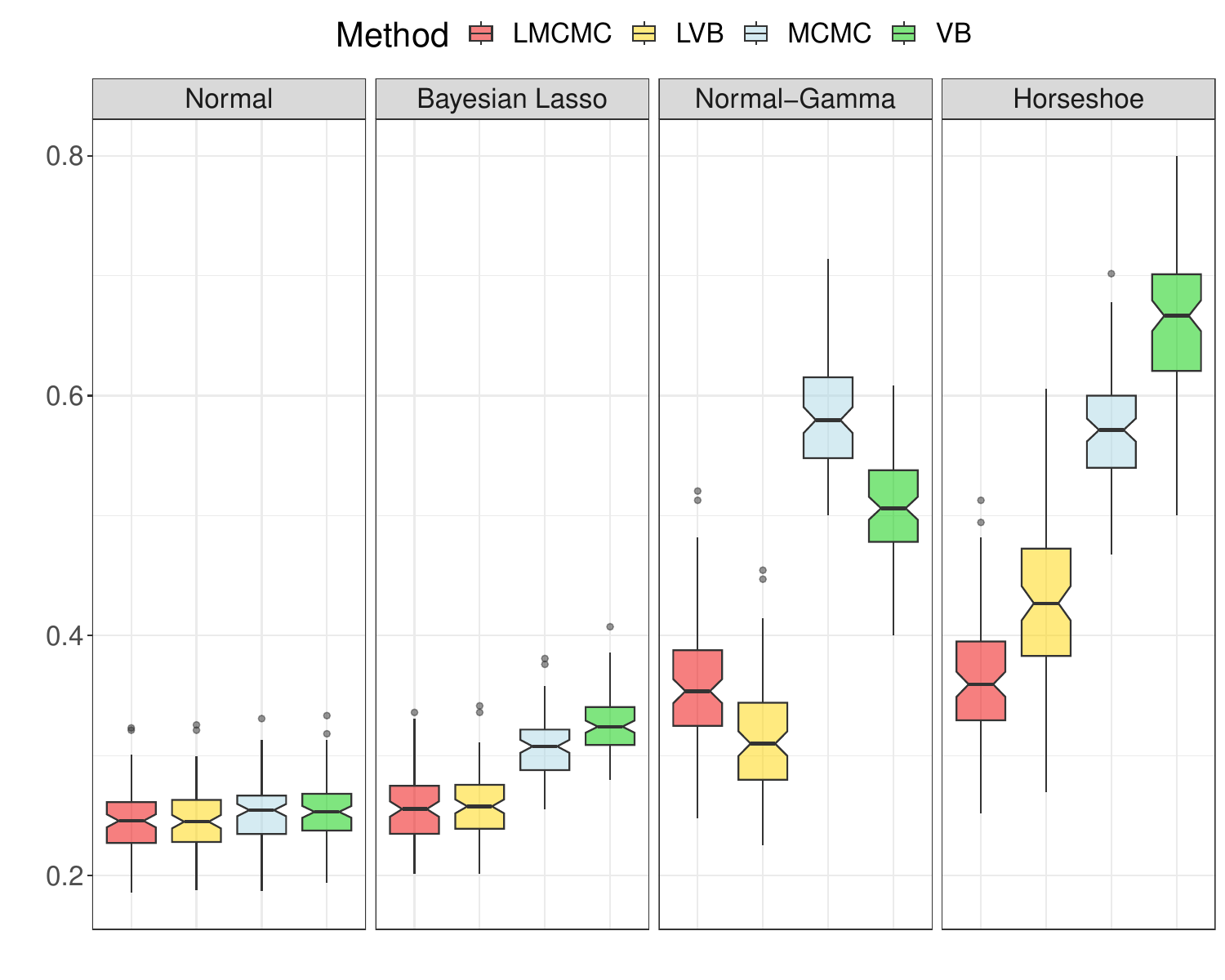}}
		\caption{\small Top panels report the Frobenius norm of $\boldsymbol{\Theta}-\widehat{\boldsymbol{\Theta}}$ for different hierarchical shrinkage priors and estimation methods. Bottom panels report the F1 score computed looking at the true non-null parameters in $\boldsymbol{\Theta}$ and the non-null parameters in the estimated matrix $\widehat{\boldsymbol{\Theta}}$.
		The box charts show the results for $N=100$ replications, $d=15$ and different levels of sparsity.}
	\label{fig:frobenius15}
\end{figure}

{\color{black}Based on the same simulation setting described above, we now investigate the performance of all estimation methods under variables permutation. Figure \ref{fig:fig3a} shows the box charts of the Frobenius norms (top panels) and F1 scores (bottom panels) for the $N=100$ replications for both moderate and high sparsity in the true $\boldsymbol{\Theta}$. For ease of exposition, we only report the case with $d=30$ predictors. We put in each figure the simulation results pertaining to the original $\mathbf{y}_t$ (solid) and its reversed order $\mathbf{y}^{rev}_t$ (shaded) next to each other. Colors/labels are the same as in the main simulation study.

\begin{figure}[t!]
	\centering

    \subfigure[Frobenius norm $d=30$, moderate sparsity]{\includegraphics[width=.45\textwidth]{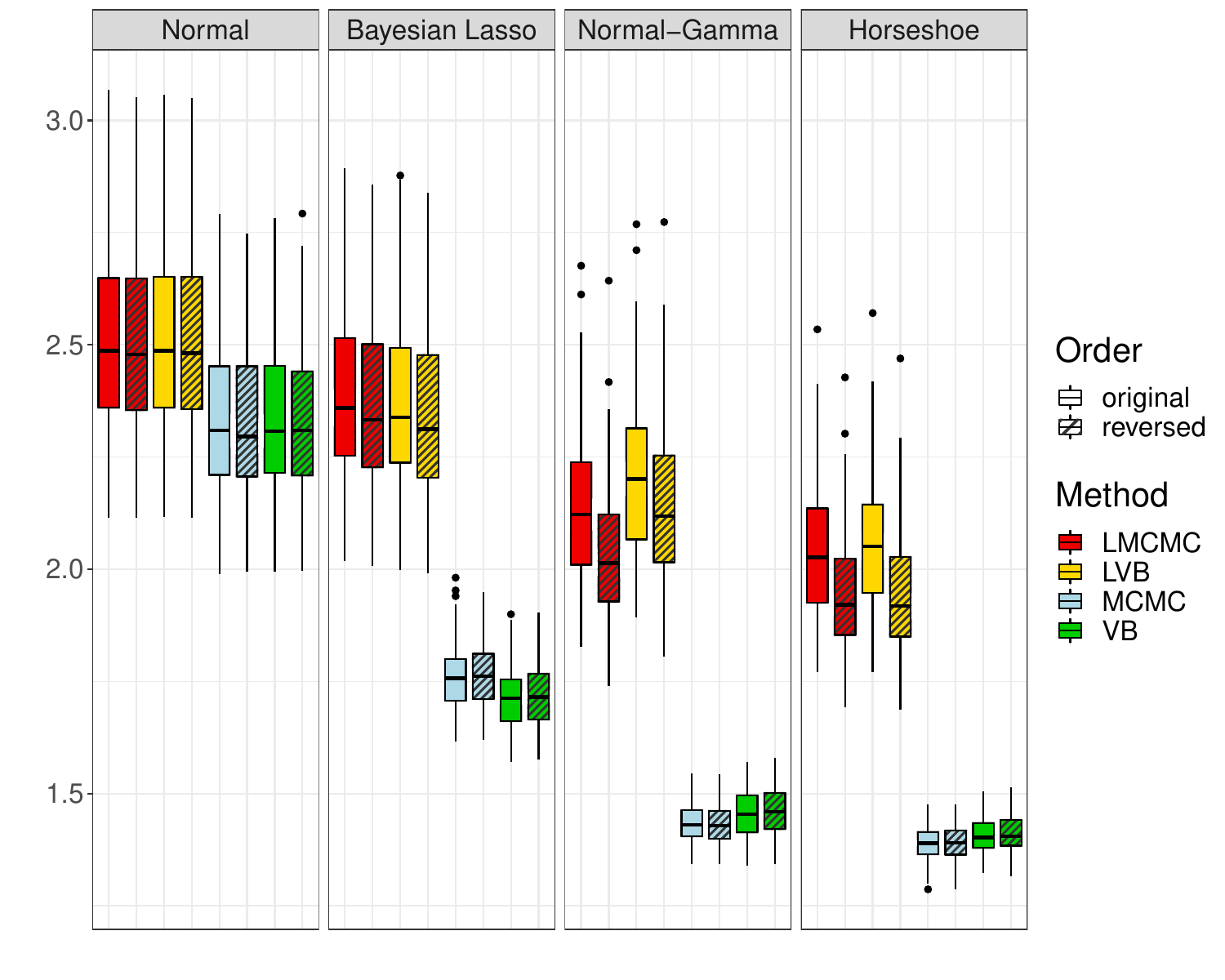}}
    \subfigure[Frobenius norm $d=30$, high sparsity]{\includegraphics[width=.45\textwidth]{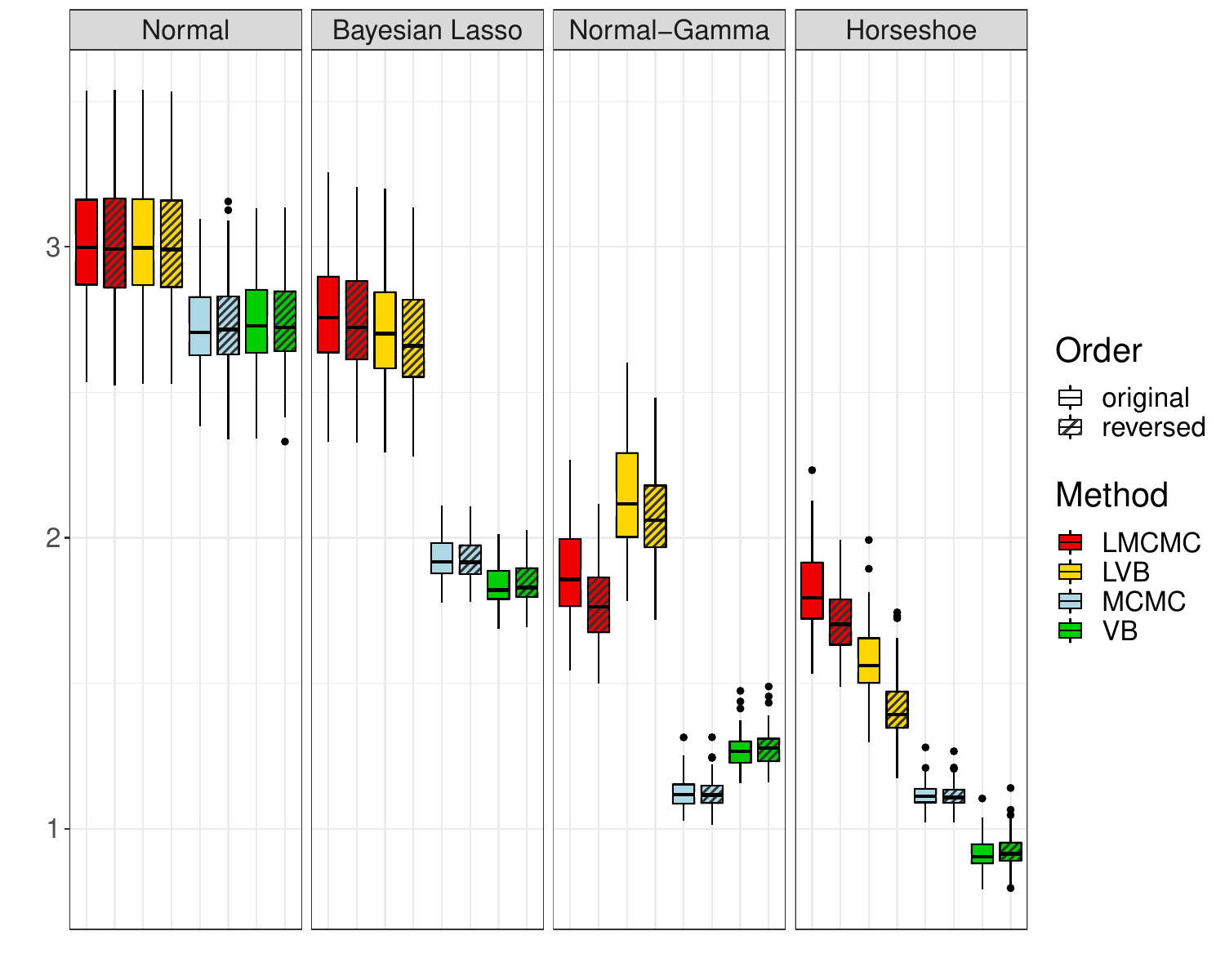}}
    
      \subfigure[F1 score $d=30$, moderate sparsity]{\includegraphics[width=.45\textwidth]{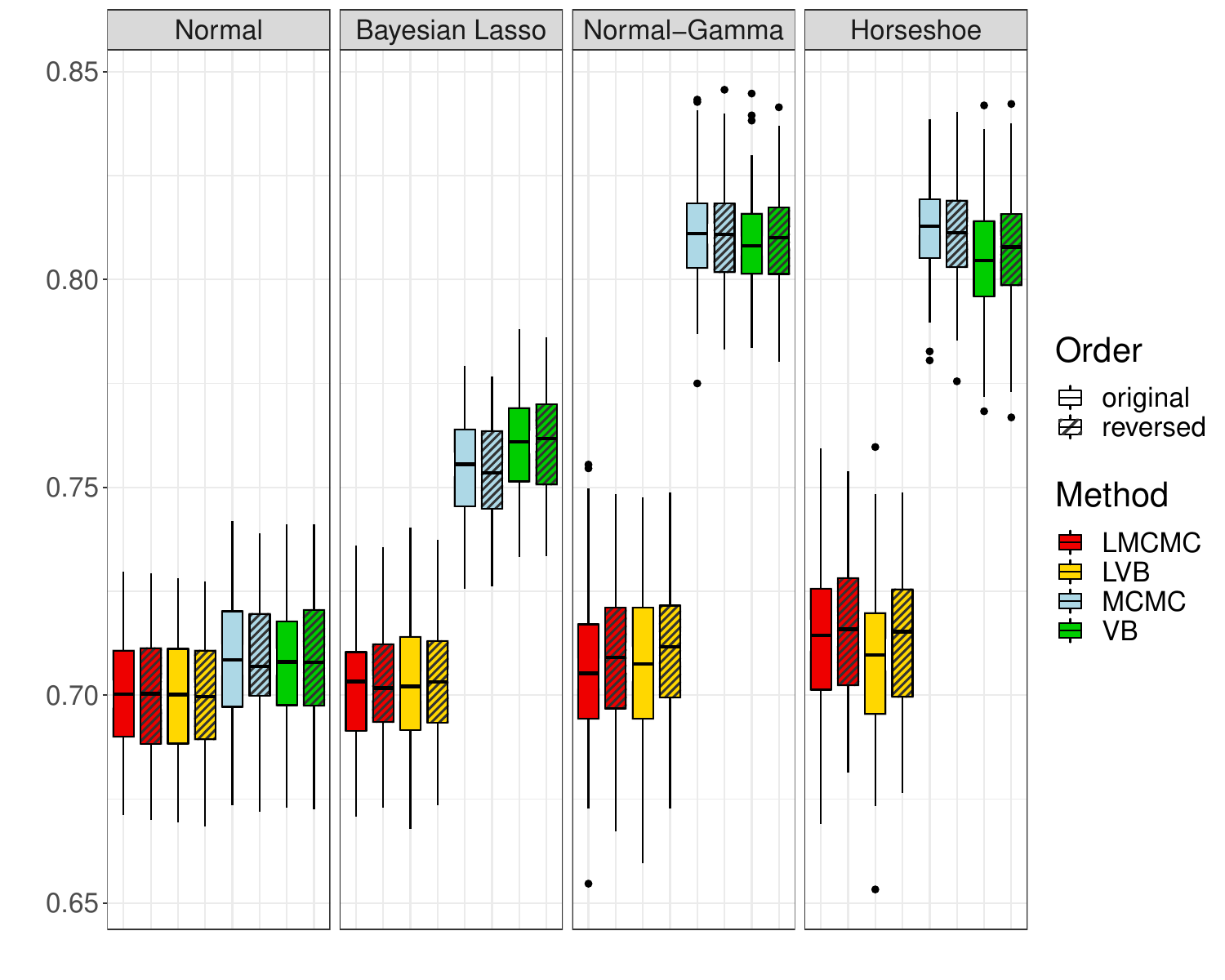}}
    \subfigure[F1 score $d=30$, high sparsity]{\includegraphics[width=.45\textwidth]{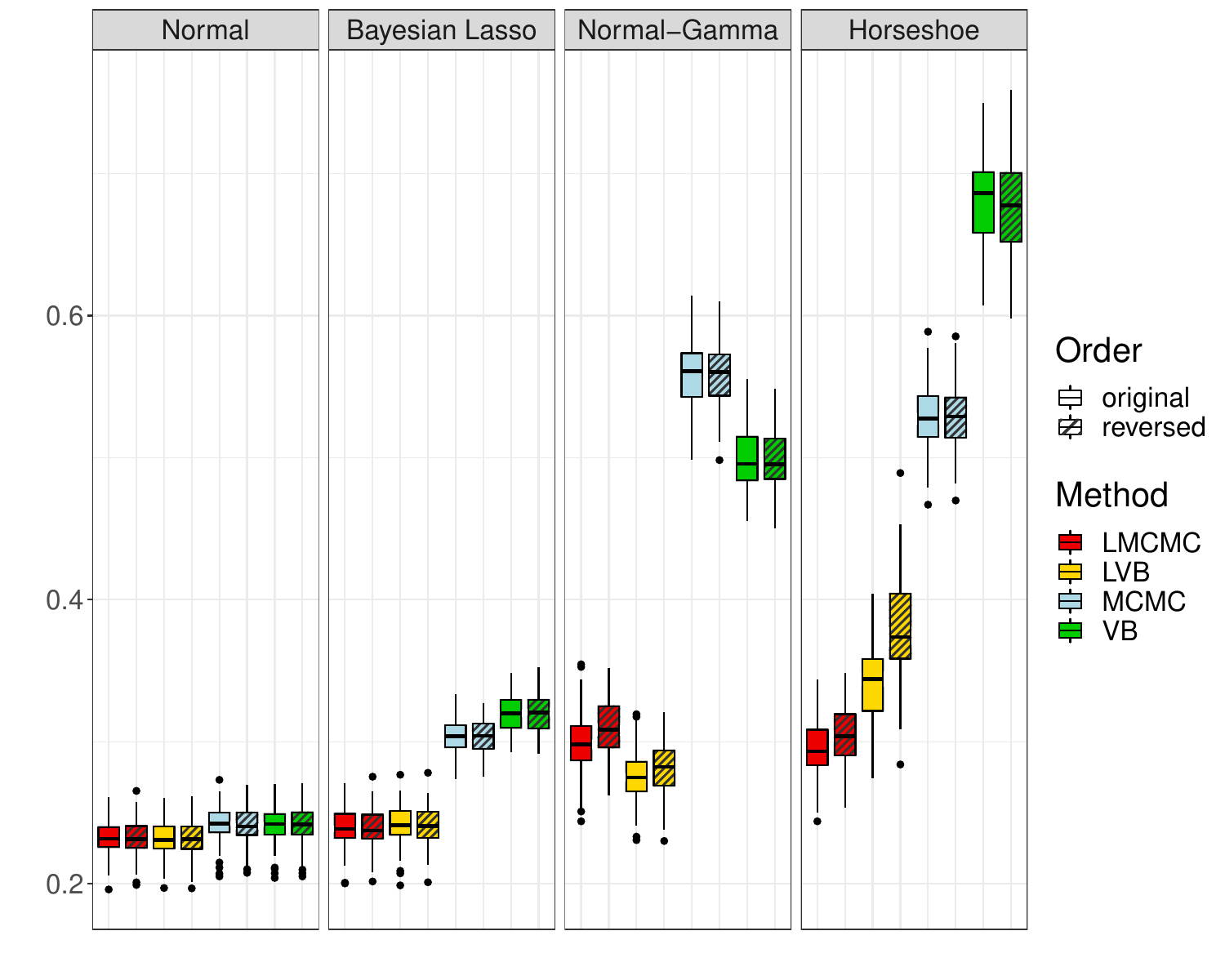}}
    
		\caption{\small Top panels report the Frobenius norm of $\boldsymbol{\Theta}-\widehat{\boldsymbol{\Theta}}$ under variables permutation for different shrinkage priors and inference approaches. Bottom panels report the F1 score computed looking at the true non-null parameters in $\boldsymbol{\Theta}$ and the non-null parameters in $\widehat{\boldsymbol{\Theta}}$.
		The box charts show the results for $N=100$ replications, $d=30$ and different levels of sparsity.}
        \label{fig:fig3a}
\end{figure}
The accuracy of the estimates of both {\tt LMCMC} and {\tt LVB} tend to deteriorate when reverting the ordering of the target variables. This is especially clear for the normal-gamma and the horseshoe priors and when the amount of zero coefficients in $\boldsymbol{\Theta}$ is more pervasive. Such performance deterioration is due to the fact that $\boldsymbol{\Theta}=\mathbf{L}^{-1}\mathbf{A}$ from the structural VAR formulation so that the posterior estimate $\widehat{\boldsymbol{\Theta}}$ changes depending on the variables ordering implied by $\mathbf{L}$. The higher the level of sparsity, the larger the disconnect between $\mathbf{A}$ and $\boldsymbol{\Theta}$. 

On the other hand, being built on the same non-linear parametrization both the {\tt MCMC} of \citet{gruber2022forecasting} and our {\tt VB} approach are substantially less sensitive to variables permutation. This applies across prior specifications, model dimension, and level of sparsity in the true matrix $\boldsymbol{\Theta}$.}

\subsection{A multivariate version of \citet{hahn2015decoupling}}
\label{subsec:postprocessing}
{\color{black}The implementation of the sparsity-inducing approach of \cite{hahn2015decoupling} to our multivariate context requires a non-trivial extension. In their original work, the authors assume a linear regression model $\mathbf{y} = \mathbf{X}\boldsymbol{\beta} + \boldsymbol{\varepsilon}$ and uncorrelated Gaussian error terms, $\boldsymbol{\varepsilon}\sim\mathsf{N}_n(0,\sigma^2\mathbf{I}_n)$. Thus, their procedure consists to run the following least-angle regression (LARS) for a grid of tuning parameters $\lambda$:
\begin{align}
\boldsymbol{\beta}_\lambda &= \arg\min_\gamma \sum_j\frac{\lambda}{|\widehat{\beta}_j|}|\gamma_j| + n^{-1}||\mathbf{X}\widehat{\boldsymbol{\beta}}-\mathbf{X}\boldsymbol{\gamma}||^2_2,
\end{align}
where $\widehat{\boldsymbol{\beta}}$ denotes the posterior mean, and, then, to compute, for each $\lambda$ and each draw $(\boldsymbol{\beta}^{(r)},\sigma^{2\,(r)})$, the variation-explained for the sparsified linear predictor $\boldsymbol{\beta}_\lambda$:
\begin{align}
\rho^{2\,(r)}_\lambda &= \frac{n^{-1}||\mathbf{X}\boldsymbol{\beta}^{(r)}||^2}{n^{-1}||\mathbf{X}\boldsymbol{\beta}^{(r)}||^2+\sigma^{2\,(r)}+n^{-1}||\mathbf{X}\boldsymbol{\beta}^{(r)}-\mathbf{X}\boldsymbol{\beta}_\lambda||^2}.
\end{align}
The selection follows a comparison between $\rho^{2}_\lambda$ and $\rho^2_{\lambda=0}$ based on the following heuristic: report the sparsified linear predictor corresponding to the smallest model whose $90\%$ $\rho^2_\lambda$ credible
interval contains $E(\rho^2_{\lambda=0})$, that is, select the smallest linear predictor whose variance-explained is not
statistically different than the full model. 

In our setting, we need to define a suitable formula to compute $\rho^{2}_\lambda$ when $\mathbf{y} = \mathbf{X}\boldsymbol{\beta} + \boldsymbol{\varepsilon}$ and the error terms are correlated, i.e. $\boldsymbol{\varepsilon}\sim\mathsf{N}_n(0,\mathbf{\Sigma})$. A natural choice appears to be:
\begin{align}\label{eq:rho2}
\rho^{2\,(r)}_\lambda = \frac{n^{-1}\boldsymbol{\beta}^{\intercal\,(r)}\mathbf{X}^{\intercal}\mathbf{\Sigma}^{-1\,(r)}\mathbf{X}\boldsymbol{\beta}^{(r)}}{n^{-1}\boldsymbol{\beta}^{\intercal\,(r)}\mathbf{X}^{\intercal}\mathbf{\Sigma}^{-1\,(r)}\mathbf{X}\boldsymbol{\beta}^{(r)}+1+n^{-1}(\mathbf{X}\boldsymbol{\beta}^{(r)}-\mathbf{X}\boldsymbol{\beta}_\lambda)^{\intercal}\mathbf{\Sigma}^{-1\,(r)}(\mathbf{X}\boldsymbol{\beta}^{(r)}-\mathbf{X}\boldsymbol{\beta}_\lambda)}.
\end{align}
Notice that, if $\mathbf{\Sigma}=\sigma^2\mathbf{I}_n$ then we obtain the original approach of \cite{hahn2015decoupling}. 

Before discussing some of the additional simulation results, two comments are in order. First, the selection from \cite{hahn2015decoupling} depends on some non-negligible arbitrariness. Specifically, the comparison between $\rho^{2}_\lambda$ and $\rho^2_{\lambda=0}$ is carried out using the selection summary plots \citep[Section 3 of][]{hahn2015decoupling}. Second, and perhaps more importantly, the post-processing approach based on SAVS is an order of magnitude faster. Indeed, the approach of \cite{hahn2015decoupling} requires the evaluation of Eq.\eqref{eq:rho2} for each $\lambda$ and each draws from the posterior. Moreover, $\lambda$ values are defined over a grid: if the latter is too coarse, then the selection procedure might be inaccurate, while if it is too dense, the computational burden suddenly increases. 

According to \cite{pallavi_battacharya2019savs}, the latter issue does not affect the SAVS procedures, which indeed does not require tuning parameters and it is computationally fast. To put things into perspective, with $d=30$, considering $5,000$ draws from the posterior after the burn-in, and a grid of $200$ values for $\lambda$, the SAVS procedure provides a sparse estimate immediately, while the \citet{hahn2015decoupling} approach takes $\approx 1$ minute.

\begin{figure}[t!]
	\centering
\subfigure[$d=30$, moderate sparsity, SAVS]{\includegraphics[width=.45\textwidth]{FiguresRevision/f1sco_30_50.pdf}}\hspace{1em}
\subfigure[$d=30$, high sparsity, SAVS]{\includegraphics[width=.45\textwidth]{FiguresRevision/f1sco_30_90.pdf}}

\subfigure[$d=30$, moderate sparsity, HC]{\includegraphics[width=.45\textwidth]{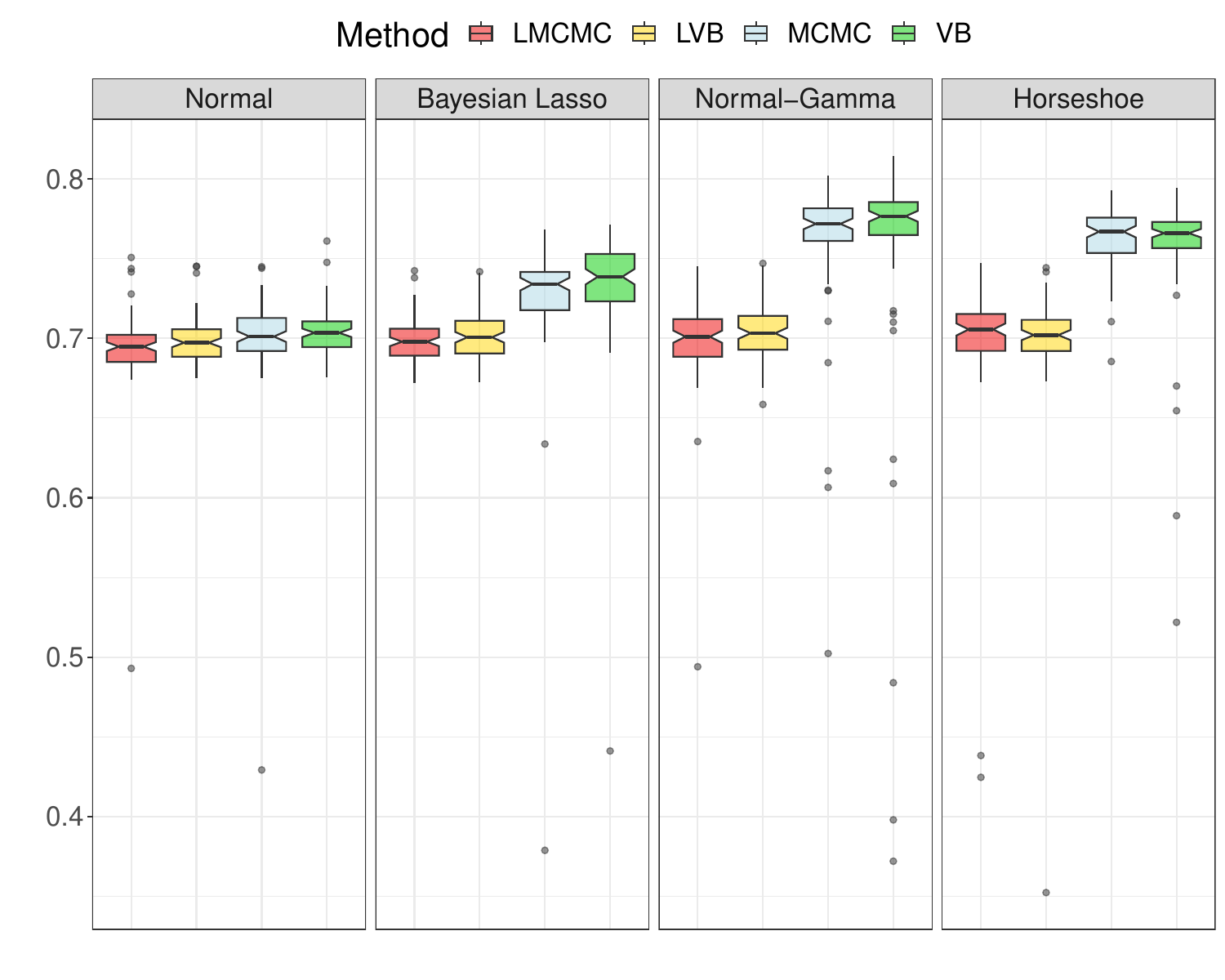}}\hspace{1em}
\subfigure[$d=30$, high sparsity, HC]{\includegraphics[width=.45\textwidth]{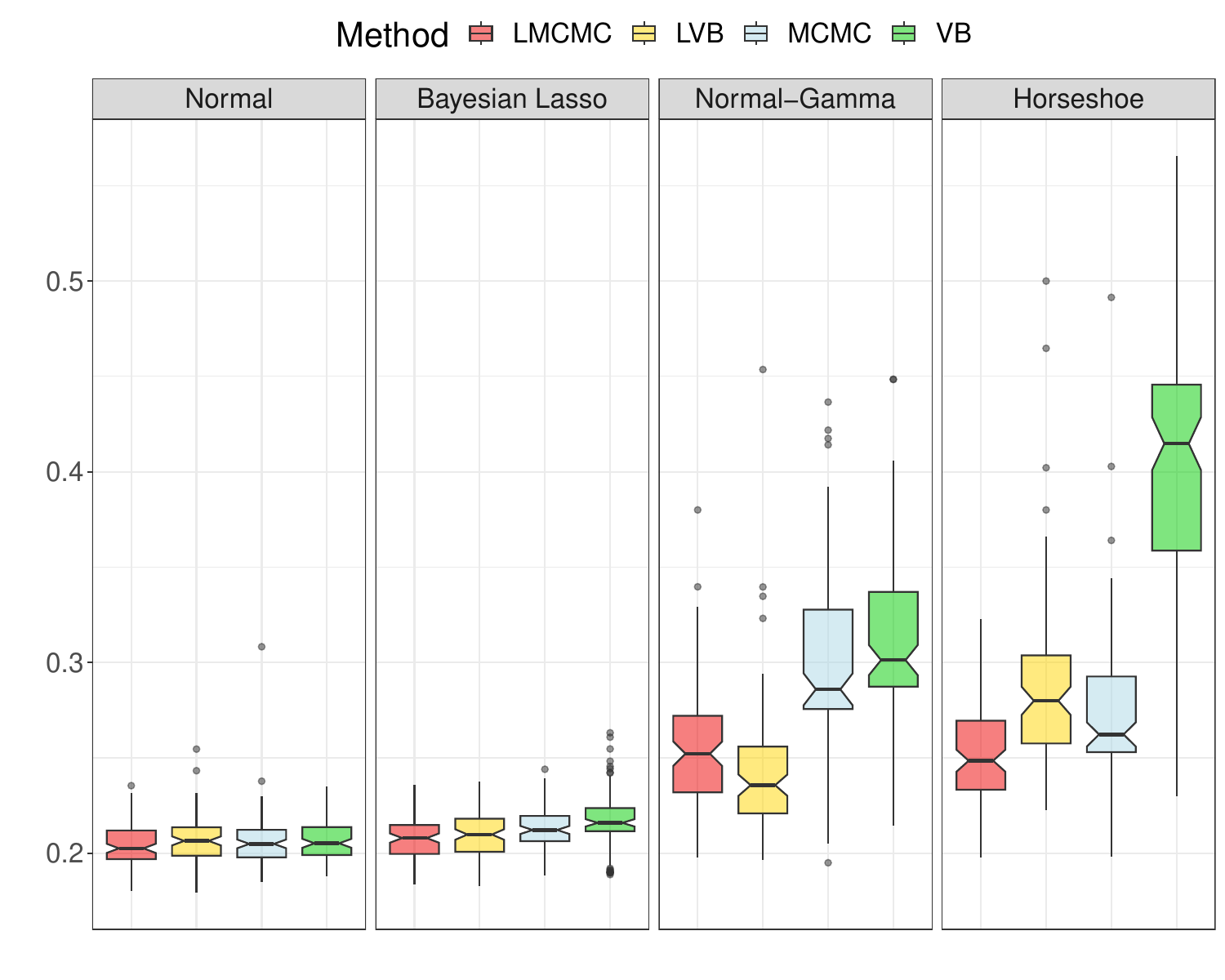}}

\caption{F1 score computed looking at the true non-null parameters in $\boldsymbol{\Theta}$ and the non-null parameters estimated based on $\widehat{\boldsymbol{\Theta}}$.}
	\label{fig:SAVS}
\end{figure}

Figure \ref{fig:SAVS} compares the F1 score based on the same posterior and variational estimates, but with either the SAVS (top panels) or the extended version of \citet{hahn2015decoupling} as outlined above across different shrinkage priors. For ease of exposition, we report uniquely the results for the $d=30$ case. The F1 scores across methods remain largely the same, in fact, the results are even more strongly in favor of our {\tt VB} compared to its {\tt MCMC} counterpart when using the extended \citet{hahn2015decoupling} approach. Specifically, our {\tt VB} is more accurate than {\tt MCMC} under the normal-gamma prior.}



\section{Additional empirical considerations}

\subsection{Computational cost of the recursive forecasts}
\label{subsec:computational cost}
{\color{black}In this section, we discuss more explicitly the qualitative differences in terms of computational efficiency across estimation methods. Starting with \cite{carriero2019large,carriero2022corrigendum}, they consider $d=20,40$ and show that the average computational time to perform 10 draws is 2.5 and 27.3 seconds, respectively, on a 3.5 GHz Intel Core i7 (see Figure 1 in \citealp{carriero2022corrigendum}). This means that for 10,000 draws (as in our case) it takes 41 minutes for $d=20$ and 7.5 hours for $d=40$ per monthly forecast. Similarly, on a 2.5 GHz Intel Xeon W-2175 with 32GB of RAM it would take approximately 40 minutes per forecast to implement the {\tt MCMC} approach of \citet{gruber2022forecasting} for a $d=30$ implementation with constant volatility. \citet{huber2019adaptive}, based on a similar non-linear MCMC algorithm for $d=20$ variables takes around 1.3 hours for 30,000 posterior draws, or 26 minutes for 10,000 draws. These results are all consistent with our own implementations of these methods. 

By comparison, our {\tt VB} with stochastic volatility takes less than 3 minutes for each recursive forecast with $d=30$. This has key implications for practical forecasting use; for instance, a recursive forecast of $d=30$ industry portfolios for 767 out-of-sample observations based on a constant-volatility specification of \citet{gruber2022forecasting} would take $20\ \text{min}\times 767\ \text{forecasts} \times 4\ \text{priors}= 76,700$ minutes, or 42 days to complete. This compares to $10\ \text{sec}\times 767\ \text{forecasts} \times 4\ \text{priors}= 511$ minutes, or almost 9 hours to complete the empirical exercise under a constant-volatility specification with our variational inference approach. 

To summarize, a substantially higher computational efficiency coupled with a comparable accuracy with complex MCMC, makes our {\tt VB} extremely competitive within the context of recursive forecasts in higher frequency data.}


\subsection{Forecasting performance over the business cycle}
Figure \ref{fig:R2oos app} reports the $R_{j,oos}^2\left(\mathcal{M}_s\right)$ (in \%) across 30 (left panel) and 49 (right panel) industry portfolios during recession periods. 

\begin{figure}[h!]
	\centering
\hspace{-1em}\subfigure[$R_{j,oos}\left(\mathcal{M}_s\right)^2$ across 30 industry portfolios]{\includegraphics[width=.42\textwidth]{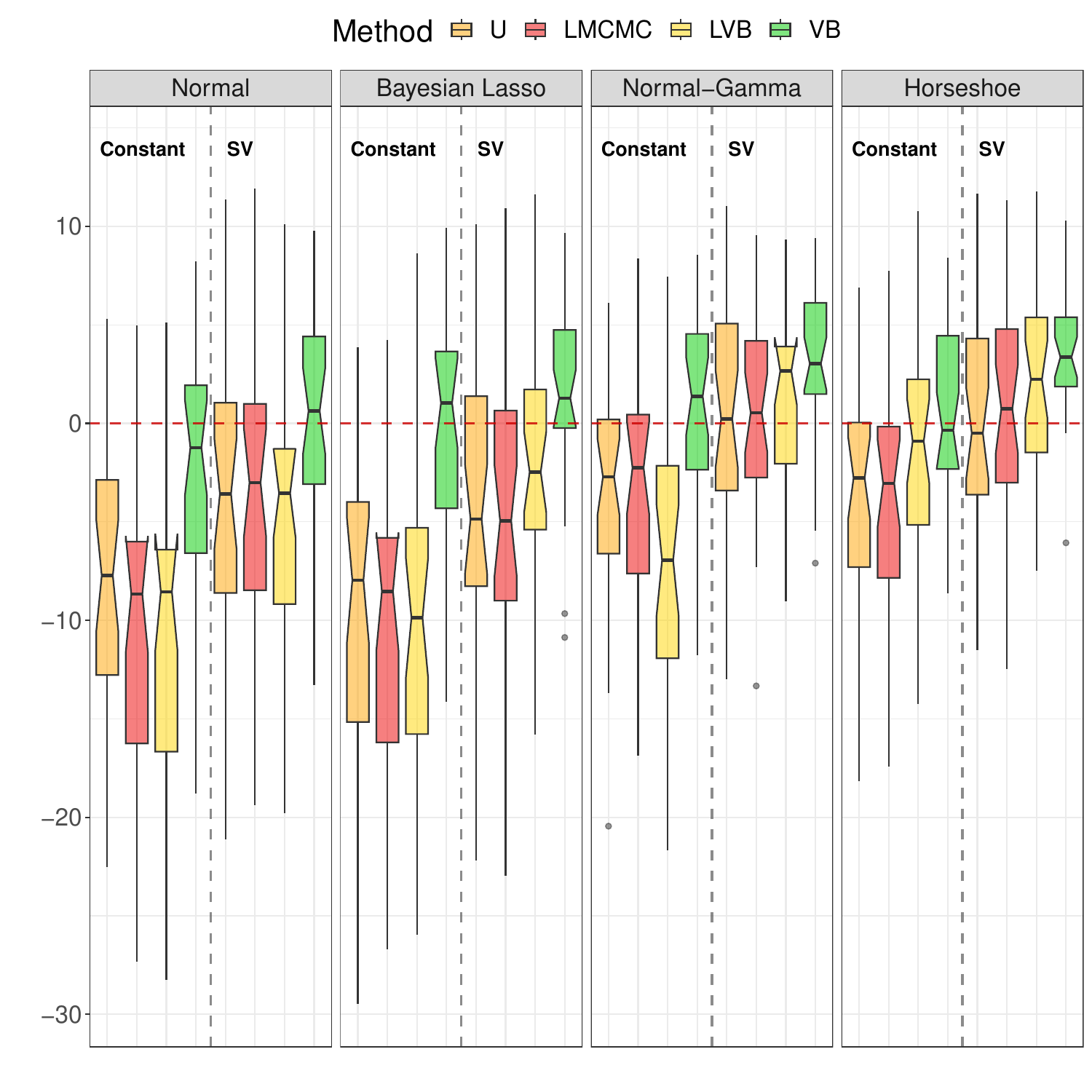}}\hspace{1em}
\subfigure[$R_{j,oos}^2\left(\mathcal{M}_s\right)$ across 49 industry portfolios]{\includegraphics[width=.42\textwidth]{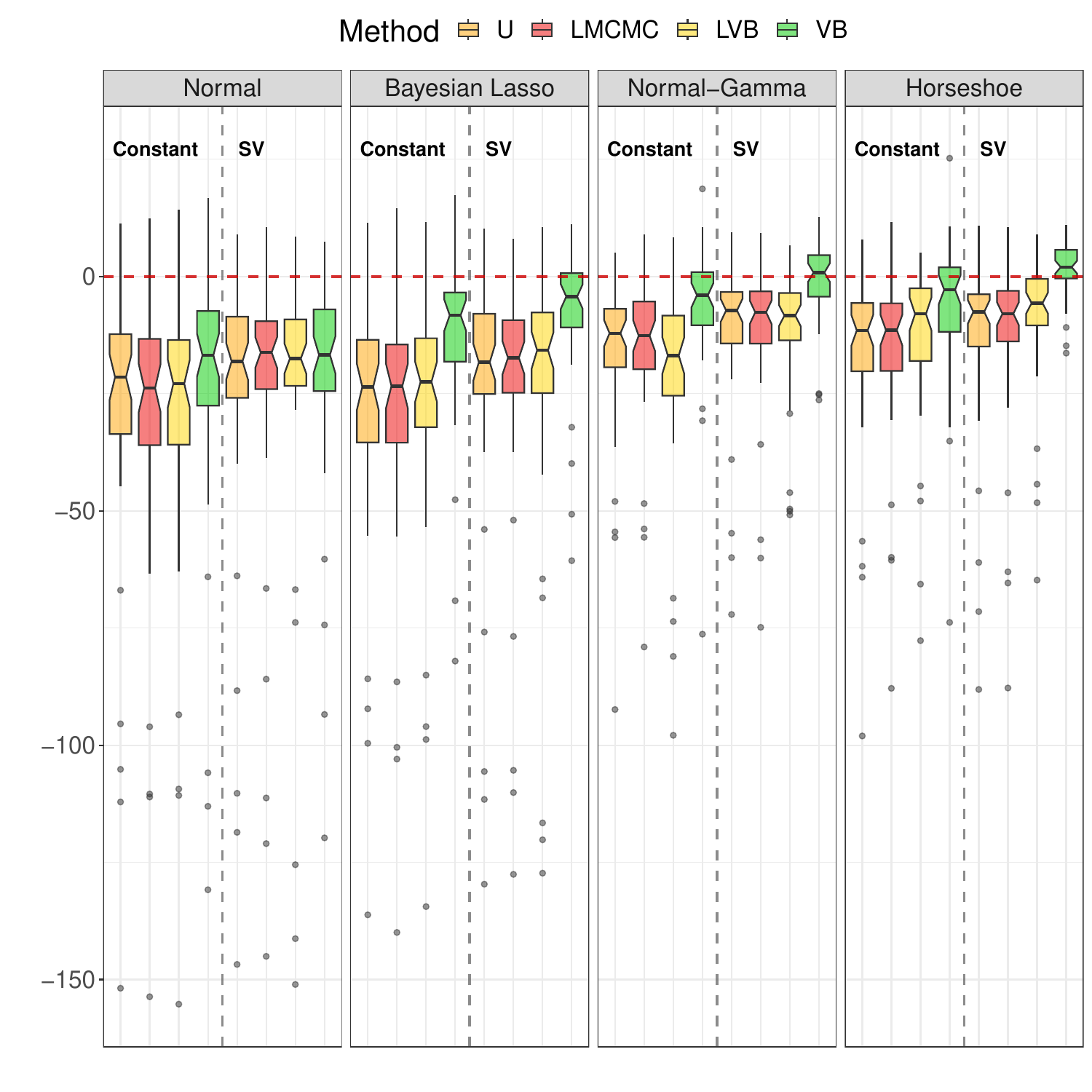}}

\caption{This figure reports the $R_{j,oos}^2\left(\mathcal{M}_s\right)$ (in \%) across 30 (left panel) and 49 (right panel) industry portfolios.}
	\label{fig:R2oos app}
\end{figure}

\subsection{Additional in-sample results}
\label{app:more_emp}
Figure \ref{fig:theta app} shows the in-sample posterior estimates estimates of the regression coefficients for the $d=30$ industry case. The in-sample estimates of $\widehat{\bTheta}$ are based on the full sample obtained from the {\tt LMCMC} and the {\tt LVB} with constant volatility, and the {\tt VB} with and without stochastic volatility. Similar to the larger-dimensional setting in the main text, the in-sample estimates highlight three key results. First, and perhaps not surprisingly, there are visible differences across shrinkage priors. For instance, the horseshoe tend to shrinkage parameters more aggressively so that $\widehat{\bTheta}$ is more sparse compared to the normal gamma. Second, the estimates of the {\tt LMCMC} and {\tt LVB} tend to be closely related, consistent with \citet{gefang2023forecasting}. Yet, the estimates for the {\tt VB} are substantially different under the same prior. This is due to the fact that $\widehat{\boldsymbol{\Theta}}=\widehat{\mathbf{L}}^{-1}\widehat{\mathbf{A}}$ in Eq.\eqref{eq:var1_orth_def2}, so that the estimated $\widehat{\mathbf{A}}$ is not translation-invariant, unlike in our approach. Third, the estimates from {\tt VB} are remarkably stable between constant vs stochastic volatility specifications, with the only exception of the adaptive lasso prior.

\begin{figure}[h!]	
\centering
\hspace{-2.5em}\subfigure[{\tt LMCMC} w/ normal]{\includegraphics[width=.25\textwidth]{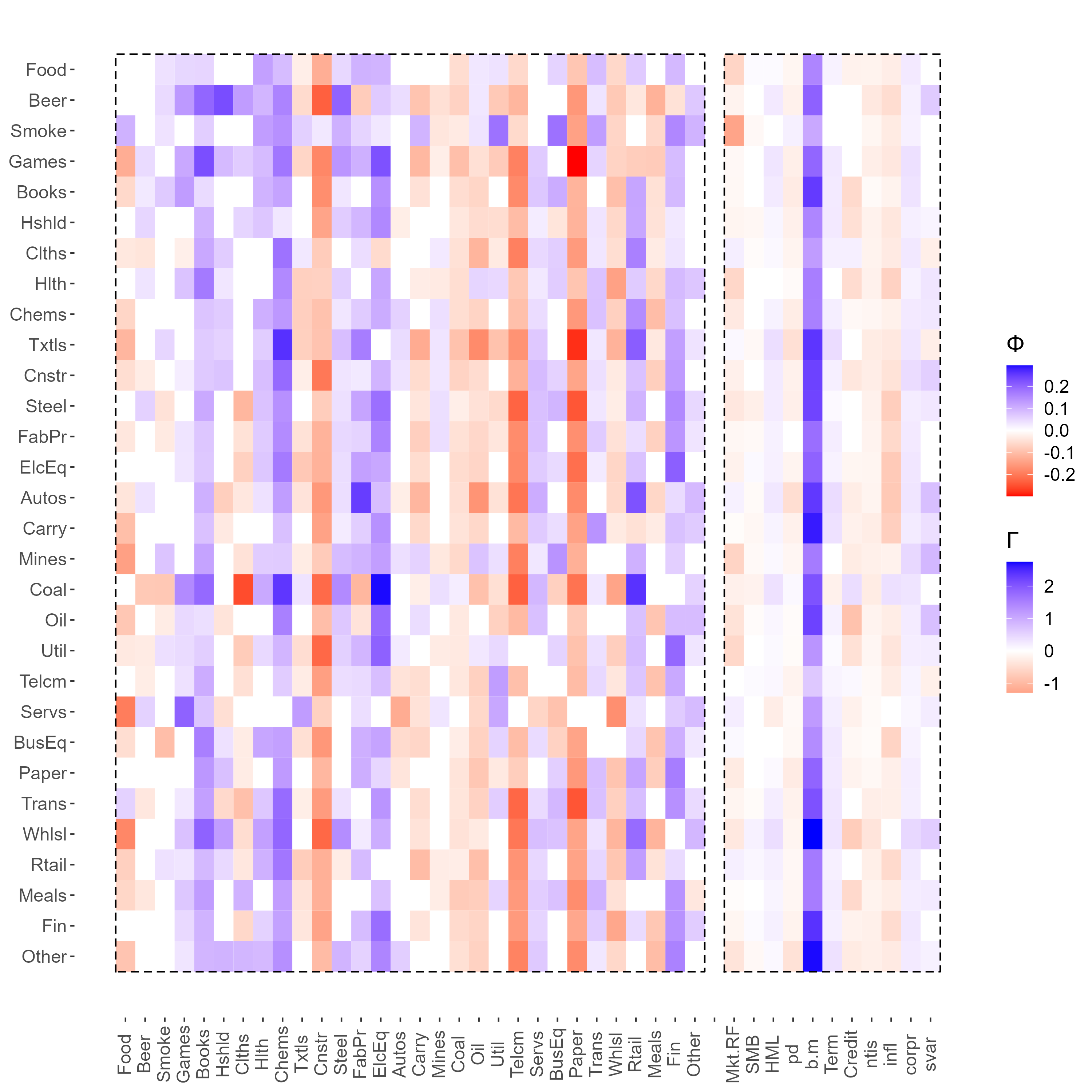}}\quad
\subfigure[{\tt LVB} w/ normal]{\includegraphics[width=.25\textwidth]{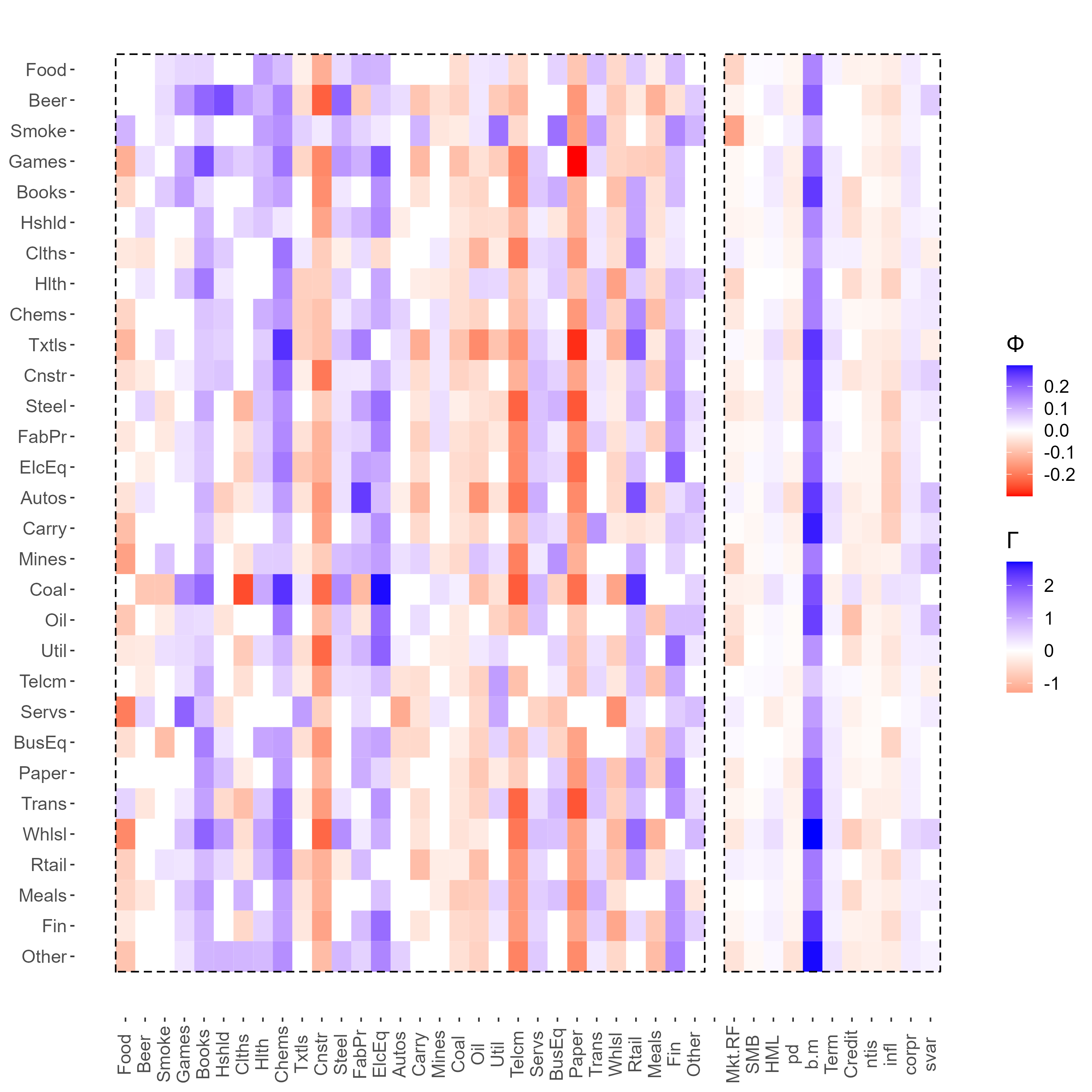}}
\subfigure[{\tt VB} w/ normal]{\includegraphics[width=.25\textwidth]{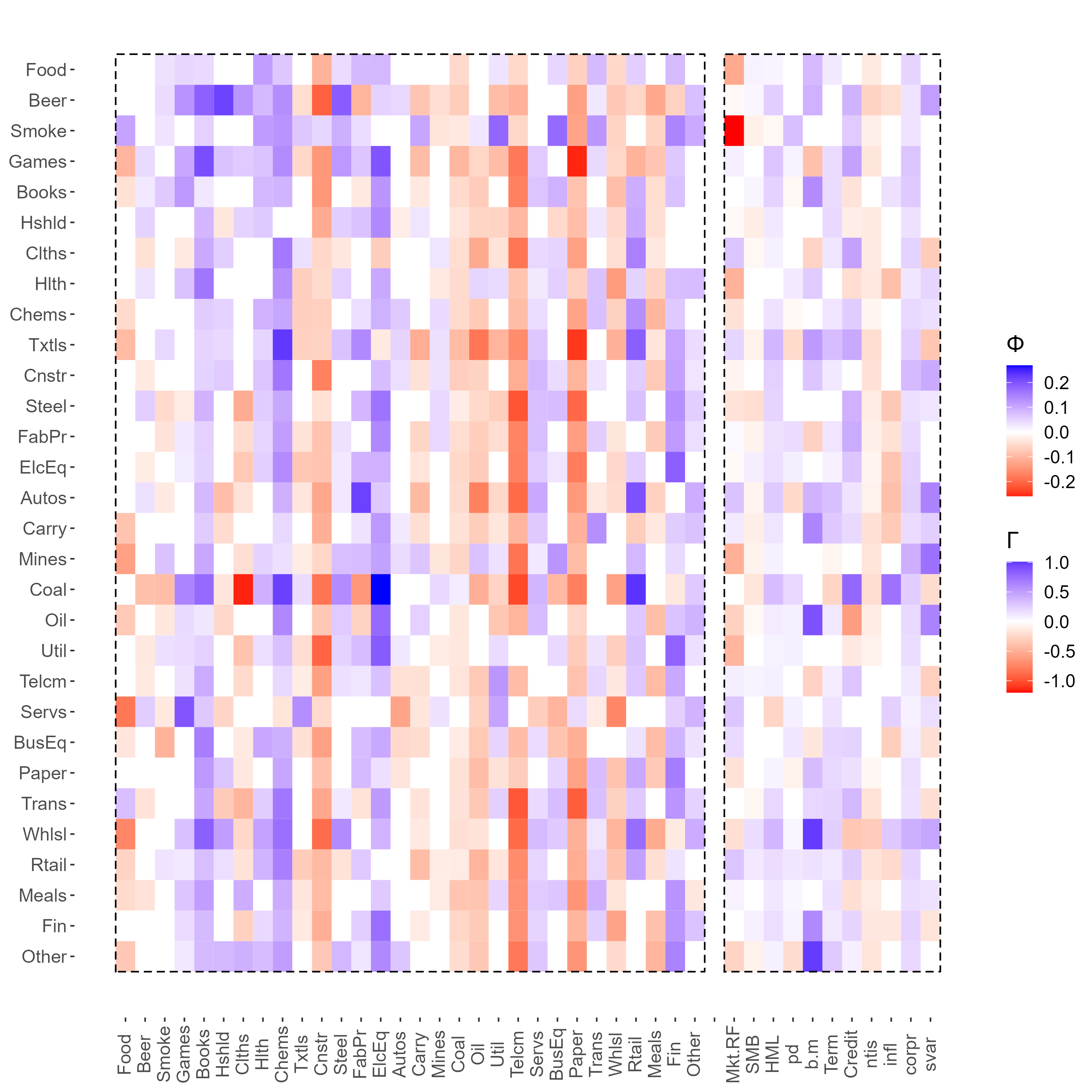}}\subfigure[{\tt VB} w/ normal + SV]{\includegraphics[width=.25\textwidth]{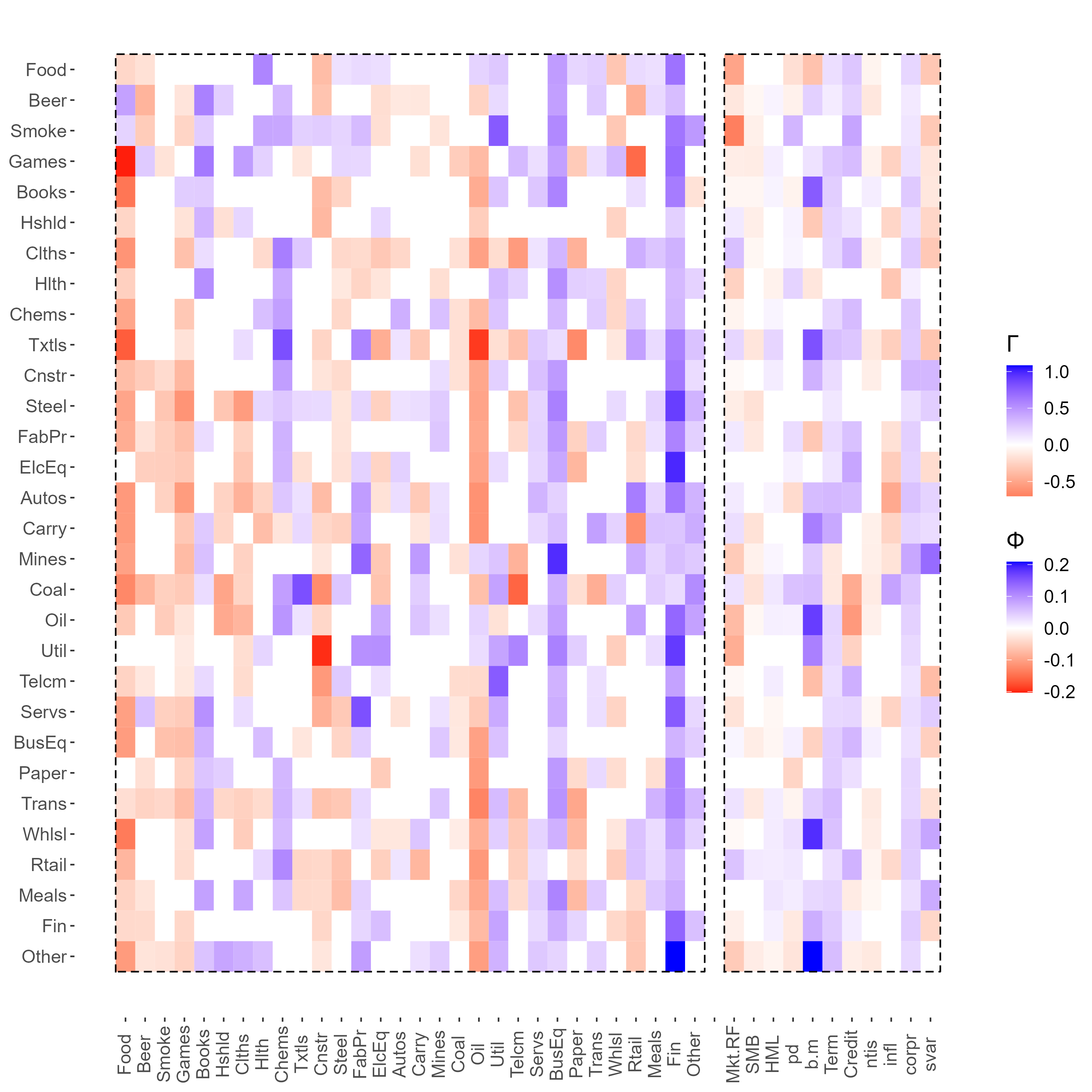}}\hspace{-2.5em}

\hspace{-2.5em}\subfigure[{\tt LMCMC} w/ Lasso]{\includegraphics[width=.25\textwidth]{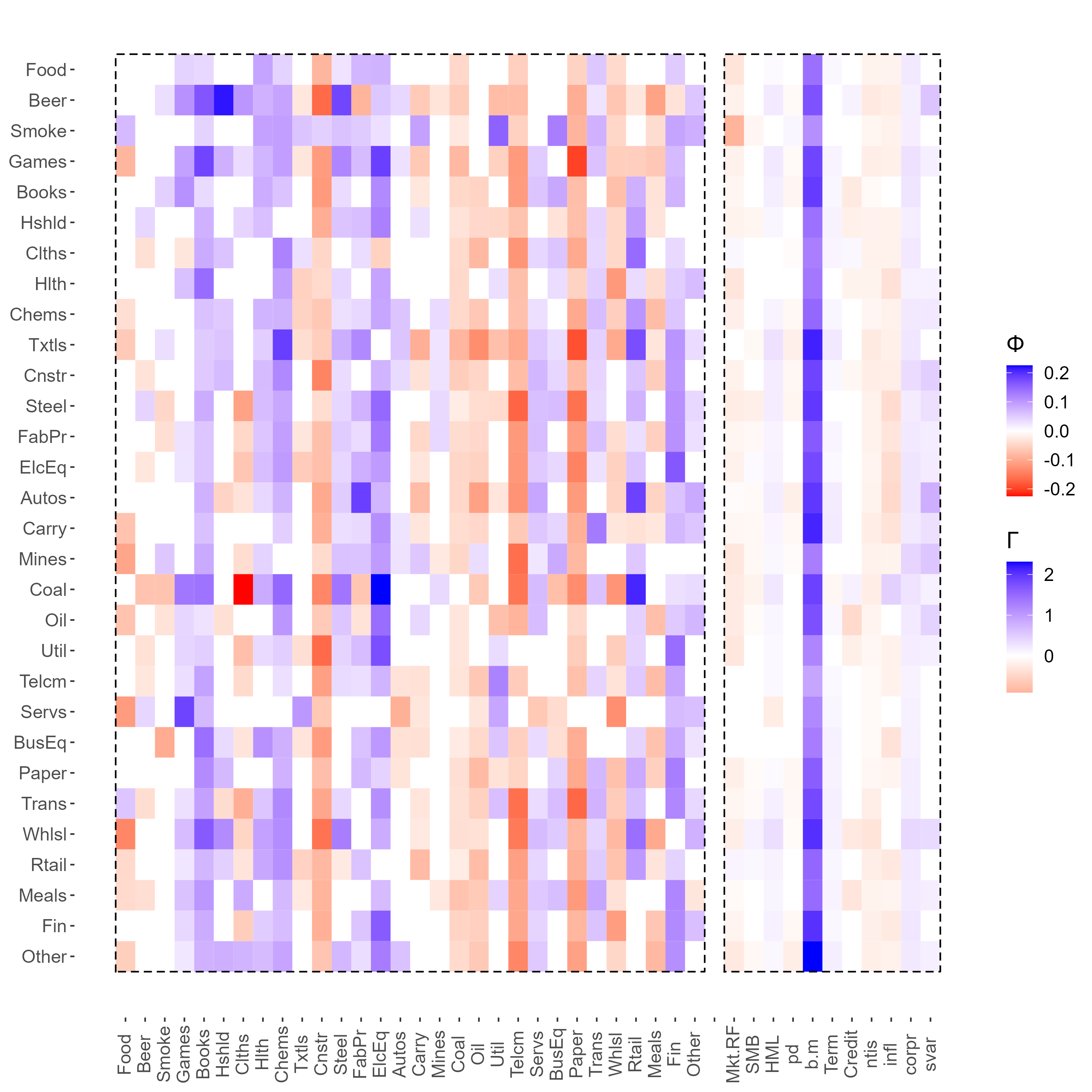}}\quad
\subfigure[{\tt LVB} w/ Lasso]{\includegraphics[width=.25\textwidth]{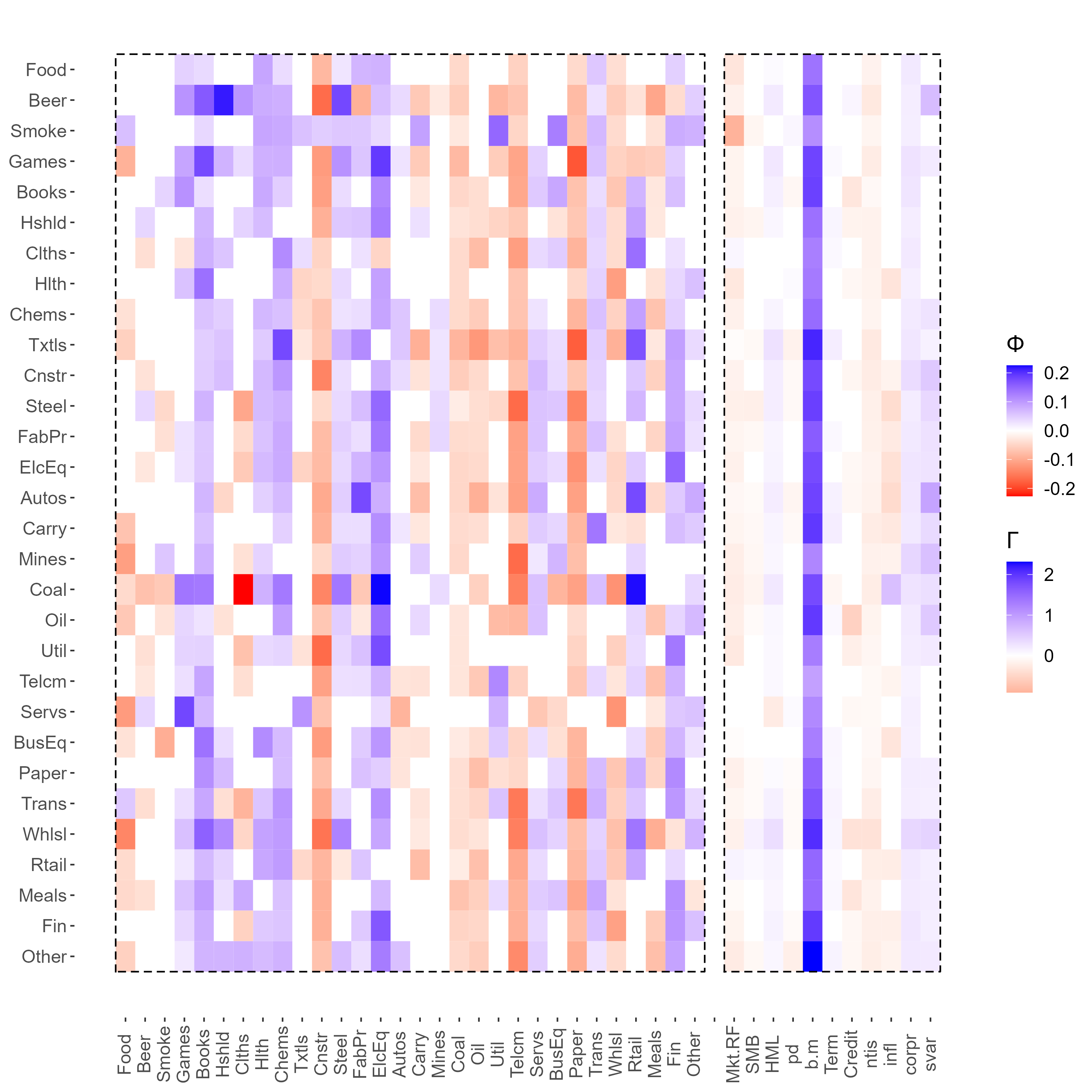}}
\subfigure[{\tt VB} w/ Lasso]{\includegraphics[width=.25\textwidth]{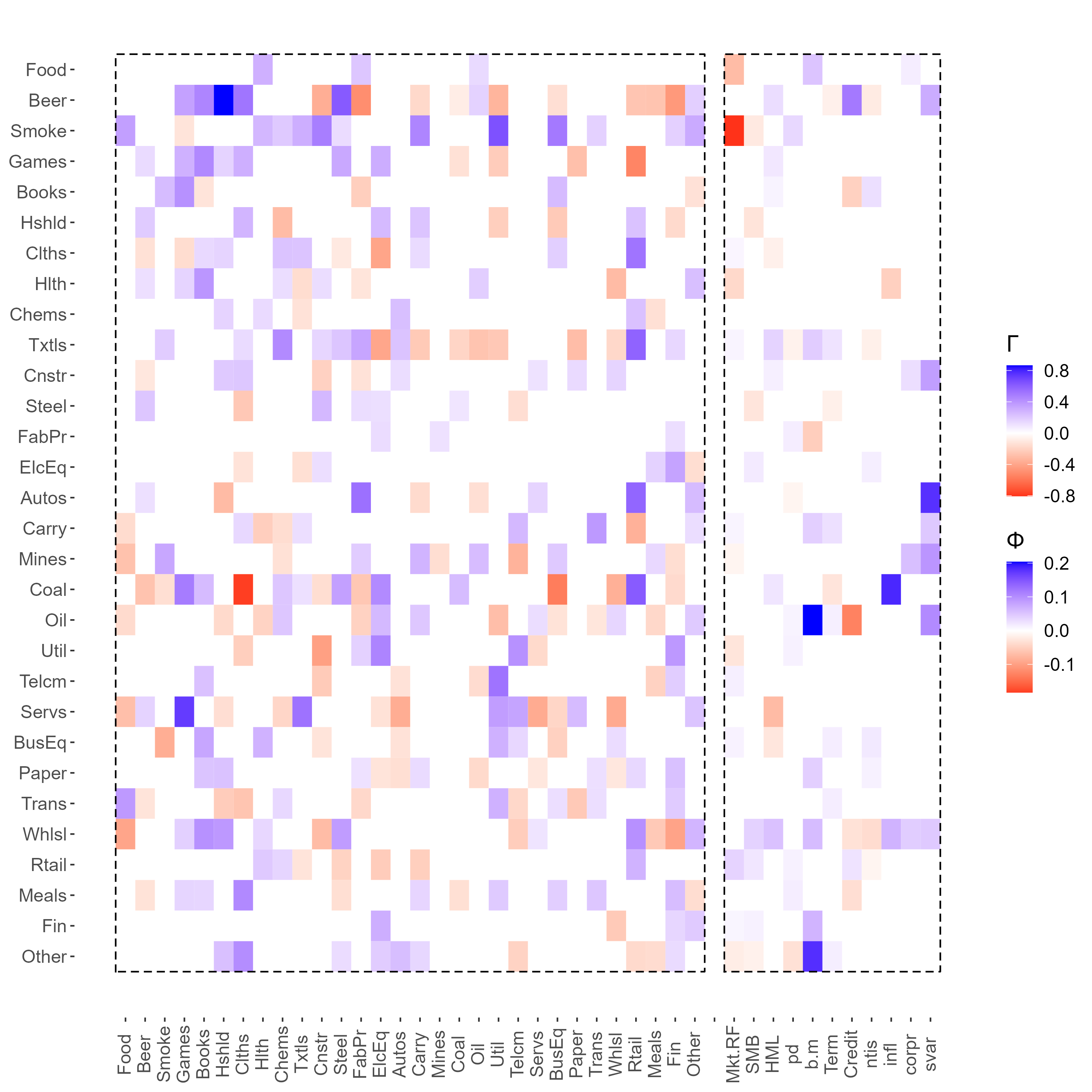}}\subfigure[{\tt VB} w/ Lasso + SV]{\includegraphics[width=.25\textwidth]{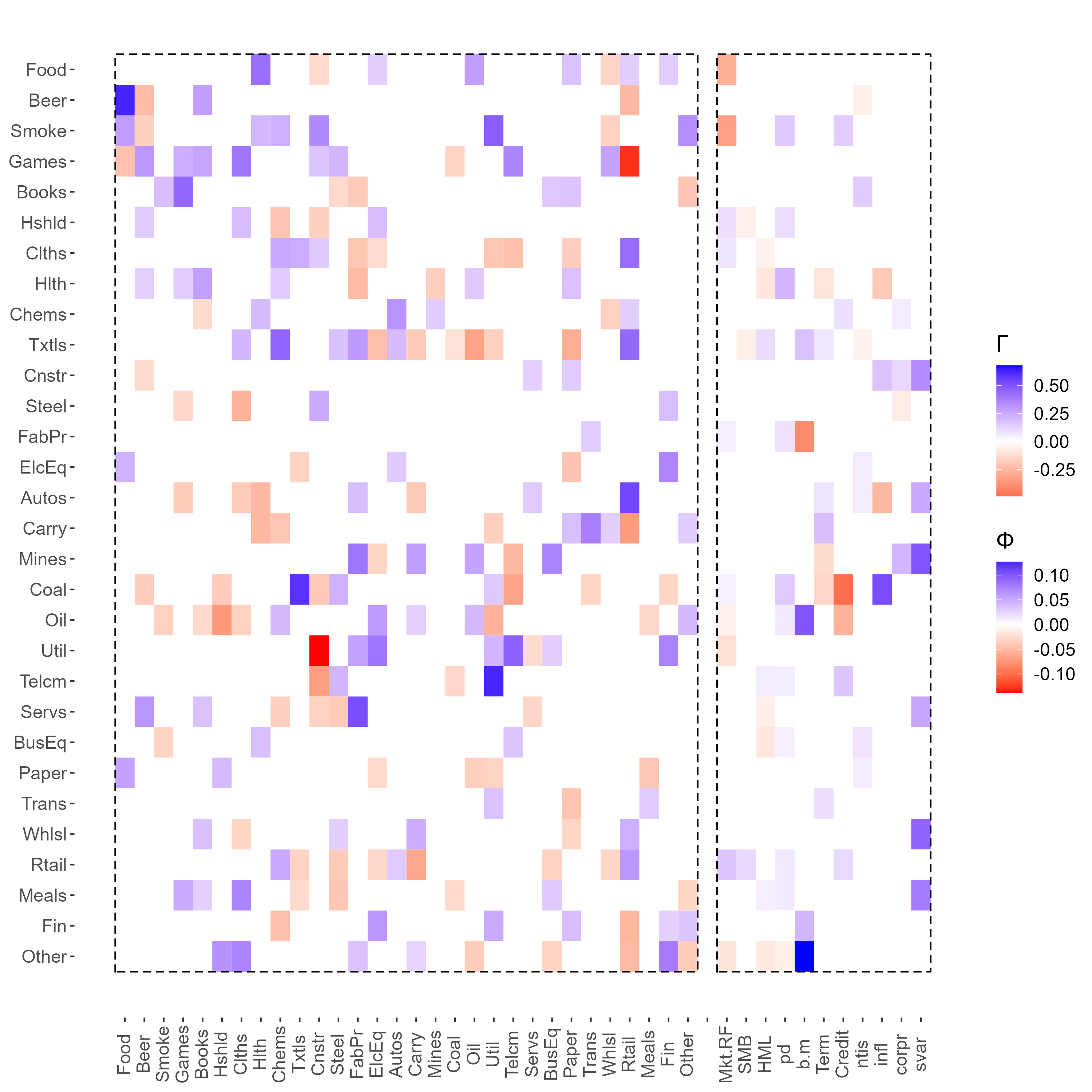}}\hspace{-2.5em}

\hspace{-2.5em}\subfigure[{\tt LMCMC} w/ HS]{\includegraphics[width=.25\textwidth]{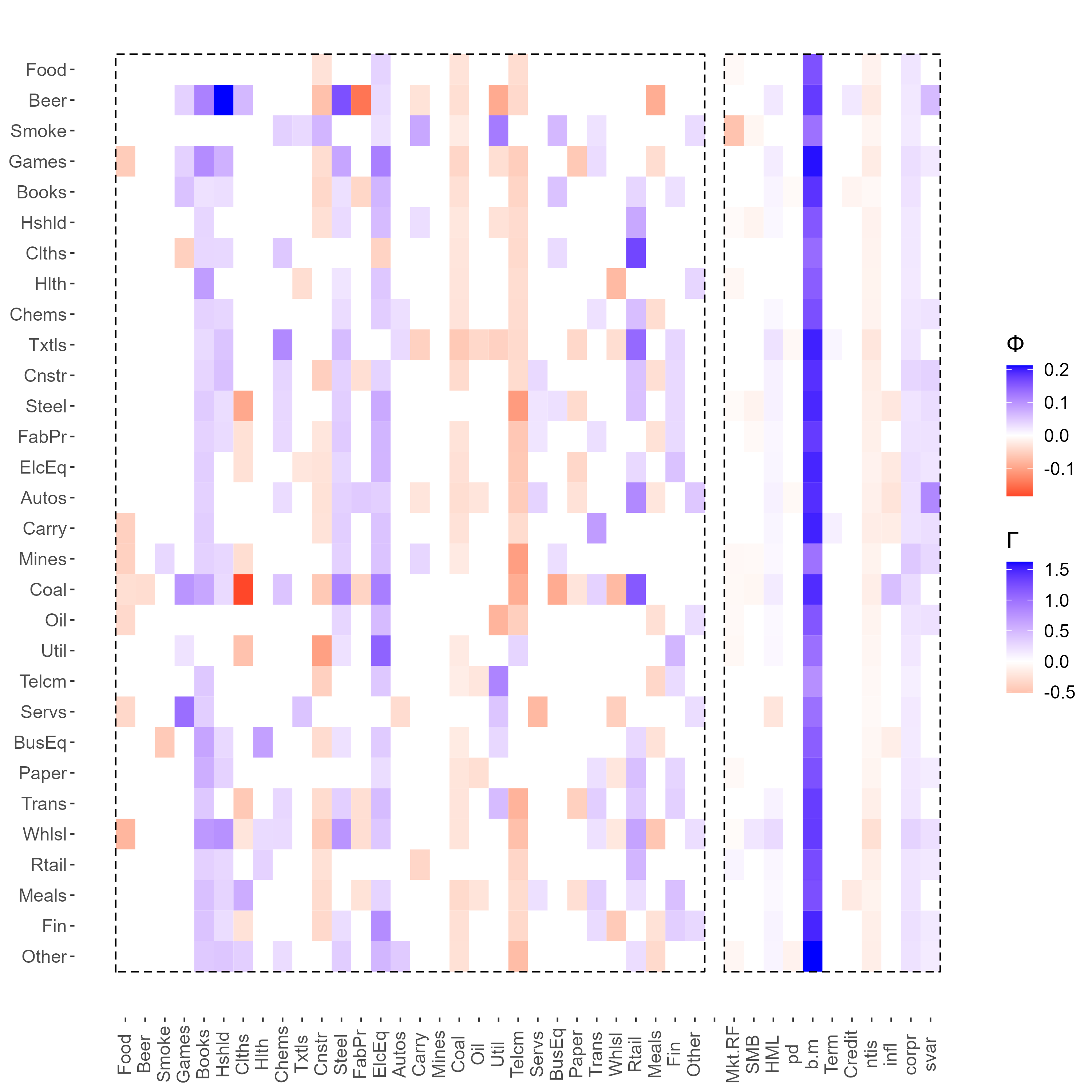}}\quad
\subfigure[{\tt LVB} w/ HS]{\includegraphics[width=.25\textwidth]{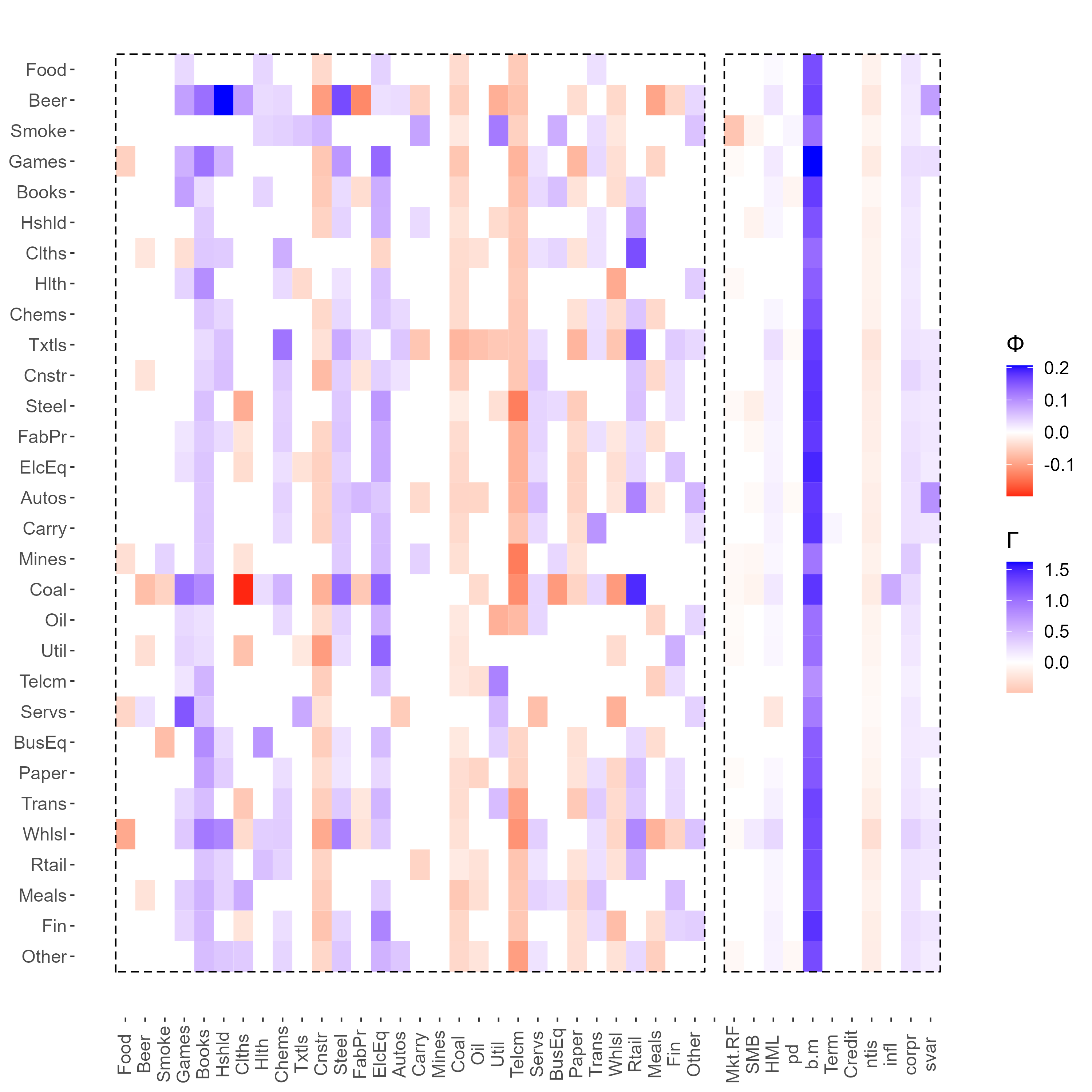}}
\subfigure[{\tt VB} w/ HS]{\includegraphics[width=.25\textwidth]{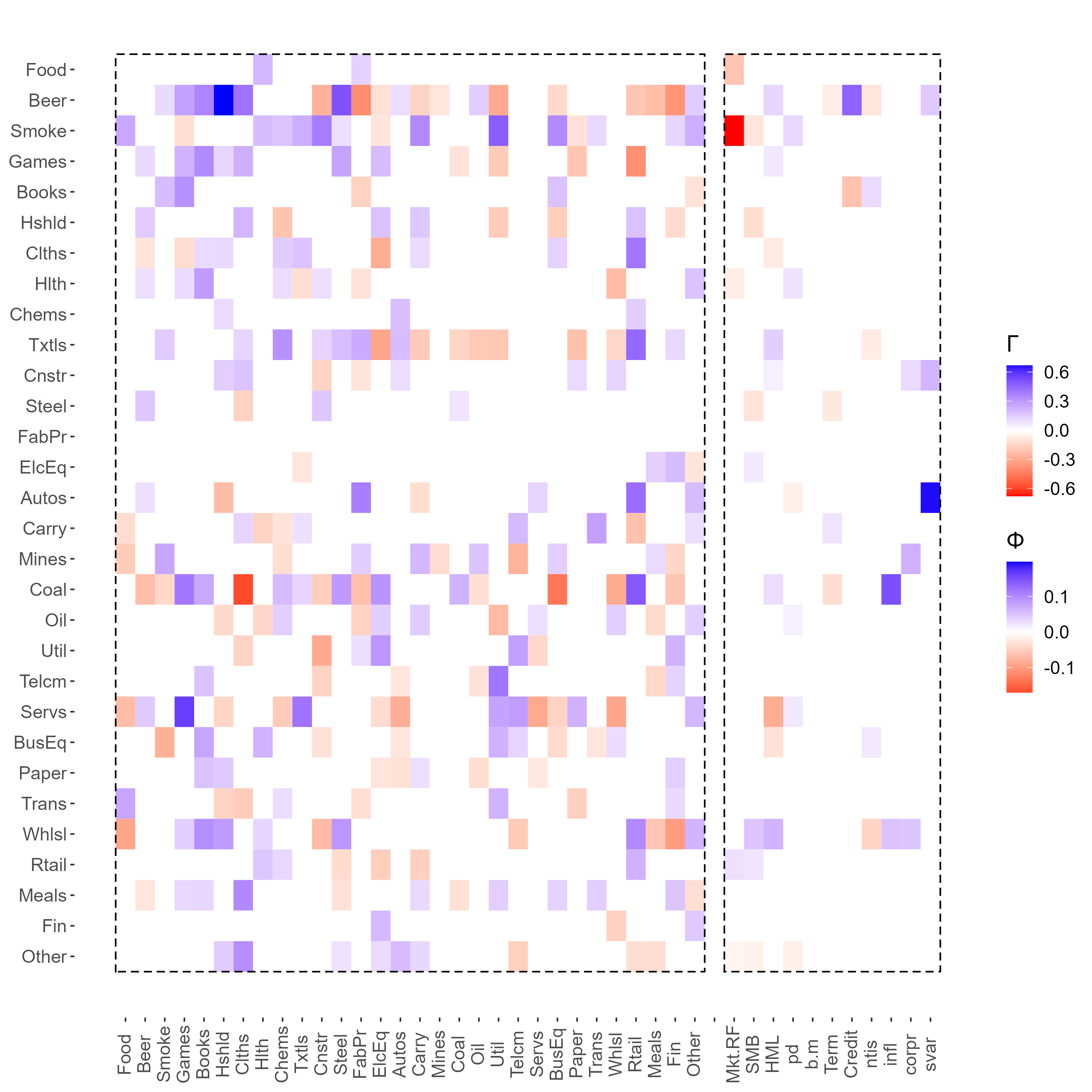}}\subfigure[{\tt VB} w/ HS + SV]{\includegraphics[width=.25\textwidth]{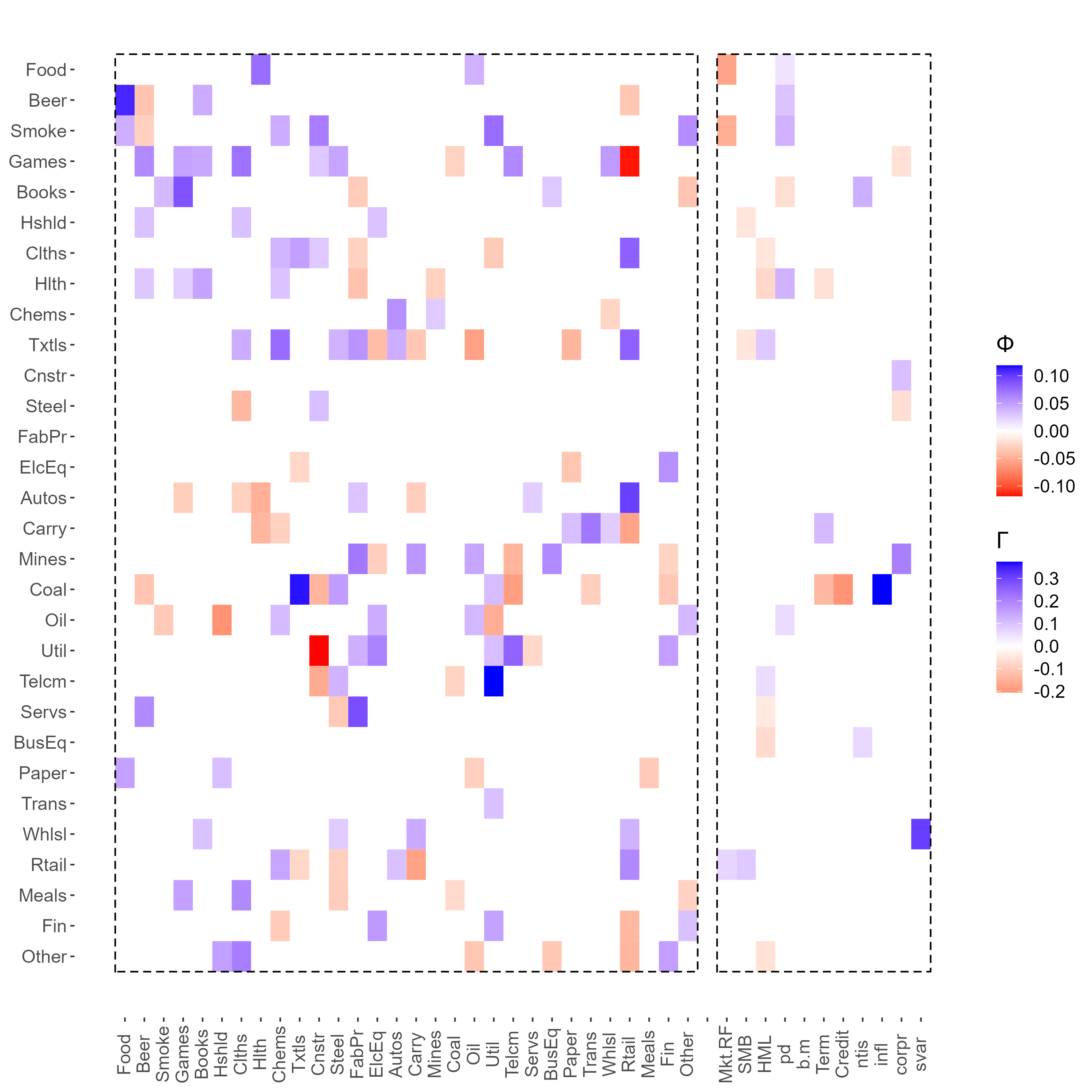}}\hspace{-2.5em}

\hspace{-2.5em}\subfigure[{\tt LMCMC} w/ NG]{\includegraphics[width=.25\textwidth]{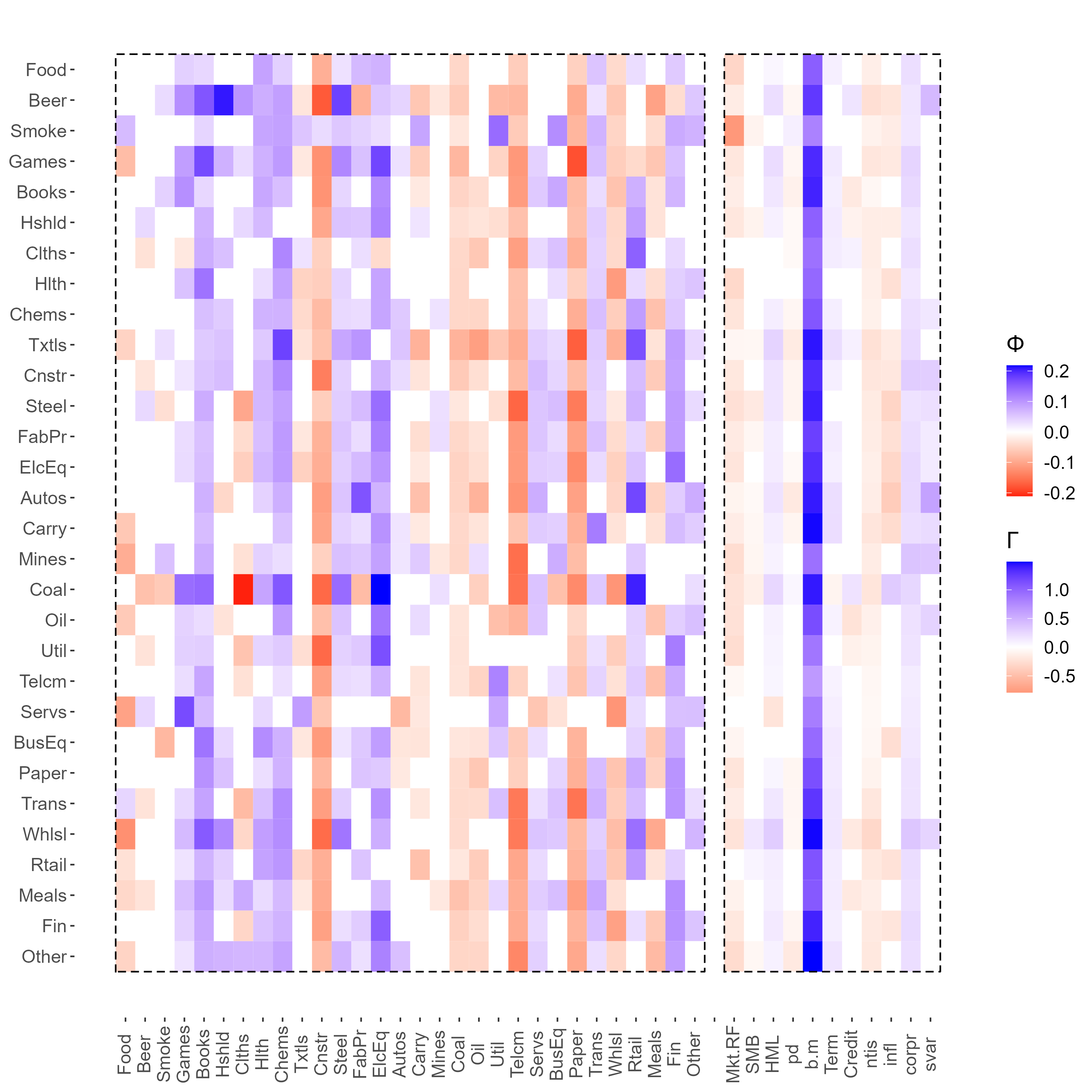}}\quad
\subfigure[{\tt LVB} w/ NG]{\includegraphics[width=.25\textwidth]{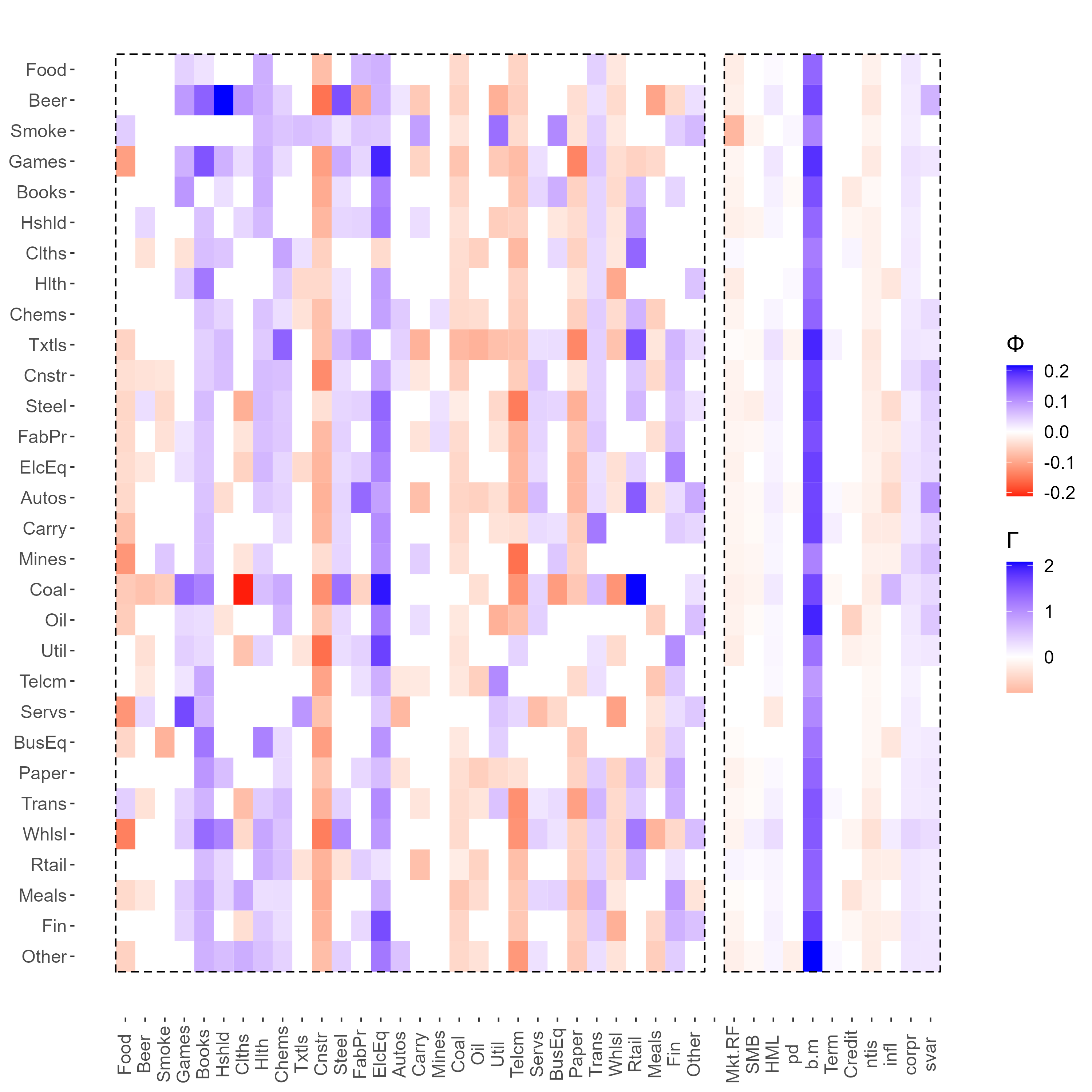}}
\subfigure[{\tt VB} w/ NG]{\includegraphics[width=.25\textwidth]{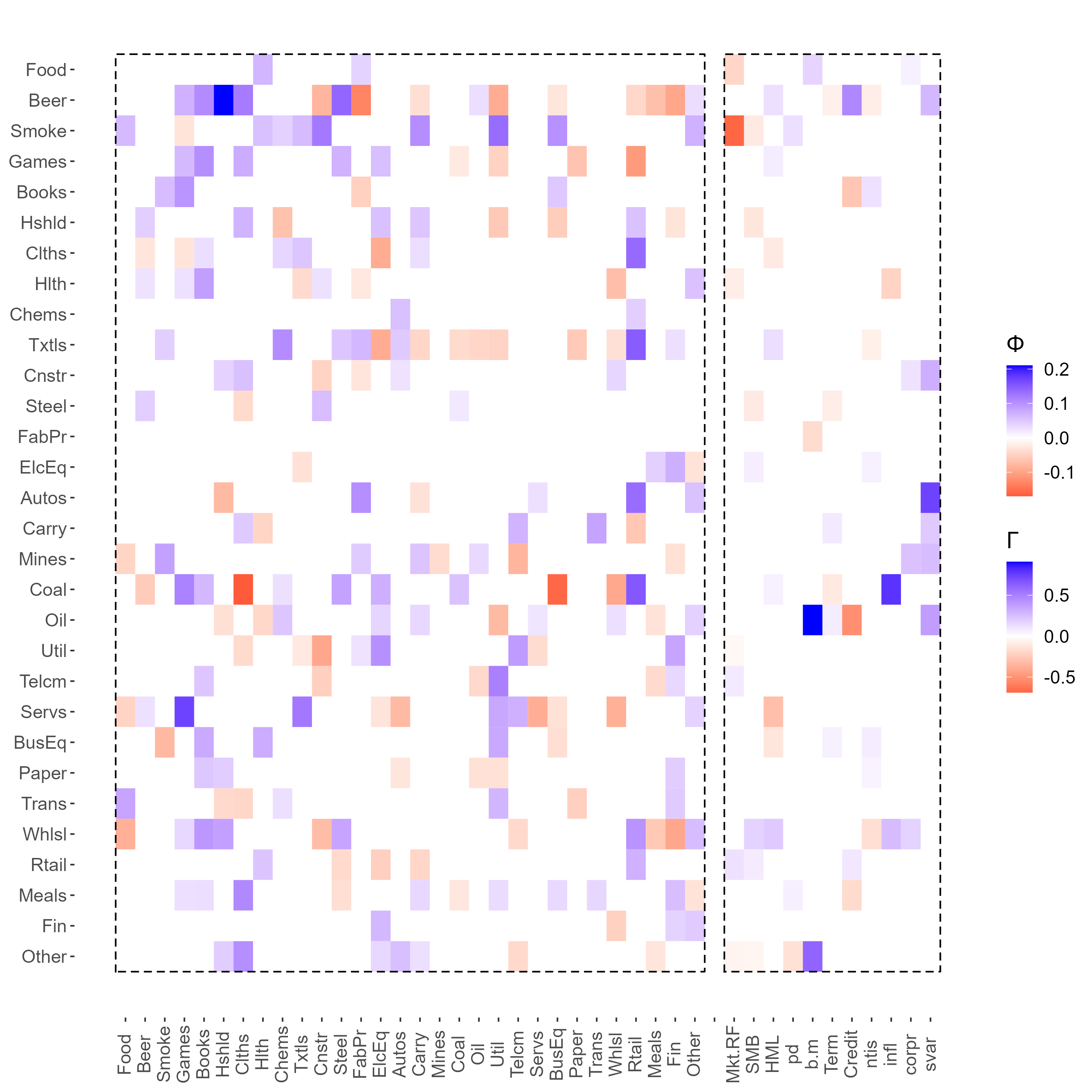}}\subfigure[{\tt VB} w/ NG + SV]{\includegraphics[width=.25\textwidth]{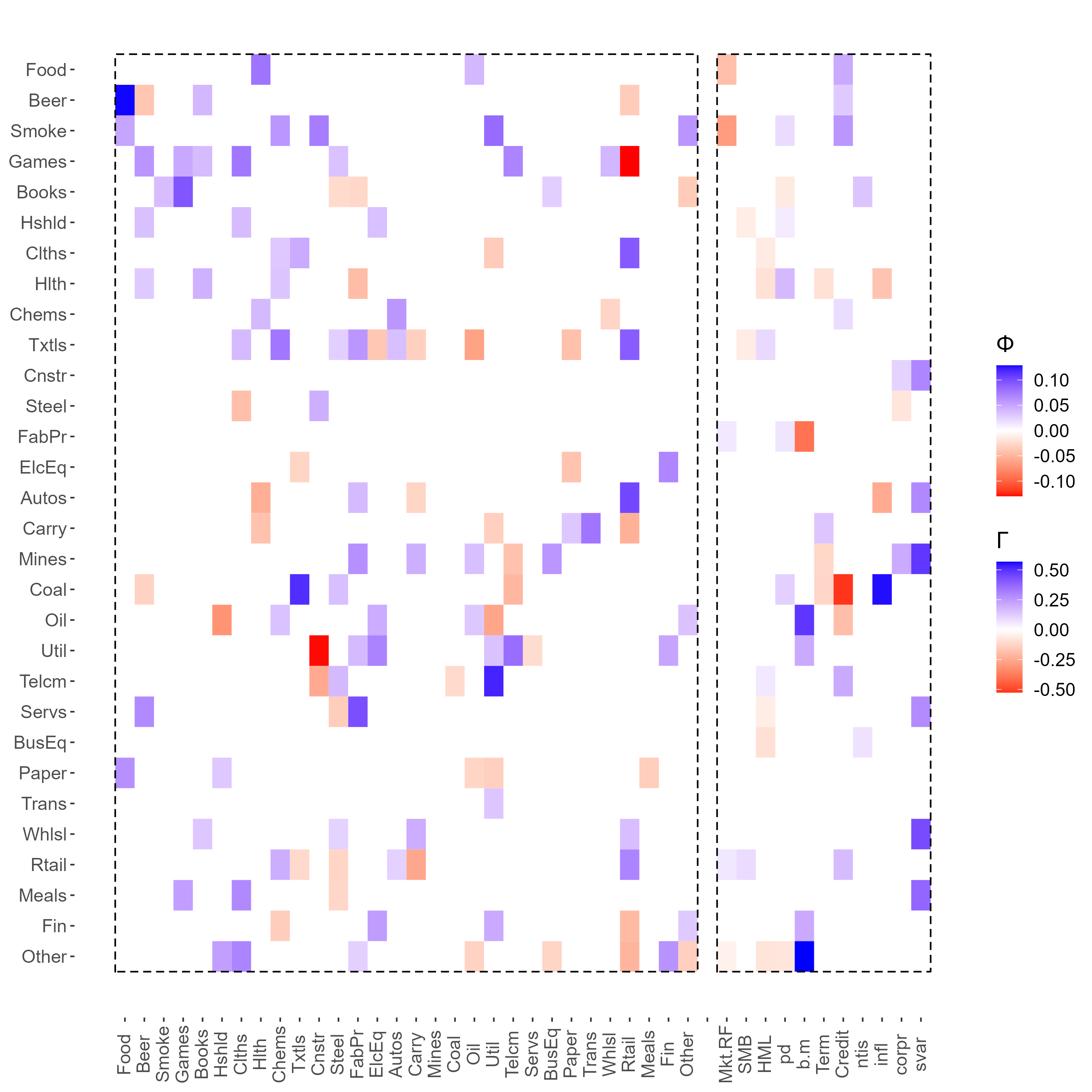}}\hspace{-2.5em}

	\caption{\small Variational Bayes estimates of the regression coefficients $\bTheta$ for different estimation methods. We report the estimates for the $d=30$ industry case obtained for all priors. We report the results for {\tt VB} with and without stochastic volatility.}
\label{fig:theta app}
\end{figure}
\end{document}